\documentclass[11pt]{article}

\usepackage{amsmath, amsthm}
\usepackage{amsmath, amsfonts}
\usepackage{amsmath, amssymb}
\usepackage{amsmath}
\usepackage{graphics}
\usepackage{graphicx}
\usepackage{color}
\usepackage{hyperref}
\usepackage[all]{xy}

\newtheorem{theorem}{Theorem}[subsection]
\newtheorem{definition}[theorem]{Definition}
\newtheorem{definition-lemma}[theorem]{Definition/Lemma}
\newtheorem{definition-explanation}[theorem]{Definition/Explanation}
\newtheorem{explanation-definition}[theorem]{Explanation/Definition}
\newtheorem{definition-fact}[theorem]{Definition/Fact}
\newtheorem{definition-notation}[theorem]{Definition/Notation}
\newtheorem{definition-conjecture}[theorem]{Definition/Conjecture}
\newtheorem{lemma}[theorem]{Lemma}
\newtheorem{lemma-definition}[theorem]{Lemma/Definition}
\newtheorem{proposition}[theorem]{Proposition}

\newtheorem{remark}[theorem]{\it Remark}
\newtheorem{remark-notation}[theorem]{\it Remark/Notation}

\newtheorem{application-lemma}[theorem]{Application/Lemma}

\newtheorem{ansatz}[theorem]{Ansatz}
\newtheorem{ansatz-definition}[theorem]{Ansatz/Definition}

\newtheorem{example}[theorem]{Example}
\newtheorem{example-definition}[theorem]{Example/Definition}

\newtheorem{definition-prototype}[theorem]{Definition-Prototype}

\numberwithin{equation}{subsection}

\newtheorem{stheorem}{Theorem}[section]

\newtheorem{sdefinition-lemma}[stheorem]{Definition/Lemma}
\newtheorem{sdefinition-explanation}[stheorem]{Definition/Explanation}
\newtheorem{sexplanation-definition}[stheorem]{Explanation/Definition}
\newtheorem{sdefinition-fact}[stheorem]{Definition/Fact}
\newtheorem{sdefinition-notation}[stheorem]{Definition/Notation}
\newtheorem{sdefinition-conjecture}[stheorem]{Definition/Conjecture}

\newtheorem{slemma-definition}[stheorem]{Lemma/Definition}

\newtheorem{sremark}[stheorem]{\it Remark}
\newtheorem{sremark-notation}[stheorem]{\it Remark/Notation}

\newtheorem{sapplication-lemma}[stheorem]{Application/Lemma}

\newtheorem{sansatz-definition}[stheorem]{Ansatz/Definition}

\newtheorem{sexample-definition}[stheorem]{Example/Definition}

\newtheorem{sdefinition-prototype}[stheorem]{Definition-Prototype}

\newtheorem{sstheorem}{Theorem}[subsubsection]
\newtheorem{ssdefinition}[sstheorem]{Definition}
\newtheorem{ssdefinition-lemma}[sstheorem]{Definition/Lemma}
\newtheorem{ssdefinition-explanation}[sstheorem]{Definition/Explanation}
\newtheorem{ssexplanation-definition}[sstheorem]{Explanation/Definition}
\newtheorem{ssdefinition-fact}[sstheorem]{Definition/Fact}
\newtheorem{ssdefinition-notation}[sstheorem]{Definition/Notation}
\newtheorem{ssdefinition-conjecture}[sstheorem]{Definition/Conjecture}
\newtheorem{sslemma}[sstheorem]{Lemma}
\newtheorem{sslemma-definition}[sstheorem]{Lemma/Definition}

\newtheorem{ssremark}[sstheorem]{\it Remark}
\newtheorem{ssremark-notation}[sstheorem]{\it Remark/Notation}

\newtheorem{ssapplication-lemma}[sstheorem]{Application/Lemma}

\newtheorem{ssansatz}[sstheorem]{Ansatz}
\newtheorem{ssansatz-definition}[theorem]{Ansatz/Definition}

\newtheorem{ssexample}[sstheorem]{Example}
\newtheorem{ssexample-definition}[sstheorem]{Example/Definition}

\newtheorem{ssdefinition-prototype}[sstheorem]{Definition-Prototype}

\newcommand{\Ad}{\mbox{\it Ad}}

\newcommand{\Aut}{\mbox{\it Aut}\,}

\newcommand{\BI}{\mbox{\it\scriptsize BI}\,}

\newcommand{\BT}{B\hspace{-.5ex}T\,}

\newcommand{\ComboMinor}{\mbox{\it ComboMinor}}
\newcommand{\Comm}{\mbox{\it Comm}\,}

\newcommand{\CSWZ}{\mbox{\scriptsize\it CS/WZ}\,}
 \newcommand{\tinyCSWZ}{\mbox{\tiny\it CS/WZ}\,}
\newcommand{\DBI}{\mbox{\it\scriptsize DBI}\,}
 \newcommand{\tinyDBI}{\mbox{\it\tiny DBI}\,}

\newcommand{\Der}{\mbox{\it Der}\,}

\newcommand{\Det}{\mbox{\it Det}\,}

\newcommand{\EM}{\mbox{\it\scriptsize EM}\,}
\newcommand{\End}{\mbox{\it End}\,}

\newcommand{\GL}{\mbox{\it GL}}

\newcommand{\Id}{\mbox{\it Id}\,}

\newcommand{\Imaginary}{\mbox{\it Im}\,}
\newcommand{\Image}{\mbox{\it Im}\,}

\newcommand{\Innsheaf}{\mbox{\it ${\cal I}\!$nn}\,}

\newcommand{\Log}{\mbox{\it Log}\,} 
\newcommand{\Map}{\mbox{\it Map}\,}

\newcommand{\Minor}{\mbox{\it Minor}}

\newcommand{\tinyNG}{\mbox{\it\tiny NG}\,}

\newcommand{\NL}{N\hspace{-.84ex}L}

\newcommand{\tinyPolyakov}{\mbox{\it\tiny Polyakov}\,}

\newcommand{\Ptn}{\mbox{\it Ptn}\,}

\newcommand{\Real}{\mbox{\it Re}\,}

\newcommand{\SAd}{\mbox{\it SAd}\,}

\newcommand{\Space}{\mbox{\it Space}\,}

\newcommand{\Spec}{\mbox{\it Spec}\,}

\newcommand{\STr}{\mbox{\it STr}\,}

\newcommand{\Supp}{\mbox{\it Supp}\,}
 \newcommand{\scriptsizeSupp}{\mbox{\scriptsize\it Supp}\,}
\newcommand{\Sym}{\mbox{\it Sym}}
 \newcommand{\scriptsizeSym}{\mbox{\scriptsize\it Sym}}
 \newcommand{\tinySym}{\mbox{\tiny\it Sym}}
\newcommand{\SymDet}{\mbox{\it SymDet}\,}

\newcommand{\Tr}{\mbox{\it Tr}\,}

\newcommand{\determinant}{\mbox{\it det}\,}
\newcommand{\dimm}{\mbox{\it dim}\,}

\newcommand{\pr}{\mbox{\it pr}}

\newcommand{\redscriptsize}{\mbox{\scriptsize\rm red}\,}

\newcommand{\transpose}{\top}

\newcommand{\boldA}{\mbox{\boldmath $A$}}
\newcommand{\boldM}{\mbox{\boldmath $M$}}
\newcommand{\boldd}{\mbox{\boldmath $d$}}
  \newcommand{\scriptsizeboldd}{\mbox{\scriptsize\boldmath $d$}}
  \newcommand{\tinyboldd}{\mbox{\tiny\boldmath $d$}}  

\newcommand{\boldt}{\mbox{\boldmath $t$}}
  \newcommand{\scriptsizeboldt}{\mbox{\scriptsize\boldmath $t$}}
  \newcommand{\tinyboldt}{\mbox{\tiny\boldmath $t$}}
\newcommand{\boldx}{\mbox{\boldmath $x$}}
  \newcommand{\scriptsizeboldx}{\mbox{\scriptsize\boldmath $x$}}
\newcommand{\boldy}{\mbox{\boldmath $y$}}
  \newcommand{\scriptsizeboldy}{\mbox{\scriptsize\boldmath $y$}}
\newcommand{\boldalpha}{\mbox{\boldmath $\alpha$}}
  \newcommand{\scriptsizeboldalpha}{\mbox{\scriptsize\boldmath $\alpha$}}
  \newcommand{\tinyboldalpha}{\mbox{\tiny\boldmath $\alpha$}}

\newcommand{\boldxi}{\mbox{\boldmath $\xi$}}

\newcommand{\longrightaarrow}{\longrightarrow\hspace{-3ex}\longrightarrow}

\newcommand{\odotwedge}{\stackrel{\mbox{\tiny $\odot$}}{\wedge}}
\newcommand{\partialboldalpha}{\partial^{\mbox{\boldmath \scriptsize $\alpha$}}}
\newcommand{\LARGEdot}{\mbox{\LARGE $\cdot$}}
  \newcommand{\Largedot}{\mbox{\Large $\cdot$}}
  
\newcommand{\partialdot}{\partial_{\mbox{\LARGE $\cdot$}}}

\newcommand{\tinybullet}{\raisebox{.2ex}{\tiny $\bullet$}}	

\begin{document}

\enlargethispage{24cm}

\begin{titlepage}

$ $

\vspace{-1.5cm} 

\noindent\hspace{-1cm}
\parbox{6cm}{\small April 2016}\
   \hspace{7cm}\
   \parbox[t]{6cm}{yymm.nnnnn [hep-th] \\
                D(13.1): Dirac-Born-Infeld  
				}

\vspace{2cm}

\centerline{\large\bf
   Dynamics of D-branes$\;$ I. The non-Abelian Dirac-Born-Infeld action,}
\vspace{1ex}
\centerline{\large\bf
  its first variation, and the equations of motion for D-branes}   
\vspace{1ex}   
\centerline{\large\bf 
 --- with remarks on the non-Abelian Chern-Simons/Wess-Zumino term}

\bigskip

\vspace{2.4em}

\centerline{\large
  Chien-Hao Liu   
            \hspace{1ex} and \hspace{1ex}
  Shing-Tung Yau
}

\vspace{2em}

\begin{quotation}
\centerline{\bf Abstract}

\vspace{0.3cm}

\baselineskip 12pt  
{\small
 In earlier works, 
  D(1) (arXiv:0709.1515 [math.AG]),
  D(11.1) (arXiv:1406.0929 [math.DG]), 
  D(11.2) (arXiv:1412.0771 [hep-th]), and 
  D(11.3.1) (arXiv:1508.02347 [math.DG]), 
 we have explained,  and shown by feature stringy examples, 
 why a D-brane in string theory, when treated as a fundamental dynamical object,
 can be described by 
 a map $\varphi$ from an Azumaya/matrix manifold $X^{\!A\!z}$ (cf.\ the D-brane world-volume) 
   with a fundamental module with a connection $(E,\nabla)$ (cf.\ the Chan-Paton bundle)
   to the target space-time $Y$.
 In this sequel,
 we construct a non-Abelian Dirac-Born-Infeld action functional
  $S_{\DBI}^{(\Phi, g, B)}(\varphi,\nabla)$ for such pairs $(\varphi,\nabla)$
  when the target space-time $Y$ is equipped with a background
   (dilaton, metric, $B$)-field $(\Phi, g,B)$ from closed strings.
 We next develop a technical tool needed to study variations of this action
   and apply it to derive the first variation $\delta S_{\DBI}^{(\Phi,g,B)}/\delta(\varphi,\nabla)$
   of $S_{\DBI}^{(\Phi,g,B)}$ with respect to $(\varphi,\nabla)$.
 The equations of motion that govern the dynamics of D-branes then follow.
 A complete action for a D-brane world-volume must include
  also the Chern-Simons/Wess-Zumino term $S_{\CSWZ}^{(C)}(\varphi,\nabla)$
  that governs how the D-brane world-volume couples with the Ramond-Ramond fields $C$ on $Y$.
 In this work, a version $S^{(C,B)}_{\CSWZ}(\varphi,\nabla)$
  of non-Abelian Chern-Simons/Wess-Zumino action functional for $(\varphi,\nabla)$
  that follows the same guide with which we construct $S^{(\Phi,g,B)}_{\DBI}(\varphi,\nabla)$
  is constructed for lower-dimensional D-branes
  (i.e.\ D(-1)-, D0-, D1-, D2-branes).
 Its first variation $\delta S^{(C,B)}_{\CSWZ}(\varphi,\nabla)/\delta(\varphi,\nabla)$
  is derived and
  its contribution to the equations of motion for $(\varphi, \nabla)$ follows.
 For D-branes of dimension $\ge 3$, an anomaly issue needs to be understood in the current context.  
 The current notes lay down a foundation toward the dynamics of D-branes along the line of this D-project.
 } 
\end{quotation}

\bigskip

\baselineskip 12pt
{\footnotesize
} 

 \bigskip

\noindent {\small MSC number 2010: 81T30, 35J20; 16S50, 14A22, 35R01.
} 

\bigskip

\baselineskip 10pt
{\scriptsize
\noindent{\bf Acknowledgements.}
We thank
 Andrew Strominger, Cumrun Vafa 
   for lectures/discussions that influence our understanding of strings, branes, and gravity.
C.-H.L.\ thanks in addition
 Joseph Polchinski 
   for his works that motivated the whole project and for consultation on a sign in his book; 
 Harald Dorn, Pei-Ming Ho, Robert Myers, Li-Sheng Tseng 
   for communications/discussions during the brewing years on Dirac-Born-Infeld action $S_{\tinyDBI}$;
 several authors (cf.\ Sec.\ 1) whose works influenced his understanding of $S_{\tinyDBI}$;  
 Tristan Collins
   for his topic course, spring 2016, discussions on analytical techniques beyond the current work, and literature guide;   
 Cumrun Vafa
   for consultation on other themes in string theory; 
 Yng-Ing Lee, Siu-Cheong Lau
   for discussions on the symplectic feature of D-branes in the past few years; 
 Daniel Cristofaro-Gardiner, Andrew Strominger 
   for other topic courses, spring 2016;
 Yu-Wei Fan, Michael Rios, Peter Smillie, Hiro Lee Tanaka, Chenglong Yu 
  for discussions on other related themes/potential applications;  
 musicians and the production team of `Les Mis\'{e}rables' 10th anniversary concert
   for the music that accompanies the typing of the notes; 
 Ann Willman for hospitality, late summer 2015;   
 Ling-Miao Chou
   for discussions on electrodynamics, comments that improve the illustrations, and moral support.
The project is supported by NSF grants DMS-9803347 and DMS-0074329.
} 

\end{titlepage}

\newpage

\enlargethispage{24cm}
\begin{titlepage}

%

\centerline{\small\it
 Chien-Hao Liu dedicates the current notes to {\it Ling-Miao Chou}}
\centerline{\small\it for her tremendous love that makes this work/project possible,}
\centerline{\small\it
 and to his parents-in-law Mr.\ \& Mrs.\ Shih-Chuan Chow (1919--2011) and Min-Chih Liu,}
\centerline{\small\it
 who handed over him their most precious gem$^{\ast}$.}

\vspace{3em}

\baselineskip 11pt

\hspace{-1.5cm}
{\footnotesize
 \noindent
 $^{\ast}$(From C.H.L.)\hspace{1em} {\sc On the Road}\\
 
 \hspace{-1.5cm}
 \parbox{17.6cm}
{\scriptsize
 \noindent
 When I first met my then father-in-law-to-be, he had already retired for a decade.
 Like many Chinese around his generation,
  he grew up in a chaotic China and
       went to schools between civil wars and the invasion of Japan.
 After his graduation from Wuhan University
    through a government aid, majoring in electrical engineering,
  he served three years at the power plant at Yibin, Sichuan.
 After World War II,  
 he got scholarships first from the then Central Government of China and later from the United Aid to China Fund, London,
  to work study in Scotland, England.
 That was a time when a trip between Sichuan, China, and Southampton, England, took about a month, including a stop at India,
  mixing airplanes, railroads, and ships.
 Seven years afterwards, when he was about to go home,
  he faced one of the most difficult decisions he had to make: Mainland China or Taiwan, 
  since China had been divided into two political entities after the four years civil war following the surrender of Japan.  
 Whatever reason behind his decision,
  he chose the latter and
  became a member of the team that constructed the modern power plants and power system at Taiwan.
 Only decades later in late 1980s 
    when the Taiwan government and the Mainland China government 
      started to build up a reconciling atmosphere,
   he got a chance to go back to see his siblings again, though his parents had long passed away.
 That is a tragedy many in his generation underwent.
 Also like many in his generation who had seen and experienced in person enough sufferings 
   but  had the luck to get to finish high-level education,
  he took his service and devotion to the common good of his country as a responsibility without leaving a name behind.
  This is a very brief story of Mr.\ Chow,
   who started his own family quite late and, for that reason, 
     held his three children, with Ling-Miao the youngest and only daughter, very dearly.
 Like Prof.\ Raoul Bott,    
  he maintained a curious mind even to his senior years 
    and did a few calculus exercises(!!!) daily to keep his mind active and alert.
   
 \hspace{1em}
 Fast forward to October 2015, completely unforeseeable while in the early summer that year,
 I was once again on the road, repeating the same route from Austin, Texas, to Boston, Massachusetts,
  I'd followed more than a decade ago with Ling-Miao.
 This is the fourth time  I made such a trip on the road across the United States
   and, as in the case of my father-in-law,
   each such trip marks a huge transition in my study and life.
 The first time was in early 1990s from Princeton, New Jersey, to Berkeley, California.
 I felt like heading toward a new world.
  That was the summer a year after I had met Ling-Miao for the first time unexpectedly in a trip to Ohio
     without knowing that she would be influencing my life forever.
 The second time was in mid-1990s from Coral Gable, Florida, to Austin, Texas.
 That marked the official end of my student years.
 The third time was at the turn of the millennium from Austin, Texas,
  to Cambridge, Massachusetts.
 That is the only road-trip I ever made with Ling-Miao.
 The fall before that trip I was once again in the job market.
  I got a surprise contact near the Thanksgiving holiday from Prof.\ Brian Greene at Columbia
    concerning a position in his group based on an earlier work I had done
 (`{\it On the isolated singularity of a $7$-space obtained by rolling Calabi-Yau threefolds
               through extremal transitions}', arXiv:hep-th/9801175v1;
			  revised in v2 with Volker Braun in Candelas' group then at U.T.\ Austin)
			and the recommendation from Prof.'s Orlando Alvarez, Jacques Distler, and Daniel Freed.
 One or two weeks later disappointing follow-up news came from Prof.\ Greene:
   That position he had kindly intended to offer requires a U.S.\ Citizenship, which I hadn't acquired at that time.
 Ordinarily, an apology as a formality is more than enough to end this contact
   since it is not his fault at all that it didn't work out.
 Yet, this is what Prof.$\,$Greene showed his nobleness and generosity:
 Rather than just ending with that, he recommended me further to Prof.$\,$Yau at Harvard.
 Thus, thanks to this unexpected twist of events and Prof.$\,$Yau's acceptance,
  I came to Harvard in pure luck as an even bigger surprise.
 
 \hspace{1em}
 The two leading institutes at the Boston area together provide a unique soil for a curious and absorbing mind:
  the frontier research and related basic and/or topic courses
    in algebraic geometry, differential geometry, and symplectic geometry on the mathematics side
   and in quantum field theory and string theory --- particularly its various geometry-related aspects ---
          on the physics side.
 The soil is further enriched through the catalyzing effect of the mathematics-physics intertwining atmosphere
    in Yau's group and regular group meetings,
  though I had to admit that it's not very easy to catch up and keep my head above water
     in such an intense environment at the beginning
   and for a while I almost got drowned.
 It is only after the birth of this project at the end of 2006
  that all this unusual, purely accidental luck given to me acquired its meaning.
 In retrospect, a project like this is very unlikely
  if not in such an intense and encompassing geometry-physics soil like the Boston area.
 
 \hspace{1em}
 This work adds another special mark to the timeline of this project.
 Though clearly not in its final form (cf.\ Remark 3.2.4),
  for the first time since the beginning of this D-project
  the dynamics of D-branes is addressed along the line of the project in a most natural and geometric way.
 This brings the study of D-branes to the same starting point as that for the fundamental string:
	namely,

 \bigskip
 
 \centerline{
 {\it
 \begin
 {tabular}{|l||l|}\hline
   \hspace{4.2em}{\bf string theory}\rule{0ex}{1.2em}
       & \hspace{10em}{\bf D-brane theory}\\[.6ex] \hline\hline
     $\begin{array}{l}
	    \mbox{string world-sheet}: \\
 		 \hspace{3em}
		 \mbox{$2$-manifold $\,\Sigma$}
	   \end{array}$
       &  $\begin{array}{l}
	           \mbox{D-brane world-volume}:\rule{0ex}{1.2em}\\
			      \hspace{3em}
	              \mbox{Azumaya/matrix manifold}\\ \hspace{4em} \mbox{with a fundamental module with a connection}\\
				  \hspace{4em}(X^{\!A\!z},E,\nabla)
			   \end{array}$	  				   \\[3.8ex] \hline
	 $\begin{array}{l}
	    \mbox{string moving in space-time $Y$}:\rule{0ex}{1.2em} \\
 		 \hspace{3em}
		 \mbox{differentiable map $f:\Sigma \rightarrow Y$}
	   \end{array}$
       &  $\begin{array}{l}
	           \mbox{D-brane moving in space-time $Y$}:\rule{0ex}{1.2em}\\
			      \hspace{3em}
	              \mbox{differentiable map $\varphi:(X^{\!A\!z},E,\nabla)\rightarrow Y$}
			   \end{array}$	  				   \\[2.4ex] \hline		
       $\begin{array}{c}\mbox{Nambu-Goto action $\:S_{\mbox{\tiny Nambu-Goto}}\:$
	                                          for $f$'s}\end{array}$\rule{0ex}{1.2em}
       &  $\begin{array}{c}\rule{0ex}{1.2em}
	           \mbox{Dirac-Born-Infeld action $\:S_{\mbox{\tiny Dirac-Born-Infeld}}\:$
			                                  for $(\varphi,\nabla)$'s}\end{array}$
                                                                     	   \\[1ex] \hline
 \end{tabular}
  }
  }

 \bigskip

 \noindent
 While there are a long list of people I need to say thanks to
  (in particular, Shiraz Minwalla and Mihnea Popa, cf.\ D(6) Dedication),
 I would not be able to survive a decade's brewing to reach D(1), 2007, and then
   another seven years' brewing to reach D(11.1), 2014, without Ling-Miao.
 In a broad sense, this piece --- indeed the whole D-project --- is also her creation through her love.
 I thus dedicate this special mark of the project to her.
 Many challenging mathematical and physical issues of D-branes along the line remain ahead,
  some of them look beyond reach at the moment of writing,
  and I am still on the road for this journey that fills with unknowns.
 
 } } 

\end{titlepage}


\newpage
$ $

\vspace{-3em}

\centerline{\sc
 Dynamics of D-branes, I: The Dirac-Born-Infeld Action
 } %

\vspace{2em}


\begin{flushleft}
{\Large\bf 0. Introduction and outline}
\end{flushleft}
In earlier works,
  [L-Y1] (arXiv:0709.1515 [math.AG], D(1)),
  [L-Y4] (arXiv:1406.0929 [math.DG], D(11.1)),
  [L-Y5] (arXiv:1412.0771 [hep-th], D(11.2)), and
  [L-Y6] (arXiv:1508.02347 [math.DG], D(11.3.1)),
 we have explained,  and shown by feature stringy examples,
 why a D-brane in string theory, when treated as a fundamental dynamical object,
 can be described by
 a map $\varphi$ from an Azumaya/matrix manifold $X^{\!A\!z}$, served as the D-brane world-volume, 
   with a fundamental module with a connection $(E,\nabla)$, served as the Chan-Paton bundle,
   to the target space-time $Y$.
 In this sequel,
 we construct a non-Abelian Dirac-Born-Infeld action functional
  $S_{\DBI}^{(\Phi, g, B)}(\varphi,\nabla)$ for such pairs $(\varphi,\nabla)$
  when the target space-time $Y$ is equipped with a background
   (dilaton, metric, $B$)-field $(\Phi, g,B)$ from closed strings; (cf.\ Sec.\ 2 \& Sec.\ 3).
We next develop a technical tool needed to study variations of this action (cf.\ Sec.\ 4)
   and apply it to derive the first variation $\delta S_{\DBI}^{(\Phi,g,B)}/\delta(\varphi,\nabla)$
   of $S_{\DBI}^{(\Phi,g,B)}$ with respect to $(\varphi,\nabla)$ (cf.\ Sec.\ 5) .
 The equations of motion that govern the dynamics of D-branes then follow.

 A complete action for a D-brane world-volume must include
  also the Chern-Simons/Wess-Zumino term $S_{\CSWZ}^{(C)}(\varphi,\nabla)$
  that governs how the D-brane world-volume couples with the Ramond-Ramond fields $C$ on $Y$.
 In the current notes, a version $S^{(C,B)}_{\CSWZ}(\varphi,\nabla)$
  of non-Abelian Chern-Simons/Wess-Zumino action functional for $(\varphi,\nabla)$
  that follows the same guide with which we construct $S^{(\Phi,g,B)}_{\DBI}(\varphi,\nabla)$
  is constructed for lower-dimensional D-branes
  (i.e.\ D(-1)-, D0-, D1-, D2-branes).
 Its first variation $\delta S^{(C,B)}_{\CSWZ}(\varphi,\nabla)/\delta(\varphi,\nabla)$
  is derived and
  its contribution to the equations of motion for $(\varphi, \nabla)$ follows; (cf.\ Sec.\ 6).
 For D-branes of dimension $\ge 3$, an anomaly issue needs to be understood in the current context.  
The current notes lay down a foundation toward the dynamics of D-branes along the line of this D-project.

 Some highlights of the history of
   how the Born-Infeld action and the Dirac-Born-Infeld action
     (cf.\ time-ordered:
	           [Mie] (1912) of Gustav Mie,
			   [Bo] (1934) of Max Born,
			   [B-I] (1934) of Born and Leopold Infeld,
			   [Di] (1962) of Paul Dirac)
   arise from open string theory
 and a list of issues one needs to resolve to convert such an action to that for coincident D-branes
 are given in Sec.~1.
They serve as a guide for our discussion.

\bigskip
\bigskip

\noindent
{\bf Convention.}
 References for standard notations, terminology, operations and facts are\\
  (1) string theory: [B-B-S], [G-S-W], [Po3];  \hspace{2em}
  (2) D-branes: [Joh],  [Po3];  \\
  (3) algebraic geometry: [Ha];  \hspace{7.8em}
  (4) $C^{\infty}$-algebraic geometry: [Joy].
 \begin{itemize}
  \item[$\cdot$]
   For clarity, the {\it real line} as a real $1$-dimensional manifold is denoted by ${\Bbb R}^1$,
    while the {\it field of real numbers} is denoted by ${\Bbb R}$.
   Similarly, the {\it complex line} as a complex $1$-dimensional manifold is denoted by ${\Bbb C}^1$,
    while the {\it field of complex numbers} is denoted by ${\Bbb C}$.
	
  \item[$\cdot$]	
  The inclusion `${\Bbb R}\subset{\Bbb C}$' is referred to the {\it field extension
   of ${\Bbb R}$ to ${\Bbb C}$} by adding $\sqrt{-1}$, unless otherwise noted.

 \item[$\cdot$]	
  The {\it real $n$-dimensional vector spaces} ${\Bbb R}^{\oplus n}$
      vs.\ the {\it real $n$-manifold} $\,{\Bbb R}^n$; \\
  similarly, the {\it complex $r$-dimensional vector space ${\Bbb C}^{\oplus r}$}
     vs.\ the {\it complex $r$-fold} $\,{\Bbb C}^r$.
   
 \item[$\cdot$]
  All manifolds are paracompact, Hausdorff, and admitting a (locally finite) partition of unity.
  We adopt the {\it index convention for tensors} from differential geometry.
   In particular, the tuple coordinate functions on an $n$-manifold is denoted by, for example,
   $(y^1,\,\cdots\,y^n)$.
  However, no up-low index summation convention is used.
   
    
  \item[$\cdot$]
  For the current notes, `{\it differentiable}', `{\it smooth}', and $C^{\infty}$ are taken as synonyms.
  
  \item[$\cdot$]
   For a smooth manifold $X$, $C^{\infty}(X):=$ the ring of smooth functions on $X$.\\
   For a vector bundle $E$ over $X$, $C^{\infty}(E):=$ the $C^{\infty}(X)$-module of smooth sections of $E$.

 \item[$\cdot$]
  {\it wedge product} convention:
   For $\alpha\in C^{\infty}(\bigwedge^pX)$, $\beta \in C^{\infty}(\bigwedge^qM)$,
	 \begin{eqnarray*}
	  \lefteqn{(\alpha\wedge\beta)(v_1,\,\cdots\,, v_{p+q})}\\[1.2ex]
	  && :=\;
	    \sum_{\mbox{\tiny
		                 $\begin{array}{c}
		                    \sigma\in\tinySym_{p+q}; \\
		                    \sigma(1)<\sigma(2)<\,\cdots\,<\sigma(p), \\
						    \sigma(p+1)<\,\sigma(p+2)<\,\cdots\,<\sigma(p+q)
						  \end{array}$}
						 }
			(-1)^{\sigma}\,
			       \alpha(v_{\sigma(1)},\,\cdots\,, v_{\sigma(p)})\,
			       \beta(v_{\sigma(p+1),\,\cdots\,, v_{\sigma(p+q)}})\,.
     \end{eqnarray*}				
   For example,
    $dx^1\wedge dx^2 \wedge \,\cdots\,\wedge dx^l
	  := \sum_{\sigma\in\scriptsizeSym_l}
	         (-1)^{\sigma}
			  dx^{\sigma(1)}\otimes dx^{\sigma(2)}\otimes\,\cdots\,
			   \otimes dx^{\sigma(l)}$.
   
  \item[$\cdot$]
   {\it a}{\sl lgebra} $A_{\varphi}$, sheaf  of algebras ${\cal A}_{\varphi}\,$
     vs.\ {\it connection $1$-form} $A_{\mu}$.
 
  \item[$\cdot$]
   {\it d}{\sl egree} $d$ vs.\ exterior {\it d}{\sl ifferential} $d$;\hspace{1ex}
   {\it d}{\sl iagonal} matrix $D$ vs.\ covariant {\it d}{\sl erivative} $D$.
  
  \item[$\cdot$]
   {\it m}{\sl atrix} $m$ vs.\ manifold of dimension $m$.
   
  \item[$\cdot$]
   the Regge slope $\alpha^{\prime}$ vs.\ dummy labelling index $\alpha$.

  \item[$\cdot$]
   {\it s}{\sl ection} $s$ of a fiber bundle vs.\ dummy labelling index $s$ vs.\ coordinate $s$.

  \item[$\cdot$]
   {\it Chan-Paton bundle} $E$  vs.\
    the combined $2$-tensor $g+B=:\sum_{i,j}E_{ij}dy^i\otimes dy^j$
	  from the metric tensor and the $B$-field.

  \item[$\cdot$]
   {\it r}{\sl ing} $R$
    vs.\ $k$-th {\it r}{\sl emainder} $R[k]$
    vs.\ {\it R}{\sl iemann} curvature tensor $R_{ijkl}$
    vs.\ index $(\,\cdot\,)^R$ for {\it r}{\sl ight} factor or component.
   
  \item[$\cdot$]
   $\Spec R $ ($:=\{\mbox{prime ideals of $R$}\}$)
         of a commutative Noetherian ring $R$  in algebraic geometry\\
   vs.\ $\Spec R$ of a $C^k$-ring $R$
  ($:=\Spec^{\Bbb R}R :=\{\mbox{$C^k$-ring homomorphisms $R\rightarrow {\Bbb R}$}\}$).

  \item[$\cdot$]
  {\it morphism} between schemes in algebraic geometry
    vs.\ {\it $C^k$-map} between $C^k$-manifolds or $C^k$-schemes
         	in differential topology and geometry or $C^k$-algebraic geometry.
   
  \item[$\cdot$]
   The `{\it support}' $\Supp({\cal F})$
    of a quasi-coherent sheaf ${\cal F}$ on a scheme $Y$ in algebraic geometry
     	or on a $C^k$-scheme in $C^k$-algebraic geometry
    means the {\it scheme-theoretical support} of ${\cal F}$
   unless otherwise noted;
    ${\cal I}_Z$ denotes the {\it ideal sheaf} of
    a (resp.\ $C^k$-)subscheme of $Z$ of a (resp.\ $C^k$-)scheme $Y$;
    $l({\cal F})$ denotes the {\it length} of a coherent sheaf ${\cal F}$ of dimension $0$.

  \item[$\cdot$]
   {\it coordinate-function index}, e.g.\ $(y^1,\,\cdots\,,\, y^n)$ for a real manifold
      vs.\  the {\it exponent of a power},
	  e.g.\  $a_0y^r+a_1y^{r-1}+\,\cdots\,+a_{r-1}y+a_r\in {\Bbb R}[y]$.
	     	
 

   
   
  \item[$\cdot$]
   The current Notes D(13.1) continues the study in
	  \begin{itemize}
	   \item[]  \hspace{-2em} [L-Y4]\hspace{1em}\parbox[t]{34em}{{\it
    	   D-branes and Azumaya/matrix noncommutative differential geometry,
        I: D-branes as fundamental objects in string theory  and differentiable maps
         from Azumaya/matrix manifolds with a fundamental module to real manifolds},
         arXiv:1406.0929 [math.DG]. (D(11.1))
		 }\\ 
		
	   \item[]  \hspace{-2em} [L-Y6]\hspace{1em}\parbox[t]{34em}{{\it
	      Further studies on the notion of differentiable maps from Azumaya/matrix manifolds,
		  I. The smooth case},
		 arXiv:1508.02347 [math.DG]. (D(11.3.1))	
		 }
	  \end{itemize}  	
   Notations and conventions follow these earlier works when applicable.
 \end{itemize}

\bigskip
\bigskip
   
\begin{flushleft}
{\bf Outline}
\end{flushleft}
\nopagebreak
{\small
\baselineskip 12pt  
\begin{itemize}
 \item[0.]
  Introduction.
   
 \item[1.]
 Coincident D-branes and issues on non-Abelian Dirac-Born-Infeld action
  \vspace{-.6ex}
  \begin{itemize}
   \item[\LARGE $\cdot$]
    Open-strings, background gauge fields, and the Born-Infeld action: The pre-D-brane era

   \item[\LARGE $\cdot$]
    D-branes, Dirac-Born-Infeld action, and its non-Abelian generalization
  
   \item[\LARGE $\cdot$]
    The construction of a non-Abelian Dirac-Born-Infeld action and its consequences
  \end{itemize}
  
 \item[2.]
 Differentiable maps from an Azumaya/matrix manifold with a fundamental module with\\ a connection
   \vspace{-.6ex}
   \begin{itemize}
    \item[2.1]
	  Differentiable maps from an Azumaya/matrix manifold with a fundamental module

    \item[2.2]
      Compatibility between the map $\varphi$ and the connection $\nabla$ from the open-string aspect	
	
	\item[2.3]
      Self-adjoint/Hermitian maps from an Azumaya/matrix manifold with a Hermitian\\ fundamental module
   \end{itemize}
   
 \item[3.]
 The Dirac-Born-Infeld action for differentiable maps from Azumaya/matrix manifolds
  \vspace{-.6ex}
  \begin{itemize}
	\item[3.1]
	 The resolution of issues toward defining the Dirac-Born-Infeld action
	  \begin{itemize}
	   \item[3.1.1]	
	   The pull-back of tensors from the target space via commutative surrogates
	
	   \item[3.1.2]
	    D-brane world-volume with constant induced-metric signature
	
	   \item[3.1.3]
	    From determinant {\it Det} to symmetrized determinant {\it SymDet}
		
	   \item[3.1.4]
        Square roots of sections of
		  $(\bigwedge^mT^{\ast}X)^{\otimes 2}\otimes_{\Bbb R}\End_{\Bbb C}(E)$
		
       \item[3.1.5]
	    The factor from the dilaton field $\Phi$ on the target space(-time)
	
	   \item[3.1.6]
        Reality of the trace
	  \end{itemize}
 	  \vspace{0ex}
	
    \item[3.2]	
     The Dirac-Born-Infeld action for admissible pairs $(\varphi, \nabla)$	
 %
  \end{itemize}
  
 \item[4.]
  Variations of $\varphi^{\sharp}$ in terms of variations of local generators
  \vspace{-.6ex}
  \begin{itemize}
	\item[4.1]
     $\varphi^{\sharp}_{\mbox{\boldmath\scriptsize $t$}}(f)$
	    in terms of
		$(\varphi^{\sharp}_{\mbox{\boldmath\scriptsize $t$}}(y^1),\,\cdots\,,\,
		      \varphi^{\sharp}_{\mbox{\boldmath\scriptsize $t$}}(y^1))$
        via Generalized Division Lemma

    \item[4.2]
     $\frac{\partial^{|\mbox{\boldmath\tiny $\alpha$}|}}
		           {\partial\mbox{\boldmath\scriptsize $t$}^{\mbox{\boldmath\tiny $\alpha$}}}
		              (\varphi^{\sharp}_{\mbox{\boldmath\scriptsize $t$}}(f) )$
		   in terms of
          $\left( \frac{\partial^{|\mbox{\boldmath\tiny $\alpha$}_1|}}
		           {\partial\mbox{\boldmath\scriptsize $t$}^{\mbox{\boldmath\tiny $\alpha$}_1}}
		              (\varphi^{\sharp}_{\mbox{\boldmath\scriptsize $t$}}(y^1) )\,,\;
					  \cdots\,,\;
         \frac{\partial^{|\mbox{\boldmath\tiny $\alpha$}_n|}}
		           {\partial\mbox{\boldmath\scriptsize $t$}^{\mbox{\boldmath\tiny $\alpha$}_n}}
		              (\varphi^{\sharp}_{\mbox{\boldmath\scriptsize $t$}}(y^n) )\right)$'s,
	   $\mbox{\boldmath $\alpha$}_1+\cdots+\mbox{\boldmath $\alpha$}_n
	       = \mbox{\boldmath $\alpha$}$		
	  \begin{itemize}
	   \item[4.2.1]	
	    Preparatory: Chain rule vs. Leibniz rule, and the increase of complexity
	
	   \item[4.2.2]
	    The case of first-order derivations
	
	   \item[4.2.3]
	   Generalization to derivations of any order	
	  \end{itemize}
 	  \vspace{0ex}		
  \end{itemize}
  
 \item[5.]
 The first variation of the Dirac-Born-Infeld action and the equations of motion for D-branes
  \vspace{-.6ex}
  \begin{itemize}   	
   \item[5.1]
    Remark on deformation problems in $C^{\infty}$-algebraic geometry
   
   \item[5.2]
    The first variation of the Dirac-Born-Infeld action
	
   \item[5.3]
    The equations of motion for D-branes
 %
  \end{itemize}
     
 \item[6.]
 Remarks on the Chern-Simons/Wess-Zumino term
  \vspace{-.6ex}
  \begin{itemize}	
   \item[6.1]
    Resolution of issues in the Chern-Simons/Wess-Zumino term
	
   \item[6.2]	
	The first variation and the contribution to the equations of motion	
	\begin{itemize}
	 \item[6.2.1]
      D$(-1)$-brane world-point $(m=0)$

     \item[6.2.2]
      D-particle world-line $(m=1)$

     \item[6.2.3]
	  D-string world-sheet $(m=2)$

     \item[6.2.4]
	  D-membrane world-volume $(m=3)$	
	\end{itemize}
  \end{itemize}	
  
 %
 %
 %
 
\end{itemize}
} 

\newpage

\section{Coincident D-branes and issues on non-Abelian Dirac-Born-\\ Infeld action}

\begin{flushleft}
{\bf Open-strings, background gauge fields, and the Born-Infeld action:\\ The pre-D-brane era}
\end{flushleft}
Consider an open string moving in  a space-time $Y$ with
 a background dilaton field $\Phi$,
 a background metric $g$, and a background $B$-field $B$, from the closed-string sector,
 and a $U(1)$ gauge field $A$, from the open-string sector.
Then, similar to the governing equations for $(\Phi,g,B)$
 from the conformal-anomaly-free conditions for the $2d$ field theory on the closed-string world-sheet,
the dynamics of $A$ is governed by the conformal-anomaly-free conditions on $2d$ field theory with boundary
 on the open-string world-sheet.
It turns out that, at least for an appropriate lowest order approximation,
 this system of differential equations can be derived from the variation of the Lagrangian density\footnote{Here,
                                                                                                  some mild natural, well-accepted updates 
																								     and change of notations
																								   are made to the quoted original works
																								   in order to fit in better the current notes.
                                                                                                  Our apology to the original authors.} 
$$
  {\cal S}^{(\Phi,g,B)}_{\BI}(\nabla)\;
    =\;   -\,T\,e^{-\Phi}\sqrt{-\,\Det(g+B+2\pi\alpha^{\prime}F_{\nabla})\,}\,.
$$
Here,
 $T$ is a physical constant,
 $F_{\nabla}$ is the curvature of $\nabla$, and
 the determinant $\Det$ is for the $2$-tensor $g+B+2\pi\alpha^{\prime}F_{\nabla}$ on $Y$.
(The fields $(\Phi,g,B)$ are governed by another Lagrangian density ${\cal S}(\Phi,g,B)$,
     we will completely omit as they are irrelevant to our discussion.)
Furthermore, through the coupling to the Chan-Paton index, the above (Abelian) Born-Infeld action
 is expected to be generalizable to a non-Abelian Born-Infeld action
 $$
  {\cal S}^{(\Phi,g,B)}_{\BI}(\nabla)\;
    =\;    -\,T\cdot \STr\!\left( e^{-\Phi}\sqrt{-\,\Det(g+B+2\pi\alpha^{\prime}F_{\nabla})\,}
	                  \right)\,.
 $$
Here, the $\STr$ is the `{\it symmetrized trace}',
 with
 \begin{itemize}
  \item[\LARGE $\cdot$]
   the {\it trace-part} $\Tr$ in $\STr$
    serving as a perturbative method to understand $\sqrt{-\Det(\,\cdots\,)}$
    in terms of an expansion at $g+B$ through powers of $\alpha^{\prime}$
    via the Taylor series of the analytic formula
    $$
     (\Det(1+x))^{\beta}\;
       =\; \exp\left(\beta   \Tr(\Log(1+x))\right)
    $$
   after some natural manipulation of the expression of ${\cal S}_{\BI}^{(\Phi,g,B)}(\nabla)$;
   this takes care of the meaning of $\sqrt{-\Det(\,\cdots\,)}$ for a Lie-algebra valued $2$-tensor
   $(\,\cdots\,)$ that appears;
  
  \item[\LARGE $\cdot$]     	
   the {\it symmetrized}-part $S$ in $\STr$ taking in addition into account the fact that
   the Lie algebra involved now is non-Abelian:
   $$
    \STr(m_1\,\cdots\,m_l)\;
     :=\; \mbox{\Large $\frac{1}{l!}$}
            \sum_{\sigma\in\scriptsizeSym_l}\Tr(m_{\sigma(1)}\,\cdots\,m_{\sigma(l)})
    $$
    for $m_1,\,\cdots\,,m_l$ in the Lie algebra in matrix form;
   not only that this is a very natural thing to do to deal with the noncommutativity issue here, 
    one anticipates that this symmetrization is also required to fit the perturbative result
	better with the open-string consideration.
 \end{itemize}
Readers are referred to, for example,\footnote{(From C.H.L.)
                                           These are some works that particularly influencd my thought
											  during my brewing years 2007-2015 on this topic.
										   There is no intention at all to make a complete survey on this topic here.
										   Readers should consult references therein and key-word search
											 for a more comprehensive understanding of this part of the history of the development.
										   Similarly, for the literature on its upgrade: the Dirac-Born-Infeld action.}
 (time-ordered)
 [F-T] (1985) of {\it Efim Fradkin} and {\it  Arkady Tseytlin},
 [D-O] (1986) of {\it Harald Dorn} and {\it Hans-Jorg Otto},
 [A-C-N-Y] (1987) of
   {\it Ahmed Abouelsaood}, {\it Curtis Callan, Jr.}, {\it Chiara Nappi}, and {\it Scott Yost}
   (preprint: 1986),
 [A-N] (1990) of {\it Philip Argyres} and {\it Nappi} (preprint: 1989)
  in the second half of 1980s for more details, further references, and issues that remain.
   
In the language of nowadays,
all these studies are in the special case
 where the space-time is filled with a D-brane-world-volume
 so that the end-point of open strings can move freely anywhere in the space-time.
(In the notation of the next theme, this means that $X=Y$,
     with the background fields $(\Phi,g,B)$ living on the space-time $Y$ and
	  $\nabla$ living on the Chan-Paton bundle $E$ over the D-brane world-volume $X$.)
Which we now turn to.

\bigskip

\begin{flushleft}
{\bf D-branes, Dirac-Born-Infeld action, and its non-Abelian generalization}
\end{flushleft}
A {\it D-brane}, in full name: {\it Dirichlet brane} , in string theory is by definition
(i.e.\ by the very word `Dirichlet') a boundary condition for the end-points of open strings.
{From} the viewpoint of the field theory on the open-string world-sheet aspect,
 it is a boundary state in the $2$-dimensional conformal field theory with boundary.
{From} the viewpoint of open-string target-space(-time) $Y$,
 its world-volume is a cycle or a union of submanifolds $X$ in $Y$
  with a Chan-Paton bundle with a $U(1)$-connection $(E,\nabla)$, suported on $X$,
  that carries the Chan-Paton index for the end-points of oriented open strings,
  cf.\ {\sc Figure}~1-1.
%

\begin{figure}[htbp]
 \bigskip
  \centering
  \includegraphics[width=0.80\textwidth]{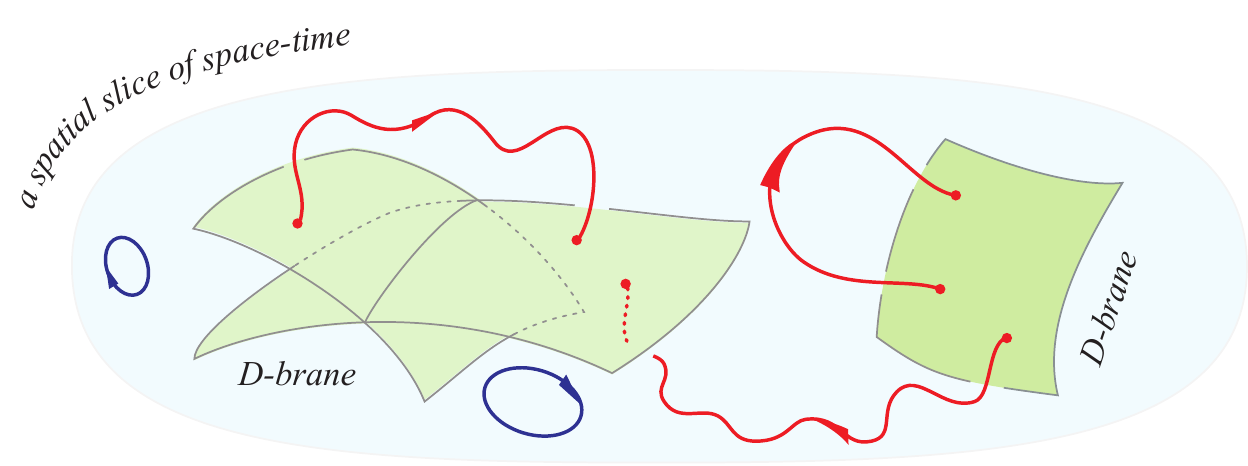}
 
  \bigskip
  \bigskip
 \centerline{\parbox{13cm}{\small\baselineskip 12pt
  {\sc Figure}~1-1.
  D-branes as boundary conditions for open strings in space-time.
  This gives rise to interactions of D-brane world-volumes with open strings.
  Properties of D-branes,
     including the quantum field theory on their world-volume and deformations of such,
   are governed by open strings via this interaction.
  Both oriented open (resp.\ closed) strings and a D-brane configuration are shown.
  }}
\end{figure}

\noindent
In the region of Wilson's theory-space of string theory where the D-brane tension is small,
 D-branes stand in an equal footing with strings as fundamental objects.
In this region, they are soft and can move around and vibrate, just like a fundamental string can,
  in the space-time $Y$.
Thus, a D-brane world-volume in this case is better described as a map $f: X\rightarrow Y$.
Such non-solitonic aspect was already taken in the original works,
  [P-C] (1989) of {\it Joseph Polchinski} and {\it Yunhai Cai}  and
  [D-L-P]  (1989) of {\it Jin Dai},  {\it Robert Leigh}, and {\it Joseph Polchinski},
 that introduced the notion of D-branes to string theory.
With the earlieir works that relate the Born-Infeld action $S_{\BI}^{(\Phi,g,B)}(\nabla)$
 to the dynamics of $\nabla$ that couple to the open-string
  in the special space-time-filling-D-brane case where $X=Y$,
 it is immediately realized by
 {\it Leigh} [Le] (1989) that
 the action for such a simple D-brane $(f,\nabla)$ is given the Dirac-Born-Infeld action density
  $$
  {\cal S}^{(\Phi,g,B)}_{\DBI}(f, \nabla)\;
     =\;   -\,T_{m-1}\,e^{-\Phi}\sqrt{-\,\Det(f^{\ast}(g+B)+2\pi\alpha^{\prime}F_{\nabla})\,}
  $$
  so that the resulting system of equations of motion for a simple D-brane $(f,\nabla)$
   coincides with the system of conformal-anomaly-free constraint equations on $(f,\nabla)$
   from the aspect of $2d$ boundary conformal field theory on the open-string world-sheet.
Here, $m=\dimm X$,
           $m-1$ is the dimension of the D-brane   and
		   $T_{m-1}$ is the D$(m-1)$-brane tension.
While this is a very natural generalization of the earlier work,
  it is worth emphasizing that
  \begin{itemize}
   \item[\LARGE $\cdot$]
   {\it This is an action for the {\it pair} $(f,\nabla)$}.
    It governs
	 not only
	   how the gauge field $\nabla$ lives on the Chan-Paton bundle $E$ on the D-brane world-volume $X$
	 but also
	  {\it how the D-brane world-volume sits in the space-time $Y$, i.e.\ the map $f:X\rightarrow Y$.}
  \end{itemize}

Fast forward now to  year 1995,
 {\it Polchinski} [Po1] (1995)
 realized that D-branes serve not only as general boundary conditions for open strings but,
 in a superstring theory, can also couiple to the Ramond-Ramond fields created by closed superstrings,
 and, hence, serve as the source for such fields.
This shifted the focus of string theory from strings to the various higher-dimensional extended objects: branes.

Something novel and mysterious at the first sight happens
 when a collection of D-branes in space-time coincide:
 (cf.\ [Wi] (1995) of {\it Edward Witten}  and
           [Po3] (1996) and [Po4] (1998) of {\it Polchinski})
\begin{itemize}
 \item[$\cdot$] \parbox[t]{37.4em}{\it
  {\bf [enhancement of scalar field on D-brane world-volume]}\hspace{1em}
   When a collection of D-branes in space-time coincide,
    the open-string-induced massless spectrum on the world-volume of the D-brane
	 is enhanced.
  In particular,
    the gauge field is enhanced to one with a larger gauge group  and
    the scalar field that describes the deformations of the brane in space-time
     is enhanced to one that is matrix-valued.}
\end{itemize}
Cf.\ {\sc Figure}~1-2.
%

\begin{figure}[htbp]
 \bigskip
  \centering
  \includegraphics[width=0.80\textwidth]{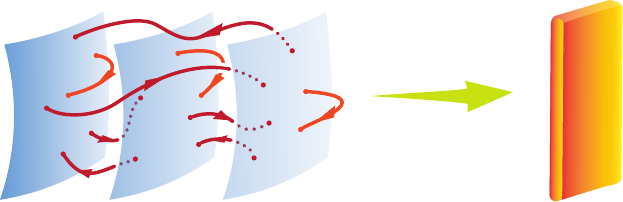}
 
  \bigskip
  \bigskip
 \centerline{\parbox{13cm}{\small\baselineskip 12pt
  {\sc Figure}~1-2.
 When a collection of D-branes coincide, the massless spectrum on the common world-volume are enhanced.
 Not only that the gauge field becomes non-Abelian, the scalar field is also enhanced and becomes non-Abelian.
  }}
\end{figure}

\noindent
This leads to the following key guiding questions:
 \begin{itemize}
  \item[\bf Q1.] \parbox[t]{37.4em}{\raggedright\it {\bf [D-brane]}\hspace{1.1ex}
   What is a D-brane as a fundamental object (as opposed to a solitonic\\ $\hspace{12.6ex}$ object)
   in string theory?}
   
  \item[\bf Q2.] \parbox[t]{37.4em}{\it {\bf [dynamics]}
   What rules govern its dynamics?}
 \end{itemize}
In other words, what is the intrinsic definition of D-branes
 so that by itself it can produce the properties of D-branes
 that are consistent with, governed by, or originally produced by open strings as well?

Leaving the first question --- which is even more fundamental of the two --- tentatively aside,
with this new motivation
 further attempts to generalize the Born-Infeld or Dirac-Born-Infeld in the Abelian case (i.e.\ for a simple D-brane)
  to a non-Abelian case (i.e.\ for coincident D-branes) followed suit immediately:
 for example,
 [Dor1] (1996) and [Dor2] (1997) of {\it Dorn},
 [Ts1] (1997) and [Ts2] (1999) of {\it Tseytlin},
 [B-dR-S1] (2000) and [B-dR-S2] (2000) of
  {\it Eric Bergshoeff}, {\it Mees de Roo}, and {\it Alexander Sevrin},
 [Schw] (2001) of {\it John Schwarz},
 [My] (2001) of {\it Robert Myers},
  and
 the thesis [S\'{e}] (2005) of {\it Emmanuel S\'{e}ri\'{e}}.
Readers are referred to these works for details, further references, and issues that remain.

\bigskip

\begin{flushleft}
{\bf The construction of a non-Abelian Dirac-Born-Infeld action and its consequences}
\end{flushleft}
Back to the two guiding questions, for the first that concerns the intrinsic nature of D-branes,
in the work [H-W: Sec.\ 5] (1996) of {\it Pei-Ming Ho} and {\it Yong-Shi Wu}
they realized that (assuming that the Chan-Paton bundle is a trivialized trivial complex vector bundle of rank $r$)
 a D-brane world-volume $X$ carries a full {\it matrix-ring structure} $M_{r\times r}(L^2(X))$,
 where $L^2(X)$ is the Hilbert space of square-integrable functions on $X$.
A decade afterwards, in late 2006 their important observation was re-picked up
 by the first author of the current notes when he re-thought of the lecture [Po2] of {\it Polchinski}
 from the viewpoint of {\it Grothendieck}'s Modern Algebraic Geometry.
Such input from algebraic geometry (cf.\ [Ha] (1977) and [Joy] (2010))
 gave rise to to a proto-typical definition of D-branes
 ([L-Y1] (D(1), 2007; [L-Y4] (D(11.1), 2014; [L-Y6] (D(11.3.1), 2015)):

\bigskip

\begin{sansatz-definition} {\bf [D-brane: prototypical]}$\;$ {\rm
 Let
  $X$ be the world-volume of a D-brane (i.e.\ a $C^{\infty}$-manifold),
  $E$ be the Chan-Paton bundle on $X$ (i.e.\ a complex vector bundle of rank $r$) on $X$.
 Then a {\it D-brane moving in a space-time} $Y$ is modelled on a `{\it map}'
   $$
     \varphi\,:\; (X^{\!A\!z},E; \nabla)\; \longrightarrow\; Y\,,
   $$
   where
	\begin{itemize}
	 \item[\LARGE $\cdot$]
	  $(X^{\!A\!z},E):= (X, C^{\infty}(\End_{\Bbb C}(E)),E)$
	   is an {\it Azumaya/matrix manifold with a fundamenatl module}
	   whose underlying topology is identical to the manifold $X$ but
	   whose function-ring is given by the {\it endomorphism-algebra}
	     $C^{\infty}(\End_{\Bbb C}(E))$ over ${\Bbb C}$,
		
     \item[\LARGE $\cdot$]
      $\nabla$ is a connection on $E$.	
	\end{itemize}
 Here, the notion of a `map' is {\it defined contravariantly by a ring-homomorphism}
  $$
    \varphi^{\sharp}\;:\; C^{\infty}(Y)\;
	  \longrightarrow\; C^{\infty}(\End_{\Bbb C}(E))
  $$
  over the canonical inclusion ${\Bbb R}\subset {\Bbb C}$.
 Cf.\ {\sc Figure}~1-3.
}\end{sansatz-definition}
%

\begin{figure}[htbp]
 \bigskip
  \centering
  \includegraphics[width=0.80\textwidth]{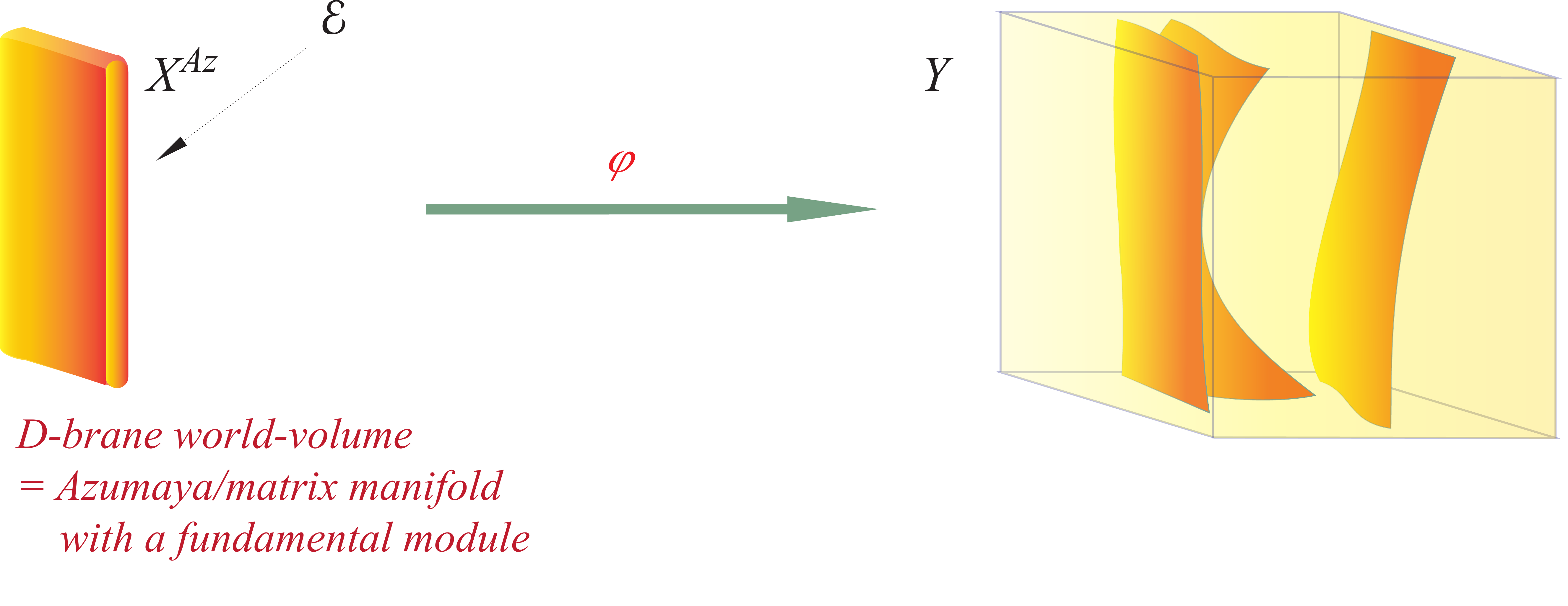}
 
  \bigskip
  \bigskip
 \centerline{\parbox{13cm}{\small\baselineskip 12pt
  {\sc Figure}~1-3.
 As a dynamical object in string theory, D-brane moving in a space-time $Y$ can be described
 as a map $\varphi$ from an Azumaya/matrix manifold $X^{\!A\!z}$
 with a fundamental module $E$ with a connection $\nabla$ to $Y$.
  }}
\end{figure}

\bigskip

\noindent
Readers are referred to Sec.~2.1 of the current notes for a terse review of the part/notation we need
 and to
    [L-Y1] (D(1)), [L-L-L-Y] (D(2)), [L-Y3] (D(6)),
	[L-Y4] (D(11.1)), [L-Y5] (D(11.2)), [L-Y6] (D(11.3.1))
 for further details, examples, and
  the justification by comparing to various D-brany phenomena in string theory
   that this definition is workable and does capture some major features of D-branes.

Once having a proto-typical answer to Question 1, we now turn to Question 2.
Taking the lesson from string-theorists
  that the dynamics of D-branes should be governed (at least at the lowest level)
  by a generalization of the Dirac-Born-Infeld action in the abelian case
  (i.e.\ the case where $E$ is a complex line bundle over $X$),
let us write down first a formal but natural expression
 for the Dirac-Born-Infeld action on the pairs $(\varphi,\nabla)\,$:
 $$
   S_{\DBI}^{(\Phi,g,B)}(\varphi,\nabla)\;
    \stackrel{\mbox{\tiny formally }}{=}\;
     -T_{m-1}\,\int_X \Tr\left(\rule{0ex}{1em}\right.
	                       e^{-\varphi^{\ast}(\Phi)}
						    \sqrt{-\,
						     \Det_X(
							  \varphi^{\ast}(g+B)\,
						       +\, 2\pi\alpha^{\prime}F_{\nabla}
							    )\,}
	                          \left.\rule{0ex}{1em}\right)\,.
 $$
 Here, $m=\dimm X$ and $T_{m-1}$ is the tension fo the $D(m-1)$-brane.
We now list all the issues that need to be resolved to make sense, or interpret correctly,
 of this formal expression:
\begin{itemize}		
 \item[(1)]  [{\it map $\varphi$}]\\
  We have settled down the notion of $\varphi$ purely algebro-geometrically,
    without having anything to do with the connection $\nabla$.
  Once the connection is brought into play,
  {\it is there a constraint on the pair $(\varphi,\nabla)$
   that comes from string-theoretical consideration?}
  Furthermore, {\it when $E$ is Hermitian and $\nabla$ is unitary,
   is there a class of maps $\varphi$ that stand out from others
   due to such additional structure on $(E,\nabla)$?}
  
 \item[(2)]  [{\it push-pull of tensor under $\varphi$}\hspace{.1ex}]\\
  The notion of {\it push-pulls under a differentiable map $\varphi$}
     from an Azumaya/matrix manifold with a fundamental module to a real manifold;
  cf.\ $\varphi^{\ast}(g+B)$.	
	
 \item[(3)] [{\it determinant of $2$-tensor on $X$}\hspace{.1ex}]\\
  The notion of {\it determinant} $\Det_X(\,\cdots\,)$ in the current context
   once Issues (1) and (2) are resolved.
   
 \item[(4)] [{\it square root of matrix section}\hspace{.1ex}]\\
 Can we take a {\it square root of a matrix}?
  When the answer is Yes, is there a unique square root?
  If not, which one to choose?
 Extension of this to {\it matrix sections}\hspace{.1ex}?

 \item[(5)]  [{\it dilaton-field factor}\hspace{.2ex}]\\
   How does {\it the factor $e^{-\varphi^{\sharp}(\Phi)}$}
    influence the interpretation of the formal expression?
	
 \item[(6)] [{\it real-valuedness}\hspace{.1ex}]\\
  Is the expression {\it real}$^{\,}$?
  
 \item[(7)] [{\it consistency with open string theory}\hspace{.2ex}]\\
  {\it How does it fit with open string theory?}
\end{itemize}

In this work, we will
 answer to and resolve Issues (1) - (6) above  in a way
    that is physically meaningful and mathematically natural     and
 construct a Dirac-Born-Infeld action $S_{\DBI}$ for $(\varphi, \nabla)\,$
 (Sec.\ 2 and Sec.\ 3).
We then develop a necessary tool (Sec.\ 4) to  carry out the first variation formula of $S_{\DBI}$
 with respect to $(\varphi,\nabla)\,$ (Sec.\ 5).
In view of Polchinski's realization ([Po1]) that
  D-brane world-volume serves as the source for Ramond-Ramond fields in superstring theory,
 the Chern-Simons/Wess-Zumino action $S_{\CSWZ}$ for D-branes is also an indispensable part
 to understand the dynamics of D-branes.
With the same essence as for the construction of $S_{\DBI}(\varphi,\nabla)$,
we
 construct the Chern-Simons/Wess-Zumino action $S_{\CSWZ}(\varphi,\nabla)$  in Sec.\ 6
    for lower-dimensional D-branes, in which cases anomaly issues do not occur,
 derive their first variation formula and, hence,
 obtain their contribution to the equations of motions for D-branes.
  
\bigskip

\begin{sremark} {\rm [\hspace{.1ex}}effect
              of $B$-field to fundamental module $E${\rm\hspace{.1ex}]}$\;$ {\rm
 The presence of a $B$-field on the space-time $Y$ can have a non-trivial twisting effect to the Chan-Paton bundle
  $E$ on the D-brane world-volume $X$,
  rendering it no longer an honest vector bundle but, rather, a `twisted vector bundle' ([Wi]).
 For better focus, we omit this effect in the current notes.
 Reader are referred to [Wi] and, e.g., [Kap] for more details and references,
  and to [L-Y2] (D(5)) for details on how this effect is taken into account in our setting.
 This twisting effect can always be added back to our presentation whenever in need. 
}\end{sremark}

%
%
%
%
%

\bigskip

\section{Differentiable maps from an Azumaya/matrix manifold with\\ a fundamental module with a connection}

Recall from Sec.\ 1 the first issue we need to understand
 before we can construct the Dirac-Born-Infeld action for D-branes in our setting:
  \begin{itemize}
   \item[(1)]  [{\it map $\varphi$}]\\
    We have settled down the notion of $\varphi$ purely algebro-geometrically,
     without having anything to do with the connection $\nabla$.
   Once the connection is brought into play,
   {\it is there a constraint on the pair $(\varphi,\nabla)$
    that comes from string-theoretical consideration?}
   Furthermore, {\it when $E$ is Hermitian and $\nabla$ is unitary,
    is there a class of maps $\varphi$ that stand out from others
    due to such additional structure on $(E,\nabla)$?}
  \end{itemize}
In this section, we
 first review very tersely the part of [L-Y4] (D(11.1)) and [L-Y6] (D(11.3.1))
   that is needed for the current notes (Sec.\ 2.1),
 then address a compatibility issue from the open-string aspect
   between our notion of maps $\varphi$ from a matrix manifold to a space-time
    and the connection $\nabla$ on the fundamental module associated to that matrix manifold (Sec.\ 2.2),
  	and
 finally bring out the additional notion of {\it self-adjoint map} in our context (Sec.\ 2.3).
This notion is what we interpret the (seemingly Lie-algebra-valued) scalar field on the common world-volume of
 coincidnet D-branes as in [Po2], [Po3], and [Wi] when the connection on the Chan-Paton bundle is unitary.

\bigskip

\subsection{Differentiable maps from an Azumaya/matrix manifold with a fundamental module}

\begin{definition} {\bf [map from Azumaya/matrix manifold]}$\;$ {\rm
 Let
  $X$ be a (real, smooth) manifold,
  $E$ be a complex vector bundle of rank $r$ over $X$,  and
  $(X^{\!A\!z}, E) := (X, C^{\infty}(\End_{\Bbb C}(E)), E)$
     be the associated Azumaya/matrix manifold with a fundamental module.
 A {\it map} (synonymously, {\it differentiable map}, {\it smooth map})
   $$
     \varphi\,:\, (X^{\!A\!z},E)\; \longrightarrow\; Y
   $$
   from $(X^{\!A\!z}, E)$ to a (real, smooth) manifold $Y$
   is defined by a ring-homomorphism
   $$
     \varphi^{\sharp}\;:\; C^{\infty}(Y)\;\longrightarrow\; C^{\infty}(\End_{\Bbb C}(E))\,.
   $$
}\end{definition}
 
\medskip

\begin{definition} {\bf [push-forward $\varphi_{\ast}{\cal E}$]}$\;$ {\rm
 Let ${\cal E}$ be the sheaf of (smooth) sections of $E$;
   it is canonically identical to the sheaf on $X$ from localizations of $C^{\infty}(E)$.
 The ring-homomorphism
   $\varphi^{\sharp}:C^{\infty}\rightarrow C^{\infty}(\End_{\Bbb C})$
  renders $C^{\infty}(E)$ a $C^{\infty}(Y)$-module.
 This defines a sheaf on $Y$, denoted by $\varphi_{\ast}{\cal E}$,
  and is called the {\it push-forward} of ${\cal E}$ under $\varphi$.
  It is an ${\cal O}_Y^{\,\Bbb C}$-module.
}\end{definition}
 
\medskip
 
\begin{proposition}  {\bf [basic properties of $\varphi^{\sharp}$]}$\;$
 $(1)$ {\rm [}\hspace{.1ex}{\sl
        realness nature of $\varphi^{\sharp}$}\hspace{.1ex}{\rm ]}\hspace{1em}
 For any $f\in C^{\infty}(Y)$ and $x\in X$ (an ${\Bbb R}$-point),
  the eigenvalues of
    $\varphi^{\sharp}(f)|_x\in \End_{\Bbb C}(E|_x)\simeq M_{r\times r}({\Bbb C})$
 are all real.\\  {\rm (Cf.\ [L-Y4: Sec.\ 3] (D(11.1).)}
 
 $(2)$ {\rm [}\hspace{.1ex}{\sl
       canonical lifting to $C^{\infty}(X\times Y)$}\hspace{.1ex}{\rm ]}\hspace{1em}
 The ring-homomorphism
  $$
    \xymatrix{
	  C^{\infty}(\End_{\Bbb C}(E))
	     && C^{\infty}(Y)\ar[ll]_-{\varphi^{\sharp}}
	  } 	
  $$
  extends canonically to a commutative diagram of ring-homomorphisms
  (over ${\Bbb R}$ or ${\Bbb R}\subset{\Bbb C}$, whichever is applicable)
  $$
    \xymatrix{
	  \;C^{\infty}(\End_{\Bbb C}(E))\;		
	     && \;\;C^{\infty}(Y)\;\;
		           \ar[ll]_-{\varphi^{\sharp}}
				   \ar@{_{(}->}[d]^-{pr^{\sharp}_{Y}}   \\
	  \;\;C^{\infty}(X)\rule{0ex}{1em}\;
	       \ar@{^{(}->}[rr]_-{pr^{\sharp}_X}      \ar@{^{(}->}[u]
	     && C^{\infty}(X\times Y)
		          \ar[llu]_-{\tilde{\varphi}^{\sharp}}\;,
    }
  $$
  where
    $\pr_X:X\times Y \rightarrow X$ and $\pr_Y:X\times Y\rightarrow Y$ are the projection maps,  and
	$C^{\infty}(X)\hookrightarrow C^{\infty}(\End_{\Bbb C}(E))$
	   follows from the inclusion of the center $C^{\infty}(X)^{\Bbb C}$ of
	   $C^{\infty}(\End_{\Bbb C}(E))$.\\
 {\rm (Cf.\ [L-Y6: Theorem 3.1.1] (D(11.3.1)).)}
\end{proposition}

\medskip

\begin{definition} {\bf [graph of $\varphi$]}$\;$ {\rm
 The above diagram of ring-homomofphisms defines a commutative diagram of maps
   $$
    \xymatrix{
	  \;(X^{\!A\!z},E)\;
	           \ar[rr]^-{\varphi_T}  \ar[rrd]^-{\tilde{\varphi}}  \ar@{->>}[d]
	       &&   \; Y\;    \\
	    \;X\;
		   &&   \;\;X\times Y
		                 \ar@{->>}[ll]^-{pr_X}
                         \ar@{->>}[u]_-{pr_Y}						 \,.
     }
   $$
 The push-forward $\tilde{\varphi}_{\ast}{\cal E}=:\tilde{\cal E}_{\varphi}$
   of ${\cal E}$ under $\tilde{\varphi}$
   is called the {\it graph} of $\varphi$.
  It is an ${\cal O}_{X\times Y}^{\,\Bbb C}$-module.
 Its {\it $C^{\infty}$-scheme-theoretical support}
  is denoted by $\Supp(\tilde{\cal E}_{\varphi})$.
}\end{definition}

\medskip

\begin{definition} {\bf [surrogate of $X^{\!A\!z}$ specified by $\varphi$]}$\;$ {\rm
 The image
   $$
     A_{\varphi}\;
	  :=\;  \Image\tilde{\varphi}^{\sharp}\;
      :=\;  \tilde{\varphi}(C^{\infty}(X\times Y))\;
	  =\; C^{\infty}(X)\langle\Image\varphi^{\sharp}\rangle\;\;
	  \subset\;  C^{\infty}(\End_{\Bbb C}(E))
   $$
   of $\tilde{\varphi}:C^{\infty}(X\times Y)\rightarrow C^{\infty}(\End_{\Bbb C}(E))$
   is a commutative $C^{\infty}(X)$-subalgebra of $C^{\infty}(\End_{\Bbb C}(E))$
 that is locally algebraically finite over $C^{\infty}(X)$.
 It defines a $C^{\infty}$-scheme
   $$
      X_{\varphi}\;:=\; \Spec^{\Bbb R}(A_{\varphi})\,,
   $$
  which is called the {\it surrogate} of $X^{\!A\!z}$ specified by $\varphi$.
 $X_{\varphi}$ is finite over $X$ and,
   by construction,
   it admits a canonical embedding
   $\tilde{f}_{\varphi}:X_{\varphi}\rightarrow X\times Y$
   into $X\times Y$ as a $C^{\infty}$-subscheme.
 The image is identical to $\Supp(\tilde{\cal E}_{\varphi})$.
  Cf.\ {\sc Figure} 2-1-1.
%
%
 \begin{figure} [htbp]
  \bigskip
  \centering

  \includegraphics[width=0.40\textwidth]{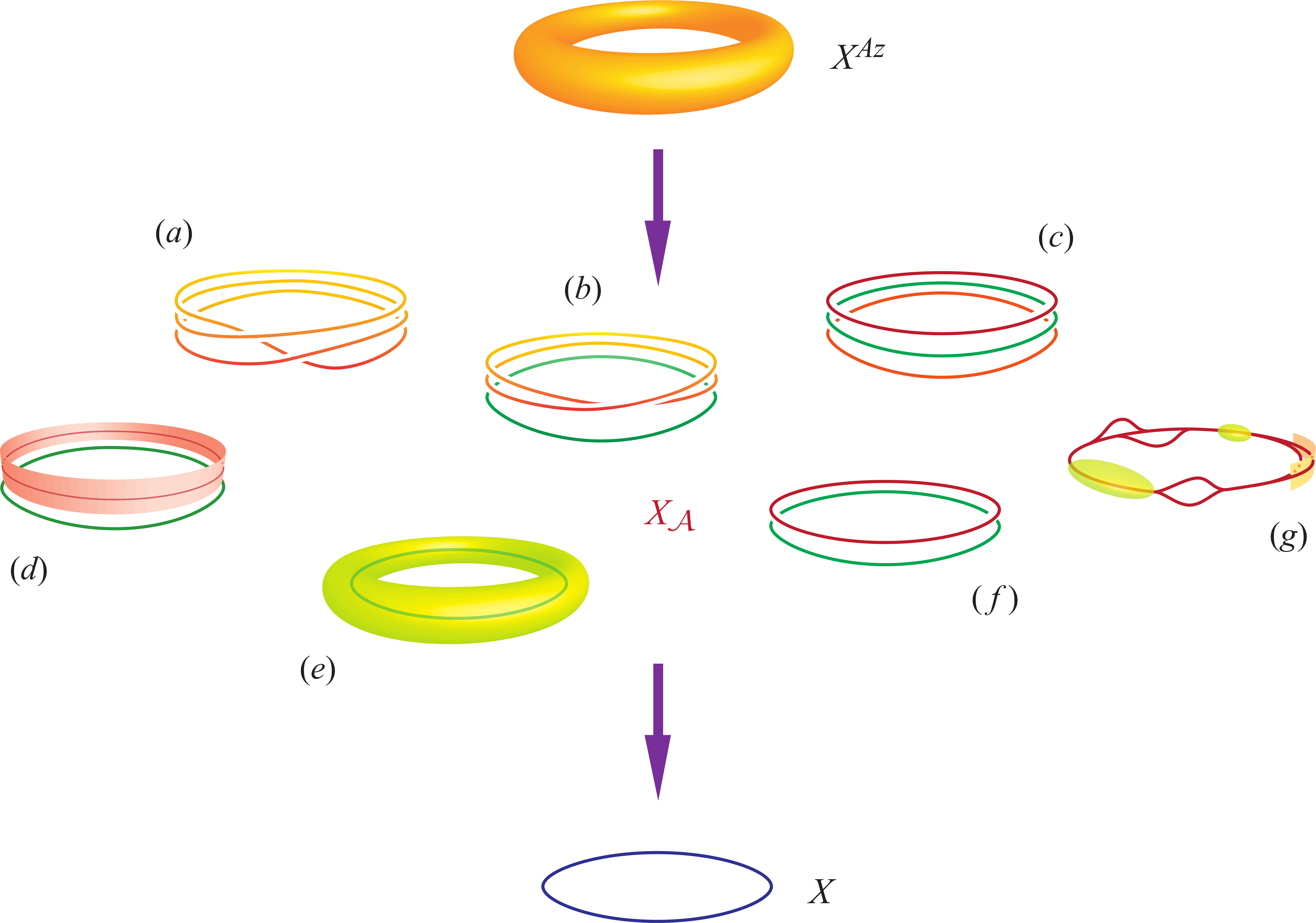}

  \bigskip
  \bigskip
  \centerline{\parbox{13cm}{\small\baselineskip 12pt
   {\sc Figure}~2-1-1.
   The $C^{\infty}$-scheme $X_{A}:=\Spec^{\Bbb R}A$
     associated to a commutative $C^{\infty}(X)$-subalgebra
	   $C^{\infty}(X)\subset A\subset C^{\infty}(\End_{\Bbb C}(E))$
	 can be thought of as interpolating between the commutative $X$ and the noncommutative $X^{\!A\!z}$
     by the built-in dominant maps
	 $$	
	    X^{\!A\!z}\; \longrightaarrow\; X_A\; \longrightaarrow\; X\,.
	 $$
   The abundance of such objects under $X^{\!A\!z}$	indicates a very rich geometric structure
     the Azumaya/matrix manifold $X^{\!A\!z}$ contains.
   In the Figure, seven surrogates $(a)$ -- $(g)$ of an Azumaya/matrix string $S^{1, A\!z}$	
     are indicated.
    They include
	  short-string sets: ($c$) and ($f$),
	  a long string ($a$),
	  a fuzzy string ($e$),
	  and various mixtures: ($b$), ($d$), ($g$).
   	In particular,
	  given a map $\varphi:(X^{\!A\!z},E)\rightarrow Y$,
	  the surrogate of $X^{\!A\!z}$ specified by $\varphi$ can be used to help capture $\varphi$ itself
	    and serve as a $\varphi$-specified medium between $X$ and $Y$.
       }}
  \bigskip
 \end{figure}	
}\end{definition}

\bigskip

One can summarize all the objects introduced into the following two diagrams
  that refine the contravariant pair of diagrams in Proposition~2.1.3 (2) and Definition~2.1.4
  respectively:	
  $$
   \xymatrix{
     C^{\infty}(E)\\
	& C^{\infty}(\End_{\Bbb C}(E))  \ar@{~>}[lu] \\
    &     A_{\varphi}\rule{0ex}{1em}
				      \ar@{^{(}->}[u]  \ar@{~>}@/^2ex/[luu]
                   				&&& C^{\infty}(Y)\ar[lllu]_-{\varphi^{\sharp}}
								                           \ar[lll]^-{f_{\varphi}^{\sharp}}
														   \ar@{_{(}->}[d]^-{pr_Y^{\sharp}}\\
    &  \;\;C^{\infty}(X)\;\; \rule{0ex}{3ex}  \ar@{^{(}->}[u]_-{\pi_{\varphi}^{\sharp}}
	                                                                 \ar@{~>}@/^5ex/[luuu]
																	 \ar@{^{(}->}[rrr]_-{pr_X^{\sharp}}
         &&&   C^{\infty}(X\times Y)\ar@{->>}[ulll]^-{\tilde{f}_{\varphi}^{\sharp}}
   }
  $$
  and
  $$
   \xymatrix{
    & {\cal E} \ar@{.>}[rd]     \ar@{.>}@/_1ex/[rdd]    \ar@{.>}@/_2ex/[rddd]   \\
    &  & X^{\!A\!z}\ar[rrrd]^-{\varphi}\ar@{->>}[d]^-{\sigma_{\varphi}}   \\
    &  & X_{\varphi}\ar[rrr]_-{f_{\varphi}}
	                                   \ar@{_{(}->} [rrrd]_{\tilde{f}_{\varphi}}
                                	   \ar@{->>}[d]^-{\pi_{\varphi}} &&& Y \\
	&  & X     &&& X\times Y \ar@{->>}[u]_-{pr_Y} \ar@{->>}[lll]^-{pr_X}  &,
    }
  $$
 with the built-in isomorphisms
 $$
   {\cal E}
      \simeq {\pi_{\varphi}}_{\ast}(_{{\cal O}_{X_{\varphi}}}{\cal E})
	  \simeq (\pi_{\varphi}\circ \sigma_{\varphi})_{\ast}
	          (_{{\cal O}^{A\!z}_X}{\cal E})
		\hspace{2em}\mbox{and}\hspace{2em}
    _{{\cal O}_{X_{\varphi}}}{\cal E}
	  \simeq {\sigma_{\varphi}}_{\ast}(_{{\cal O}^{A\!z}_X}{\cal E})\,.	
 $$
   
 \bigskip
 
At this point, readers may feel that such a notion of maps is too abstract to perceive.
Recall then that in the ordinary differential topology or geometry,
 a map $f$ from a manifold $X$ to ${\Bbb R}^n$, with coordinates $(y^1,\,\cdots\,,\, y^n)$,
  is determined by specifying its projection to each coordinate, i.e.\
  $f=(f^1,\,\cdots\,,\, f^n):X\rightarrow {\Bbb R}^n$. Each $f^i$ is now in $C^{\infty}(X)$.
Thus,
 in terms of function-rings, this means that $f$ is determiend by the $n$-tuple
  $(f^{\sharp}(y^1),\,\cdots\,,\, f^{\sharp}(y^n))$, which is exactly the $n$-tuple
  $(f^1,\,\cdots\,,\,f^n)$ above.
The following proposition from [L-Y6] (D(11.3.1)) says that
 a very  similar statement holds for our noton of maps from Azumaya/matrix manifolds:

\bigskip

\begin{proposition} {\bf [map from Azumaya/matrix manifold to ${\Bbb R}^n$]}$\;$
 {\rm ([L-Y6: Theorem 3.2.1] (D(11.3.1)).)}
 Let
  $X$ be a smooth manifold  and
  $E$ be a complex smooth vector bundle of rank $r$ on $X$.
 Let
  $(y^1,\,\cdots\,, y^n)$ be a global coordinate system on ${\Bbb R}^n$,
  as a smooth manifold, and
  $$
    \eta\;:\; y^i\; \longmapsto\; m_i\,\in\, C^{\infty}(\End_{\Bbb C}(E))\,,\;\;
	i\;=\;1,\,\ldots\,,n\,,
  $$
 be an assignment  such that
  \begin{itemize}
   \item[$(1)$]
     $\;m_im_j\;=\;m_jm_i$, for all $i,\,j\,$;

   \item[$(2)$]
    for every $p\in X$,
	 the eigenvalues of the restriction
	   $m_i(p)\in \End_{\Bbb C}(E|_p)\simeq M_{r\times r}({\Bbb C})$
	   are all real.
 \end{itemize}
 Then,
  $\eta$ extends to a unique ring-homomorphism
  $$
    \varphi_{\eta}^{\sharp}\;:\;
	 C^{\infty}({\Bbb R}^n)\; \longrightarrow\; C^{\infty}(\End_{\Bbb C}(E))
  $$
  over ${\Bbb R}\subset{\Bbb C}$ and, hence,
  defines a map $\varphi_{\eta}:(X^{\!A\!z},E)\rightarrow {\Bbb R}^n$.
\end{proposition}
  
\bigskip

This is a consequence of the Malgrange Division Theorem ([Mal]; also [Br\"{o}], [Mat1], [Mat2], [Ni]).
Note that Conditions (1) and (2) in Proposition~2.1.6
  are necessary conditions for $\eta$ to be extendable to a full ring-homomorphism.
The proposition says that they are also sufficient and the extension is unique.
Due to its importance as a technical tool for our study later,
 we will highlight its proof in Sec.\ 4.1 in the form we need and then
 generalize it to a similar statement for derivatives
  $\partial^{\scriptsizeboldalpha}\varphi^{\sharp}$ of $\varphi^{\sharp}$
  to all orders $|\boldalpha|$ in Sec.\ 4.2.
It is through this proposition that we can almost visualize $\varphi^{\sharp}$ as we would for $f$.
  
\bigskip

\begin{example} {\bf [D$0$-brane on ${\Bbb R}^2$, deformation, Higgsing/un-Higgsing]}$\;$ {\rm
 D$0$-branes (or D$(-1)$-brane world-points) on ${\Bbb R}^2$
  can be described by maps from Azumaya/matrix points
  $\varphi:(p^{A\!z},{\Bbb C}^{\oplus r})\rightarrow {\Bbb R}^2$,
  defined by ring-homomorphisms
   $$
     \varphi^{\sharp}\,:\; C^{\infty}({\Bbb R}^2)\; \longrightarrow\;
	    M_{r\times r}({\Bbb C})\,.
   $$
 The latter is determined by the value
   $(m^1,m^2)
       :=(\varphi^{\sharp}(y^1),\varphi^{\sharp}(y^2))
       \in M_{r\times r}({\Bbb C})\times M_{r\times r}({\Bbb C})$
   of $\varphi^{\sharp}$ on the coordinates $(y^1,y^2)$ of ${\Bbb R}^2$.
 Any pair $(m^1, m^2)$ of $r\times r$ matrices (with entries in ${\Bbb C}$)
  that commute and with each matrix having only real eigenvalues defines a $\varphi$.
 Deformations of $\varphi$  may create
   various Higgsing/un-Higgsing phenomena of D0-branes on ${\Bbb R}^2$.
 Cf.\ {\sc Figure} 2-1-2.
%

 \begin{figure} [htbp]
  \bigskip
  \centering

  \includegraphics[width=0.60\textwidth]{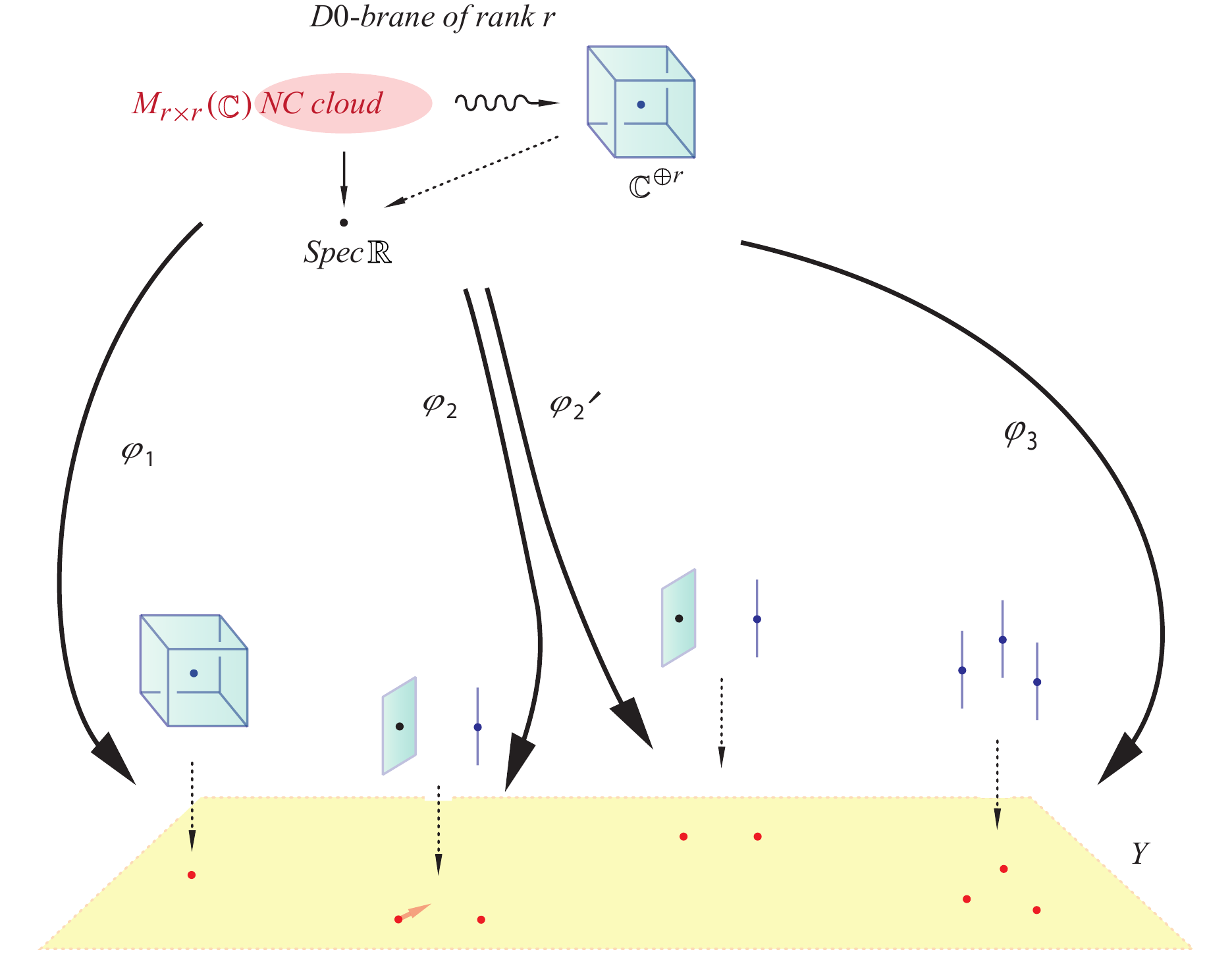}

  \bigskip
  \centerline{\parbox{13cm}{\small\baselineskip 12pt
   {\sc Figure}~2-1-2.
    Four examples of maps
     $\varphi:(p^{A\!z},{\Bbb C}^{\oplus r})\rightarrow {\Bbb R}^2$
      from an Azumaya/matrix point with a fundamental module to ${\Bbb R}^2$
	  are illustrated.
	The nilpotency of the image scheme $\Image\varphi$ in ${\Bbb R}^2$ is bounded by $r$.	
    In the figure, the push-forward $\varphi_{\ast}({\Bbb C}^{\oplus r})$
	of the fundamental module in each example is also indicated.
       }}
  \bigskip
 \end{figure}	
}\end{example}

\vspace{6em}

\begin{remark}{\rm [\hspace{.1ex}}D$p$-brane from smearing D$0$-branes{\rm\hspace{.1ex}]}\; {\rm
 Functionally, a D$p$-brane can be thought of as from smearing a jam of D$0$-branes along a $p$-cycle $X$.
 A map $\varphi: X^{\!A\!z}\rightarrow Y$ may be thought of
  as an $X$-family of maps $\varphi_x:p^{\!A\!z}\rightarrow Y$ from an Azumaya/matrix point $p^{\!A\!z}$ to $Y$.
 This gives another way to visualize $\varphi$,
    in addition to the picture by the surrogate $X_{\varphi}$ of $X^{\!A\!z}$ specified by $\varphi$.
  Cf.\ {\sc Figure}~2-1-3 and Example~2.1.7.
%

\begin{figure}[htbp]
 \bigskip
  \centering
  \includegraphics[width=0.80\textwidth]{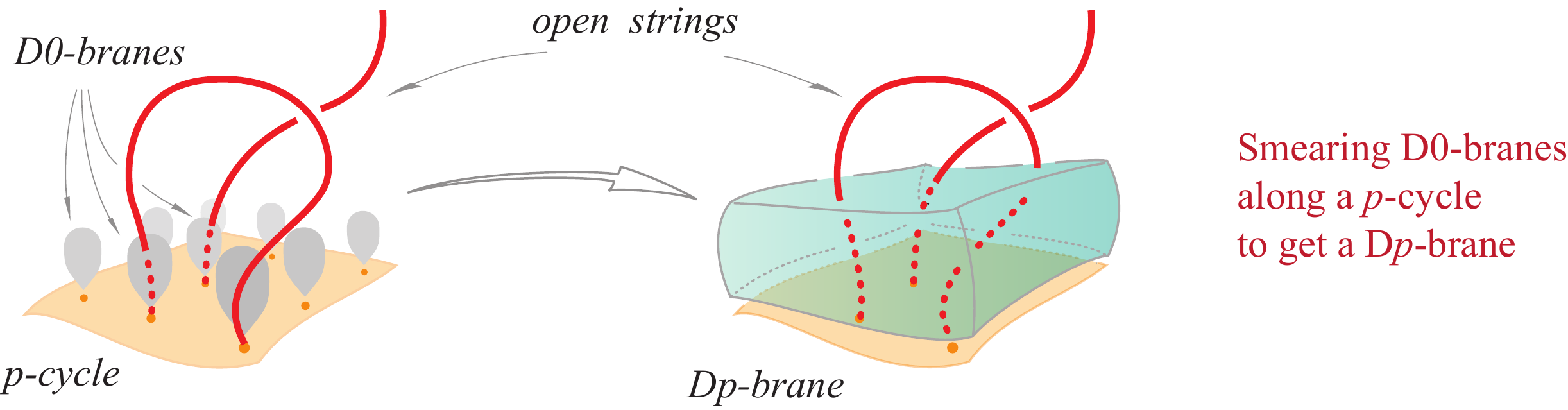}
 
  \bigskip
  \bigskip
 \centerline{\parbox{13cm}{\small\baselineskip 12pt
  {\sc Figure}~2-1-3.
  A D$p$-brane from smearing a jam of D$0$-branes along a $p$-cycle.
  }}
\end{figure}
}\end{remark}

\bigskip

\subsection{Compatibility between the map $\varphi$ and the connection $\nabla$ from the open-string aspect}
 
Recall ([DV-M: Proposition 3]; see ibidem and [L-Y4: Sec.\ 4] for more references and discussions) that
a connection $\nabla$ on a complex vector bundle $E$ over $X$ induces canonically a connection $D$ 
 on the endomorphism bundle $\End_{\Bbb C}(E)$
 ($\simeq E\otimes_{\Bbb C}E^{\vee}$ canonically, with $E^{\vee}$ the dual vector bundle of $E$).
Let
 $\pi_X^{\sharp}: C^{\infty}(X)\hookrightarrow C^{\infty}(\End_{\Bbb C}(E))$
   be the inclusion that follows from the inclusion of center
   $C^{\infty}(X)^{\Bbb C}\subset C^{\infty}(\End_{\Bbb C}(E))$,
 ${\cal T}_{\ast}X$, ${\cal T}_{\ast}X^{\!A\!z}$ be the tangent sheaf
  (i.e.\ sheaf of derivations) of $X$ and $X^{\!A\!z}$ respectively, and
 $\Innsheaf({\cal O}_X^{A\!z})$ be the sheaf of inner derivations of the structure sheaf
  ${\cal O}_X^{A\!z}$.
Then there is a natural exact sequence of ${\cal O}_X^{\,\Bbb C}$-modules
 $$
    0\; \longrightarrow\; \Innsheaf({\cal O}_X^{A\!z})\;
	     \longrightarrow \; {\cal T}_{\ast}X^{\!A\!z}\;
         \longrightarrow\; {\cal T}_{\ast}X^{\Bbb C}\;  \longrightarrow\; 0\,.
 $$
Furthermore,
 since
 $$
   D_{\tinybullet}(m_1m_2)\;=\; (D_{\tinybullet}m_1)\,m_2\, +\, m_1\,D_{\tinybullet}m_2
 $$
 for all $m_1,\, m_2\in C^{\infty}(\End_{\Bbb C}(E))$,
 the connection
    $\nabla:{\cal E}\rightarrow {\cal T}_{\ast}X\otimes_{{\cal O}_X^{\,\Bbb C}}{\cal E}$
 on ${\cal E}$ induces an embedding
 $$
   \xymatrix{
	 \iota^{\nabla}\;:\; {\cal T}_{\ast}X^{\Bbb C}\;  \ar@{^{(}->}[r]
	 & \;{\cal T}_{\ast}X^{\!A\!z}
	 }
 $$
 as ${\cal O}_X^{\,\Bbb C}$-modules,
   with  $\xi\mapsto  D_{\xi}$,
  that splits the above short exact sequence.
$D$ satisfies the property that
 $$
   \pi_X^{\sharp}\, df(\,\tinybullet\,) \;
     =\; D_{\tinybullet}\,\pi_X^{\sharp}(f)
 $$
 for all $f\in C^{\infty}(X)$.
All these serve as an indication that, in a sense, $X^{\!A\!z}$ is flatly uniform over $X$ differentially topologically.
Which renders the differential calculus on $X^{\!A\!z}$ accessible with respect to that on $X$,
 despite being noncommutative.

In contrast, when given a map $\varphi:(X^{\!A\!z},E)\rightarrow Y$,
 the surrogate $X_{\varphi}$ of $X^{\!A\!z}$ specified by $\varphi$,
  though a commutative $C^{\infty}$-scheme finite over $X$,
  may not be flat or uniform over $X$.
And it has nothing to do with $\nabla$ at all.
Since the fundamental (left) ${\cal O}_{X^{\!A\!z}}$-module ${\cal E}$
   descends canonically to an ${\cal O}_{X_{\varphi}}$-module
      $ _{{\cal O}_{X_{\varphi}}}{\cal E}$,
 one would like $\nabla$ induces canonically a connection $^{\varphi}\nabla$
  on $ _{{\cal O}_{X_{\varphi}}}{\cal E}$ at least over an open-dense subset of $X$.
When that happens, the curvature $F_{\nabla}$ of $\nabla$
 should behave also as a tensor, possibly with singularity, on $X_{\varphi}$.
  
These together motivate us the following definition:

\bigskip

\begin{definition} {\bf [admissible pair $(\varphi,\nabla)$]}$\;$ {\rm
 Let $\varphi:(X^{\!A\!z},E)\rightarrow Y$ be a differentiable map  and
  $\nabla$ be a connection on $E$.
 The pair $(\varphi,\nabla)$ is called {\it admissible} if the following two conditions are satisfied
  over an open-dense subset of $X$:
  %
  %
   $$
    (1)\;\;\;
     DA_{\varphi}\;\subset\; C^{\infty}(\Omega_X)\otimes_{C^{\infty}(X)}A_{\varphi}
	  \hspace{2em}\mbox{and}\hspace{2em}
	(2)\;\;\;
	 F_{\nabla}\;
	    \subset\;  C^{\infty}(\Omega_X^2)\otimes_{C^{\infty}(X)} \Comm(A_{\varphi})\,.
   $$
 Here,
  $\Comm(A_{\varphi})$ is the commutant of $A_{\varphi}$
     in $C^{\infty}(\End_{\Bbb C}(E))$.
 For convenience, we say also that
  $\varphi$ is an {\it admissible map} from $(X^{\!A\!z},E;\nabla)$ to $Y$,   or that
  $\varphi:(X^{\!A\!z},E)\rightarrow Y$ is a map that is {\it admissible} to $\nabla$ on $E$,   or that
  $\nabla$ is a connection on $E$ that is {\it admissible} to $\varphi:(X^{\!A\!z},E)\rightarrow Y$.
}\end{definition}

\bigskip

Further illuminations of Definition~2.2.1 are given in the following two remarks:

\bigskip

\begin{remark} {\rm [\hspace{.1ex}}on
      Admissibility Condition $(1):$ generic uniformality of $X_{\varphi}$
	  over $X${\rm \hspace{.1ex}]}$\;$
{\rm
 The Admissible Condition (1) says that
  \begin{itemize}
   \item[\LARGE $\cdot$] {\it
    The commutative $C^{\infty}(X)$-subalgebra $A_{\varphi}$
         of $C^{\infty}(\End_{\Bbb C}(E))$
      is covariantly invariant under the induced connection $D$ on $\End_{\Bbb C }(E)$
	  over an open-dense subset $U$ of $X$.}
  \end{itemize}
  This defines an embedding
    $$
	   \Der (C^{\infty}(U))\; \subset\;  \Der(A_{\varphi}|_U)
	$$
	and, hence, a connection on $X_{\varphi}|_U$ over $U$.
  In terms of the above inclusion,
  \begin{itemize}
   \item[\LARGE $\cdot$] {\it
    One can associate to a tensor of type $(0,d)$ on $X_{\varphi}|_U$
        an $A_{\varphi}|_U$-valued tensor of the same type $(0,d)$ on $U$,
    which is then canonically
      an $C^{\infty}(\End_{\Bbb C}(E|_U))$-valued tensor of type $(0,d)$  on $U$
      through the localization of the built-in embedding
	  $A_{\varphi}\subset C^{\infty}(\End_{\Bbb C}(E))$ to over $U$.}
  \end{itemize}
  %
%
%
%
}\end{remark}

\medskip
 
\begin{remark} {\rm [\hspace{.1ex}}on Admissibility Condition $(2):$
       massless condition on $\nabla$ with respect to open strings{\rm \hspace{.1ex}]}$\;$ {\rm
 When $D$-brane is treated as a fundamental dynamical object,
  its interaction with open strings is through its image in the space-time $Y$.
 For simplicity,
  assume that
    $\pr_Y:X\times Y\rightarrow Y$ pushes $\tilde{\cal E}_{\varphi}$
	   to $\varphi_{\ast}(\cal E)$ isomorphically and, hence,
    $X_{\varphi}\simeq\Supp(\varphi_{\ast}({\cal E}))$.
 Then, if $\nabla$ is to be massless from the aspect of open strings moving in $Y$,
   $\nabla$ must be descendable to a connection $_{\varphi}\!\nabla$
    on $_{{\cal O}_{X_{\varphi}}}{\cal E}$ over $X_{\varphi}$.
 When that happens,
  its curvature $F_{\,_{\varphi}\!\nabla}$ becomes a $(0,2)$-tensor on $X_{\varphi}$	
   and, hence,
  takes values in the commutant $\Comm(A_{\varphi})$ of
   $A_{\varphi}$ in $C^{\infty}(\End_{\Bbb C}(E))$.
 When, in addition, the Admissible Condition (1) holds,
   to $F_{\,_{\varphi}\!\nabla}$ is associated a $\End_{\Bbb C}(E)$-valued $2$-form on $X$,
   which is nothing but $F_{\nabla}$.
 Thus, one has the Admissible Condition (2):
   $F_{\nabla}  \subset
     C^{\infty}(\Omega_X^2)\otimes_{C^{\infty}(X)} \Comm(A_{\varphi})$.
 This reasoning indicates  that
   \begin{itemize}
    \item[\LARGE $\cdot$] {\it
   Admissible Condition (2) has a concrete open-string-theoretical meaning of
   requiring $\nabla$ to be massless from the viewpoint of open strings in $Y$ via $\varphi$.}
   \end{itemize}
}\end{remark}
 
%
%
%
%

\bigskip

\subsection{Self-adjoint/Hermitian maps from an Azumaya/matrix manifold with a Hermitian fundamental module}

When the complex vector bundle $E$ over $X$ is Hermitian,
 i.e.\ $E$ is equipped with a smooth map
    $$
	  \langle\:,\:\rangle\; :\;  E\times _XE \;\longrightarrow\; {\Bbb C}
	$$
  that gives a Hermitian inner product on each fiber of $E$ over $X$, 	
one can require that
  the connection $\nabla$ under consideration be unitary with respet to the Hermitian structure
 $\langle\:,\:\rangle$.
This is a compatibility condition between
  the connection $\nabla$ on $E$ and the Hermitian metric $\langle\:,\:\rangle$ on $E$
  in the sense that
   the parallel transport defined by $\nabla$ of a  pair of elements in a fiber of $E$ along a path in $X$
   would then preserve their inner product under $\langle\:,\:\rangle$.

Very naturally, one may ask:
 \begin{itemize}
  \item[{\bf Q.}] \parbox[t]{37em}{\it
   Is there a condition on maps $\varphi:(X^{\!A\!z},E)\rightarrow Y$ one can impose as well
     so that such $\varphi$ can be thought of as being compatible with $\langle\:,\:\rangle$ in some sense?}
 \end{itemize}
In this subsection, we answer this question affirmatively.

\bigskip

\begin{flushleft}
{\bf The adjoint $\varphi^{\dagger}$ of $\varphi$ with respect to $\langle\:,\:\rangle$ on $E$}
\end{flushleft}
The Hermitian structure $\langle\:,\:\rangle : E\times _XE  \rightarrow{\Bbb C}$
 induces an anti-linear isomorphism $E\simeq E^{\vee}$, $v\mapsto \langle v,\,\cdot\,\rangle$,
  as smooth complex vector bundles on $X$
  and hence an anti-linear anti-isomorphism
   $\End_{\Bbb C}(E)\simeq \End_{\Bbb C}(E^{\vee})$,
   as ${\Bbb C}$-algebra bundles, given by $s\mapsto s^{\dagger}$
   with
   $\langle s^{\dagger}(v), w \rangle = \langle v, s(w) \rangle$.
 With respect  to a local trivialization of $E$ by unitary frames with respect to $\langle\:,\:\rangle$,
  the adjoint $s^{\dagger}$ of $s$ is the transpose of the complex-conjugate of $s$.
  
\bigskip

\begin{lemma}
 Let $\varphi:(X^{\!A\!z},E)\rightarrow Y$ be a differentiable map defined by a ring-homomorphism
   $\varphi^{\sharp}:C^{\infty}(Y)\rightarrow C^{\infty}(\End_{\Bbb C}(E))$
   over ${\Bbb R}\subset {\Bbb C}$.
 With the above anti-linear anti-isomorphism,
 consider the specification
  $$
   \begin{array}{cccccc}
   {\varphi^{\sharp}}^{\dagger} & :
        & C^{\infty}(Y)   & \longrightarrow    & C^{\infty}(\End_{\Bbb C}(E))  \\[1.2ex]
	 && f  	 & \longmapsto    & (\varphi^{\sharp}(f))^{\dagger}  &,	
   \end{array}
  $$
 where $(\varphi^{\sharp}(f))^{\dagger}$ is the adjoint of $\varphi^{\sharp}(f)$
  with respect to $\langle\:,\:\rangle$.
 Then,
  ${\varphi^{\sharp}}^{\dagger}$
   is a ring-homomorphism over ${\Bbb R}\subset {\Bbb C}$.
\end{lemma}
 
\medskip

\begin{proof}
 Since
    $(m_1m_2)^{\dagger}=m_2^{\dagger}m_1^{\dagger}$
	   for all $m_1, m_2 \in C^{\infty}(\End_{\Bbb C}(E))$,
  ${\varphi^{\sharp}}^{\dagger}$ is a ring-anti-homomorphism by nature.
 However, $C^{\infty}(Y)$ is commutative;
  this renders ${\varphi^{\sharp}}^{\dagger}$ a ring-homomorphism.
  
\end{proof}
 
\bigskip

It follows that
 ${\varphi^{\sharp}}^{\dagger}:
        C^{\infty}(Y)\rightarrow C^{\infty}(\End_{\Bbb C}(E))$
 defines a differentiable map
  $$
    \varphi^{\dagger}\;:\; (X^{\!A\!z},E)\;\longrightarrow\; Y\,.
  $$		
   
\bigskip

\begin{definition}  {\bf [adjoint/Hermitian conjugate of $\varphi$]}$\;$ {\rm
  The map $\varphi^{\dagger}:(X^{\!A\!z},E)\rightarrow Y$  thus defined
   is called the {\it adjoint}, or synonymously the {\it Hermitian conjugate}, of $\varphi$
   with respect to the Hermitian structure $\langle\:,\:\rangle$ on $E$.
}\end{definition}

\medskip

\begin{lemma} {\bf [basic properties of $\varphi^{\dagger}$]}$\;$
  $(1)$
    The support of the graph of $\varphi^{\dagger}$ and $\varphi$ are identical, i.e.\
     $\varGamma_{\varphi^{\dagger}}= \varGamma_{\varphi}$  as subschemes of $X\times Y$.
  $\;(2)$	
    The graph of $\varphi^{\dagger}$ and the graph of $\varphi$ differ by an antiisomorphism, i.e.\
	 $\tilde{\cal E}_{\varphi^{\dagger}}\simeq \overline{\tilde{\cal E}_{\varphi}}$
	   as ${\cal O}_{X\times Y}^{\,\Bbb C}$-modules.
	
   As a consequence, 	
	  $\Image \varphi =\Image \varphi^{\dagger}$ as subschemes of $Y$,  and
	  $\varphi^{\dagger}_{\;\ast}(\cal E)
	     \simeq \overline{\varphi_{\ast}(\cal E)}$ as ${\cal O}_Y^{\,\Bbb C}$-modules.		
\end{lemma}

\medskip

\begin{proof}
 Statement (1) follows from the observation that, as $C^{\infty}(X)$-algebras,
  $A_{\varphi}\simeq A_{\varphi^{\dagger}}$.
 Statement (2) follows from
   the observation that
     if two matrices $m_1$ and  $m_2$ have their eigenvalues all real and $m_2=m_1^{\dagger}$,
	 then, their Jordan form can be made identical,
  and
 the fact that, by construction,
  $A_{\varphi^{\dagger}}= (A_{\varphi})^{\dagger} $
   as subalgebras of $C^{\infty}(\End_{\Bbb C}(E))$.

\end{proof}

\bigskip

\begin{flushleft}
{\bf Self-adjoint/Hermitian maps}
\end{flushleft}
 Let $E$ be equipped with a Hermitian structure $\langle\:,\:\rangle$.
 
\begin{definition} {\bf [self-adjoint/Hermitian map]}$\;$ {\rm
 A differentiable map $\varphi:(X^{\!A\!z},E)\rightarrow Y$ is called
  {\it self-adjoint},
    or synonymously {\it Hermitian}, with respect to $\langle\:,\:\rangle$
   if $\,\varphi^{\dagger}=\varphi$.
}\end{definition}

\begin{lemma} {\bf [characterization by coordinate functions]}$\;$
 Let
   $y^1,\, \cdots\,, y^n$ be a set of coordinate functions of ${\Bbb R}^n$ as a $C^{\infty}$-manifold
   and $\varphi:(X^{\!A\!z},E)\rightarrow {\Bbb R}^n$ be a differentiable map.
 Then $\varphi$ is self-adjoint if and only if
   each of $\varphi^{\sharp}(y^i)$, $i=1,\,\ldots\,,n$, is Hermitian.
\end{lemma}
 
\begin{proof}
 We only need to prove the if-part.
 This follows from the proof of Proposition~2.1.6,
  cf.\ [L-Y6: Theorem 3.2.1]	(D(11.3.1)), reviewed in Sec.~4.1.
 In essence, as a consequence of the Malgrange Division Theorem,
  for any $f\in C^{\infty}(Y)$ and at the level of germs over $X$,
    $\varphi^{\sharp}(f)$ is expressible as a polynomial
	$P(\varphi^{\sharp}(y^1),\,\cdots\,,\, \varphi^{\sharp}(y^n))$
	with coefficients elements in a germ of smooth functions on $X$.
 The multi-degree of $P$ $\le (r-1,\,\cdots\,, r-1)$
   and the coefficients depend on $f$ and location of the germ on $X$.
 Since the addition and the multiplication of commuting Hermitian matrices remain Hermitian,
   $\varphi^{\sharp}(f)$ is Hermitian,
   i.e.\ $(\varphi^{\sharp}(f))^{\dagger}=\varphi^{\sharp}(f)$.
 As $f\in C^{\infty}(Y)$ is arbitrary,
   this implies that $(\varphi^{\sharp})^{\dagger}=\varphi^{\sharp}$.
 This concludes the proof.
   
\end{proof}

\medskip

\begin{lemma} {\bf [$X_{\varphi}$ generically reduced]}$\;$
 Let $\varphi:(X^{\!A\!z},E;\nabla)\rightarrow Y$ be a Hermitian map.
 Then $X_{\varphi}$ is generically reduced.
\end{lemma}

\medskip

\begin{proof}
 Commutating Hermitian matrices are simultaneously diagonalizable by a common unitary frame.
 Thus over each $p\in X$,
  the finite-dimensional ${\Bbb R}$-algebra $A_{\varphi}|_p$ has no nilpotent elements.
 The lemma follows.

\end{proof}

\medskip
  
\begin{remark} {\rm [\hspace{.1ex}}meaning of enhanced Lie-algrebra-valued massless spectrum
       on D-brane world-volume{\rm\hspace{.1ex}]}$\;$ {\rm
 Recall from [Po3], [Po4], [Wi] that for coincident D-brane world-volume $X$ of multiplicity $r$,
  the massless spectrum thereupon from excitations of oriented open strings consists of
  a $u(r)$-valued gauge field and a $u(r)$-valued scalar field.
 The former corresponds to $(E,\nabla)=$ a Hermitian vector bundle with a unitary connection,
  which describes the Spin-$1$ degrees of freedom on the D-brane
 while the latter corrsponds a self-adjoint map $\varphi:(X^{\!A\!z},E)\rightarrow Y$
   that describes the Spin-$0$ degrees of freedom on the D-brane.
 This gives a precise interpretation of the related paragraphs in the above work of Polchinski and Witten.
 Here, the Lie algebra $u(r)$ from the unitary group $U(r)$
  is identified with the fibers of
  the ${\Bbb R}$-vector subbundle $\SAd(E, \langle\,,\,\rangle)$ of $\End_{\Bbb C}(E)$
  that consists of self-ajoint endomorphisms on fibers of $(E,\langle\,,\,\rangle)$.
 While $\SAd(E)$ is {\it not} a bundle of rings,
  it makes sense to talk about {\it
   ring-homomorphisms
     $\varphi^{\sharp}:C^{\infty}(Y)\rightarrow C^{\infty}(\End_{\Bbb C})$
	 with values in $C^{\infty}(\SAd(E,\langle\,,\,\rangle))$}.
 They define precisely the self-adjoint maps in Definition~2.3.4.
}\end{remark}

\medskip

\begin{remark} {\rm[\hspace{.1ex}}admissible Hermitian map{\rm\hspace{.1ex}]}$\;$
 An admissible Hermitian map $\varphi:(X^{\!A\!z},E;\nabla)\rightarrow Y$
  has the special property that
    over some open-dense subset $U\subset X$,
	 $X_{\varphi}|_U$ is a covering space over $U$
     under the built-in map $X_{\varphi}|_U\rightarrow U$
   and that
    $E|_U$ has an orthogonal decomposition,
    	from the built-in isomorphism
		${\cal E}\simeq
		   {\pi_{\varphi}}_{\ast}(_{{\cal O}_{X_{\varphi}}}{\cal E})$,	
	with each summand covariantly invariant under $\nabla$.
 They should be studied in more detail.
 Cf.\ Remark~3.2.5 and Remark~5.3.3.
\end{remark}
 
\bigskip

\noindent
{\bf Convention.} {\rm
 For simplicity and a better focus on other issues that also occur,
  we'll assume for the rest of the work that
  the Admissibility Conditions (1) and (2) in Definition~2.1.1 apply to all over $X$.
}

\bigskip

\section{The Dirac-Born-Infeld action for differentiable maps from\\ Azumaya/matrix manifolds}

Recall from Sec.\ 1 Issues (2) -- (6) in the list one needs to understand to make sense of
 the formal expression of the Dirac-Born-Infeld action for stacked D-branes
 $$
   S_{\DBI}^{(\Phi,g,B)}(\varphi,\nabla)\;
    \stackrel{\mbox{\tiny formally }}{=}\;
     -T_{m-1}\,\int_X \Tr\left(\rule{0ex}{1em}\right.
	                       e^{-\varphi^{\ast}(\Phi)}
						    \sqrt{-\,
						     \Det_X(
							  \varphi^{\ast}(g+B)\,
						       +\, 2\pi\alpha^{\prime}F_{\nabla}
							    )\,}
	                          \left.\rule{0ex}{1em}\right)\,.
 $$
 \begin{itemize}
  \item[(2)]  [{\it push-pull of tensor under $\varphi$}$^{\,}$]\\
  The notion of {\it push-pulls under a differentiable map $\varphi$}
     from an Azumaya/matrix manifold with a fundamental module to a real manifold;
  cf.\ $\varphi^{\ast}(g+B)$.	
	
 \item[(3)] [{\it determinant of $2$-tensor on $X$}$^{\,}$]\\
  The notion of {\it determinant} $\Det_X(\,\cdots\,)$ in the current context
   once Issues (1) and (2) are resolved.
   
 \item[(4)] [{\it square root of matrix section}$^{\,}$]\\
 Can we take a {\it square root of a matrix}?
  When the answer is Yes, is there a unique square root?
  If not, which one to choose?
 Extension of this to {\it matrix sections}\hspace{.1ex}?
 
 \item[(5)]  [{\it dilaton-field factor}$^{\,}$]\\
   How does {\it the factor $e^{-\varphi^{\sharp}(\Phi)}$}
    influence the interpretation of the formal expression?
	
 \item[(6)] [{\it real-valuedness}$^{\,}$]\\
  Is the expression {\it real}$^{\,}$?
 \end{itemize}
 
We now proceed to resolve all these issues (Sec.\ 3.1) and, hence,
 construct the Dirac-Born-Infeld action for admissible pairs $(\varphi,\nabla)\;$ (Sec.\ 3.2).

\bigskip

\subsection{The resolution of issues toward defining the Dirac-Born-Infeld action}

In this subsection we resolve Issues (2) -- (6) in the list subsubsection  by subsubsection.

\bigskip
 
\subsubsection{The pull-back of tensors from the target space via commutative surrogates}

We introduce the notion of `pull-push' that works naturally for $(\varphi,\nabla)$ admissible
 and then discuss its basic properties and introduce its characteristic tensors.
For the purpose of the current work,
 all the tensors on $Y$ considered are of type $(0,\,\tinybullet\,)$,
  i.e.\ sections of $\otimes^{\tinybullet}\,T^{\ast}Y$.

\bigskip

\begin{flushleft}
{\bf The notion of `pull-push' }
\end{flushleft}
By itself, there is no problem to define the notion of tensors on an abstract ``space" associated to
 a general (unital, associative but not necessarily commutative) ring.
However, when compared  with the definition for the same on a space associated to a commutative ring,
 the latter is a quotient of the former with additional relators
  arising from the commutativity relation of the commutative ring.
Due to this,
 for a map between two spaces, with each associated to a ring,
 $$
   f\::\; \Space(R)\; \longrightarrow\; \Space(S)\,
 $$
 defined through a ring-homomorphism $f^{\sharp}:S\rightarrow R$,
 with $S$ commutative and $R$ noncommutative,
 there is no canonical/natural notion of a pull-back $f^{\ast}$ that takes tensors on $\Space(S)$
  to that on $\Space(R)$.
(See [L-Y4: Sec.\ 4.1] (D(11.1)), in particular, [ibidem: Example~4.1.20] for more details.)
This is what happens in our situation for a map $\varphi: (X^{\!A\!z},E)\rightarrow Y$,
 defined by a ring-homomorphism
  $\varphi^{\sharp}:C^{\infty}(Y)\rightarrow C^{\infty}(\End_{\Bbb C}(E))$
  over ${\Bbb R}\subset{\Bbb C}$.

On the other hand, while $X^{\!A\!z}$ is noncommutative,
 the surrogate $X_{\varphi}$ of $X^{\!A\!z}$ associated to $\varphi$ is commutative
  and fits into the following diagram that is canonically associated to $\varphi$
  $$
    \xymatrix{
	 X^{\!A\!z} \ar[rrrd]^-{\varphi} \ar@{->>}[d]   \\
	 X_{\varphi} \ar[rrr]_-{f_{\varphi}} \ar@{->>}[d]^-{\pi_{\varphi}}         &&& \;Y\;. \\
	 X
     }
  $$
The existing notion of pull-back of tensors can be applied to define the pull-back
 $f_{\varphi}^{\ast}\varXi$ on $X_{\varphi}$ of a tensor $\varXi$ on $Y$ under $f_{\varphi}$.
Furthermore,
 when $E$ is equipped with a connection $\nabla$  and the pair $(\varphi,\nabla)$ is admissible,
 $f_{\varphi}^{\ast}\varXi$ can be naturally realized as a $\End_{\Bbb C}(E)$-valued tensor
 ${\pi_{\varphi}}_{\ast}f_{\varphi}^{\ast}\varXi$ on $X$;
(cf.\ Sec.\ 2.2).
 
\bigskip

\begin{ssdefinition} {\bf [pull-push from $Y$ to $X$ by $\varphi$]}$\;$ {\rm
 Let $(\varphi,\nabla)$ be admissible and $\varXi$ a tensor on $Y$ as above.
 We will denote ${\pi_{\varphi}}_{\ast}f_{\varphi}^{\ast}\varXi$,
  which comes from the pull-push (i.e.\ first pulling back, then pushing forward)
    along the diagram associated to $\varphi$,
  by $\varphi^{\diamond}\varXi$ and call it the {\it pull-push} of the tensor $\varXi$ on $Y$
   to $X$ by $\varphi$.
}\end{ssdefinition}

\bigskip

In particular, for the $2$-tensors metric $g$ and $B$-field $B$ on $Y$,
 $$
   \varphi^{\diamond}(g+B)\; :=\; {\pi_{\varphi}}_{\ast}f_{\varphi}^{\ast}(g+B)
 $$
 is a well-defined $\End_{\Bbb C}(E)$-valued $2$-tensor on $X$.
This can then be added to a multiple of
  the curvature $2$-tensor $F_{\nabla}$ on $X$ of the connection $\nabla$ on $E$
 to give the $\End_{\Bbb C}(E)$-valued $2$-tensor
 $$
   \varphi^{\diamond}(g+B)\,+\, 2\pi\alpha^{\prime}\,F_{\nabla}
 $$
 on $X$.
This is what we will interpret the object
 ``\hspace{.1ex}$\varphi^{\ast}(g+B)+2\pi\alpha^{\prime}F_{\nabla}$\hspace{.1ex}"
 in the formal expression of the Dirac-Born-Infeld action
  $S_{\DBI}^{(\Phi,g,B)}(\varphi,\nabla)$, Sec.\ 1,   in our context.
 
\bigskip
 
This resolves Issue (2) in the list.

\bigskip

\begin{flushleft}
{\bf Basic properties and the characteristic tensor of the pull-push}
\end{flushleft}

\begin{lemma} {\bf [pull-push of symmetric tensor or alternating tensor]}$\;$
 Given an admissible pair $(\varphi,\nabla)$ as above, let $\varXi$ be a tensor,
  say of degree $l$, on $Y$.
 Then:
    \begin{itemize}
	 \item[$(1)$]
	   If $\varXi$ is a symmetric $l$-tensor on $Y$, then
	    $\varphi^{\diamond}\varXi$ is an $\End_{\Bbb C}(E)$-valued symmetric $l$-tensor on $X$.
    
      \item[$(2)$]
	   If $\varXi$ is an alternating $l$-tensor (i.e.\ an $l$-form) on $Y$, then
	    $\varphi^{\diamond}\varXi$ is an $\End_{\Bbb C}(E)$-valued alternating $l$-tensor
		  (i.e.\ an $\End_{\Bbb C}(E)$-valued $l$-form) on $X$.
    \end{itemize}
\end{lemma}

\smallskip

\begin{proof}
 The issue is local.
 Thus, for any $p\in X$, consider a small enough coordinate neighborhood $U$
  (with coordinate functions $\mbox{\boldmath $x$}=(x^1,\,\cdots\,,\,x^m)$)
  of $p$ such that
  $\varphi(U)$ lies in a local chart $V$ of $Y$
  (with coordinate functions $\mbox{\boldmath $y$}=(y^1,\,\cdots\,,\,y^n)$).
 Let
  $\varXi|_V
     =\sum_{i_1,\,\cdots\,,\, i_l}\,
	      \varXi_{i_1\,\cdots\,i_l}\,dy^{i_1}\otimes\,\cdots\,\otimes dy^{i_l}$.
 Recall the connection $D$ on $\End_{\Bbb C}(E)$ canonically induced by $\nabla$ on $E$.
 Denote $D_{\partial/\partial x^{\mu}}$, $\mu=1,\,\ldots\,,\, m$,  by $D_{\mu}$.
 Then, by definition,
 {\small
  $$
   \begin{array}{rll}
    (\varphi^{\diamond}\varXi)|_{U}
	     & = & \sum_{\mu_1,\,\cdots\,,\, \mu_l=1}^m \left(
		              \sum_{i_1,\,\cdots\,,\, i_l=1}^n
		                 \varphi^{\sharp}(\varXi_{i_1\,\cdots\,,i_l})
						   D_{\mu_1}\varphi^{\sharp}(y^{i_1})\,\cdots\,
						   D_{\mu_l}\varphi^{\sharp}(y^{i_l})   \right)
						    dx^{\mu_1}\otimes\,\cdots\,\otimes dx^{\mu_l}\\[2ex]
         & =: &  	 \sum_{\mu_1,\,\cdots\,,\, \mu_l=1}^m
		                   a_{\mu_1\,\cdots\,\mu_l}						
						    dx^{\mu_1}\otimes\,\cdots\,\otimes dx^{\mu_l}\,.
   \end{array}							
  $$}
 Note that since $(\varphi,\nabla)$ is admissible,
   namely $D_{\!\mbox{\LARGE $\cdot$}\,}A_{\varphi}\subset A_{\varphi}$,
  all the elements
   $$
      \varphi^{\sharp}(\varXi_{i_1\,\cdots\,,i_l})\,,\;\;
	   D_{\mu_1}\varphi^{\sharp}(y^{i_1})\,,\;\;   \cdots\,,\;\;
	   D_{\mu_l}\varphi^{\sharp}(y^{i_l})\;
       \in\;  A_{\varphi}|_U \; \subset \; C^{\infty}(\End_{\Bbb C}(E|_U))\,,
   $$
   with $i_1,\,\cdots\,,\, i_l =1,\,\ldots\,,\, n$, commute.
  
 If $\varXi$ is symmetric, then, for example,
   $$
    \begin{array}{rcl}
	 a_{\mu_2\mu_1\mu_3\,\cdots\,\mu_l}
	   & = &    \sum_{i_1,\,\cdots\,,\, i_l=1}^n
		                 \varphi^{\sharp}(\varXi_{i_1i_2i_3\,\cdots\,,i_l})
						   D_{\mu_2}\varphi^{\sharp}(y^{i_1})\,
						   D_{\mu_1}\varphi^{\sharp}(y^{i_2})\,
						   D_{\mu_3}\varphi^{\sharp}(y^{i_3})\,\cdots\,
						   D_{\mu_l}\varphi^{\sharp}(y^{i_l}) \\[2ex]
       & = &    \sum_{i_1,\,\cdots\,,\, i_l=1}^n
		                 \varphi^{\sharp}(\varXi_{i_1i_2i_3\,\cdots\,,i_l})
						   D_{\mu_1}\varphi^{\sharp}(y^{i_2})\,
						   D_{\mu_2}\varphi^{\sharp}(y^{i_1})\,
						   D_{\mu_3}\varphi^{\sharp}(y^{i_3})\,\cdots\,
						   D_{\mu_l}\varphi^{\sharp}(y^{i_l}) \\[2ex]							
       & = &    \sum_{i_1,\,\cdots\,,\, i_l=1}^n
		                 \varphi^{\sharp}(\varXi_{i_2i_1i_3\,\cdots\,,i_l})
						   D_{\mu_1}\varphi^{\sharp}(y^{i_2})\,
						   D_{\mu_2}\varphi^{\sharp}(y^{i_1})\,
						   D_{\mu_3}\varphi^{\sharp}(y^{i_3})\,\cdots\,
						   D_{\mu_l}\varphi^{\sharp}(y^{i_l})               \\[2ex]				
       & = &	 a_{\mu_1\mu_2\mu_3\,\cdots\,\mu_l}\,; 						 	
    \end{array}
   $$
   and similarly for other exchanges of indices of $a_{\mu_1\,\cdots\,\mu_l}$.
 This proves that $\varphi^{\diamond}\varXi$ is symmetric and concludes Statement $(1)$.

 Statement $(2)$ follows by a similar argument.
  
\end{proof}
 
\bigskip

For $(\varphi,\nabla)$ admissible, let $\varXi$ be a tensor of degree $l$ on $Y$.
Then, since $\varphi^{\diamond}\varXi$ is $A_{\varphi}$-valued,
 in any local expression of $\varphi^{\diamond}\varXi$ in terms of local coordinate functions,
  $$
    \begin{array}{c}
  	  \varphi^{\diamond}\varXi|_U\;
	     =\; \sum_{\mu_1,\,\cdots\,,\, \mu_l=1}^m
		                      a_{\mu_1\,\cdots\,\mu_l}						
			   			    dx^{\mu_1}\otimes\,\cdots\,\otimes dx^{\mu_l}\,,
    \end{array}							
  $$
  the coefficients
   $a_{\mu_1\,\cdots\,\mu_l}
       \in A_{\varphi}|_U  \subset C^{\infty}(\End_{\Bbb C}(E|_U))$,
  $\mu_1,\,\cdots\,,\,\mu_l=1,\,\ldots\,,\,m$,
  commute with each other  and, hence, at each $p\in U$, can be simultaneously triangulated:
  $$
    a_{\mu_1\,\cdots\,\mu_l} \;
	=\;     G_p\,\cdot\,  \left[
	             \begin{array}{ccc}
	                \lambda_{\mu_1\,\cdots\,\mu_l}^{(1)}(p)   & \ast & \ast   \\[1.2ex]
				 0    & \ddots            & \ast                                                                  \\[1.2ex]
				 0    & 0 &  \lambda_{\mu_1\,\cdots\,\mu_l}^{(r)}(p)
			     \end{array}	
                              	     \right]  \,\cdot\, G_p^{\;\;-1}\,,
  $$
  where $G_p\in \Aut(E|_p)$.
 The set of (ordinary) tensors  at $p\in X$
  $$
   \begin{array}{c}
     \Lambda_{\varphi^{\diamond}\varXi}(p) \;
	  :=\; \{\, \sum_{\mu_1,\,\cdots\,,\, \mu_l}\,
	                  \lambda_{\mu_1\,\cdots\,\mu_l}^{(s)}(p)
	                             (dx^{\mu_1}\otimes \,\cdots\,\otimes dx^{\mu_l})|_p \;|\; s=1,\,\ldots\,,r \,\}\;
      \subset\; (\otimes^l\,T^{\ast}X)|_p								
   \end{array}								
  $$
  is invariant under changes of coordinates on $X$  and the local trivializations of $E$.
 As $p$ varies, this defines
  a $r$-multi-section $\Lambda_{\varphi^{\diamond}\varXi}$ of $\otimes^l\,T^{\ast}X$.
  
\bigskip

\begin{ssdefinition} {\bf [characteristic tensor of pull-push]}$\;$ {\rm
 With some abuse of the word `tensor',
  the multi-section $\Lambda_{\varphi^{\diamond}\varXi}$ of $\otimes^l\,T^{\ast}X$ thus defined
  is called the {\it characteristic tensor} of the pull-push $\varphi^{\diamond}\varXi$ of an $l$-tensor
  $\varXi$ on $Y$ under an admissible map $\varphi:(X^{\!A\!z},E;\nabla)\rightarrow Y$.
}\end{ssdefinition}

\bigskip

The same proof as that of Lemma~3.1.1 gives:
 
\bigskip

\begin{sslemma} {\bf [characteristic tensor of pull-push of symmetric tensor or alternating tensor]}$\;$
 Continuing the setting in Lemma~3.1.1.
  %
  \begin{itemize}
   \item[$(1)$]
    If $\varXi$ is a symmetric tensor on $Y$,
	  then $\Lambda_{\varphi^{\diamond}\varXi}$ is a symmetric multi-valued tensor on $X$.
   									
   \item[$(2)$]
    If $\varXi$ is an alternating tensor on $Y$,
	 then $\Lambda_{\varphi^{\diamond}\varXi}$ is an alternating multi-valued tensor on $X$.
  \end{itemize}									
\end{sslemma}
 
\smallskip

\begin{ssremark} {\it $[\,\Lambda_{\varphi^{\diamond}\varXi}$
               from aspect of $C^{\infty}$-algebraic geometry$\,]$}$\;$ {\rm
 As a subobject in the total space (denoted the same) of $\otimes^l\,T^{\ast}X$,
  $\Lambda_{\varphi^{\diamond}\varXi}$ is a $C^{\infty}$-subscheme of $\otimes^l\,T^{\ast}X$
   that is algebraic and finite over $X$ of relative length $r$.
 Its detail can be complicated.
 For the current notes, we use only its pointwise property over $X$ in a few occasions.
}\end{ssremark}
  
\smallskip
 
\begin{ssexample} {\bf [pull-push of metric tensor on $Y$]}$\;$ {\rm
 Let
   $(Y,g)$ be either a Riemannian manifold or a Lorentzian manifold   and
  $\varphi: (X^{\!A\!z},E;\nabla)\rightarrow Y$ be an admissible map.
 Then $\varphi^{\diamond}g$ is an $\End_{\Bbb C}(E)$-valued symmetric $2$-tensor on $X$.
}\end{ssexample}
  
\smallskip

\begin{ssexample} {\bf [pull-push of $B$-field on $Y$]}$\;$ {\rm
 Continuing Example~3.1.1.5.
 Let $B$ be a $2$-form on $Y$.
 Then $\varphi^{\diamond}B$ is an $\End_{\Bbb C}(E)$-valued $2$-form on $X$.
}\end{ssexample}

\bigskip

\subsubsection{D-brane world-volume with constant induced-metric signature}

For a simple D-brane moving in a space-time $Y$,
 by definition it sweeps out a D-brane world-volume that has a Lorentzian induced metric.
This gives the simplest picture of D-brane world-volume:
 A Lorentzian submanifold (with a Chan-Paton bundle with a connection, ... ) in a space-time $Y$.
Now that
 we generalize the notion of a submanifold to the notion of differentiable map
   $\varphi:(X^{\!A\!z},E)\rightarrow Y$,
  a question arises immediately:
 \begin{itemize}
  \item[\bf Q.] \parbox[t]{37em}{\it
    Given a Lorentzian manifold $(Y,g)$,
	 in what sense can one say that $\varphi:(X^{\!A\!z},E)\rightarrow Y$
	  is Lorentzian (or equivalently, timelike), or spacelike, or null?}
 \end{itemize}
A geometrically reasonable, though naive, approach is to consider
 the $C^{\infty}$-subscheme
  $$
      \varGamma_{\varphi}\; :=\;  \Supp(\tilde{\cal E}_{\varphi})\; \subset\;  X\times Y\,,
  $$
   which is canonically isomorphic to the surrogate $X_{\varphi}$,
  and look at the restriction $(\pr_Y^{\ast}g)|_{\varGamma_{\varphi}}$ or
  $(\pr_Y^{\ast}g)|_{(\varGamma_{\varphi})_{red}}$.
 Here, $\pr_Y:X\times Y\rightarrow Y$ is the projection map.
In [L-Y4: Sec.\ 6.3 ] (D(11.1)), we took such an approach to define notions such as
   `{\sl Lagrangian maps}' to a symplectic manifold or `{\sl special Lagrangian maps}' to a Calabi-Yau manifold.
The setting is independent of the connection $\nabla$ on $E$.

However, in the course of understanding the Dirac-Born-Infeld action in our context,
 it turns out that,
  for an admissible $(\varphi,\nabla)$, the following definition is algebraically and technically more natural:
(cf.\ Lemma~3.1.4.7)
  
\bigskip

\begin{ssdefinition} {\bf [Lorentzian/timelike, spacelike, null map]}$\;$ {\rm
 Let $(Y,g)$ be a Lorentzian manifold (of signature $(-, +,\,\cdots\,,+)$).
 An admissible map $\varphi:(X^{\!A\!z},E;\nabla)\rightarrow Y$ is said to be
  {\it Lorentzian},, or equivalently {\it timelike}, (resp.\ {\it spacelike}, {\it null})
  if for any $p\in X$,
   each symmetric 2-tensor in the characteristic tensor-set $\Lambda_{\varphi^{\diamond}g}(p)$
    (cf.\ Definition~3.1.1.2)
	  defines a Minkowskian
	(resp.\ Euclidean, degenerate with signature $(0,+,\,\cdots\,,+)$) inner product on $T_pX$.
}\end{ssdefinition}
 
\smallskip
 
\begin{ssdefinition} {\bf [Riemannian map]}$\;$ {\rm
 Let $(Y,g)$ be a Riemannian manifold.
 An admissible map $\varphi:(X^{\!A\!z},E;\nabla)\rightarrow Y$ is said to be
  {\it Riemannian}
  if for any $p\in X$,
   each symmetric 2-tensor in the characteristic tensor-set $\Lambda_{\varphi^{\diamond}g}(p)$
   defines a Euclidean inner product on $T_pX$.
}\end{ssdefinition}
 
\bigskip

The relation, or discrepancy,
  between the setting following [L-Y4: Sec.\ 6.3] (D(11.1)) and the setting in the above two definitions
 should be investigated further.

\bigskip

\subsubsection{From determinant $\Det$ to symmetrized determinant $\SymDet$}

For comparison and motivation, we
 review first the defining properties of the determinant function $\Det$ over a commutative ring   and then
 generalize it to the notion of symmetrized determinant $\SymDet$ in the noncommutative case.
This is then applied to define the notion of symmetrized determinant of
 an $\End_{\Bbb C}(E)$-valued $2$-tensor on $X$.
 
\bigskip

\noindent
{\bf Convention [oriented manifold and compatible system of coordinate functions]}$\;$
 To have a globally well-define volumed form, rather than just a density or measure,
 for the rest of the notes, we assume:
  \begin{itemize}
   \item[\LARGE $\cdot$]
    Both $X$ and $Y$ are oriented manifolds.
	
   \item[\LARGE $\cdot$]	
    Whenever a system of local coordinate functions are chosen, e.g.
       $\mbox{\boldmath $x$}=(x^1,\,\cdots\,,\,x^m)$ for some local chart $U\subset X$ and
       $\mbox{\boldmath $y$}=(y^1,\,\cdots\,,\,y^n)$ for some local chart $V\subset Y$,
    the order of these functions  is chosen so that
      $dx^1\wedge\,\cdots\,\wedge dx^m$ specifies the orientation on $U$    and
	  $dy^1\wedge\,\cdots\,\wedge dy^n$ specifies the orientation on $V$.
  \end{itemize}

\bigskip

\begin{flushleft}
{\bf The determinant function $\Det$ over a commutative ring}
\end{flushleft}
We summarize the defining properties of the determinant function $\Det$ over a commutative ring
 into the following two definitions and theorem.
Readers are referred to [H-K: Chapter 5] for details.

\bigskip

\begin{ssdefinition} {\bf [multi-linear alternating function on matrices]}$\;$ {\rm
  Let
    $R$ be a commutative ring with the identity element $1$,
	$M_{l\times l }(R)$ be the ring of $l\times l$ matrices with entries in $R$.
  A function
   $$
     f\,:\; M_{l\times l}(R)\;\longrightarrow\; R
   $$
   is called {\it $l$-linear alternating}
  if
    \begin{itemize}
	 \item[\LARGE $\cdot$] [\hspace{.1ex}{\it $l$-linear}$\,$]$\;$
      For each $i$, $1\le i\le l$,
       $f$ is an $R$-linear function of the $i$-th row when the other $(l-1)$ rows are held fixed.
		
	 \item[\LARGE $\cdot$] [{\it alternating}$\,$]$\;$
      The following two conditions are satisfied:
	  \begin{itemize}
	   \item[\LARGE $\cdot$]
	    $f(m)=0$ whenever two rows of $m\in M_{l\times l}(R)$ are equal.
		
	   \item[\LARGE $\cdot$]
       If $m^{\prime}$ is obtained from $m\in M_{l\times l}(R)$ by interchanging two rows of $m$,\\
       then $f(m^{\prime})= - f(m)$.	   	
	  \end{itemize}	
	\end{itemize}	
}\end{ssdefinition}

\smallskip

\begin{ssdefinition} {\bf [determinant function]}$\;$ {\rm
 Continuing the setting of Definition~3.1.3.1.
 A function $f:M_{l\times l}(R)\rightarrow R$ is called a {\it determinant function}
  if $f$ is $l$-linear, alternating, and $f(\Id_{l\times l})=1$.
  Here, $\Id_{l\times l}$ is the identity matrix in $M_{l\times l}(R)$.
}\end{ssdefinition}

\smallskip

\begin{sstheorem} {\bf [existence and uniqueness of determinant function]}$\;$
 {\rm Continuing the setting of Definition~3.1.3.1.}
 There exists a unique determinant function $M_{l\times l}(R)\rightarrow R$.
 Denote this function by $\Det$.
 Then, for $m=(m_{ij})_{ij}\in M_{l\times l}(R)$,
  $$
    \Det(m)\;=\; \sum_{\sigma\in\scriptsizeSym_l}
	  (-1)^{\sigma}
	     m_{1\sigma(1)}\,\cdots\, m_{l\sigma(l)}\,.
  $$
 Here,
   $\Sym_l$ is the permutation group on $l$-many letters, and
   $(-1)^{\sigma}=1$ (resp.\ $-1$) if $\sigma$ is an even (resp.\ odd) permutation.
\end{sstheorem}

\bigskip

With the above review, the question now is:
  \begin{itemize}
   \item[\bf Q.] \parbox[t]{37em}{\it
     Can the above functorial definition of the determinant function $\Det$ be generalized to
	 the case where $R$ is noncommutative?}
  \end{itemize}

\vspace{6em}

\begin{flushleft}
{\bf The symmetrized determinant over a noncommutative ring}
\end{flushleft}
Let $R$ be an (associative, unital) ring that is not necessarily commutative.
Our goal now is to generalize the determinant function $\Det$ above for $R$ commutative to the current case.
When $R$ is noncommutative,
 one learns from experience that
  it is very restrictive to require a function $f:M_{l\times l}(R)\rightarrow R$ to be multi-$R$-linear
  and it is more practical to demand only that
  $f:M_{l\times l}(R)\rightarrow R$ be multi-$C(R)$-linear, where $C(R)$ is the center of $R$.
This motivates the following definition:

\bigskip

\begin{ssdefinition} {\bf [multi-central-linear alternating function on matrices]}$\;$ {\rm
  Let
    $R$ be a (unital associative) ring with the identity element $1$.
  Denote by $C(R)$ the center of $R$.	
	$M_{l\times l }(R)$ be the ring of $l\times l$ matrices with entries in $R$.
  A function
   $$
     f\,:\; M_{l\times l}(R)\;\longrightarrow\; R
   $$
   is called {\it $l$-central linear alternating}
  if
    \begin{itemize}
	 \item[\LARGE $\cdot$] [{\it $l$-central linear}$\,$]$\;$
      For each $i$, $1\le i\le l$,
       $f$ is an $C(R)$-linear function of the $i$-th row when the other $(l-1)$ rows are held fixed.
		
	 \item[\LARGE $\cdot$] [{\it alternating}$\,$]$\;$
      The following two conditions are satisfied:
	  \begin{itemize}
	   \item[\LARGE $\cdot$]
	    $f(m)=0$ whenever two rows of $m\in M_{l\times l}(R)$ are equal.
		
	   \item[\LARGE $\cdot$]
       If $m^{\prime}$ is obtained from $m\in M_{l\times l}(R)$ by interchanging two rows of $m$,\\
       then $f(m^{\prime})= - f(m)$.	   	
	  \end{itemize}	
	\end{itemize}	
}\end{ssdefinition}

\smallskip

\begin{ssdefinition} {\bf [determinant function -- noncommutative case]}$\;$ {\rm
 Continuing the setting of Definition~3.1.3.4.
 A function $f:M_{l\times l}(R)\rightarrow R$ is called a {\it determinant function}
  if $f$ is $l$-central linear, alternating, and $f(\Id_{l\times l})=1$.
  Here, $\Id_{l\times l}$ is the identity matrix in $M_{l\times l}(R)$.
}\end{ssdefinition}

\bigskip

The following definition and lemma answer the existence part of a determinant function
 in the noncommutative case:

\bigskip

\begin{ssdefinition} {\bf [symmetrized determinant]}$\;$ {\rm
 Let
   $R$ be a (unital associative) ring with the identity element $1$   and
   $M_{l\times l}(R)$ be the ring of $l\times l$ matrices with entries in $R$.
 Define the {\it symmetrized determinant function}
  $$
    \SymDet\;:\; M_{l\times l}(R)\; \longrightarrow\; R
  $$
  by the assignment to $m=(m_{ij})_{ij}\in M_{l\times l}(R)$ the following element in $R$
  $$
    \SymDet(m)\;
	  :=\;
	     \sum_{\sigma\in\scriptsizeSym_l}	  (-1)^{\sigma}\,
		     m_{1\sigma(1)}\odot\,\cdots\,\odot m_{l\sigma(l)}\,,	 		
  $$
  where
    $$
      r_1\odot\,\cdots\,\odot r_l\;
	    :=\; 	\frac{1}{l!}\,
		          \sum_{\sigma^{\prime}\in\scriptsizeSym_l}\,
		           r_{\sigma^{\prime}(1)}\,\cdots\,r_{\sigma^{\prime}(l)}
	$$
	is the {\it symmetrized product} of $r_1,\,\cdots\,, r_l\in R$.
 Here, $\Sym_l$ is the permutation group on $l$ letters.	
}\end{ssdefinition}
  
\bigskip

The lemma below justifies the name:

\bigskip
 
\begin{sslemma} {\bf [$\SymDet$ as generalization of $\Det$]}$\;$
 {\rm Continuing the setting in Definition~3.1.3.6.}
 The correspondence $\;\SymDet:M_{l\times l}(R)\rightarrow R\;$
   is a determinant function.
 Furthermore, when $R$ is commutative, $\SymDet$ and $\Det$ coincide.
\end{sslemma}

\medskip

\begin{proof}
 That $\SymDet$ is $l$-central linear and
  that $\SymDet(\Id_{l\times l})=1$
   are immediate.
 To show that $\SymDet$ is alternating,
  observe that
   $$
     r_1\odot \,\cdots\,\odot r_l\;=\;
	    r_{\sigma^{\prime}(1)}\,\odot\,\cdots\,\odot r_{\sigma^{\prime}(l)}
   $$
  for any $\sigma^{\prime}\in \Sym_l$.
 Consequently,
  the proof that $\SymDet$ is alternating follows exactly the same proof that $\Det$ is alternating
  since in the latter case only
     the commutativity of factors in the $l$-products in the expansion of $\Det(\,\cdot\,)$    and 
	 the sign $(-1)^{\sigma}=\pm 1$, $\sigma\in \Sym_l$, before the $l$-products
	 are used in the proof.
  That $\SymDet$ and $\Det$ coincide when $R$ is commutative is clear by the definition of $\SymDet$.
    
\end{proof}

\medskip

\begin{sslemma} {\bf [$\SymDet$ in terms of $\Det$]}$\;$
 {\rm Continuing the setting in Definition~3.1.3.6.}
  Let
    $$
	   m\;
	     =\; \left[ \begin{array}{c} m_{(1)}\\  \vdots  \\ m_{(l)} \end{array}  \right]\;
	     =\; [m_{(1)}^{\transpose},\,\cdots\,,\,m_{(l)}^{\transpose}]^{\transpose}
	$$
  be the presentation of an $l\times l$ matrix $m$ in terms of its row vectors
    $m_{(1)},\,\cdots\,,\,m_{(l)}$.	
 {\rm Here, $[\,\cdot\,]^{\transpose}$ denotes the transpose of a matrix $[\,\cdot\,]$.}
 Then,
   $$
     \SymDet(m)\;
	  =\;   \frac{1}{l!}
	              \sum_{\sigma\in\scriptsizeSym_l}
			       (-1)^{\sigma}
			        \Det([m_{\sigma(1)}^{\transpose},\,\cdots\,,\,
					            m_{(\sigma(l))}^{\transpose}]^{\transpose})\,,
   $$					
      where  we define
	     $\Det(m):=\sum_{\sigma\in\scriptsizeSym_l}
                                    (-1)^{\sigma}
                 m_{1\,\sigma(1)}\cdots m_{l\,\sigma(l)}$.	
\end{sslemma}

\medskip

\begin{proof}
 This follows directly from the definition of $\SymDet$ and $\Det$.

\end{proof}

\bigskip

Caution that for $R$ noncommutative,
  $\Det$ as defined is, in general, not a determinant function in the sense of Definition~3.1.3.5

\bigskip

\begin{ssremark} {\it $[\,$uniqueness$\,]$}$\;$ {\rm
 It is not clear to us whether
  the symmetrized determinant $\SymDet$ is the only determinat function
    (in the sense of Definition~3.1.3.5)
  that can be defined on $M_{l\times l}(R)$ for $R$ noncommutative.
}\end{ssremark}

\medskip

\begin{ssremark} {\it $[\,$on the altered ring $(R,+,\odot)\,]$}$\:$ {\rm
 Caution that we directly define the symmetrized product $r_1\odot\,\cdots\,\odot r_l$ for an $l$-tuple
  $(r_1,\,\cdots\,,\,r_l )$ of elements in $R$ in the definition of $\SymDet$,
  rather than building it up through a binary operation.
 This is all we need and used.
 The ring $(R,+, \odot)$ altered from the original $R$ is commutative and unital,
   but in general no longer associative: For example, the three products
 $$
   r_1\odot r_2\odot r_3\,,\;\;\;  (r_1\odot r_2)\odot r_3\,,\;\;\;  r_1\odot(r_2\odot r_3)
 $$
  in general are all different.
 In particular, a property or an identity related to $\Det(\,\cdot\,)$
   that relies only on the commutativity of the underlying ring
     automatically passes over to $\SymDet(\,\cdot\,)$ in the noncommutative case,
 while a property or an identity related to $\Det(\,\cdot\,)$ that involves also the associativity 
   of the underlying ring either fails or requires to be checked independently.
 For this reason, one does not have a simple formula that expresses $\SymDet(m)$
  in terms of  a row or column of $m$ and the corresponding $(l-1)\times (l-1)$ minors of $m$.
}\end{ssremark}

\bigskip

\begin{flushleft}
{\bf The symmetrized determinant $\SymDet_X(\Xi)$
          of an $\End_{\Bbb C}(E)$-valued $2$-tensor $\Xi$ on $X$}
\end{flushleft}
We are now ready to address the notion of `determinant' that appears in the formal expression of the Dirac-Born-Infeld action
 $S_{\DBI}^{(\Phi,g,B)}(\varphi,\nabla)$, Sec.\ 1.

\bigskip

\begin{ssansatz} {\bf [$\SymDet$ in the Dirac-Born-Infeld action]}$\;$ {\rm
 We interpret the determinant
   that appears in the formal expression of the Dirac-Born-Infeld action $S_{\DBI}^{(\Phi,g,B)}(\varphi,\nabla)$,
   Sec.\ 1, as the {\it symmetrized determinant} $\SymDet$
         that applies to an $\End_{\Bbb C}(E)$-valued $2$-tensor $\Xi$ on $X$.
}\end{ssansatz}
 
\bigskip

\noindent
We now explain the details of this determinant in our context of D-branes.

\bigskip
 
Let
 $\Xi\in  C^{\infty}( (T^{\ast}X)^{\otimes 2}\otimes_{\Bbb R}\End_{\Bbb C}(E))$
 be an $\End_{\Bbb C}(E)$-valued $2$-tensor on $X$.
Locally on a coordinate chart $U\subset X$ (with coordinates $(x^1,\,\cdots\,,\, x^m)$ )
 $\Xi$ has an expression of the form
 $$
    \sum_{\mu,\,\nu=1}^m\,\Xi_{\mu\nu} dx^{\mu}\otimes dx^{\nu}\,,
 $$
 with the coefficients
   in the (unital, associative) endomorphism ring (with the identity element $\Id_{r\times r}$):
  $$
    \Xi_{\mu\nu}\;  \in
     C^{\infty}(\End_{\Bbb C}(E|_U))\,.
  $$
The local coefficients form a $m\times m$ matrix, with the $(\mu,\nu)$-entry $\Xi_{\mu\nu}$:
 $$
    \hat{\Xi}_U\;:=\;
	  (\Xi_{\mu\nu})_{\mu\nu}\; \in\;
	   M_{m\times m}(C^{\infty}(\End_{\Bbb C}(E|_U)))\,.
 $$

\medskip
 
\begin{ssdefinition} {\bf [symmetrized determinant of $\End_{\Bbb C}(E)$-valued $2$-tensor]}$\;$
{\rm
 With the notation from above, the {\it symmetrized determinant}
  $$
    \SymDet_X(\Xi)\;\in\;
	     C^{\infty}(
		   (\mbox{$\bigwedge$}^mT^{\ast}X)^{\otimes 2}
 		        \otimes_{\Bbb R}\End_{\Bbb C}(E))
  $$
  of the $\End_{\Bbb C}(E)$-valued $2$-tensor $\Xi$ is defined to be
  the $\End_{\Bbb C}(E)$-valued $2m$-tensor on $X$, locally defined by
  $$
    \SymDet(\hat{\Xi}_U)\, (dx^1\wedge\,\cdots\,\wedge dx^m)^{\otimes 2}
  $$
  on a coordinate chart $U\subset X$ with coordinate functions $(x^1,\,\cdots\,,\,x^m)$.
}\end{ssdefinition}
  
\smallskip

\begin{sslemma} {\bf [well-definedness of $\SymDet(\Xi)$]}$\;$
 The symmetrized determinat $\SymDet(\Xi)$ of $\Xi$, as defined in Definition~3.1.3.12, is well-defined.
\end{sslemma}

\begin{proof}
 We only need to show that the local expressions of $\SymDet(\Xi)$ transform from one chart to another
  under the local coordinate transformation on $X$ and the accompanying local transition on $E$.
 This is a standard computation.
 Let $U_{\alpha}$ and $U_{\beta}$ be overlapping local charts on $X$ with coordinates
  $\mbox{\boldmath $x$}_{\alpha}:=(x_{\alpha}^1,\,\cdots\,,\, x_{\alpha}^m)$ and
   $\mbox{\boldmath $x$}_{\beta}
       :=(x_{\beta}^1,\,\cdots\,,\, x_{\beta}^m)$
	respectively.
 Let
  $$
    \begin{array}{ccccc}
	 \phi_{\alpha\beta} & : &
   	 (E|_{U_{\alpha}})|_{U_{\alpha}\cap U_{\beta}}  & \longrightarrow
	       & (E|_{U_{\beta}})|_{U_{\beta}\cap U_{\alpha}}\\[1.2ex]
     && (\mbox{\boldmath $x$}_{\alpha},\mbox{\boldmath $v$}_{\alpha})
   		  & \longmapsto
		  & (\mbox{\boldmath $x$}_{\beta}, \mbox{\boldmath $v$}_{\beta})\;
		       =\; (h_{\alpha\beta}(\mbox{\boldmath $x$}_{\alpha}),
			           \hat{h}_{\alpha\beta}(\mbox{\boldmath $x$}_{\alpha})
					     (\mbox{\boldmath $v$}_{\alpha})	)
	\end{array}
  $$
  be the transition map.
 Then the transition of local expressions of $\Xi$ is given by
  $$
   \begin{array}{rcl}
   \phi_{\alpha\beta}^{\ast}(\Xi_{\beta})
     &  =   &  \phi_{\alpha\beta}^{\ast}
	                 (\sum_{\mu^{\prime},\, \nu^{\prime}}
			                   \Xi_{\beta,\,\mu^{\prime}\nu^{\prime}}(\mbox{\boldmath $x$}_{\beta})
				                dx_{\beta}^{\mu^{\prime}}\otimes dx_{\beta}^{\nu^{\prime}})\\[1.2ex]
     &  =   & \sum_{\mu,\,\nu}\left(
	                \sum_{\mu^{\prime},\,\nu^{\prime}}
	                  \hat{h}_{\alpha\beta}(\mbox{\boldmath $x$}_{\alpha})^{-1}\,
					  \Xi_{\beta,\,\mu^{\prime}\nu^{\prime}}
					                       (h_{\alpha\beta}(\mbox{\boldmath $x$}_{\alpha}))\,
					  \hat{h}_{\alpha\beta}(\mbox{\boldmath $x$}_{\alpha})\,				
					    \frac{\partial h^{\mu^{\prime}}_{\alpha\beta}
						                          (\mbox{\scriptsize\boldmath $x$}_{\alpha})}
						         {\partial x_{\alpha}^{\mu}}\,
						\frac{\partial h^{\nu^{\prime}}_{\alpha\beta}
						                          (\mbox{\scriptsize\boldmath $x$}_{\alpha})}
						         {\partial x_{\alpha}^{\nu}} \right)
						dx_{\alpha}^{\mu}\otimes dx_{\alpha}^{\nu}   \\[2.4ex]
     &  =   & \sum_{\mu,\,\nu}\,
	                  \Xi_{\alpha,\,\mu\nu}(\mbox{\boldmath $x$}_{\alpha})\,
					   dx_{\alpha}^{\mu}\otimes dx_{\alpha}^{\nu}  \;.
    \end{array}	
  $$
 In terms of the local coefficient matrix of $\End_{\Bbb C}(E)$-valued $2$-tensors on $X$,
  this says that
  $$
   \begin{array}{c}
    \hat{\Xi}_{U_{\alpha}}\;
	 =\;    {\frac{\partial \mbox{\scriptsize\boldmath $x$}_{\beta} }
	                    {\partial\mbox{\scriptsize\boldmath $x$}_{\alpha}}}^T
  	                 \Ad_{\hat{h}_{\alpha\beta}^{\;\;-1}}(\hat{\Xi}_{U_{\beta}})\,
			      \frac{\partial \mbox{\scriptsize\boldmath $x$}_{\beta} }
	                    {\partial\mbox{\scriptsize\boldmath $x$}_{\alpha}}\,,		
   \end{array}	
  $$
 where
     the Adjoint $\Ad_{\hat{h}_{\alpha\beta}^{\;\;-1}}$
       acts on the $\End_{\Bbb C}(E|_{U_{\alpha}\cap U_{\beta}})$-valued entries of
	   $\hat{\Xi}_{U_{\beta}}$, and
     $\frac{\partial \mbox{\scriptsize\boldmath $x$}_{\beta} }
	                      {\partial\mbox{\scriptsize\boldmath $x$}_{\alpha}}$
       is the $m\times m$ Jacobian matrix of $h_{\alpha\beta}$	
	  with ${\frac{\partial \mbox{\scriptsize\boldmath $x$}_{\beta} }
	                    {\partial\mbox{\scriptsize\boldmath $x$}_{\alpha}}}^T$ its transpose. 	
 
 It follows that
  $$
    \begin{array}{l}
	  \hspace{-2em}
      \SymDet(\hat{\Xi}_{U_{\alpha}})\,
	         (dx^1\wedge\,\cdots\,\wedge dx^m)^{\otimes 2}                              \\[1.2ex]
	   = \; \SymDet\left(
	           {\frac{\partial \mbox{\scriptsize\boldmath $x$}_{\beta} }
	                    {\partial\mbox{\scriptsize\boldmath $x$}_{\alpha}}}^T
  	                 \Ad_{\hat{h}_{\alpha\beta}^{\;\;-1}}(\hat{\Xi}_{U_{\beta}})\,
			      \frac{\partial \mbox{\scriptsize\boldmath $x$}_{\beta} }
	                    {\partial\mbox{\scriptsize\boldmath $x$}_{\alpha}}
	                            \right)	
            (dx^1\wedge\,\cdots\,\wedge dx^m)^{\otimes 2}								\\[2.4ex]
	  =\;  \Ad_{\hat{h}_{\alpha\beta}^{\;\;-1}}
	                 \left(\SymDet(\hat{\Xi}_{U_{\beta}})\right)\,
					(\Det(\frac{\partial\mbox{\scriptsize\boldmath $x$}_{\beta} }
			                {\partial\mbox{\scriptsize\boldmath $x$}_{\alpha}}))^2\,
             (dx^1\wedge\,\cdots\,\wedge dx^m)^{\otimes 2}                               \\[2.4ex]
       =\;  \phi_{\alpha\beta}^{\ast}
	           \left(  \SymDet(\hat{\Xi}_{U_{\beta}})\,
			                  (dx_{\beta}^1\wedge\,\cdots\,\wedge dx_{\beta}^m)^{\otimes 2}
				 \right)\,.
    \end{array}
  $$
 In other words,
  the collection
   $$
      \left\{ \SymDet(\hat{\Xi}_{U_{\alpha}})\,
                  (dx_{\alpha}^1\wedge\,\cdots\,\wedge dx_{\alpha}^m)^{\otimes 2})
				 \right\}_{\alpha}
   $$
   of local sections of
    $(\bigwedge^mT^{\ast}X)^{\otimes 2}\otimes_{\Bbb R}\End_{\Bbb C}(E)$
   glue to a global section of 	
    $(\bigwedge^mT^{\ast}X)^{\otimes 2}\otimes_{\Bbb R}\End_{\Bbb C}(E)$.
 This concludes the lemma.
  
\end{proof}

\bigskip

This resolves Issue (3) in the list.

\bigskip

\subsubsection{Square roots of sections of
  $(\bigwedge^mT^{\ast}X)^{\otimes 2}\otimes_{\Bbb R}\End_{\Bbb C}(E)$}

We give first a general study of square roots of matrices in $M_{r\times r}(\Bbb C)$
 and then apply it to understand the square roots of sections of
  $(\bigwedge^mT^{\ast}X)^{\otimes 2}\otimes_{\Bbb R}\End_{\Bbb C}(E)$.

\vspace{6em}
  
\begin{flushleft}
{\bf Square roots of matrices in $M_{r\times r}({\Bbb C})$}
\end{flushleft}
Note first that for an arbitrary $r\times r$ matrix $m\in M_{r\times r}({\Bbb C})$,
 there may not be an $m^{\prime}\in M_{r\times r}({\Bbb C})$
 that satisfies $(m^{\prime})^2=m$.
In other words, a square root of $m$ may not exist.
This is illustrated by the following example:
 
\bigskip

\begin{ssexample} {\bf [matrix with no square root]}$\;$ {\rm
 Let $r\ge 2$ and $m\in M_{r\times r}({\Bbb C})$ be a nilpotent matrix of nilpotency $r$.
 If  a square root $m^{\prime}$ of $m$ exists, then $m^{\prime}$ must also be nilpotent, of nilpotency $\le r$.
 But this implies in turn that the nilpotency of $m=(m^{\prime})^2$ must be strictly less than $r$, 
 which is a contradiction.
}\end{ssexample}

\bigskip

On the other hand, one has the following affirmative situation:

\bigskip

\begin{ssexample} {\bf [neighborhood of diagonalizable matrices with nonzero eigenvalues]}$\;$ {\rm
 First notice that a diagonalizable matrix in $M_{r\times r}({\Bbb C})$ with all its eigenvalues nonzero
  has $2^r$-many square roots.
 Using the Implicit Function Theorem, one can show then that
  any matrix in a small enough neighborhood of such a matrix in $M_{r\times r}({\Bbb C})$
    has also $2^r$-many square roots.
}\end{ssexample}
 
\bigskip

Indeed, motivated by
  how a commutative subalgebra of $M_{r\times r}({\Bbb C})$ is canonically a $C^{\infty}$-ring
  (cf.\ [L-Y6: Sec.\ 2] (D(11.3.1))),
one can prove a stronger result than Example~3.1.4.2:

\bigskip

\begin{sslemma} {\bf [existence of $2^r$-many square roots of invertible matrix]}$\;$
 Let $\GL_r({\Bbb C})=\{m\,|\,\determinant m\ne 0 \}$
  be the open subset of $M_{r\times r}({\Bbb C})$ that consists of invertible $r\times r$ matrices, 
   with the subset topology from the isomorphism
   $M_{r\times r}({\Bbb C})\simeq {\Bbb C}^{\oplus r^2}$ as ${\Bbb C}$ vector spaces.
 Then,
  $$
   \begin{array}{ccccc}
     \varUpsilon &:  & \GL_r({\Bbb C}) & \longrightarrow  & \GL_r({\Bbb C})   \\[1.2ex]
              && m^{\prime}   & \longmapsto   &  (m^{\prime})^2
   \end{array}
  $$
  is a covering map of degree $2^r$.
 It follows that for $m\in \GL_r({\Bbb C})$, $m$ has exactly $2^r$-many distinct square roots.
\end{sslemma}
 
\begin{proof}
 Let $m\in\GL_r({\Bbb C})$.
 We construct first $2^r$-many local inverses $\varUpsilon^{-1}(m)$ to $\varUpsilon$ at $m$ as follows.
 Let
   $\upsilon:{\Bbb C}\rightarrow {\Bbb C}$ be the map $z\mapsto z^2$,
     $z_0\in{\Bbb C}-\{0\}$,   and
	 $\sqrt{\upsilon\,}$ be either of the square root of $\upsilon$,
	     defined and analytic on a simply connected region
	    $\Omega\subset {\Bbb C}-\{0\}$ that contains $z_0.$	
   Consider the Taylor expansion of $\sqrt{\upsilon\,}$ at $z_0$ for $z$ with $|z-z_0|\le |z_0|$:
    $$
	  \sqrt{\upsilon\,}(z)\;=\;
	     \sum_{l=0}^{r-1}
		    \frac{(-1)^{l-1}(2l-3)!!}{l!\cdot 2^l}\,\frac{\sqrt{z_0}}{z_0^l}\cdot    (z-z_0)^l\,
 			+\, O((z-z_0)^r)\,,
	$$
   where
     $(2l-3)!!:= \prod_{i=1}^{l-1}(2i-1)$ for $l\ge 2$,   and
     $(-3)!!=-1$, $(-1)!!=1$ by convention.
   Let $m=G_mJ_mG_m^{\,-1}$, where $G_m\in  \GL_r({\Bbb C})$
     and $J_m$ is the Jordan form of $m$.
   Then $J_m=D_m+N_m$,
      where $D_m$ is diagonal and $N_m$ is upper triangular and nilpotent, such that $D_mN_m=N_mD_m$.
   In terms of this,
    $$
      \sqrt{\upsilon\,}(m)\;=\;
	     G_m\,
		   \left(
             \sum_{l=0}^{r-1}
		     \frac{(-1)^{l-1}(2l-3)!!}{l!\cdot 2^l}\,\frac{\sqrt{D_m}}{D_m^{\, l}}
		     N_m^{\,l}
		   \right)\,   G_m^{\,-1}\,.
    $$	
  This reduce the problem of defining $\sqrt{\upsilon\,}(m)$ to the existence of $\sqrt{D_m}$.
  The latter holds, since $D_m$ is diagonal, and has $2^r$-many choices.
  This says that $m$ has $2^r$-many inverses, counted with multiplicity, under $\varUpsilon$.
  
 Since $m\in\GL_r({\Bbb C})$,
   all the diagonal entries in the diagonal matrix $D_m$ is non-zero.
 Thus, all these inverses must be simple (i.e.\ distinct of multiplicity $1$).
 It follows that the construction can be extended to a small enough neighborhood of $m\in GL_r({\Bbb C})$
  to define $2^r$-many distinct local inverses to $\varUpsilon$ around $m$.
 This concludes the lemma.
   
\end{proof}
 
\smallskip

\begin{ssdefinition} {\bf [principal square root]}$\;$ {\rm
 (1) For a daigonal matrix $D\in M_{r\times r}({\Bbb C})$ with all the diagonal entries positive,
            we define the {\it principal square root} of $D$, in notation $\sqrt{D}$,
			   to be the unique square root of $D$ that has all the dagonal entries positive as well.
 (2) For $m\in \GL_r({\Bbb C})$
                that lies in a small enough neighborhood of the conjugacy class
				   of a diagonal matrix in $\GL_r({\Bbb C})$ with all the diagonal entries positive,
            we define the {\it principal square root} of $m$, in notation $\sqrt{m}$,	
			   to be the unique square root of $m$	
			   that has all its eigenvalues $\lambda_i$ satisfying $\Real\lambda_i>0$.            			  
}\end{ssdefinition}

\bigskip

\begin{flushleft}
{\bf Principal square root of elements in $A_{\varphi}$}.
\end{flushleft}
Let $\varphi: (X^{\!A\!z},E)\rightarrow Y$ be a differentiable map,
  defined by a ring-homomorphism
   $\varphi^{\sharp}:C^{\infty}(Y)\rightarrow C^{\infty}(\End_{\Bbb C}(E))$
   over ${\Bbb R}\subset {\Bbb C}$..
As an intermediate step, consider the notion of `square roots' of elements in
 $A_{\varphi}:= C^{\infty}(X)\langle \Image(\varphi^{\sharp})\rangle$.
Through
  the built-in inclusion $A_{\varphi}\subset C^{\infty}(\End_{\Bbb C}(E))$
  and the study of the previous theme `{\sl Square roots of matrices in $M_{r\times r}({\Bbb C})$}',
 one learns that
   an element $s\in A_{\varphi}$ may not have a square root in $A_{\varphi}$;
    namely, there may be no element $s^{\prime}\in A_{\varphi}$ such that $(s^{\prime})^2=s$.
However, from the proof of Lemma~3.1.4.3,
  one learns that fiberwise over $p$, if a principal square root of $s(p)\in \End_{\Bbb C}(E|_p)$ exists,
    it comes from the $C^{\infty}$-ring structure of $A_{\varphi}|_p$.
It follows that for $s\in A_{\varphi}$,
 if there is an $s^{\prime}\in C^{\infty}(\End_{\Bbb C}(E))$
 such that for all $p\in X$,  $s^{\prime}(p)$ is the principal square root of $s$,
 then $s^{\prime}$ must lie in $A_{\varphi}$.

\bigskip

\begin{ssdefinition} {\bf [principal square root of element in $A_{\varphi}$]}$\;$ {\rm
  An $s^{\prime}\in A _{\varphi}$ is the {\it principal square root} of $s\in A_{\varphi}$
   if for all $p\in X$, $s^{\prime}(p)$ is the principal square root of $s(p)$.
}\end{ssdefinition}
  
\smallskip
 
\begin{sslemma} {\bf [criterion for existence of principal square root]}$\;$
 Let $s\in A_{\varphi}\subset \End_{\Bbb C}(E)$.
 Then,
   the principal square root $\sqrt{s\,}$ of $s$ exists in $A_{\varphi}$
    if and only if, for all $p\in X$, the principal square root $\sqrt{s(p)\,}\in \End_{\Bbb C}(E|_p)$ exists.
\end{sslemma}

\medskip

\begin{proof}
 We only neeed to show the if-part.
 For that, one only needs to show
   that the correspondence $X\rightarrow \End_{\Bbb C}(E)$ with $p\mapsto \sqrt{s(p)\,}$ is smooth.
 The latter follows from the observation
   that in the current situation
     $s \in C^{\infty}(\Aut_{\Bbb C}(E)) \subset C^{\infty}(\End_{\Bbb C}(E))$   and
   that the  square map $\Aut_{\Bbb C}(E)\rightarrow \Aut_{\Bbb C}(E)$, $h\mapsto h^2$,
     over $X$ is a smooth covering map and, hence, its local inverse, which takes in particular $s$ to $\sqrt{s}$,
	 must be a diffeomorphism.
   
\end{proof}

\bigskip

\begin{flushleft}
{\bf Square roots of sections of
          $(\bigwedge^mT^{\ast}X)^{\otimes 2}\otimes_{\Bbb R}\End_{\Bbb C}(E)$}
\end{flushleft}
It follows from the previous theme `{\sl Square roots of matrices in $M_{r\times r}({\Bbb C})$}'
 that for a general section $s$ of
 $(\bigwedge^mT^{\ast}X)^{\otimes 2}\otimes_{\Bbb R}\End_{\Bbb C}(E)$,
 the square root of $s$ is only a rational multi-section of
 $\bigwedge^mT^{\ast}X\otimes_{\Bbb R}\End_{\Bbb C}(E)$,
 defined on an open subset of $X$.

\bigskip

\begin{sslemma} {\bf [principal square root of symmetrized determinant of pull-push of metric tensor]}$\;$
 $(1)$
  Let
   $(Y,g)$ be a Lorentzian $n$-manifold   and
   $\varphi:(X^{\!A\!z},E;\nabla)\rightarrow Y$ be a Lorentzian admissible map.
  Then
   $$
     -\,\SymDet_X(\varphi^{\diamond}g)\;
	        \in\;   C^{\infty}((\mbox{$\bigwedge$}^mT^{\ast}X)^{\otimes 2}
	                                          \otimes_{\Bbb R}\End_{\Bbb C}(E))
   $$
   has a well-defined principal square root
   $$
     \sqrt{
         -\, \SymDet_X(\varphi^{\diamond}g)\,}\;
		    \in\;   C^{\infty}
			            (\mbox{$\bigwedge$}T^{\ast}X\otimes_{\Bbb R}\End_{\Bbb C}(E))\,.
   $$
 $(2)$
  Let either
    $(Y,g)$ be a Lorentzian $n$-manifold   and
       $\varphi:(X^{\!A\!z},E;\nabla)\rightarrow Y$ be a spacelike admissible map,
  or $(Y,g)$ be a Riemannian $n$-manifold   and
       $\varphi:(X^{\!A\!z},E;\nabla)\rightarrow Y$ be a Riemannian admissible map.
  Then	
   $$
       \SymDet_X(\varphi^{\diamond}g)\;
	        \in\;   C^{\infty}((\mbox{$\bigwedge$}^mT^{\ast}X)^{\otimes 2}
	                                          \otimes_{\Bbb R}\End_{\Bbb C}(E))
   $$
   has a well-defined principal square root
   $$
     \sqrt{
        \SymDet_X(\varphi^{\diamond}g)\,}\;
		    \in\;   C^{\infty}
			            (\mbox{$\bigwedge$}^mT^{\ast}X \otimes_{\Bbb R}\End_{\Bbb C}(E))\,.
   $$
\end{sslemma}

\medskip

\begin{proof}
 Let $\Aut_{\Bbb C}(E)\subset \End_{\Bbb C}(E)$ be the automorphism bundle of $E$.
 Then it follows from Lemma~3.1.4.3
   that the map
      $\Aut_{\Bbb C}(E)\rightarrow \Aut_{\Bbb C}(E)$, $h\mapsto h^2$,
	  over $X$
    is a covering map of degree $2^r$, where is the rank of $E$ as a complex vector bundle over $X$.
 It follows that
  as a long as a section $s$  in
   $(\mbox{$\bigwedge$}^mT^{\ast}X)^{\otimes 2}
	                                          \otimes_{\Bbb R}\End_{\Bbb C}(E)$
   lies in the open subset
       $(\mbox{$\bigwedge$}^mT^{\ast}X)^{\otimes 2}
	                                          \otimes_{\Bbb R}\Aut_{\Bbb C}(E)$
   and for each point $p\in X$, the principal square root $\sqrt{s(p)\,}$ exists,
 then $\sqrt{s}$ exists as a smooth section of
   $\mbox{$\bigwedge$}^mT^{\ast}X \otimes_{\Bbb R}\Aut_{\Bbb C}(E)
      \subset  \mbox{$\bigwedge$}^mT^{\ast}X
	                                          \otimes_{\Bbb R}\End_{\Bbb C}(E)$.
	
 For Statement (1),
  for each $p\in X$, let
   $$
     \lambda^{(1)},\,\cdots\,,\,\lambda^{(r)} \;
	   \in\;  \Lambda_{\varphi^{\diamond}g}(p)
   $$
  give the characteristic tensor of $\varphi^{\diamond}g$ over $p$.
 Counted with multilicity, they defines $r$-many inner products on $T_pX$.
 By construction,
  $$
    \SymDet((\varphi^{\diamond}g)(p))\;
	 =\; \Det((\varphi^{\diamond}g)(p))\;
	 =\; G_p\cdot \left[
	           \begin{array}{ccc}
			     \Det(\lambda^{(1)})  & \ast   & \ast   \\
				    0   &  \ddots   & \ast                                           \\[1.2ex]
					0   &  0   & \Det(\lambda^{(r)})			
	           \end{array}
	                     \right]
	         \cdot G_p^{\;\;-1}\,,
  $$
  for some $G_p\in \Aut(E|_p)$.
 Since $\varphi$ is Lorentzian, each $\lambda^{i}$ defines
  a (non-degenerate) Minkowskian inner product on $T_pX$, for $i=1,\,\ldots\,, r$.
 It follows that the proof of Lemma~3.1.4.3 that
  $$
      -\, \SymDet((\varphi^{\diamond}g)(p))\;
	 =\; -\, \Det((\varphi^{\diamond}g)(p))\;
	 =\; G_p\cdot \left[
	           \begin{array}{ccc}
			     -\, \Det(\lambda^{(1)})  & -\, \ast   & -\, \ast   \\
				    0   & -\, \ddots   & -\, \ast                                           \\[1.2ex]
					0   &  0   & -\, \Det(\lambda^{(r)})			
	           \end{array}
	                     \right]
	         \cdot G_p^{\;\;-1}
  $$
  admits a principal square root
   $\sqrt{-\,\SymDet(\varphi^{\diamond}g)\,}$ of the form
   $$
      \sqrt{-\, \SymDet((\varphi^{\diamond}g)(p))\,}
	 =\; G_p\cdot \left[
	           \begin{array}{ccc}
			     \sqrt{-\, \Det(\lambda^{(1)})\,}  &  \ast^{\prime}   & \ast^{\prime}  \\
				    0   &  \ddots   &   \ast^{\prime}                                           \\[1.2ex]
					0   &  0   &  \sqrt{-\, \Det(\lambda^{(r)})\,}			
	           \end{array}
	                     \right]
	         \cdot G_p^{\;\;-1}\,.
   $$
 This proves Statement (1)
   
 Statement (2) is proved by a similar argument.

\end{proof}

\bigskip

It follows that, for a Lorentzian map $\varphi:(X^{\!A\!z},E;\nabla)\rightarrow (Y,g,B)$,
 \begin{itemize}
  \item[\LARGE $\cdot$] {\it
   If $B$ and $F_{\nabla}$ are small,
   the $\End_{\Bbb C}(E)$-valued $2$-tensor
     $\;-\SymDet_X(\varphi^{\diamond}(g+B)+2\pi\alpha^{\prime}F_{\nabla})$,
      now regarded as from a deformation of $\;-\SymDet_X(\varphi^{\diamond}g)$,
    has a well-defined principal square root
   $\;\sqrt{-\SymDet(\varphi^{\diamond}(g+B)+2\pi\alpha^{\prime}F_{\nabla})\,}$.}
 \end{itemize}
Similarly, for the other two situations.
 
\medskip
 	
This resolves Issue (4) in the list.
 
\bigskip

\begin{ssremark} {\it $[\,$Where the tensors take their value$\,]$}$\;$ {\rm
 Let $\varXi$ be a $2$-tensor on $Y$.
 By construction,
  both $\varphi^{\diamond}\varXi$ and, hence, its symmetrized determinant
    $\SymDet_X(\varphi^{\diamond}\varXi$) are $A_{\varphi}$-valued tensors on $X$.
 It follows from Lemma~3.1.4.6
 and the construction of $\sqrt{\pm \SymDet_X(\varphi^{\diamond}\varXi)}$ that
  if the principal square root $\sqrt{\pm \SymDet_X(\varphi^{\diamond}\varXi)}$
       exists as an $\End_{\Bbb C}(E)$-valued tensor on $X$,
  then it must be indeed $A_{\varphi}$-valued.
 When in addition $F_{\nabla}$ is taken into account,
 in general
  $\varphi^{\diamond}\varXi+2\pi\alpha^{\prime}F_{\nabla}$,
   and, hence,
      $\SymDet_X(\varphi^{\diamond}\varXi+2\pi\alpha^{\prime}F_{\nabla})$     and
	  $\sqrt{\pm \SymDet_X(\varphi^{\diamond}\varXi+2\pi\alpha^{\prime}F_{\nabla})\,}$
      (if defined)
   are only $\End_{\Bbb C}(E)$-valued.
 This applies when $\varXi=g$, $B$, or $g+B$.
}\end{ssremark}

\bigskip
 
\subsubsection{The factor from the dilaton field $\Phi$ on the target space(-time)}

The dilaton field $\Phi$ is a scalar field on $Y$.
We will take $\Phi$ as smooth, i.e.\ $\Phi \in C^{\infty}(Y)$.
Then, by the definition of pull-push under $\varphi:(X^{\!A\!z},E)\rightarrow Y$,
 $$
   \varphi^{\diamond}\Phi\;  =\; \varphi^{\sharp}(\Phi)
    \hspace{2em}\mbox{and}\hspace{2em}
   e^{-{\varphi^{\diamond}\Phi}}\;
      =\; \varphi^{\sharp}(e^{-\Phi})\;
	  =\;  e^{-\varphi^{\sharp}(\Phi)} \;\;\;
   \in\;  A_{\varphi}\,.
 $$
 Here,
  $e^{-\varphi^{\sharp}(\Phi)}$ is defined
   through the $C^{\infty}$-ring structure of $A_{\varphi}$.
 As noted in Remark~3.1.4.8,
  $\sqrt{-\SymDet_X(\varphi^{\diamond}(g+B))\,}$ is $A_{\varphi}$-valued,
   and, hence,
  \begin{itemize}
   \item[\LARGE $\cdot$] {\it
     The factor $e^{-\varphi^{\diamond}\Phi}$ and
	 the principal square root
	   $\sqrt{-\SymDet_X(\varphi^{\diamond}(g+B))\,}$
	 commute.}
  \end{itemize}
Once the gauge curvature $F_{\nabla}$ is also taken into account,
   $e^{-\varphi^{\diamond}\Phi}$ and\\ 	
	   $\sqrt{-\SymDet_X(\varphi^{\diamond}(g+B)+2\pi\alpha^{\prime}F_{\nabla})\,}\;$
	 may not commute for a general $F_{\nabla}$.

\bigskip

\begin{sslemma} {\bf [commutativity with dilaton factor]}$\;$
 Let $(Y,g)$ be Lorentzian and $\varphi$ is admissible to $\nabla$ and Lorentzian.
 Then
   $e^{-\varphi^{\diamond}\Phi}$ and
	   $\sqrt{-\SymDet_X(\varphi^{\diamond}(g+B)+2\pi\alpha^{\prime}F_{\nabla})\,}\;$
   commute.
 Similarly, for
   the case $(Y,g)$ Lorentzian and $\varphi$ admissible and spacelike,    and
   the case $(Y,g)$  Riemannian and $\varphi$ admissible and Riemannian.
\end{sslemma}
	
\begin{proof}
 Since this a pointwise issue over $X$,
 we only need to prove the following statement:
  \begin{itemize}
   \item[\LARGE $\cdot$] {\it
    Let $m_0, \, m_1\in M_{r\times r}({\Bbb C})$ commute.
	Assume that $\Det(m_1)\ne 0$ and let $\sqrt{m_1}$ be any of its square root.
   Then, $m_0$ and $\sqrt{m_1}$  also commute.}
  \end{itemize}
 
 Since $m_0$ and  $m_1$ commute,
  there exists a decomposition
   $$
     m_0\;=\; a_0+n_0  \hspace{2em} m_1\;=\; a_1+n_1
   $$
  such that
    $a_0$ and $a_1$ are diagonalizable,
	$n_0$ and $n_1$ are nilpotent, and
	$a_0,\, a_1,\, n_0,\, n_1$ commute with each other.
  Since $\sqrt{m_1}$ can be expressed as a polynomial in $n_1$
   with coefficients the evaluation of smooth functions on $a_1$,
  the statement follows.
\end{proof}
 
\bigskip

\begin{ssremark}{\rm [\hspace{.1ex}}dilaton factor{\hspace{.1ex}\rm ]}$\;$ {\rm
 In proving Lemma~3.1.5.1 above,
  the fact that $F_{\nabla}$ lies in the commutant of $A_{\varphi}$ for $\varphi$ admissible to $\nabla$
  is used.
 When this condition is not satisfied, one may have to consider taking the symmetrized product of
     $e^{-\varphi^{\diamond}\Phi}$ and
	   $\sqrt{-\SymDet_X(\varphi^{\diamond}(g+B)+2\pi\alpha^{\prime}F_{\nabla})\,}$.
}\end{ssremark}
 
\bigskip

This resolves Issue (5) in the list.

\bigskip
  
\subsubsection{Reality of the trace}

Let $\varphi:(X^{\!A\!z},E;\nabla)\rightarrow (Y, g, B, \Phi)$
  be an admissible Lorentzian map to a Lorentzian manifold with a $B$-field $B$ and a dilaton field $\Phi$.
As noted in Remark~3.1.4.8,
  the $\End_{\Bbb C}(E)$-valued tensor $\sqrt{-\SymDet_X(\varphi^{\diamond}g)\,}$
     is indeed $A_{\varphi}$-valued,
   and, hence,
  \begin{itemize}
   \item[\LARGE $\cdot$] {\it
     The factor $e^{-\varphi^{\diamond}\Phi}$ and
	 the principal square root
	   $\sqrt{-\SymDet_X(\varphi^{\diamond}g)\,}$
	 commute.}
 \end{itemize}
It follows then from
     the proof of Lemma~3.1.4.7,
     the positivity of eigenvalues of $e^{-\varphi^{\diamond}\Phi}$,    and             								
     the commutivity of
	    $e^{-\varphi^{\diamond}\Phi}$ and  $\sqrt{-\SymDet_X(\varphi^{\diamond}g)\,}$
 that
  \begin{itemize}
   \item[\LARGE $\cdot$]  {\it
    $\Tr (e^{-\varphi^{\diamond}\Phi}\,
     	\sqrt{-\SymDet_X(\varphi^{\diamond}g)})$
	is positive real-valued (when applied to a frame on $X$ that is compatible with the orientation).}
   \end{itemize}
 
 Recall from the near end of Sec.\ 3.1.4, as a consequence of Lemma~3.1.4.7,
   that
    if $\nabla$ and $B$ are such that
     $F_{\nabla}$  and $\varphi^{\diamond}B$ are small in the sense that they are close enough
        to the zero-section of  $(T^{\ast}X)^{\otimes 2}\otimes_{\Bbb R}\End_{\Bbb C}(E)$
		(with respect to the natural topology on the total space thereof),
     then the principal square root
	  $$
	    \sqrt{-\SymDet_X(\varphi^{\diamond}(g+B)+2\alpha^{\prime}F_{\nabla})\,}\;    \in\;
		C^{\infty}(\mbox{$\bigwedge^m$}T^{\ast}X \otimes_{\Bbb R}\End_{\Bbb C}(E)  )
	  $$
	  exists.
 It follows that, as a deformation of
   $\Tr(e^{-\varphi^{\diamond}\Phi}\,\sqrt{-\SymDet_X(\varphi^{\diamond}g)})$,
  $$
    \Tr \left( e^{-\varphi^{\diamond}\Phi}\,
	   \sqrt{-\,\SymDet_X(\varphi^{\diamond}(g+B)+2\alpha^{\prime}F_{\nabla})\,}\,
	     \right)\;
		 \in\;     C^{\infty}(\mbox{$\bigwedge$}^mT^{\ast}X)^{\Bbb C}\,,
  $$
   has positive real part if $B$ and $F_{\nabla}$ are small enough.
However, it may not be real itself.
This can be remedied by taking only the real part of the resulting trace as the action functional.

 
Similarly, for
   the case where $(Y,g)$ is Lorentzian and $(\varphi, \nabla)$ is admissible spacelike  or
   the case where $(Y,g)$ is Riemannian and $(\varphi,\nabla)$ is admissible Riemannian.

\bigskip
 	
This resolves Issue (6) in the list.
 
\bigskip

We have  thus resolved all of Issues (2) -- (6) in the list.

\bigskip

\subsection{The Dirac-Born-Infeld action for admissibles pairs $(\varphi, \nabla)$}

With the preparations in Sec.\ 2.2 and Sec.\ 3.1,
 we can now define the Dirac-Born-Infeld action for D-branes
 along the line of [L-Y1] (D(1)) and [L-Y4] (D(11.1)).
 
\bigskip

\begin{definition} {\bf [Dirac-Born-Infeld action for admissible $(\varphi,\nabla)$]}$\;$ {\rm
 (1)
 Let
   $(Y, \Phi, g, B)$ be a Lorentzian manifold $(Y,g)$ with a $B$-field $B$ and a dilaton field $\Phi$,  and
   $\varphi:(X^{\!A\!z}, E;\nabla)\rightarrow Y$ be an admissible Lorentzian map.
 We assume that $B$ and the curvature $F_{\nabla}$ are small enough.
 Then the {\it Dirac-Born-Infeld action} $S_{\DBI}^{(\Phi,g,B)}$ for the pair $(\varphi,\nabla)$
  is defined to be
 $$
   S_{\DBI}^{(\Phi,g,B)}(\varphi,\nabla)\;
    :=\;  -\, T_{m-1}\,
	         \int_X  \Real \left(
             	\Tr   \left(e^{-\varphi^{\diamond}\Phi}\,
  				\sqrt{-\,\SymDet_X
				   (\varphi^{\diamond}(g+B)\, +\, 2\pi\alpha^{\prime}\, F_{\nabla} )\,}\,
                          \right)             \right)											\,,
 $$
 where
    $m=\dimm X$,
	$T_{m-1}$ is the D$(m-1)$-brane tension,
     $\alpha^{\prime}$ is the Regge slope,   and
	 $2\pi\alpha^{\prime}$ is the inverse of the open-string tension.

  \smallskip
  
 (2)
 Let either
     $(Y,g)$ Lorentzian   and
	 $\varphi:(X^{\!A\!z},E;\nabla)\rightarrow (Y,\Phi, g,B)$ admissible and spacelike
  or 	
    $(Y,g)$ Riemannian  and
	 $\varphi:(X^{\!A\!z},E;\nabla)\rightarrow (Y,\Phi, g, B)$ admissible and Riemannian.
 Assume also that $B$ and $F_{\nabla}$ are small enough.
 Then the {\it Dirac-Born-Infeld action} $S_{\DBI}^{(\Phi,g,B)}$ for the pair $(\varphi,\nabla)$
  is defined to be
  $$
    S_{\DBI}^{(\Phi,g,B)}(\varphi,\nabla)\;
     :=\;  T_{m-1}\,
	         \int_X  \Real \left(
             	\Tr  \left(	e^{-\varphi^{\diamond}\Phi}\,
				  \sqrt{\SymDet_X
				   (\varphi^{\diamond}(g+B)\, +\, 2\pi\alpha^{\prime}\, F_{\nabla} )\,}\,
                             \right)          \right)											\,.
  $$
}\end{definition}

\bigskip

In local coordinate chart $U\subset X$, $V\subset Y$, this is explicitly

 {\footnotesize
  \begin{eqnarray*}
   \lefteqn{
    S_{\DBI}^{(\Phi,g,B)}(\varphi|_U,\nabla|_U)  }\\
    &&	=\;
       \mp\,T_{m-1}\,\int_U
   	                       \Tr
						    e^{-\varphi^{\sharp}(\Phi)}
						    \sqrt{\mp\,
						      \SymDet_U \left( \rule{0ex}{1em} \right.
		                         \sum_{i,j}
								    \varphi^{\sharp}(g_{ij}+ B_{ij})\,
								      D_{\mu}\varphi^{\sharp}(y^i)\,D_{\nu}\varphi^{\sharp}(y^j)\,
		                              +\,  2\pi\alpha^{\prime}\, F_{\mu\nu}
                                       \left. \rule{0ex}{1em}	                                    
 										\right)_{\mu\nu} }\;   d^{\,m}x\,.	 	
  \end{eqnarray*}}

\bigskip

\begin{theorem} {\bf [non-Abelian Dirac-Born-Infeld action for D-brane world-volume]}$\;$
 Under the assumption of
	  the enough weakness of the $B$-field $B$ on $Y$ and the gauge curvature $F_{\nabla}$ on $X$,
 the Dirac-Born-Infeld action $S_{\DBI}^{(\Phi,g,B)}$ for an admissible pair $(\varphi,\nabla)$
   in each setting in Definition~3.2.1
  is well-defined.	
 Furthermore,
  when the rank $r$ of $E$, as a complex vector bundle over $X$, is $1$
  (i.e.\ the case of a simple D-brane where the Chan-Paton bundle $E$ is a complex line bundle),
  the action $S_{\DBI}^{(\Phi,g,B)}$ as defined therein
   resumes to the standard Dirac-Born-Infeld action in the string-theory literature action
     for simple D-branes moving in a space-time with a background metric, $B$-field, and dilaton field,
	 e.g., {\rm [Po3: vol.\ I, Eqn.\ (8.7.2)]}.
\end{theorem}

\medskip

\begin{proof}
 The discussions in Sec.~2.2 on admissible pairs $(\varphi,\nabla)$ and in Sec.\ 3.1.1 -- Sec.\ 3.1.4  and
    the assumption that $B$ and $F_{\nabla}$ are weak enough
  imply that the principal square root
    $$
	    \sqrt{\pm\,
	                \SymDet_X
				   (\varphi^{\diamond}(g+B)\, +\, 2\pi\alpha^{\prime}\, F_{\nabla} )\,}
	$$
   in each case is well-defined and has positive real part.
 This imples the well-definedness and ${\Bbb R}_{> 0}\,$-valuedness of the whole integrand
  in each case of the statement.
  
 That $S_{\DBI}^{(\Phi,g,B)}$ as defined is a generalization of the standard Dirac-Born-Infeld action in,
   for example, the quoted textbook by Polchinski, is immediate.

\end{proof}

\medskip

\begin{remark} {\rm [\hspace{.1ex}}on
                         the overall sign in $S_{\DBI}^{(\Phi,g,B)}${\rm\hspace{.1ex}]}$\;$
{\rm
 Recall that in electrodynamics,
  the action for a relativistic charged particle of mass $m$ and electric charge $e$
   moving in a space-time $(Y,g)$ with a background $U(1)$-gauge field $A^{\prime}$ on $Y$,
     whose curvature gives the electromagnetic field on $Y$,
  is given by the Lorentz-invariant action
  $$
    S_{\EM}(\gamma)\;
	 =\;
     -\, m\,\int_{{\Bbb R}^1}
	         \sqrt{-\,\mbox{$ g(\frac{d\gamma}{d\tau}, \frac{d\gamma}{d\tau})$ } }\,d\tau\,
	 +\,  e\,\int_{{\Bbb R}^1}\gamma^{\ast}A^{\prime}\,.
  $$
  Here, $\gamma:{\Bbb R}^1\rightarrow Y$  is the world-line of the particle, parameterized by $\tau$.
 In comparison with the situation for D-branes,
  $\gamma$ here corresponds to the admissible map $\varphi$
  $A^{\prime}$ here plays the role of a Ramond-Ramond field $C$,
  the first term in $S_{\EM}$ corresponds to the Dirac-Born-Infeld action $S_{\DBI}^{(g)}$,
   and the second term in $S_{\EM}$ corresponds to the Chern=Simons/Wess-Zumino action
    $S_{\CSWZ}^{(C)}$  (cf.\ Sec.\ 6).
  See, e.g., [Ja: Chapter 12].
 This comparison sets the overall sign in $S_{\DBI}^{(\Phi,g,B)}$ to be  $\,-\,T_{m-1}\int_X(\cdots)$,
   rather than $\,+\,T_{m-1}\int_X(\cdots)$ for $\varphi$ Lorentzian.
 
 For $\varphi$ spacelike or Riemannian,
    we set the sign to be $+$, by convention, to fit in with the study of minimal submanifolds and harmonics maps
	(e.g.\ [La], [L-W]).
}\end{remark}

%
%
%
   
\medskip

\begin{remark}{$[\,$open-string-compatible quantizable action and super generalization$\,]$.} \rm
 For the first time since the beginning of this D-project,
  the dynamics of D-branes is addressed along the line of the project in a most natural and geometric way.
 This brings the study of D-branes truly to the same starting point as that for the fundamental string:

 \bigskip
 
 \centerline{
{\footnotesize\it
 \begin{tabular}{|l||l|}\hline
   \hspace{5em}{\bf string theory}\rule{0ex}{1.2em}
       & \hspace{10em}{\bf D-brane theory}\\[.6ex] \hline\hline
     $\begin{array}{l}
	    \mbox{string world-sheet}: \\
 		 \hspace{3em}
		 \mbox{$2$-manifold $\,\Sigma$}
	   \end{array}$
       &  $\begin{array}{l}
	           \mbox{D-brane world-volume}:\rule{0ex}{1.2em}\\
			      \hspace{3em}\mbox{Azumaya/matrix manifold}\\
                  \hspace{3em}\mbox{with a fundamental module with a connection}\\
				  \hspace{4em}(X^{\!A\!z},E,\nabla)
			   \end{array}$	  				   \\[3.8ex] \hline
	 $\begin{array}{l}
	    \mbox{string moving in space-time $Y$}:\rule{0ex}{1.2em} \\
 		 \hspace{3em}
		 \mbox{differentiable map $f:\Sigma \rightarrow Y$}
	   \end{array}$
       &  $\begin{array}{l}
	           \mbox{D-brane moving in space-time $Y$}:\rule{0ex}{1.2em}\\
			      \hspace{3em}
	              \mbox{differentiable map $\varphi:(X^{\!A\!z},E,\nabla)\rightarrow Y$}
			   \end{array}$	  				   \\[2.4ex] \hline		
       $\begin{array}{c}
	          \mbox{Nambu-Goto action $\,S_{\tinyNG}\,$  for $f$'s}\end{array}$\rule{0ex}{1.2em}
       &  $\begin{array}{c}\rule{0ex}{1.2em}
	           \mbox{Dirac-Born-Infeld action $\,S_{\tinyDBI}\,$ for $(\varphi,\nabla)$'s}\end{array}$
                                                                     	   \\[1ex] \hline
 \end{tabular}
  } 
 }

 \bigskip

 On the other hand, from the lesson in string theory
  one learns that this action is quite unworkable for quantization of the theory.
 Thus, the above table should be extended immediately to the following not-yet-completed table
  as a guide for further studies:

 \bigskip

 \centerline{
{\footnotesize\it
 \begin{tabular}{|l||l|}\hline
   \hspace{7em}{\bf string theory}\rule{0ex}{1.2em}
       & \hspace{6.4em}{\bf D-brane theory}\\[.6ex] \hline\hline
     $\begin{array}{l}
	    \mbox{string world-sheet}: \\
 		 \hspace{3em}
		 \mbox{$2$-manifold $\,\Sigma$}
	   \end{array}$
       &  $\begin{array}{l}
	           \mbox{D-brane world-volume}:\rule{0ex}{1.2em}\\
			      \hspace{3em}\mbox{Azumaya/matrix manifold}\\
 				  \hspace{3em}\mbox{with a fundamental module with a connection}\\
				  \hspace{4em}\mbox{$(X^{\!A\!z},E,\nabla)$}
			   \end{array}$	  				                                                                      \\[3.8ex] \hline
	 $\begin{array}{l}
	    \mbox{string moving in space-time $Y$}:\rule{0ex}{1.2em} \\
 		 \hspace{3em}
		 \mbox{differentiable map $f:\Sigma \rightarrow Y$}
	   \end{array}$
       &  $\begin{array}{l}\rule{0ex}{1.2em}
	           \mbox{D-brane moving in space-time $Y$}:\\
			      \hspace{3em}
	              \mbox{differentiable map $\varphi:(X^{\!A\!z},E,\nabla)\rightarrow Y$}
			   \end{array}$	  				                                                                             \\[2.4ex] \hline		
       $\begin{array}{c}
	        \mbox{Nambu-Goto action $\,S_{\tinyNG}\,$ for $f$'s}\end{array}$\rule{0ex}{1.2em}
           &  $\begin{array}{c}\rule{0ex}{1.2em}
	                \mbox{Dirac-Born-Infeld action $\,S_{\tinyDBI}\,$ for $(\varphi,\nabla)$'s}\end{array}$
                                                                     	                                                           \\[1ex] \hline
       $\begin{array}{l}\rule{0ex}{1.2em}
	        \mbox{Polyakov action $\,S_{\tinyPolyakov}\,$ for bosonic strings}
	       \end{array}$
	        & \hspace{1.2em}
		        $\begin{array}{l}\rule{0ex}{1.2em}
				 \mbox{\rm ???,
				                      required to be open-string compatible}\end{array}$\\[1ex] \hline
       $\begin{array}{l}\rule{0ex}{1.2em}
	        \mbox{action for Ramond-Neveu-Schwarz superstrings}
	       \end{array}$
	       & \hspace{1.2em}
		      $\begin{array}{l}\rule{0ex}{1.2em}
			       \mbox{\rm ???,$\;$ cf.\ [L-Y5: Sec.\ 5.1] (D(11.2))}\end{array}$
			                                                                                                           \\[1ex] \hline		
       $\begin{array}{l}\rule{0ex}{1.2em}
	        \mbox{action for Green-Schwarz superstrings}
	       \end{array}$
	       & \hspace{1.2em}
		       $\begin{array}{l}\rule{0ex}{1.2em}
			     \mbox{\rm ???,$\;$ cf.\ [L-Y5: Sec.\ 5.1] (D(11.2))}\end{array}$\\[1ex] \hline		   
       $\begin{array}{l}\rule{0ex}{1.2em}
	        \mbox{quantization}
	       \end{array}$
	       & \hspace{1.2em}
		      $\begin{array}{l}\rule{0ex}{1.2em}???\end{array}$\\[1ex] \hline		
 \end{tabular}
  }
 }
 
 \bigskip

 \noindent
 Recall how the Dirac-Born-infeld action for a simple D-branes arises from
   the anomaly-free condition for the world-sheet of open-strings with end-points on such D-brane; 
   cf.\ [Le].
 Here, `open-string compatible' means that the new quantizable action for D-branes
        is required to produce the same anomaly-free conditions for open strings.
\end{remark}

\medskip

\begin{remark} {\rm [\hspace{.1ex}}when in addition
 $\varphi$ is Hermitian and $\nabla$ unitary{\rm\hspace{.1ex}]}$\;$ {\rm
 With the setting in Definition~3.2.1,
 let $E$ be equipped with a Hermitian structure $\langle\:,\:\rangle$.
 If in addition $\varphi$ is Hermitian and $\nabla$ is unitary with respect to $\langle\:,\:\rangle$,
 then one can check that
  $$
     e^{-\varphi^{\sharp}(\Phi)}\sqrt{
	              \mp\,
	                \SymDet_X  \left(
		           \varphi^{\diamond}(g+B))\,
		              +\,  2\pi\alpha^{\prime}\, F_{\nabla}
	                                         \right) }
  $$
   is Hermitian-matrix-valued (with respect to any local unitary frame on $(E,\langle\:,\:\rangle)$).
 %
 %
 %
 %
 %
 %
  and, hence,
    the Dirac-Born-Infeld action for an admissible Hermitian pair $(\varphi,\nabla)$ is simply
     $$
      S_{\DBI}(\varphi,\nabla)\;    :
	   =\; \mp\,
	        T_{m-1}\,\int_X
		     e^{-\varphi^{\diamond}(\Phi)}
    	      \Tr   \sqrt{\mp\,
			                       \SymDet_X \left(
 		                           \varphi^{\diamond}(g+B))\,
 		                              +\,  2\pi\alpha^{\prime}\, F_{\nabla}
 	                                             \right) }\;.
      $$
  Cf.\ Remark~2.3.8 and Remark~5.3.3.
}\end{remark}

\bigskip

\section{Variations of $\varphi^{\sharp}$ in terms of variations of local generators}

Recall from [L-Y6: Sec.\ 2] (D(11.3.1)) that
  \begin{itemize}
   \item[\LARGE $\cdot$] \it
    Let  $(y^1,\,\cdots\,, y^n)$ be a coordinate system on ${\Bbb R}^n$ and
	 $\varphi^{\sharp}:C^{\infty}({\Bbb R}^n)\rightarrow C^{\infty}(\End_{\Bbb C}(E))$
       be a ring-homomorphism over ${\Bbb R}\subset {\Bbb C}$.
    Then for any $f\in C^{\infty}({\Bbb R}^n)$,
      $$
         \varphi^{\sharp}(f)\;=\; f(\varphi^{\sharp}(y^1), \,\cdots\,, \varphi^{\sharp}(y^n))\,.
      $$	
  \end{itemize}
 Here for the Right Hand Side of the equality,
  $\varphi^{\sharp}(y^i)\in C^{\infty}(\End_{\Bbb C}(E))$, for $i=1, \ldots\,, n$,  and
  the value $f(\varphi^{\sharp}(y^1), \,\cdots\,, \varphi^{\sharp}(y^n))$
   is computed pointwise-over-$X$
   through the built-in/canonical $C^{\infty}$-ring structure of the commutative subalgebra
   generated by
      $\varphi^{\sharp}(y^1)(x), \,\cdots\,,\varphi^{\sharp}(y^n)(x)
	       \in \End_{\Bbb C}(E|_x) $ for all $x\in X$.
 That the result lies in $C^{\infty}(\End_{\Bbb C}(E))$
    and coincides with $\varphi^{\sharp}(f)$
  is proved using the Generalized Division Lemma as a consequence of the Malgrange Division Theorem,
  (cf.\ [L-Y6: Step (b) in Proof of Theorem 3.1.1] (D(11.3.1))).
 The equality says that
   \begin{itemize}
    \item[\LARGE $\cdot$] \it
	 A differentiable map $\varphi:(X^{\!A\!z},E)\rightarrow {\Bbb R}^n$,
	   defined by a ring-homorphism
	   $\varphi^{\sharp}: C^{\infty}({\Bbb R}^n)
	          \rightarrow C^{\infty}(\End_{\Bbb C}(E))$,
     is determined by the value of $\varphi^{\sharp}$ on the coordinate functions
	  $y^1,\,\cdots\,, y^n$ of ${\Bbb R}^n$; namely by
      $$
	      \varphi^{\sharp}(y^1)\,,\;  \cdots\,,\;   \varphi^{\sharp}(y^n)\;
		     \in\;  C^{\infty}(\End_{\Bbb C}(E))\,.
	  $$	
   \end{itemize}
([L-Y6: Theorem 3.2.1] (D(11.3.1)).)
  
To calculate the variation of the Dirac-Born-Infeld $S_{\DBI}$ under variations of $(\varphi,\nabla)$,
 one needs to address the following generalization of the above order-$0$ result to higher orders:
  \begin{itemize}
   \item[{\bf Q.}] \parbox[t]{37em}{Let
    $T=(-\varepsilon,\varepsilon)^l\subset {\Bbb R}^l$ be the base manifold and
	 $\varphi_{\mathbf{t}}:(X^{\!A\!z},E)\rightarrow {\Bbb R}^n$,
	     $\mbox{\boldmath $t$}\,:=(t^1,\,\cdots\,,t^l)\in T$,
     	 be a $T$-family of differentiable maps from $(X^{\!A\!z},E)$
  		 to ${\Bbb R}^n$ (with coordinate functions
		 $\mbox{\boldmath $y$}\,:=(y^1,\,\cdots\,,y^n)$),
	defined by
	  $$
	     \varphi_{\mbox{\boldmath \scriptsize $t$}} ^{\sharp}\; :\;
   		  C^{\infty}({\Bbb R}^n)\; \longrightarrow\;
          C^{\infty}(\End_{\Bbb C}(E))\,.		
	  $$	 	
	For $\mbox{\boldmath $\alpha$}\,=(\alpha_1,\,\cdots\,,\alpha_l)$,
	   with $\alpha_1,\,\cdots\,,\alpha_l\in {\Bbb Z}_{\ge 0}$,
   	 let $|\mbox{\boldmath $\alpha$}|:=\alpha_1+\,\cdots\,+\alpha_l$ and
	  $$
	    \mbox{\Large
		  $\frac{\partial^{|\mbox{\boldmath\tiny $\alpha$}|}}
		                 {\partial \mbox{\boldmath\small $t$}^{\mbox{\boldmath\tiny $\alpha$}}}$}\;
		  :=\; \mbox{\Large
		         $\frac{\partial^{|\mbox{\boldmath\tiny $\alpha$}|}}
		        {\rule{0ex}{0.7em}
				   {\partial t^1}^{\,\alpha_1}\, \cdots\, {\partial t^l}^{\,\alpha_l}}$}\,.
	  $$	
    Consider the derivation of order $ |\mbox{\boldmath $\alpha$}| \ge 1$
	  $$
	   \begin{array}{cccccc}
	      \mbox{\Large
		          $\frac{\partial^{|\mbox{\boldmath\tiny $\alpha$}|}}
		           {\partial\mbox{\boldmath\small $t$}^{\mbox{\boldmath\tiny $\alpha$}}}$}
		        \varphi^{\sharp}_{\mbox{\boldmath\scriptsize $t$}} &  :
		     & C^{\infty}({\Bbb R}^n)    & \longrightarrow
		                               & C^{\infty}(\End_{\Bbb C}(E))   \\[1.2ex]
          && f	  & \longmapsto
		     & \frac{\partial^{|\mbox{\boldmath\tiny $\alpha$}|}}
		           {\partial\mbox{\boldmath\scriptsize $t$}^{\mbox{\boldmath\tiny $\alpha$}}}
		              (\varphi^{\sharp}_{\mbox{\boldmath\scriptsize $t$}}(f) )        &.
       \end{array}								
	  $$	
	 Then, for $f\in C^{\infty}({\Bbb R}^n)$,
      $$
	   \begin{array}{l}
     	\mbox{\it Can
          $\; \frac{\partial^{|\mbox{\boldmath\tiny $\alpha$}|}}
		           {\partial\mbox{\boldmath\scriptsize $t$}^{\mbox{\boldmath\tiny $\alpha$}}}
		              (\varphi^{\sharp}_{\mbox{\boldmath\scriptsize $t$}}(f) )\;$
		    be expressed in terms of
          $\; \left(
		          \frac{\partial^{|\mbox{\boldmath\tiny $\alpha$}_1|}}
		           {\partial\mbox{\boldmath\scriptsize $t$}^{\mbox{\boldmath\tiny $\alpha$}_1}}
		              (\varphi^{\sharp}_{\mbox{\boldmath\scriptsize $t$}}(y^1) )\,,\;
					  \cdots\,,\;
         \frac{\partial^{|\mbox{\boldmath\tiny $\alpha$}_n|}}
		           {\partial\mbox{\boldmath\scriptsize $t$}^{\mbox{\boldmath\tiny $\alpha$}_n}}
		              (\varphi^{\sharp}_{\mbox{\boldmath\scriptsize $t$}}(y^n) )
					  \right)$'s}                 \\
        \mbox{\it with
	      $\; \mbox{\boldmath $\alpha$}_1 +\,\cdots\,+ \mbox{\boldmath $\alpha$}_n
	          =  \mbox{\boldmath $\alpha$}\,$?}
	   \end{array}
      $$	
	   } 
  \end{itemize}
In this section, we answer this question affirmatively at the level of germs of differentiable functions,
 following a similar reasoning as in [L-Y6] (D(11.3.1)).
The result will be used to calculate the first variation of  $S_{\DBI}^{(\Phi,g,B)}$ in the current notes
 and the second variation of  $S_{\DBI}^{(\Phi,g,B)}$ in a sequel.

\bigskip
\bigskip

\begin{flushleft}
{\bf\large 4.1\hspace{0.6em}
   $\varphi^{\sharp}_{\mbox{\boldmath\scriptsize $t$}}(f)$
		 in terms of
		  $(\varphi^{\sharp}_{\mbox{\boldmath\scriptsize $t$}}(y^1),\,\cdots\,
			  \varphi^{\sharp}_{\mbox{\boldmath\scriptsize $t$}}(y^n))$
          via Generalized Division Lemma}
\end{flushleft}		
With the notations from above,
we recall from [L-Y6] (D(11.3.1)) how the right-hand-side of the equality
   $$
     \varphi^{\sharp}_{\mbox{\boldmath\scriptsize $t$}}(f)
	  \; =\;
	    f(\varphi^{\sharp}_{\mbox{\boldmath\scriptsize $t$}}(y^1)\,,\,\cdots\,,\,
             \varphi^{\sharp}_{\mbox{\boldmath\scriptsize $t$}}(y^n))
   $$
   is expressed in terms of
     $f$ and
     $\varphi^{\sharp}_{\mbox{\boldmath\scriptsize $t$}}(y^1)\,,\,\cdots\,,\,
             \varphi^{\sharp}_{\mbox{\boldmath\scriptsize $t$}}(y^n)$
    via the Generalized Division Lemma, a corollary of the Malgrange Division Theorem
	([Mal]; see also [Br\"{o}], [Mat1], [Mat2], and [Ni]).		
Readers are referred to ibidem for more details.

\bigskip

\begin{flushleft}
{\bf The basic setup}
\end{flushleft}
For convenience, consider the projection map
 $$
    X_T\, :=\,  X\times T\; \longrightarrow\; X
 $$
 and denote the pull-back of $E$ to $X\times T$ by $E_T$.
Then the $T$-family of differentiable maps
  $$
    \{\varphi_{\mbox{\boldmath\scriptsize $t$}}:(X^{\!A\!z},E)\rightarrow {\Bbb R}^n\;|\;
	     \mbox{\boldmath $t$} \in T \}\,,
  $$
  with $\varphi_{\mbox{\boldmath\scriptsize $t$}}$
   defined by a ring-homomorphism
    $\varphi^{\sharp}_{\mbox{\boldmath\scriptsize $t$}}:
	   C^{\infty}({\Bbb R}^n)\rightarrow C^{\infty}(\End_{\Bbb C}(E))$
  over ${\Bbb R}\subset {\Bbb C}$,	
 defines a differentiable map
  $$
    \xymatrix{
    (X_T^{\!A\!z}, E_T)\ar[rr]^-{\varphi_T}    &&{\Bbb R}^n
	}
  $$
  that is defined by the ring-homomorphism
  $$
    \xymatrix{
	  C^{\infty}(\End_{\Bbb C}(E_T))
	     && C^{\infty}({\Bbb R}^n)\ar[ll]_-{\varphi^{\sharp}_T}
	  } 	
  $$
   over ${\Bbb R}\subset{\Bbb C}$
  that restricts to $\varphi^{\sharp}_{\mbox{\boldmath\scriptsize $t$}}$,
    for all $\mbox{\boldmath $t$}\,\in T$.
 This extends canonically to a commutative diagram of ring-homomorphisms
  (over ${\Bbb R}$ or ${\Bbb R}\subset{\Bbb C}$, whichever is applicable)
  ([L-Y6: Theorem 3.1.1] (D(11.3.1)))
  $$
    \xymatrix{
	  \;C^{\infty}(\End_{\Bbb C}(E_T))\;		
	     && \;\;C^{\infty}({\Bbb R}^n)\;\;
		           \ar[ll]_-{\varphi^{\sharp}_T}
				   \ar@{_{(}->}[d]^-{pr^{\sharp}_{{\Bbb R}^n}}   \\
	  \;\;C^{\infty}(X_T)\rule{0ex}{1em}\;
	       \ar@{^{(}->}[rr]_-{pr^{\sharp}_{X_T}}      \ar@{^{(}->}[u]
	     && C^{\infty}(X_T\times{\Bbb R}^n)
		          \ar[llu]_-{\tilde{\varphi}^{\sharp}_T}\;,
    }
  $$
  where
    $\pr_{X_T}:X_T\times {\Bbb R}^n \rightarrow X_T$ and
	$\pr_{{\Bbb R}^n}:X_T\times{\Bbb R}^n\rightarrow{\Bbb R}^n$ are the projection maps,  and
	$C^{\infty}(X_T)\hookrightarrow C^{\infty}(\End_{\Bbb C}(E))$
	   follows from the inclusion of the center $C^{\infty}(X_T)^{\Bbb C}$ of
	   $C^{\infty}(\End_{\Bbb C}(E_T))$.
 This in turn defines the following diagrams of differentiable maps that extends $\varphi_T$	   
   $$
    \xymatrix{
	  \;(X_T^{\!A\!z},E_T)\;
	           \ar[rr]^-{\varphi_T}  \ar[rrd]^-{\tilde{\varphi}_T}  \ar@{->>}[d]
	       &&   \;{\Bbb R}^n\;    \\
	    \;X_T\;
		   &&   \;\;X_T\times {\Bbb R}^n
		                 \ar@{->>}[ll]^-{pr_{X_T}}
                         \ar@{->>}[u]_-{pr_{{\Bbb R}^n}}						 \,.	
     }
   $$
 Let ${\cal E}_{T}$ be the sheaf of  $C^{\infty}$-sections of	$E_T$.
 Then,
  the ${\cal O}_{X_T\times {\Bbb R}^n}^{\,\Bbb C}$-module
  $$
     \tilde{\cal E}_{\varphi_T}\;  :=\;   \mbox{$\tilde{\varphi}_T$}_{\ast}({\cal E}_T)
  $$
  defines the {\it graph} of  $\varphi_T$.
  Its $C^{\infty}$-scheme-theoretical support
    $$
	   \varGamma_{\varphi_T}   \; :=\;  \Supp(\tilde{\cal E}_{\varphi_T})\;
	     \subset\; X_T\times{\Bbb R}^n
	$$
  is finite and algebraic over  $X_T$
    under the restriction of the projection $\pr_{X_T}:X_T\times{\Bbb R}^n\rightarrow X_T$.

\bigskip

\begin{flushleft}
{\bf  $\varGamma_{\varphi_T}$ and the spectral subscheme $\varSigma_{\varphi_T}$
           of $\varphi^{\sharp}_T$ in $X_T\times{\Bbb R}^n$}	
\end{flushleft}	
Under the projection map $\pr_{{\Bbb R}^n}:X_T\times{\Bbb R}^n\rightarrow{\Bbb R}^n$,
the generators $y^1\,,\,\cdots\,,\, y^n$ of $C^{\infty}({\Bbb R}^n)$, as a $C^{\infty}$-ring,
 pull back to elements in $C^{\infty}(X_T\times{\Bbb R}^n)$.
They will still be denoted by $y^1\,,\,\cdots\,,\, y^n$.
Define the {\it spectral subscheme} of $\varphi_T$
  ({\it associated to the generating set $\{y^1\,,\,\cdots\,,\,y^n\}$ of $C^{\infty}({\Bbb R}^n)$})
  to be the subscheme
  $$
    \varSigma_{\varphi_T; \{y^1\,,\,\cdots\,,\,y^n\}}\;
	  \subset\; X_T\times {\Bbb R}^n
  $$
  defined by the ideal
  $$
    I_{\varphi_T;\{y^1\,,\,\cdots\,,\,y^n\}}\;
	=\; (\,
	         \determinant(y^i\cdot\Id_{r\times r}- \varphi^{\sharp}_T(y^i)) \,|\;
              	i=1,\,\cdots\,,\, n \,)\;
		\subset\; C^{\infty}(X_T\times{\Bbb R}^n)  \,.
  $$
 Here,
   $r$ is the rank of $E$ as a complex vector bundle over $X$, and
   $\Id_{r\times r}\in C^{\infty}(\End_{\Bbb C}(E))$ the identity endomorphism.
 Then,
  $$
    \varGamma_{\varphi_T}\;\subset\; \varSigma_{\varphi_T;\{y^1\,,\,\cdots\,,\,y^n\}}
      \hspace{1em}\mbox{with}\hspace{1em}
   	(\varGamma_{\varphi_T})_{\redscriptsize}\;
	    =\; (\varSigma_{\varphi_T;\{y^1\,,\,\cdots\,,\,y^n\}})_{\redscriptsize}\,,
  $$
  where $(\,\cdot\,)_{\redscriptsize}$ denotes the reduced subscheme of a scheme in the sense
   $C^{\infty}$-algebraic geometry.
In terms of this picture, one has
 \begin{itemize}
  \item[\LARGE $\cdot$]
  {\it
   For any $f\in C^{\infty}({\Bbb R}^n)$,
     denote its pull-back to $C^{\infty}(X_T\times{\Bbb R}^n)$ still by $f$.
   Then,
     $\varphi_T^{\sharp}(f)$ depends only the restriction of $f$
       to $\varSigma_{\varphi_T;\{y^1\,,\,\cdots\,,\,y^n\}}$.}	
   In other words,
   {\it  if $f_1,\,f_2\in C^{\infty}({\Bbb R}^n)$ satisfy that,
          after being pulled back to $X_T\times{\Bbb R}^n$,
            $f_1-f_2\in I_{\varphi_T,\{y^1\,,\,\cdots\,,\,y^n\}}$,
	       then $\varphi^{\sharp}_T(f_1)=\varphi^{\sharp}_T(f_2)$.}
 \end{itemize}			
  
It is important to note that
 $\varSigma_{\varphi_T;\{y^1\,,\,\cdots\,,\,y^n\}}$ is a $C^{\infty}$-subscheme,
   rather than just a subset,  of $X_T\times{\Bbb R}^n$.
Thus, the restriction of $f|_{\varSigma_{\varphi_T;\{y^1\,,\,\cdots\,,\,y^n\}}}$
  of $f$ on $\varSigma_{\varphi_T;\{y^1\,,\,\cdots\,,\,y^n\}}$ captures
   not only the value of $f$ at ${\Bbb R}$-points on $\varSigma_{\varphi_T;\{y^1\,,\,\cdots\,,\,y^n\}}$
  but also the behavior of $f$  in an infinitesimal neighborhood
    of $\varSigma_{\varphi_T;\{y^1\,,\,\cdots\,,\,y^n\}}$ up to a finite order.
Cf.\ {\sc Figure}~4-1-1.
 %

 \begin{figure} [htbp]
  \bigskip
  \centering

  \includegraphics[width=0.55\textwidth]{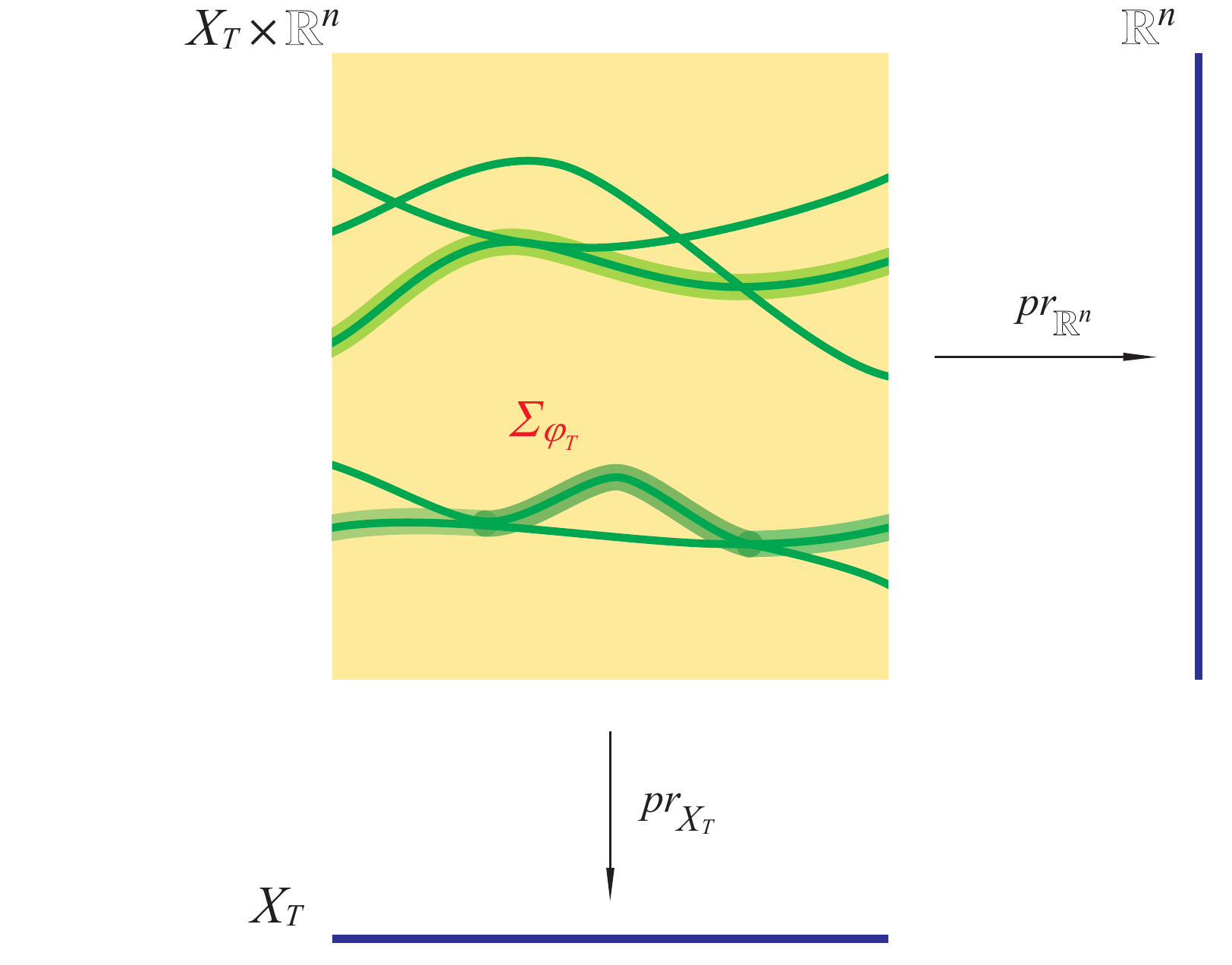}

  \bigskip
  \bigskip
  \centerline{\parbox{13cm}{\small\baselineskip 12pt
   {\sc Figure}~4-1-1.
     The spectral subscheme $\varSigma_{{\varphi}_T}$
	   (in green color, with the green shade indicating the nilpotent structure/cloud
	       on $\varSigma_{{\varphi}_T}$)
	   in $X_T\times {\Bbb R}^n$
	  associated to a ring-homomorphism
	   $\varphi^{\sharp}_T:C^{\infty}({\Bbb R}^n)\rightarrow C^{\infty}(\End_{\Bbb C}(E_T))$.
	 More than just a point-set with topology,
	   it is a $C^{\infty}$-scheme that is finite over $X_T$.
       }}
  \bigskip
 \end{figure}

\bigskip

\begin{flushleft}
 {\bf Function germs at  $\varSigma_{\varphi_T;\{y^1,\,\cdots\,,y^n\}}$
           from the Generalized Division Lemma\\  \`{a} la Malgrange}
\end{flushleft}	
Note that
 the polynomials
  $\determinant(y^i\cdot\Id_{r\times r}- \varphi^{\sharp}_T(y^i))
     \in C^{\infty}(X_T)[y^i] \subset C^{\infty}(X_T\times{\Bbb R}^n)$, $i=1,\,\ldots\,,\, n$,
 	 are of degree $r$.
As a consequence of the Generalized Division Lemma, which is a corollary of the Malgrange Division Theorem,
 one has
   \begin{itemize}
    \item[\LARGE $\cdot$] {\it
	  For $q\in \varSigma_{\varphi_T;\{y^1\,,\,\cdots\,,\,y^n\}}$  and
	         $h\in C^{\infty}(X_T\times {\Bbb R}^n)$,
	   there exists a neighborhood $U_h^{\prime}\times V_h$ of $q$ in $X_T\times{\Bbb R}^n$  and
    	  a polynomial of $(y^1,\,\cdots\,,\, y^n)$-degree $\le (r-1,\,\cdots,\,r-1)$
	     $$
		  R_{h;q}\; :=\;
		    \sum_{\mbox{\tiny $\begin{array}{c}0\le d_i\le r-1 \\[1.2ex]  1\le i\le n \end{array}$}}
			  a^{h;q}_{(d_1,\,\cdots\,,\,d_n)}\cdot(y^1)^{d_1}\,\cdots\,(y^n)^{d_n}\;
			 \in\; C^{\infty}(U^{\prime}_h)[y^1,\,\cdots\,,y^n]\;
			 \subset\; C^{\infty}(U^{\prime}_h\times V_h)
		 $$
     such that
        $$
          h|_{U_h^{\prime}\times V_h}\; =\; R_{h;q}\, +\, h^{\prime}\,,
        $$		
		where
		 $h^{\prime}\in I_{\varphi_T;\{y^1,\,\cdots\,,\,y^n\}}|_{U^{\prime}_h\times V_h}$. }
   \end{itemize}

\bigskip

\begin{flushleft}
{\bf The explicit form for $\varphi^{\sharp}_T(f)$ from the Generalized Division Lemma}
\end{flushleft}
For $f\in C^{\infty}({\Bbb R}^n)$ and $p\in X_T$,
 let
 \begin{itemize}
  \item[\LARGE $\cdot$]
   $\{q_1,\,\cdots\,,q_s\}$ be the set of ${\Bbb R}$-points in the $0$-dimensional subscheme
    $\pr_{X_T}^{-1}(p)\cap  \varSigma_{\varphi_T;\{y^1,\,\cdots\,,y^n\}}$
    of $X_T\times{\Bbb R}^n$,
 \end{itemize}	
Then,
 for a neighborhood $U^{\prime}$ of $p\in X_T$ sufficiently small,
 $\pr_{X_T}^{-1}(U)\cap \varSigma_{\varphi_T;\{y^1,\,\cdots\,,y^n\}}$ has
 exactly $s$-many connected components
 $$
    \pr_{X_T}^{-1}(U)\cap \varSigma_{\varphi_T;\{y^1,\,\cdots\,,y^n\}}\;     =\;
	\varSigma_{(U^{\prime}; 1)}\,\sqcup\,\cdots\,\sqcup\, \varSigma_{(U^{\prime};\,s)}
     \hspace{1em}\mbox{with}\hspace{1em}
      q_j\in \varSigma_{(U^{\prime};\, j)}\,.	
 $$
This implies that
   the support  $\Supp(\tilde{\cal E}_{\varphi_T}|_{U^{\prime}\times {\Bbb R}^n})$
     of $\tilde{\cal E}_{\varphi_T}|_{U^{\prime}\times{\Bbb R}^n}$
     has also exactly $s$-many connected components.
Let
 \begin{itemize}
  \item[\LARGE $\cdot$] 	
	    $E_T|_{U^{\prime}}
		   =  E|_{U^{\prime}}^{(1)}\oplus\,\cdots\, \oplus E_T|_{U^{\prime}}^{(s)}$
		 be the decomposition of $E_T|_{U^{\prime}}$
		 associated to the decomposition of
		 $\tilde{\cal E}_{\varphi_T}|_{U^{\prime}\times {\Bbb R}^n}$
		 into the direct sum of its restriction to the connected components of
		 $\Supp(\tilde{\cal E}_{\varphi_T}|_{U^{\prime}\times {\Bbb R}^n})$.  		
 \end{itemize}		
In terms of this,
 \begin{itemize}
  \item[\LARGE $\cdot$] {\it
    Over $U^{\prime}$, $\varphi^{\sharp}_T$ is decomposed into
     $$
        \varphi^{\sharp}_T\;
   	      =\; (\varphi_T^{\sharp, (1)},\,\cdots\,, \varphi_T^{\sharp, (s)})\;
	      =\;   \varphi_T^{\sharp, (1)}\oplus\,\cdots\,\oplus \varphi_T^{\sharp, (s)} \,,
     $$
    with
     $$
        \varphi_T^{\sharp,(j)}\;:\;  C^{\infty}({\Bbb R}^n)\; \longrightarrow\;
         C^{\infty}(\End_{\Bbb C}(E_T|_{U^{\prime}}^{(j)}))\,,
     $$
     for $j=1,\,\cdots\,,s$.
  }
 \end{itemize}
 
 Now for each
    $\varphi_T^{\sharp,(j)}:  C^{\infty}({\Bbb R}^n) \rightarrow
         C^{\infty}(\End_{\Bbb C}(E_T|_{U^{\prime}}^{(j)}))$,
 one can apply the result in the previous theme on the Generalized Division Lemma
   to express the germ of $\varphi_T^{\sharp,(j)}(f)$ over $p\in X_T$
   in terms of $\varphi_T^{\sharp,(j)}(y^1)\,,\,\cdots\,,\, \varphi_T^{\sharp,(j)}(y^n)$
   as follows.
  \begin{itemize}
   \item[(1)]
     First, it follows from the previous theme on the Generalized Division Lemma, with $q=q_j$ and $h=f$, that
	   there exists a neighborhood
		   $U_f^{\prime, (j)}\times V_f^{(j)}$ of $q_j$ in $U^{\prime}\times{\Bbb R}^n$  and
    	  a polynomial of $(y^1,\,\cdots\,,\, y^n)$-degree $\le (r-1,\,\cdots,\,r-1)$
	     $$
		  R_{f;q_j}\; :=\;
		    \sum_{\mbox{\tiny $\begin{array}{c}0\le d_i\le r-1 \\[1.2ex]  1\le i\le n \end{array}$}}
			  a^{f;q_j}_{(d_1,\,\cdots\,,\,d_n)}\cdot(y^1)^{d_1}\,\cdots\,(y^n)^{d_n}
         $$
		   in $C^{\infty}(U^{\prime, (j)}_f)[y^1,\,\cdots\,,y^n]
			 \subset C^{\infty}(U^{\prime, (j)}_f\times V_f^{(j)})$
       such that
        $$
          f|_{U_f^{\prime,(j)}\times V_f^{(j)}}\; =\; R_{f;q_j}\, +\, f^{\prime}_{(j)}\,,
        $$		
		where
		 $f^{\prime}_{(j)} \in
		  I_{\varphi_T;\{y^1,\,\cdots\,,\,y^n\}}|_{U^{\prime, (j)}_f\times V_f^{(j)}}$.
  
   \item[(2)]
    Then, over $U^{\prime,(j)}_f$,
	 $\varphi_T^{\sharp,(j)}(f)$ has the following expression in terms of
	   $\varphi_T^{\sharp,(j)}(y^1)\,,\,\cdots\,,\,\varphi_T^{\sharp,(j)}(y^n)\,$:   	
	  \begin{eqnarray*}
	   \lefteqn{\varphi_T^{\sharp,(j)}(f)\;
	      =\;    \varphi_T^{\sharp,(j)}(R_{f;q_j}) }\\[1.2ex]
          && =\;
		  	 \sum_{\mbox{\tiny $\begin{array}{c}0\le d_i\le r-1 \\[1.2ex]  1\le i\le n \end{array}$}}
			  a^{f;q_j}_{(d_1,\,\cdots\,,\,d_n)}\cdot
			    (\varphi_T^{\sharp,(j)}(y^1))^{d_1}\,\cdots\,
				(\varphi_T^{\sharp,(j)}(y^n))^{d_n}  \\[1.2ex]		
		  && \in \;\;\;
  		         C^{\infty}
		             (\End_{\Bbb C}
					    ((E_T|_{U^{\prime}}^{(j)})|_{U^{\prime, (j)}_f}))\;.
	 \end{eqnarray*}
	
   \item[(3)]
    Finally, let $U^{\prime}_{f;p}:=\bigcap_{j=1}^s U^{\prime,(j)}_f$.
	Then, over $U^{\prime}_{f;p}$,
	  $$
	    \varphi_T^{\sharp}(f)\;
		 =\; \varphi_T^{\sharp,(1)}(f)\,\oplus\,\cdots\,\oplus\,\varphi_T^{\sharp,(s)}(f)\;	
	    \in\; C^{\infty}	 (\End_{\Bbb C}(E_T|_{U^{\prime}_{f;p}}))\,.
      $$
 \end{itemize}
  
\bigskip
 
\noindent
{\bf Remark/Notation 4.1.1. [equivalent form: zero-th order]}$\;$
 At the level of germs over $X$,
 one can re-write the above expression for $\varphi_T^{\sharp}(f)$
    in terms of  the evaluation of $\varphi_T^{\sharp}$ on local coordinate functions
  more compactly as
  $$
	\varphi_T^{\sharp}(f)\;
	=\;
		 \sum_{\mbox{\tiny $\begin{array}{c}0\le d_i\le r-1 \\[1.2ex]  1\le i\le n \end{array}$}}
			  a^f_{(d_1,\,\cdots\,,\,d_n)}\cdot
			    (\varphi_T^{\sharp}(y^1))^{d_1}\,\cdots\,
				(\varphi_T^{\sharp}(y^n))^{d_n}              \;\;\;\;
      \in\;\;  C^{\infty}(\End_{\Bbb C}(E_T|_{U^{\prime}_f}))\,,
  $$
  where
  $$
    a^f_{(d_1,\,\cdots\,,\, d_n)}\;
	 :=\; \sum_{j=1}^s
	          a^{f;q_j}_{(d_1,\,\cdots\,,\,d_n)}
			  \cdot \Id_{(E_T|^{(j)}_{U^{\prime}})|_{U^{\prime}_f}   }\,.
  $$
 Here, for $j=1,\,\ldots\,,\,s$,
  we identify the identity map
   $\Id_{(E_T|^{(j)}_{U^{\prime}})|_{U^{\prime}_f}   }$
     on  $(E_T|^{(j)}_{U^{\prime}})|_{U^{\prime}_f}   $
   as an idempotent map on $E_T|_{U^{\prime}_f}$
   through its extension-by-zero on all
   $(E_T|^{(j^{\prime})}_{U^{\prime}})|_{U^{\prime}_f}$ for $j^{\prime}\ne j$.

\bigskip
\bigskip

\begin{flushleft}
{\bf\large 4.2\hspace{0.6em}
      $\frac{\partial^{|\mbox{\boldmath\tiny $\alpha$}|}}
		           {\partial\mbox{\boldmath\scriptsize $t$}^{\mbox{\boldmath\tiny $\alpha$}}}
		              (\varphi^{\sharp}_{\mbox{\boldmath\scriptsize $t$}}(f) )$
		   in terms of
          $\left(
		      \frac{\partial^{|\mbox{\boldmath\tiny $\alpha$}_1|}}
		           {\partial\mbox{\boldmath\scriptsize $t$}^{\mbox{\boldmath\tiny $\alpha$}_1}}
		              (\varphi^{\sharp}_{\mbox{\boldmath\scriptsize $t$}}(y^1) )\,,\,
					  \cdots\,,\,
         \frac{\partial^{|\mbox{\boldmath\tiny $\alpha$}_n|}}
		           {\partial\mbox{\boldmath\scriptsize $t$}^{\mbox{\boldmath\tiny $\alpha$}_n}}
		              (\varphi^{\sharp}_{\mbox{\boldmath\scriptsize $t$}}(y^n) )\right)$'s,
	   $\mbox{\boldmath $\alpha$}_1+\cdots +\mbox{\boldmath $\alpha$}_n
	      = \mbox{\boldmath $\alpha$}$}
\end{flushleft}		

\noindent 		
With the preparation/review in Sec.~4.1,
 we now study
   how to compute/express
     $\frac{\partial^{|\mbox{\boldmath\tiny $\alpha$}|}}
		           {\partial\mbox{\boldmath\scriptsize $t$}^{\mbox{\boldmath\tiny $\alpha$}}}
		              (\varphi^{\sharp}_{\mbox{\boldmath\scriptsize $t$}}(f) )$
		   in terms of
          $\left(\frac{\partial^{|\mbox{\boldmath\tiny $\alpha$}_1|}}
		           {\partial\mbox{\boldmath\scriptsize $t$}^{\mbox{\boldmath\tiny $\alpha$}_1}}
		              (\varphi^{\sharp}_{\mbox{\boldmath\scriptsize $t$}}(y^1) )\,,\;
					  \cdots\,,\;
         \frac{\partial^{|\mbox{\boldmath\tiny $\alpha$}_n|}}
		           {\partial\mbox{\boldmath\scriptsize $t$}^{\mbox{\boldmath\tiny $\alpha$}_n}}
		              (\varphi^{\sharp}_{\mbox{\boldmath\scriptsize $t$}}(y^n) )\right)$'s
	   with
	   $\mbox{\boldmath $\alpha$}_1+ \cdots+\mbox{\boldmath $\alpha$}_n
	      = \mbox{\boldmath $\alpha$}$.

\bigskip

\begin{flushleft}
{\bf 4.2.1\hspace{0.6em} Preparatory: Chain rule vs.\ Leibniz rule, and the increase of complexity}
\end{flushleft}		
In the commutative case, this is simply the consequence of the chain rule of differentiations.
However, in the noncommutative case, the formal commutative chain rule is not correct
 and there is no obvious chain rule to use.
Neverless, the Leibniz rule
 $$
   \partial_{\mbox{\LARGE $\cdot$}}(ab)\;
     =\;   (\partial_{\mbox{\LARGE $\cdot$}}a)\,b\,
	          +\,  a\, \partial_{\mbox{\LARGE $\cdot$}}b
 $$
 still holds.
Thus, for polynomial type functions with noncommutative arguments,
 one can still work things out.
This is why the polynomial type expression
 (with coefficients in the germs of differentiable functions on $X_T$)
 of germs of $\varphi_T^{\sharp}(f)$ over $X_T$
  in terms of germs of $\varphi_T^{\sharp}(y^1),\,\cdots\,,\varphi_T^{\sharp}(y^n)$
  in the previous subsection
  is fundamental.
 
\bigskip

\noindent
{\bf Example 4.2.1.1. [violation of formal chain rule in the noncommutative case]}$\;$
 Consider $m(t)=a+bt$, where $ab\ne ba$ (but $ta=at$ and $tb=bt$).
 Let $f\in C^{\infty}({\Bbb R}^1)$ defined by $f(y)= y^k$, $k\ge 2$.
 Then $\frac{d}{dy}f=ky^{k-1}$.
 The formal chain rule would give
  $$
   \begin{array}{c}
    \frac{d}{dt}(f(m(t)))\;
	  =\; (\frac{d}{dy}f)(m(t))\,\frac{d}{dt}m(t)\;
	  =\; k\,(a+bt)^{k-1}\,b\,.
   \end{array}
  $$
 However, in truth,
  $$
   \begin{array}{c}
     \frac{d}{dt}(f(m(t)))\; =\; \frac{d}{dt}
	  (\,\overbrace{(a+bt)\,\cdots\,(a+bt)}^{\mbox{\scriptsize $k$-times}}\,)\;
	 =\; \sum_{k^{\prime}=0}^{k-1}(a+bt)^l\,b\,(a+bt)^{k-1-k^{\prime}}\,.	
   \end{array}
  $$
  As $b$ and $a+bt$ do not commute,
   $$
    \begin{array}{c}
     k\,(a+bt)^{k-1}\,b  \;
	   \ne\; \sum_{k^{\prime}=0}^{k-1}(a+bt)^{k^{\prime}}\,b\,(a+bt)^{k-1-k^{\prime}}
	\end{array}
   $$
  in general.
  
\bigskip

\noindent
{\it Remark 4.2.1.2. $[$complexity$\,]$}\;
 Example~4.2.1.1 serves to illustrate also
   the complexity of the expression in the noncommutative case versus the commutative case.
 In general,
   for $\mbox{\boldmath $\alpha$}=(\alpha_1,\,\cdots\,,\,\alpha_l)$  and
   a monomial
   $$
     P(\xi^1,\,\cdots\,,\,\xi^n)\;  =\;  (\xi^1)^{d_1}\,\cdots\,(\xi^n)^{d_n}
   $$
   of degree $d=d_1+\,\cdots\,+d_n$ in noncommuting variables $\xi^1,\,\cdots\,\xi^n$,
 suppose  that
  $\xi^i=\xi^i(\mbox{\mbox{\boldmath $t$}})$,
    where $\mbox{\boldmath $t$}=(t^1,\,\cdots\,,\,t^l)\in T$   and
 {\it assume that the order of differentiations play no role}.		
 Then,
   the direct expansion of the composition
    $$
	  \partial^{\mbox{\scriptsize\boldmath $\alpha$}}_{\mbox{\scriptsize\boldmath $t$}}
	     P(\xi^1(\mbox{\boldmath $t$})\,,\,\cdots\,,\, \xi^n(\mbox{\boldmath $t$}))	  \;
	  :=\;
	  \mbox{\Large $\frac{\partial^{|\mbox{\tiny\boldmath $\alpha$}|}}
	        {\partial \mbox{\boldmath\small $t$}^{\mbox{\tiny\boldmath $\alpha$}} }$}
		  P(\xi^1(\mbox{\boldmath $t$})\,,\,\cdots\,,\, \xi^n(\mbox{\boldmath $t$}))	
	$$
   via the Leibniz rule would have $d^{|\mbox{\scriptsize\boldmath $\alpha$}|}$-many terms.
   Which collapse to
    $$	
	  \prod_{l^{\prime}=1}^l
	     \left(\mbox{\small $\begin{array}{c}    \alpha_{l^{\prime}}+d -1\\
		                                                                           d-1 \end{array}$}\right)\;
	  =\;   \prod_{l^{\prime}=1}^l
	                    \left(\mbox{\small $\begin{array}{c}    \alpha_{l^{\prime}}+d -1\\
		                                                                           \alpha_{l^{\prime}} \end{array}$}\right)
	$$
	-many terms, after collecting like terms,
   of the form
   
   {\small
	$$
	   \partial^{\mbox{\tiny\boldmath $\alpha$}_{(1)}}_{\mbox{\tiny\boldmath $t$}}
		  \xi^1(\boldt)\,
		  \cdots\,
	   \partial^{\mbox{\tiny\boldmath $\alpha$}_{(d_1)}}_{\mbox{\tiny\boldmath $t$}}
		  \xi^1(\boldt)\,
	   \partial^{\mbox{\tiny\boldmath $\alpha$}_{(d_1+1)}}_{\mbox{\tiny\boldmath $t$}}
		  \xi^2(\boldt)\, \cdots
      \partial^{\mbox{\tiny\boldmath $\alpha$}_{(d_1+d_2)}}_{\mbox{\tiny\boldmath $t$}}
		  \xi^2(\boldt)\,		  		
		 \cdots\cdots\,
      \partial^{\mbox{\tiny\boldmath $\alpha$}_{(d_1+\,\cdots\,+d_{n-1}+1)}}
	      _{\mbox{\tiny\boldmath $t$}}
		  \xi^n(\boldt)\,\cdots\,
		\partial^{\mbox{\tiny\boldmath $\alpha$}_{(d)}}_{\mbox{\tiny\boldmath $t$}}
		  \xi^n(\boldt)\,,		
	$$}
   where
     $\boldalpha=\boldalpha_{(1)}+\,\cdots\,+\boldalpha_{(d)}$,
	   $\boldalpha_{(1)},\,\cdots\,, \boldalpha_{(d)}\in ({\Bbb Z}_{\ge 0})^l$,
     is an ordered partition $(\boldalpha_{(1)},\,\cdots\,,\boldalpha_{(d)})$ of $\boldalpha$
	 into $d$-many summands,
 with the coefficients from the coefficients of the expansion of the product
   $\prod_{l^{\prime}=1}^l
     (z^1_{l^{\prime}}+\,\cdots\,+z^d_{l^{\prime}})^{\alpha_{l^{\prime}}}$.

\bigskip

\noindent
{\bf Example 4.2.1.3. [differentiation via Leibniz rule]}$\;$
 Let
  $\boldt=(t^1,t^2)$  and
  $a(\boldt),\, b(\boldt)$ be functions of $\boldt$ with values in a noncommutative ring.
 Then
 
 {\scriptsize
  $$
   \begin{array}{rcll}
     \partial_1^3\partial_2^2(a(\boldt)^2 b(\boldt))
	   & = & \partial_1^3\partial_2^2 a(\boldt)\,a(\boldt)b(\boldt)\;
	                  +\;  \partial_2^2a(\boldt)\, \partial_1^3 a(\boldt)\, b(\boldt) \;
					  +\;  \partial_2^2a(\boldt)\,a(\boldt)\,\partial_1^3 b(\boldt)                      \\[.8ex]
   	   &&		  +\; 3\,\partial_1^2\partial_2^2a(\boldt) \partial_1a(\boldt )\, b(\boldt)\;
	                  +\; 3\,\partial_1^2\partial_2^2a(\boldt)\,a(\boldt)\, \partial_1 b(\boldt)\;
                      +\; 3\,\partial_2^2a(\boldt)\, \partial_1^2 a(\boldt)\partial_1 b(\boldt) \\[.8ex]					  					
	   &&	      +\; 3\,\partial_1 \partial_2^2a(\boldt)\partial_1^2 a(\boldt)\,b(\boldt)\;
					  +\; 3\,\partial_1\partial_2^2 a(\boldt)\,a(\boldt)\,\partial _1^2 b(\boldt)\;
                      +\; 3\,\partial_2^2a(\boldt)\,\partial_1 a(\boldt)\partial_1^2 b(\boldt)  \\[.8ex]
       &&         +\; 6\,\partial_1 \partial_2^2a(\boldt)\partial_1 a(\boldt)\partial_1 b(\boldt)		\\[.8ex]
	   && +\; \partial_1^3 a(\boldt)\partial_2^2a(\boldt)\,b(\boldt)\;
	                  +\;  a(\boldt)\, \partial_1^3\partial_2^2 a(\boldt)\, b(\boldt) \;
					  +\;  a(\boldt)\,\partial_2^2a(\boldt)\,\partial_1^3 b(\boldt)                      \\[.8ex]
   	   &&		  +\; 3\,\partial_1^2a(\boldt) \partial_1\partial_2^2a(\boldt )\, b(\boldt)\;
	                  +\; 3\,\partial_1^2a(\boldt)\,\partial_2^2a(\boldt)\, \partial_1 b(\boldt)\;
                      +\; 3\,a(\boldt)\, \partial_1^2 \partial_2^2a(\boldt)\partial_1 b(\boldt) \\[.8ex]					  					
	   &&	      +\; 3\,\partial_1 a(\boldt)\partial_1^2 \partial_2^2a(\boldt)\,b(\boldt)\;
					  +\; 3\,\partial_1 a(\boldt)\,\partial_2^2a(\boldt)\,\partial _1^2 b(\boldt)\;
                      +\; 3\,a(\boldt)\,\partial_1 \partial_2^2a(\boldt)\partial_1^2 b(\boldt)  \\[.6ex]
       &&         +\; 6\,\partial_1 a(\boldt)\partial_1\partial_2^2 a(\boldt)\partial_1 b(\boldt)		\\[.8ex]
	   && +\; \partial_1^3 a(\boldt)\,a(\boldt)\partial_2^2b(\boldt)\;
	                  +\;  a(\boldt)\, \partial_1^3 a(\boldt)\, \partial_2^2b(\boldt) \;
					  +\;  a(\boldt)^2\,\partial_1^3\partial_2^2 b(\boldt)                      \\[.8ex]
   	   &&		  +\; 3\,\partial_1^2a(\boldt) \partial_1a(\boldt )\,\partial_2^2 b(\boldt)\;
	                  +\; 3\,\partial_1^2a(\boldt)\,a(\boldt)\, \partial_1 \partial_2^2b(\boldt)\;
                      +\; 3\,a(\boldt)\, \partial_1^2 a(\boldt)\partial_1\partial_2^2 b(\boldt) \\[.8ex]					  					
	   &&	      +\; 3\,\partial_1 a(\boldt)\partial_1^2 a(\boldt)\,\partial_2^2b(\boldt)\;
					  +\; 3\,\partial_1 a(\boldt)\,a(\boldt)\,\partial _1^2\partial_2^2 b(\boldt)\;
                      +\; 3\,a(\boldt)\,\partial_1 a(\boldt)\partial_1^2 b\partial_2^2(\boldt)  \\[.6ex]
       &&         +\; 6\,\partial_1 a(\boldt)\partial_1 a(\boldt)\partial_1\partial_2^2 b(\boldt)		  \\[.8ex]
	   && +\; 2\,\partial_1^3\partial_2 a(\boldt)\,\partial_2a(\boldt)b(\boldt)\;
	                  +\; 2\,\partial_2 a(\boldt)\, \partial_1^3\partial_2 a(\boldt)\, b(\boldt) \;
					  +\; 2\,\partial_2 a(\boldt)\partial_2a(\boldt)\,\partial_1^3 b(\boldt)                      \\[.8ex]
   	   &&		  +\; 3\!\cdot\! 2\,\partial_1^2\partial_2a(\boldt) \partial_1\partial_2a(\boldt )\, b(\boldt)\;
	                  +\; 3\!\cdot\! 2\,\partial_1^2\partial_2a(\boldt)\,\partial_2a(\boldt)\, \partial_1 b(\boldt)\;
                      +\; 3\!\cdot\! 2\,\partial_2a(\boldt)\, \partial_1^2\partial_2 a(\boldt)\partial_1 b(\boldt) \\[.8ex]					  					
	   &&	      +\; 3\!\cdot\! 2\,\partial_1\partial_2 a(\boldt)\partial_1^2\partial_2 a(\boldt)\,b(\boldt)\;
					  +\; 3\!\cdot\! 2\,\partial_1\partial_2 a(\boldt)\,\partial_2a(\boldt)\,\partial _1^2 b(\boldt)\;
                      +\; 3\!\cdot\! 2\,\partial_2a(\boldt)\,\partial_1\partial_2 a(\boldt)\partial_1^2 b(\boldt)  \\[.6ex]
       &&         +\; 6\!\cdot\! 2\,\partial_1\partial_2 a(\boldt)\partial_1\partial_2 a(\boldt)\partial_1 b(\boldt)
                                             	   \\[.8ex]
	   && +\; 2\,\partial_1^3\partial_2 a(\boldt)\,a(\boldt)\partial_2b(\boldt)\;
	                  +\; 2\, \partial_2a(\boldt)\, \partial_1^3 a(\boldt)\,\partial_2 b(\boldt) \;
					  +\; 2\,\partial_2a(\boldt)\,a(\boldt)\,\partial_1^3\partial_2 b(\boldt)                      \\[.8ex]
   	   &&		  +\; 3\!\cdot\! 2\,\partial_1^2\partial_2a(\boldt) \partial_1a(\boldt )\,\partial_2 b(\boldt)\;
	                  +\; 3\!\cdot\! 2\,\partial_1^2\partial_2a(\boldt)\,a(\boldt)\, \partial_1\partial_2 b(\boldt)\;
                      +\; 3\!\cdot\! 2\,\partial_2a(\boldt)\, \partial_1^2 a(\boldt)\partial_1\partial_2 b(\boldt) \\[.8ex]					  					
	   &&	      +\; 3\!\cdot\! 2\,\partial_1\partial_2 a(\boldt)\partial_1^2 a(\boldt)\,\partial_2 b(\boldt)\;
					  +\; 3\!\cdot\! 2\,\partial_1\partial_2 a(\boldt)\,a(\boldt)\,\partial _1^2\partial_2 b(\boldt)\;
                      +\; 3\!\cdot\! 2\,\partial_2a(\boldt)\,\partial_1 a(\boldt)\partial_1^2\partial_2 b(\boldt)  \\[.6ex]
       &&         +\; 6\!\cdot\! 2\,\partial_1\partial_2 a(\boldt)\partial_1 a(\boldt)\partial_1\partial_2 b(\boldt)
                                            	   \\[.8ex]
	   && +\; 2\,\partial_1^3 a(\boldt)\,\partial_2a(\boldt)\partial_2b(\boldt)\;
	                  +\; 2\, a(\boldt)\, \partial_1^3\partial_2 a(\boldt)\,\partial_2 b(\boldt) \;
					  +\; 2\, a(\boldt)\,\partial_2a(\boldt)\,\partial_1^3\partial_2 b(\boldt)                      \\[.8ex]
   	   &&		  +\; 3\!\cdot\! 2\,\partial_1^2a(\boldt) \partial_1\partial_2a(\boldt )\, \partial_2b(\boldt)\;
	                  +\; 3\!\cdot\! 2\,\partial_1^2a(\boldt)\,\partial_2a(\boldt)\, \partial_1\partial_2 b(\boldt)\;
                      +\; 3\!\cdot\! 2\,a(\boldt)\, \partial_1^2\partial_2 a(\boldt)\partial_1\partial_2 b(\boldt) \\[.8ex]					  					
	   &&	      +\; 3\!\cdot\! 2\,\partial_1 a(\boldt)\partial_1^2\partial_2 a(\boldt)\,\partial_2b(\boldt)\;
					  +\; 3\!\cdot\! 2\,\partial_1 a(\boldt)\,\partial_2a(\boldt)\,\partial _1^2\partial_2 b(\boldt)\;
                      +\; 3\!\cdot\! 2\,a(\boldt)\,\partial_1\partial_2 a(\boldt)\partial_1^2\partial_2 b(\boldt)  \\[.6ex]
       &&         +\; 6\!\cdot\! 2\,\partial_1 a(\boldt)\partial_1\partial_2 a(\boldt)\partial_1\partial_2 b(\boldt)&,				  	
   \end{array}
  $$}
  which has
   $(\!\!\mbox{\tiny $\begin{array}{c}5\\ 2\end{array}$}\!\!)
      (\!\!\mbox{\tiny $\begin{array}{c}4\\ 2\end{array}$}\!\!)=60$-many terms
	  (collapsed from an expansion of $3^5=243$-many terms).

\bigskip

\noindent
{\bf Notation 4.2.1.4. [differentiation via Leibniz rule on monomial]}\;
 Continuing Remark 4.2.1.2 and Example 4.2.1.3.
 While the Leibniz rule itself is straightforward,
  the resulting expression after its repeated applications can be a burdon.
 Some notations are introduced below for easier bookkeeping.
 \begin{itemize}
  \item[(1)]  ({\it differentiation of any order})\hspace{2em}
  Let
  $\vec{\Ptn}(\boldalpha, d)$ be the {\it set of ordered partitions of $\boldalpha$} into $d$-many summands.
  For a $\vec{\pi}=(\boldalpha_{(1)},\,\cdots\,,\boldalpha_{(d)})\in \vec{\Ptn}(\boldalpha, d)$, 
   denote
   \begin{eqnarray*}
     && \hspace{-1.6em}
	  [\partial^{\vec{\pi}}_{\mbox{\scriptsize\boldmath $t$}}]
	    P(\xi^1(\boldt),\,\cdots\,,\xi^n(\boldt))    \\
     &&	 :=\;
	   \partial^{\mbox{\scriptsize\boldmath $\alpha$}_{(1)}}_{\mbox{\scriptsize\boldmath $t$}}
		  \xi^1(\boldt)\,
		  \cdots\,
	   \partial^{\mbox{\scriptsize\boldmath $\alpha$}_{(d_1)}}_{\mbox{\scriptsize\boldmath $t$}}
		  \xi^1(\boldt)\:
	   \partial^{\mbox{\scriptsize\boldmath $\alpha$}_{(d_1+1)}}_{\mbox{\scriptsize\boldmath $t$}}
		  \xi^2(\boldt)\,
		  \cdots
      \partial^{\mbox{\scriptsize\boldmath $\alpha$}_{(d_1+d_2)}}_{\mbox{\scriptsize\boldmath $t$}}
		  \xi^2(\boldt)\:    \\
	 &&	\hspace{16em}
		 \cdots\cdots\,
      \partial^{\mbox{\scriptsize\boldmath $\alpha$}_{(d_1+\,\cdots\,+d_{n-1}+1)}}
	      _{\mbox{\scriptsize\boldmath $t$}}
		  \xi^n(\boldt)\,\cdots\,
		\partial^{\mbox{\scriptsize\boldmath $\alpha$}_{(d)}}_{\mbox{\scriptsize\boldmath $t$}}
		  \xi^n(\boldt)\,.
	\end{eqnarray*}
 Then, in terms of these,
  $$
     \partial^{\mbox{\scriptsize\boldmath $\alpha$}}_{\mbox{\scriptsize\boldmath $t$}}
	    P(\xi^1(\boldt),\,\cdots\,,\xi^n(\boldt))\;
	=\;
	  \sum_{\vec{\pi}
	                   \in \vec{\mbox{\scriptsize\it Ptn}}
					                (\mbox{\scriptsize\boldmath $\alpha$}, d)}\,
	         m_{\vec{\pi}}
			\cdot
         	[\partial^{\vec{\pi}}_{\mbox{\scriptsize\boldmath $t$}}]
	            P(\xi^1(\boldt),\,\cdots\,,\xi^n(\boldt))\,,
  $$
  where $m_{\vec{\pi}}\in {\Bbb Z}_{>0}$ are the coefficients in the expansion.
  
 \item[(2)]  ({\it first-order differentiation})\hspace{2em}
 For the case $|\alpha|=1$,
  let $\partial_t$ be the corresponding differentiable operator from the list $\partial_{t^1},\,\cdots\,,\partial_{t^l}$.
 Then,
   a $\vec{\pi}=(\boldalpha_{(1)},\,\cdots\,,\boldalpha_{(d)})\in \vec{\Ptn}(\boldalpha, d)$
   has exactly one summand $\boldalpha_{(j^{\prime})}$ that is non-zero.
 Assume that $d_1+\,\cdots\,+ d_{i-1} < j^{\prime} \le  d_1+\,\cdots\,+d_{i-1}+d_i$ and
  let $1\le j:=j^{\prime}-(d_1+\,\cdots\,+d_{i-1})\le d_i$.
 Then,
   \begin{eqnarray*}
     && \hspace{-1.6em}
	  [\partial^{\vec{\pi}}_{\mbox{\scriptsize\boldmath $t$}}]
	    P(\xi^1(\boldt),\,\cdots\,,\xi^n(\boldt))    \\
     &&	 =\;\;
		  \xi^1(\boldt)^{d_1}\, \cdots\,  \xi^{i-1}(\boldt)^{d_{i-1}}\:
		  \xi^i(\boldt)^{j-1}\,  \partial_t\xi^i(\boldt)\, \xi^i(\boldt)^{d_i-j}\;
          \xi^{i+1}(\boldt)^{d_{i+1}}\,\cdots\, \xi^n(\boldt)^{d_n}  \\
	 &&	=:\;
     ([\partial^{\vec{\pi}}_{\xi^{i_{\vec{\pi}}}}]P)^L(\xi^1(\boldt),\,\cdots\,, \xi^n(\boldt))
	   \cdot
	   \partial_t\xi^{i_{\vec{\pi}}}(\boldt)
	   \cdot
	 ([\partial^{\vec{\pi}}_{\xi^{i_{\vec{\pi}}}}]P)^R(\xi^1(\boldt),\,\cdots\,, \xi^n(\boldt))\,.	
	\end{eqnarray*}
 In terms of these,
  \begin{eqnarray*}
   \lefteqn{
    \partial_t P(\xi^1(\boldt),\,\cdots\,,\xi^n(\boldt))} \\
	&&
    =\; 	\partial^{\mbox{\scriptsize\boldmath $\alpha$}}_{\mbox{\scriptsize\boldmath $t$}}
	              P(\xi^1(\boldt),\,\cdots\,,\xi^n(\boldt))\;\;\;
	=\;
	  \sum_{\vec{\pi}
	                   \in \vec{\mbox{\scriptsize\it Ptn}}
					                (\mbox{\scriptsize\boldmath $\alpha$}, d)}\,
         	[\partial^{\vec{\pi}}_{\mbox{\scriptsize\boldmath $t$}}]
	            P(\xi^1(\boldt),\,\cdots\,,\xi^n(\boldt))  \\
   &&
    =\; 	
	   \sum_{\vec{\pi}
	                   \in \vec{\mbox{\scriptsize\it Ptn}}
					                (\mbox{\scriptsize\boldmath $\alpha$}, d)}\,
	  ([\partial^{\vec{\pi}}_{\xi^{i_{\vec{\pi}}}}]P)^L(\xi^1(\boldt),\,\cdots\,, \xi^n(\boldt))
	   \cdot
	   \partial_t\xi^{i_{\vec{\pi}}}(\boldt)
	   \cdot
	 ([\partial^{\vec{\pi}}_{\xi^{i_{\vec{\pi}}}}]P)^R(\xi^1(\boldt),\,\cdots\,, \xi^n(\boldt)) \\
	&& \hspace{16em}
	      \mbox{(a summation of $d$-many terms)}\\[.6ex]
	&& =\;
	   \sum_{i=1}^n
	    \sum_{\;\vec{\pi}
	                      \in \vec{\mbox{\scriptsize\it Ptn}}
					                (\mbox{\scriptsize\boldmath $\alpha$}, d);\,
						   i_{\vec{\pi}}=i  }\,
	  ([\partial^{\vec{\pi}}_{\xi^i}]P)^L(\xi^1(\boldt),\,\cdots\,, \xi^n(\boldt))
	   \cdot
	   \partial_t\xi^i(\boldt)
	   \cdot
	 ([\partial^{\vec{\pi}}_{\xi^i}]P)^R(\xi^1(\boldt),\,\cdots\,, \xi^n(\boldt))\\
	&& \hspace{16em}
	      \mbox{(an expression closer to the usual chain rule)}\,.
  \end{eqnarray*}
 \end{itemize}

\bigskip

\noindent
{\it Remark 4.2.1.5. $[\,$would-be chain rule in the noncommutative case$\,]\;$}
 Readers are recommended to compare the last expression in Notation~4.2.1.4, Item (2),
  with the expression for the chain rule in the commutative case
  $$
    \partial _t  f(y^1(\boldt),\,\cdots\,,y^n(\boldt))\;
     =\;  \sum_{i=1}^n\,
	        (\partial_{y^i}f)(y^1(\boldt),\,\cdots\,,y^n(\boldt))
	      \cdot \partial_ty^i(\boldt)\,.
  $$		
 In a sense, noncommutativity brings into the problem
  the necessity to distinguish `{\it which $\xi^i$ is involved}'
  when we takes the differentiation $\partial_t(P(\xi^1(\boldt)),\,\cdots\,,\,\xi^n(\boldt) )$.
 For that reason, each $\partial_{\xi^i}P(\xi^1,\,\cdots\,,\xi^n)$ splits into two factors,
 the {\it left factor $(\,\cdots\,)^L$} and the {\it right factor $(\,\cdots\,)^R$}
 that depends on this additional detail.
  
\bigskip

\noindent
{\it Remark 4.2.1.6. $[$standard presentation for mixed case$]$}\;
 For the situation to appear in the current notes,
   the values $\xi^1(\boldt), \,\cdots\,\xi^n(\boldt)$ for any $\boldt$ commute among themselves
   but not necessarily with their differentiations with respect to $\boldt$.
 The above discussion still applies to such situations.
 However, the explicit expression for
  $\partial^{\mbox{\scriptsize\boldmath $\alpha$}}_{\mbox{\scriptsize\boldmath $t$}}
	    P(\xi^1(\boldt),\,\cdots\,,\xi^n(\boldt))$
  after the expansion by Leibniz rule
 depends on how we represent $P(\xi^1(\boldt),\,\cdots\,,\xi^d(\boldt))$
  in terms of a product of $d_1$-many $\xi^1(\boldt)$'s, $\cdots$, $d_n$-many $\xi^n(\boldt)$'s.
 By convention, we will take
  $\xi^1(\boldt)^{d_1}\,\cdots\,\xi^n(\boldt)^{d_n}$
  as the standard presentation for $P(\xi^1(\boldt),\,\cdots\,,\xi^n(\boldt))$
  and the resulting expansion the standard expression for
   $\partial^{\mbox{\scriptsize\boldmath $\alpha$}}_{\mbox{\scriptsize\boldmath $t$}}
	    P(\xi^1(\boldt),\,\cdots\,,\xi^n(\boldt))$ and
   $[\partial^{\vec{\pi}}_{\mbox{\scriptsize\boldmath $t$}}]
	            P(\xi^1(\boldt),\,\cdots\,,\xi^n(\boldt))$'s respectively.	

\bigskip

\noindent
{\bf Notation 4.2.1.7. [first-order differentiation via Leibniz rule on polynomial]}\;
 Continuing Remark 4.2.1.2, Example 4.2.1.3, and Notation 4.2.1.4, Item (2);
 and recall $\bold\alpha$ and the associated $\partial_t$.			
 For a multi-degree $\boldd=(d_1,\,\cdots\,,d_n)\in {\Bbb Z}_{\ge 0}^n$,
  denote
   the monomial
   $(\xi^1)^{d_1}\,\cdots\,(\xi^n)^{d_n}$
    by $\boldxi^{\scriptsizeboldd}$ or $P_{\scriptsizeboldd}(\boldxi)$  interchangeably,
   $\xi^1(\boldt)^{d_1}\,\cdots\,\xi^n(\boldt)^{d_n}$
    by $\boldxi(\boldt)^{\scriptsizeboldd}$ or $P_{\scriptsizeboldd}(\boldxi(\boldt))$
	interchangeably,  and the total degree $|\boldd|:= d_1+\,\cdots\,+d_n$. 	
 Let
   $$
     P(\boldxi)\;\;
	  :=\;P(\xi^1,\,\cdots\,, \xi^n)\;
	  =\;  \sum_{d=0}^{\mbox{\tiny $\bullet$}}
	          \sum_{\;\,\scriptsizeboldd,\, |\scriptsizeboldd|=d}
	            c_{\scriptsizeboldd}\,\boldxi^{\scriptsizeboldd}\;
	  =:\;  \sum_{d=0}^{\mbox{\tiny $\bullet$}}
	          \sum_{\;\,\scriptsizeboldd,\, |\scriptsizeboldd|=d}
	            c_{\scriptsizeboldd}\,  P_{\scriptsizeboldd}(\boldxi)\;		
   $$			
   be a polynomial in $(\xi^1,\,\cdots\,,\xi^n)$ with coefficients commutative with all of
    $\xi^1, \,\cdots\,,\xi^n$.
 Then
  \begin{eqnarray*}
   \lefteqn{
    \partial_t P(\boldxi(\boldt))\;
	 =\;  \sum_{d=0}^{\mbox{\tiny $\bullet$}}
	          \sum_{\;\,\scriptsizeboldd, |\scriptsizeboldd|=d}
	            c_{\scriptsizeboldd}\,
				\partial_t(\boldxi(\boldt)^{\scriptsizeboldd})  }\\
    &&
      =\;	
	  \sum_{d=0}^{\mbox{\tiny $\bullet$}}
	          \sum_{\;\,\scriptsizeboldd,\, |\scriptsizeboldd|=d}
	            c_{\scriptsizeboldd}\,			
		\sum_{\vec{\pi}
	         \in \vec{\mbox{\scriptsize\it Ptn}}
				    (\mbox{\scriptsize\boldmath $\alpha$}, d)}\,
	  ([\partial^{\vec{\pi}}_{\xi^{i_{(\vec{\pi},\,\tinyboldd)}}}]
	             P_{\scriptsizeboldd})^L(\boldxi(\boldt))
	   \cdot
	   \partial_t\xi^{i_{(\vec{\pi},\,\tinyboldd)}}(\boldt)
	   \cdot
	 ([\partial^{\vec{\pi}}_{\xi^{i_{(\vec{\pi},\,\tinyboldd)}}}]
	             P_{\scriptsizeboldd})^R(\boldxi(\boldt))\\		
   &&
      =\;	
	  \sum_{d=0}^{\mbox{\tiny $\bullet$}}
	   \sum_{\;\,\scriptsizeboldd,\, |\scriptsizeboldd|=d}\,
		\sum_{\vec{\pi}
	         \in \vec{\mbox{\scriptsize\it Ptn}}
				    (\mbox{\scriptsize\boldmath $\alpha$}, d)}\,
	  ([\partial^{\vec{\pi}}_{\xi^{i_{(\vec{\pi},\,\tinyboldd)}}}]
	             P_{(\scriptsizeboldd)})^L(\boldxi(\boldt))
	   \cdot
	   \partial_t\xi^{i_{(\vec{\pi},\,\tinyboldd)}}(\boldt)
	   \cdot
	 ([\partial^{\vec{\pi}}_{\xi^{i_{(\vec{\pi},\,\tinyboldd)}}}]
	             P_{(\scriptsizeboldd)})^R(\boldxi(\boldt))\\
   &&	
     =\;
     \sum_{i=1}^n\,
	   \sum_{d=0}^{\mbox{\tiny $\bullet$}}
	   \sum_{\;\,\scriptsizeboldd,\, |\scriptsizeboldd|=d}
	   \sum_{\;\vec{\pi}
	                         \in \vec{\mbox{\scriptsize\it Ptn}}
				            (\mbox{\scriptsize\boldmath $\alpha$}, d),\,i_{(\vec{\pi},\,\tinyboldd)}=i}\,							
	  ([\partial^{\vec{\pi}}_{\xi^i}]
	             P_{(\scriptsizeboldd)})^L(\boldxi(\boldt))
	   \cdot
	   \partial_t\xi^i(\boldt)
	   \cdot
	 ([\partial^{\vec{\pi}}_{\xi^i}]
	             P_{(\scriptsizeboldd)})^R(\boldxi(\boldt))\,.
  \end{eqnarray*}	
  Here,
   $P_{(\scriptsizeboldd)}:= c_{\scriptsizeboldd}P_{\scriptsizeboldd}$
   is the multi-degree $\boldd$ component of the polynomial $P$.
 This is the expression --- a ``virtual chain rule" in some sense ---  we will need for the current notes.

\bigskip

\begin{flushleft}
{\bf The main technical issue}
\end{flushleft}
Continuing now the setting and the study in Sec.\ 4.1, via the Leibniz rule one has at first
 an expansion to
   \begin{eqnarray*}
     \left(\rule{0ex}{.9em}\right.\!
	   \mbox{\Large $\frac{\partial^{|\mbox{\boldmath\tiny $\alpha$}|}}
	                 {\partial \mbox{\boldmath\small $t$}^{\mbox{\boldmath\tiny $\alpha$}}}$}
	   \varphi_T^{\sharp}\!\left.\rule{0ex}{.9em}\right)(f)
	   & =
	   & \mbox{\Large$\frac{\partial^{|\mbox{\boldmath\tiny $\alpha$}|}}
	                 {\partial \mbox{\boldmath\small $t$}^{\mbox{\boldmath\tiny $\alpha$}}}$}
	    (\varphi_T^{\sharp}(f))            \\[.8ex]
	 & =	
	    & \sum_{j=1}^s
		    \sum_{\mbox{\tiny $\begin{array}{c}0\le d_i\le r-1 \\[1.2ex]  1\le i\le n \end{array}$}}
		    \mbox{\Large
			   $\frac{\partial^{|\mbox{\boldmath\tiny $\alpha$}|}}
	                    {\partial \mbox{\boldmath\small $t$}^{\mbox{\boldmath\tiny $\alpha$}}}$}
              \left(					
			   a^{f;\,q_j}_{(d_1,\,\cdots\,,\,d_n)}\cdot
			      (\varphi_T^{\sharp, (j)}(y^1))^{d_1}\,\cdots\,
			  	  (\varphi_T^{\sharp, (j)}(y^n))^{d_n}
				  \right)
  \end{eqnarray*}
  over $U^{\prime}_f$.
Since the coefficients $a^{f;\,q_j}_{(d_1,\,\cdots\,,\,d_n)}$
  may depend also on $\mbox{\boldmath $t$}$,
 the expansion to the above summation via repeated applications of the Leibniz rule
 involves, in general, terms that depend on
    $\left(\frac{\partial^{|\mbox{\boldmath\tiny $\alpha$}_1|}}
		       {\partial\mbox{\boldmath\scriptsize $t$}^{\mbox{\boldmath\tiny $\alpha$}_1}}
		              (\varphi^{\sharp}_{\mbox{\boldmath\scriptsize $t$}}(y^1) )\,,\;
					  \cdots\,,\;
         \frac{\partial^{|\mbox{\boldmath\tiny $\alpha$}_n|}}
		           {\partial\mbox{\boldmath\scriptsize $t$}^{\mbox{\boldmath\tiny $\alpha$}_n}}
		              (\varphi^{\sharp}_{\mbox{\boldmath\scriptsize $t$}}(y^n) )\right)$'s
	   with
	   $\mbox{\boldmath $\alpha$}_1+ \cdots+\mbox{\boldmath $\alpha$}_n
	      < \mbox{\boldmath $\alpha$}$.
The main technical issue is:
  \begin{itemize}
    \item[\LARGE $\cdot$]  \parbox[t]{37em}{\it
    How to express the summation of such lower-order derivative terms in terms of\\
     $\left(\frac{\partial^{|\mbox{\boldmath\tiny $\alpha$}_1|}}
		           {\partial\mbox{\boldmath\scriptsize $t$}^{\mbox{\boldmath\tiny $\alpha$}_1}}
		              (\varphi^{\sharp}_{\mbox{\boldmath\scriptsize $t$}}(y^1) )\,,\;
					  \cdots\,,\;
         \frac{\partial^{|\mbox{\boldmath\tiny $\alpha$}_n|}}
		           {\partial\mbox{\boldmath\scriptsize $t$}^{\mbox{\boldmath\tiny $\alpha$}_n}}
		              (\varphi^{\sharp}_{\mbox{\boldmath\scriptsize $t$}}(y^n) )\right)$'s
	   with
	   $\mbox{\boldmath $\alpha$}_1+ \cdots+\mbox{\boldmath $\alpha$}_n
	      = \mbox{\boldmath $\alpha$}$
    alone\\  so that, in the end,
	  $(\frac{\partial^{|\mbox{\boldmath\tiny $\alpha$}|}}
	                 {\partial\mbox{\boldmath\scriptsize $t$}^{\mbox{\boldmath\tiny $\alpha$}}}
	        \varphi_T^{\sharp})(f)$
	  depends only on
	        $\left(\frac{\partial^{|\mbox{\boldmath\tiny $\alpha$}_1|}}
		           {\partial\mbox{\boldmath\scriptsize $t$}^{\mbox{\boldmath\tiny $\alpha$}_1}}
		              (\varphi^{\sharp}_{\mbox{\boldmath\scriptsize $t$}}(y^1) )\,,\;
					  \cdots\,,\;
         \frac{\partial^{|\mbox{\boldmath\tiny $\alpha$}_n|}}
		           {\partial\mbox{\boldmath\scriptsize $t$}^{\mbox{\boldmath\tiny $\alpha$}_n}}
		              (\varphi^{\sharp}_{\mbox{\boldmath\scriptsize $t$}}(y^n) )\right)$'s
	   with
	   $\mbox{\boldmath $\alpha$}_1+ \,\cdots\,+\mbox{\boldmath $\alpha$}_n
	      = \mbox{\boldmath $\alpha$}$.
	}
  \end{itemize}

\bigskip

\begin{flushleft}
{\bf  4.2.2\hspace{0.6em} The case of first-order derivations}
\end{flushleft}
{\bf Proposition 4.2.2.1. [first order derivation]}$\;$ {\it
 Denote by $\partial_{\mbox{\LARGE $\cdot$}}$
   any of $\partial/\partial t^1,\,\cdots\,,\, \partial/\partial t^l$.
 Let $f\in C^{\infty}(Y)$,
   regarded also as an element in $C^{\infty}(X_T\times Y)$
      through the inclusion $\pr_Y^{\sharp}:C^{\infty}(Y)\rightarrow C^{\infty}(X_T\times Y)$
	  whenever necessary.
 Then
   $(\partial_{\mbox{\LARGE $\cdot$}}\varphi_T^{\sharp})(f)\;$
     ($=\partial_{\mbox{\LARGE $\cdot$}}(\varphi_T^{\sharp}(f))$)
   depends
   on $((\partial_{\mbox{\LARGE $\cdot$}} \varphi_T^{\sharp})(y^1),\,
                    \cdots\,,\, (\partial_{\mbox{\LARGE $\cdot$}}\varphi_T^{\sharp})(y^n))\;$
   ($=(\partial_{\mbox{\LARGE $\cdot$}}(\varphi_T^{\sharp}(y^1)),\,
                    \cdots\,,\,  \partial_{\mbox{\LARGE $\cdot$}}(\varphi_T^{\sharp}(y^n)))$).
 More explicitly,
  recall
   the defining equations
    $$
      \chi^{(i)}_{\varphi_T}\;
	  :=\;    \determinant(y^i\cdot\Id_{r\times r}-\varphi_T^{\sharp}(y^i))\;
	   =\; (y^i)^r + a^{(i)}_{r-1}(y^i)^{r-1}+\,\cdots\,+ a^{(i)}_1y^i+a^{(i)}_0\,,
    $$	
      $i = 1,\,\ldots\,,\, n$,
      for the spectral locus $\varSigma_{\varphi_T;\{y^1,\,\cdots\,,\, y^n\}}\subset X_T\times Y$
      of $\varphi_T$ over $U^{\prime}_f$
   and that,
     around $q_j\in \varSigma_{\varphi_T;\{y^1,\,\cdots\,,\, y^n\}}$,
     there exist
	  $a^{f;q_j}_{(d_1,\,\cdots\,, d_n)}\in C^{\infty}(U^{\prime}_f)$ and	
	  $Q^{f;\,q_j}_{(i)}\in C^{\infty}(U^{\prime}_f\times V_f^{(j)})$,
	  $0\le d_i\le r-1$, $i=1,\,\ldots\,,\, n$,
	 such that
   	  $$
		  f\;=\;
		    \sum_{\mbox{\tiny $\begin{array}{c}0\le d_i\le r-1 \\[1.2ex]  1\le i\le n \end{array}$}}
			    a^{f;\,q_j}_{(d_1,\,\cdots\,,\,d_n)}
				  \cdot (y^1)^{d_1}\,\cdots\,(y^n)^{d_n} \;
				+\;  \sum_{i=1}^n\, Q^{f;\,q_j}_{(i)}\,\chi_{\varphi_T}^{(i)}\,.
      $$
  Then, in terms of these data encoded in $\varphi_T$ and $f$, one has
   \begin{eqnarray*}
    (\partialdot \varphi_T^{\sharp})(f)
	    & = &
          \sum_{j=1}^s
		    \sum_{\mbox{\tiny $\begin{array}{c}0\le d_i\le r-1 \\[1.2ex]  1\le i\le n \end{array}$}}
			    a^{f;\,q_j}_{(d_1,\,\cdots\,,\,d_n)}\cdot
				 \partialdot
				   \!\left(
			       (\varphi_T^{\sharp, (j)}(y^1))^{d_1}\,\cdots\,
				   (\varphi_T^{\sharp, (j)}(y^n))^{d_n}
				   \right)    \\
     && +\; \sum_{j=1}^s\,
		    \sum_{i=1}^n\,
               \tilde{\varphi}_T(Q_{(i)}^{f;\,q_j})\,
               \sum_{d=1}^r\,			
			    a^{(i)}_d\cdot
				 \partialdot
				   \!\left(
			       (\varphi_T^{\sharp, (j)}(y^i))^d
				   \right)\,,      		
   \end{eqnarray*}
 from which
  one can use the Leibniz rule to further express
     $\partialdot ((\varphi_T^{\sharp, (j)}(y^1))^{d_1}\,\cdots\,
				                (\varphi_T^{\sharp, (j)}(y^n))^{d_n})$'s
   and
     $ \partialdot ((\varphi_T^{\sharp, (j)}(y^i))^d)$'s
   into the desired form.
 (Here, $a_r^{(i)}=1$ by convention, for $i=1,\,\ldots\,,\,n$.)
} 
  
\bigskip

\noindent{\it Proof.}
 We proceed in three steps.
 
 \medskip
 
 \noindent
  {\it Step $(1):\,$ Identifying the seemingly problematic terms}\hspace{1em}
 Over $U^{\prime}_f$, one has
  
   {\footnotesize
    \begin{eqnarray*}
   (\partialdot \varphi_T^{\sharp})(f)\;
	   & =
	   &  \sum_{j=1}^s
		    \sum_{\mbox{\tiny $\begin{array}{c}0\le d_i\le r-1 \\[1.2ex]  1\le i\le n \end{array}$}}
		    \partialdot
               \left(					
			    a^{f;\,q_j}_{(d_1,\,\cdots\,,\,d_n)}\cdot
			       (\varphi_T^{\sharp, (j)}(y^1))^{d_1}\,\cdots\,
				   (\varphi_T^{\sharp, (j)}(y^n))^{d_n}
				   \right)    \\[.8ex]
     &  & \hspace{-4em}=\;\; \sum_{j=1}^s
		    \sum_{\mbox{\tiny $\begin{array}{c}0\le d_i\le r-1 \\[1.2ex]  1\le i\le n \end{array}$}}
		    \partialdot	a^{f;\,q_j}_{(d_1,\,\cdots\,,\,d_n)}
			  \cdot (\varphi_T^{\sharp, (j)}(y^1))^{d_1}\,\cdots\,
				               (\varphi_T^{\sharp, (j)}(y^n))^{d_n}                   \\[.8ex]
     && \hspace{-2em}+\;
  	        \sum_{j=1}^s
		    \sum_{\mbox{\tiny $\begin{array}{c}0\le d_i\le r-1 \\[1.2ex]  1\le i\le n \end{array}$}}
			\sum_{d_1^{\prime}=0}^{d_1-1}
			    a^{f;\,q_j}_{(d_1,\,\cdots\,,\,d_n)}\cdot
			       (\varphi_T^{\sharp, (j)}(y^1))^{d_1^{\prime}} \cdot
				      \partialdot \varphi_T^{\sharp, (j)}(y^1)\cdot
					  (\varphi_T^{\sharp, (j)}(y^1))^{d_1-1-d_1^{\prime}}
				     \,\cdots\,
				   (\varphi_T^{\sharp, (j)}(y^n))^{d_n}                          \\[.8ex]       	 				   
     && \hspace{-2em}+\;\cdots\;   \\[2.4ex]
     && \hspace{-2em}+\;
  	        \sum_{j=1}^s
		    \sum_{\mbox{\tiny $\begin{array}{c}0\le d_i\le r-1 \\[1.2ex]  1\le i\le n \end{array}$}}
			\sum_{d_n^{\prime}=0}^{d_n-1}
			    a^{f;\,q_j}_{(d_1,\,\cdots\,,\,d_n)}\cdot
			       (\varphi_T^{\sharp, (j)}(y^1))^{d_1}
				     \,\cdots\,
				   (\varphi_T^{\sharp, (j)}(y^n))^{d_n^{\prime}}
				      \cdot   \partialdot\varphi_T^{\sharp, (j)}(y^n)
					  \cdot (\varphi_T^{\sharp, (j)}(y^n))^{d_n-1-d_n^{\prime}}\,.					
  \end{eqnarray*}}
 Terms that are not manifestly multilinear in
  $(\partialdot(\varphi_T^{\sharp}(y^1)),\,
                    \cdots\,,\,  \partialdot(\varphi_T^{\sharp}(y^n)))$
 lie in the first cluster 		
     $\sum_{j=1}^s
		    \sum_{\mbox{\tiny $\begin{array}{c}0\le d_i\le r-1 \\[1.2ex]  1\le i\le n \end{array}$}}
		    \partialdot	a^{f;\,q_j}_{(d_1,\,\cdots\,,\,d_n)}
			  \cdot (\varphi_T^{\sharp, (j)}(y^1))^{d_1}\,\cdots\,
				               (\varphi_T^{\sharp, (j)}(y^n))^{d_n}$.
							
 Since terms with different $j$'s are independent from each other and, hence, can be treated separately,
  to avoid the burden of notation and without loss of generality,
  we assume from now on in the proof that $s=1$ and drop the $j$ label altogether.
  
 \medskip

 \noindent
 {\it Step $(2):\;$ Defining equations of the spectral locus come to play}\hspace{1em}
 Recall the generators
  $$
    \begin{array}{rll}
    \chi^{(i)}_{\varphi_T}
	  & :=
	     & \determinant(y^i\cdot\Id_{r\times r}-\varphi_T^{\sharp}(y^i))  \\[.8ex]
	 & \;=:\;
	     &  (y^i)^r + a^{(i)}_{r-1}(y^i)^{r-1}+\,\cdots\,+ a^{(i)}_1y^i+a^{(i)}_0\\[1.2ex]
	  && \in\;  C^{\infty}(U_f^{\prime})[y^i]\;
	           \subset\; C^{\infty}(U_f^{\prime}\times V_f)\,,
    \end{array}	
  $$	
  for $i = 1,\,\ldots\,,\, n$,
  of the ideal $I_{\varphi_T;\{y^1,\,\cdots\,,\, y^n\}}|_{U_f^{\prime}\times V_f}$
  that defines the spectral locus $\varSigma_{\varphi_T;\{y^1,\,\cdots\,,\, y^n\}}$
  of $\varphi_T$ over $U_f^{\prime}$.
 By construction,
   there exist
    $Q^{f;\,q}_{(1)},\,\cdots\,,\, Q^{f;\,q}_{(n)}
	      \in C^{\infty}(U^{\prime}_f\times V_f)$
  such that
  $$
    f\; =\;
      \sum_{\mbox{\tiny $\begin{array}{c}0\le d_i\le r-1 \\[1.2ex]  1\le i\le n \end{array}$}}
			    a^{f;\,q}_{(d_1,\,\cdots\,,\,d_n)}\cdot(y^1)^{d_1}\,\cdots\,(y^n)^{d_n}\;
      +\; \sum_{i^{\prime}=1}^n
	            Q^{f;\,q}_{(i^{\prime})}\,\chi^{(i^{\prime})}_{\varphi_T}
  $$
  in $C^{\infty}(U^{\prime}_f\times V_f)$.
Since $f\in C^{\infty}(V_f)$ and, hence, has no dependence on {\boldmath  $t$},
 applying $\partialdot$ to both sides of the above identity gives
 $$
   0\;=\;
     \sum_{\mbox{\tiny $\begin{array}{c}0\le d_i\le r-1 \\[1.2ex]  1\le i\le n \end{array}$}}
			    \partialdot a^{f;\,q}_{(d_1,\,\cdots\,,\,d_n)}\cdot(y^1)^{d_1}\,\cdots\,(y^n)^{d_n}\;
         +\; \sum_{i^{\prime}=1}^n
		               \partialdot Q^{f;\,q}_{(i^{\prime})}
		                                  \cdot  \chi^{(i^{\prime})}_{\varphi_T}\;
         +\; \sum_{i^{\prime\prime}=1}^n
		              Q^{f;\,q}_{(i^{\prime\prime})}\cdot
		                                  \partialdot  \chi^{(i^{\prime\prime})}_{\varphi_T}\,.
 $$
 Applying $\tilde{\varphi}_T^{\sharp}$ to this identity   and
   noticing that $\tilde{\varphi}_T^{\sharp}(\chi_{\varphi_T}^{(i^{\prime})})=0$
      for $i^{\prime}=1,\,\ldots\,,n$,
 one now has
 $$
   \sum_{\mbox{\tiny $\begin{array}{c}0\le d_i\le r-1 \\[1.2ex]  1\le i\le n \end{array}$}}
	  \partialdot a^{f;\,q}_{(d_1,\,\cdots\,,\,d_n)}\cdot
	   (\varphi_T^{\sharp}(y^1))^{d_1}\,\cdots\,(\varphi_T^{\sharp}(y^n))^{d_n}\;
   =\;  -\,\sum_{i^{\prime\prime}=1}^n
		            \tilde{\varphi}_T^{\sharp}(Q^{f;\,q}_{(i^{\prime\prime})})
					    \cdot   \tilde{\varphi}_T^{\sharp}
						                      (\partialdot  \chi^{(i^{\prime\prime})}_{\varphi_T})\,.
 $$	

 \medskip
 
 \noindent
 {\it Step $(3):\;$ Understanding
      $\,\tilde{\varphi}_T^{\sharp}
	      (\partialdot\chi_{\varphi_T}^{(i^{\prime\prime})})$}\hspace{1em}
 Explicitly, one has
  $$
    \partialdot\chi_{\varphi_T}^{(i^{\prime\prime})}\;
	 =\;  \partialdot a^{(i^{\prime\prime})}_{r-1}\cdot(y^{i^{\prime\prime}})^{r-1}\,
			+\,\cdots\,
	        +\, \partialdot a^{(i^{\prime\prime})}_1\cdot y^{i^{\prime\prime}}\,
			+\,\partialdot a^{(i^{\prime\prime})}_0
  $$
 and, hence,
  $$
    \tilde{\varphi}_T^{\sharp}
	  (\partialdot\chi_{\varphi_T}^{(i^{\prime\prime})})\;
	 =\;   \partialdot a^{(i^{\prime\prime})}_{r-1}
			        \cdot(\varphi_T^{\sharp}(y^{i^{\prime\prime}}))^{r-1}\,
			+\,\cdots\,
	        +\, \partialdot a^{(i^{\prime\prime})}_1
			        \cdot \varphi_T^{\sharp}(y^{i^{\prime\prime}})\,
			+\,\partialdot a^{(i^{\prime\prime})}_0\,,
  $$
  for $i^{\prime\prime}=1,\,\ldots\,,\, n$.
 On the other hand,
  $$
    \tilde{\varphi}_T^{\sharp}
	  (\chi_{\varphi_T}^{(i^{\prime\prime})})\;
	 =\; (\varphi_T^{\sharp}(y^{i^{\prime\prime}}))^r \,
	        +\, a^{(i^{\prime\prime})}_{r-1}
			        \cdot(\varphi_T^{\sharp}(y^{i^{\prime\prime}}))^{r-1}\,
			+\,\cdots\,
	        +\,  a^{(i^{\prime\prime})}_1
			        \cdot \varphi_T^{\sharp}(y^{i^{\prime\prime}})\,
			+\,  a^{(i^{\prime\prime})}_0\;
	 =\; 0
  $$
  and, hence,
  $$
     \partialdot(
	   \tilde{\varphi}_T^{\sharp}
	  (\chi_{\varphi_T}^{(i^{\prime\prime})}))\;=\; 0\,.
  $$
  Which gives
   \begin{eqnarray*}
     \tilde{\varphi}_T^{\sharp}
	  (\partialdot\chi_{\varphi_T}^{(i^{\prime\prime})})\;
      & =
	  & -\,\left[\,
	       \partialdot(\varphi_T^{\sharp}(y^{i^{\prime\prime}}))^r \,
	        +\, a^{(i^{\prime\prime})}_{r-1}
			        \cdot \partialdot(\varphi_T^{\sharp}(y^{i^{\prime\prime}}))^{r-1}\,
			+\,\cdots\,
	        +\,  a^{(i^{\prime\prime})}_1
			        \cdot \partialdot\varphi_T^{\sharp}(y^{i^{\prime\prime}})  \,\right]  \\
     && \hspace{-7em}
			=\;\;
	        -\,\left[\rule{0ex}{1.6em}\right.\,
	           \sum_{k_r=0}^{r-1}
			     (\varphi_T^{\sharp}(y^{i^{\prime\prime}}))^{k_r}
				   \cdot   \partialdot \varphi_T^{\sharp}(y^{i^{\prime\prime}})
				   \cdot   (\varphi_T^{\sharp}(y^{i^{\prime\prime}}))^{r-1-k_r}  \\				   
     &&  \hspace{-3em}
   	        +\;  a_{r-1}^{(i^{\prime\prime})}\cdot
			       \sum_{k_{r-1}=0}^{r-2}
			     (\varphi_T^{\sharp}(y^{i^{\prime\prime}}))^{k_{r-1}}
				   \cdot   \partialdot \varphi_T^{\sharp}(y^{i^{\prime\prime}})
				   \cdot   (\varphi_T^{\sharp}(y^{i^{\prime\prime}}))^{r-2-k_{r-1}}  \;
				   +\; \cdots\; \\		
     && \hspace{-2.4em}
	        + \; a_2^{(i^{\prime\prime})}\cdot
			       \left(\partialdot\varphi_T^{\sharp}(y^{i^{\prime\prime}})
				        \cdot \varphi_T^{\sharp}(y^{i^{\prime\prime}})\,
                        +\, 	\varphi_T^{\sharp}(y^{i^{\prime\prime}})
				                  \cdot \partialdot \varphi_T^{\sharp}(y^{i^{\prime\prime}})	\right)\;
            +\;  a^{(i^{\prime\prime})}_1
			        \cdot \partialdot\varphi_T^{\sharp}(y^{i^{\prime\prime}}) 	
            \left.\rule{0ex}{1.6em}\right].	
   \end{eqnarray*}

  This concludes the proof of the proposition.
  
\noindent\hspace{40.8em}$\square$

\bigskip

\begin{flushleft}
{\bf  4.2.3\hspace{0.6em} Generalization to derivations of any order}
\end{flushleft}
Thinking deep enough of the case of first order derivations leads one to its generalization
 to all higher-order situations.
To state the proposition, with the notation from the previous theme,
 note that
  by shrinking $U^{\prime}_f$ if necessary and applying the Generalized Division Lemma repeatingly, 
   first to $f$, then to $Q^{f,q_j}_{(i)}$ , $\cdots$, and so on,
   one has
   \begin{eqnarray*}
   f  & =
      &  R^{f;\,q_j}_0\; + \; \sum_{i=1}^n R^{f;\,q_j}_{(i)}\,\chi_{\varphi_T}^{(i)}\;
	         +\;  \sum_{i_1,\,i_2=1}^n
	                    R^{f;\,q_j}_{(i_1,\,i_2)}\,
	                      \chi_{\varphi_T}^{(i_1)}\chi_{\varphi_T}^{(i_2)}\; +\;\cdots\cdots\; \\
   &&  \hspace{1em}
          +\; \sum_{i_1,\,\cdots\,,\,i_k=1}^n
                   R^{f;\,q_j}_{(i_1,\,\cdots\,,\,i_k)}\,
	                     \chi_{\varphi_T}^{(i_1)}\,\cdots\,\chi_{\varphi_T}^{(i_k)}\;
	      +\; \sum_{i_1,\,\cdots\,,\, i_{k+1}=1}^n	
   	               Q^{f;\,q_j}_{(i_1,\,\cdots\,,\,i_{k+1})}\,
	                     \chi_{\varphi_T}^{(i_1)}\,\cdots\,\chi_{\varphi_T}^{(i_{k+1})}\\
   &\;=:
      & R^{f;\,q_j}[k]\;
		  +\; \sum_{i_1,\,\cdots\,,\, i_{k+1}=1}^n	
   	               Q^{f;\,q_j}_{(i_1,\,\cdots\,,\,i_{k+1})}\,
	                     \chi_{\varphi_T}^{(i_1)}\,\cdots\,\chi_{\varphi_T}^{(i_{k+1})}
  \end{eqnarray*}
 around each $q_j$, for any $k\in {\Bbb Z}_{\ge 1}$.
Here,
   $R^{f;\,q_j}_0,\, R^{f;\,q_j}_{(i)},\,\cdots\,,\, R^{f;q_j}_{(i_1,\,\cdots\,,\,i_k)}$,
   and, hence, $R^{f;\,q_j}[k]$ are all in
   $C^{\infty}(U^{\prime}_f)[y^1,\,\cdots\,, y^n]
    \subset C^{\infty}(U^{\prime}_f\times V^{(j)}_f)$.
Denote the multi-degree of a summand of such a polynomial by $\boldd=(d_1,\,\cdots\,, d_n)$
 and $(y^1)^{d_1}\,\cdots\,(y^n)^{d_n}$ by $\boldy^{\scriptsizeboldd}$.
	
\bigskip

\noindent
{\bf Proposition 4.2.3.1. [derivation of any order]}$\;$ {\it
 Let
  $\partial^{\mbox{\boldmath \scriptsize $\alpha$}}$	be a derivation of order $k\ge 1$
  with respect to the coordinates $\boldt=(t^1,\,\cdots\,,t^l)$ of $T$ .
 Then,  over $U_f^{\prime}$,
  $$
    (\partialboldalpha\varphi_T^{\sharp})(f)\;
	  :=\;  \partialboldalpha(\varphi_T^{\sharp}(f))\;
	  =\; \sum_{j=1}^s
   	           R^{f;q_j}[k]
			   \!\!\left.\rule{0ex}{0.8em}\right|
			                    _{\scriptsizeboldy^{\tinyboldd}
 			                              \rightsquigarrow\,
								          \partial^{\tinyboldalpha}
								   (\varphi_T^{\sharp, (j)}(\scriptsizeboldy^{\tinyboldd}))}
  $$
 and, hence, depends only on
	        $\left(\frac{\partial^{|\tinyboldalpha_1|}}
		           {\partial\scriptsizeboldt^{\tinyboldalpha_1}}
		              (\varphi^{\sharp, (j)}_{\scriptsizeboldt}(y^1) )\,,\;
					  \cdots\,,\;
         \frac{\partial^{|\tinyboldalpha_n|}}
		           {\partial\scriptsizeboldt^{\tinyboldalpha_n}}
		              (\varphi^{\sharp, (j)}_{\scriptsizeboldt}(y^n) )\right)$'s,
      $j=1,\,\ldots\,,\, s$,					
	   with
	   $\boldalpha_1+ \,\cdots\,+\boldalpha_n = \boldalpha$.
 Here,
   $\boldy^{\scriptsizeboldd}
 	   \rightsquigarrow
	     \partial^{\scriptsizeboldalpha}
				(\varphi_T^{\sharp, (j)}(\boldy^{\scriptsizeboldd}))$
  means
   ``{\sl
      the replacement  of  $\boldy^{\scriptsizeboldd}$
 	   by $\partial^{\scriptsizeboldalpha}
				(\varphi_T^{\sharp, (j)}(\boldy^{\scriptsizeboldd}))$
       for all multi-degree-$\boldd$ summands of $R^{f;q_j}[k]$,
	   $\boldd$ running from $(0,\,\cdots\,,0)$ to $((k+1)r-1,\,\cdots\,,\,(k+1)r-1)$}".
}
  
\bigskip

\begin{proof}
 Since $|\boldalpha|=k$,
 there is always at least one
  $\chi_{\varphi_T}^{(\,\mbox{\tiny $\bullet$}\,)}$ factor left
  in the summands of the final expansion of
   $\partialboldalpha
     ( Q^{f;\,q_j}_{(i_1,\,\cdots\,,\,i_{k+1})}\,
	                     \chi_{\varphi_T}^{(i_1)}\,\cdots\,\chi_{\varphi_T}^{(i_{k+1})} ) $
   through repeating the Leibniz rule.
 Together with the identity
   $\tilde{\varphi}_T^{\sharp}(\chi_{\varphi_T}^{(\,\mbox{\tiny $\bullet$})})=0$,
 one has
 $$
   (\partialboldalpha\varphi_T^{\sharp})(f)\;
    =\;  \sum_{j=1}^s\partialboldalpha (R^{f;\,q_j}[k])\,.
 $$
 By construction, for each $j$,
  $R^{f;\,q_j}[k]$ is in
  $C^{\infty}(U^{\prime}_f)[y^1,\,\cdots\,,\, y^n]
     \subset C^{\infty}(U_f^{\prime}\times V_f^{(j)})$
  of $(y^1,\,\cdots\,,\, y^n)$-multi-degree $\le ((k+1)r-1, \,\cdots\,,\, (k+1)r-1) $.
 For convenience and all we need in the proof,
  we will write $R^{f;\,q_j}[k]$ as a polynomial in $\boldy=(y^1,\,\cdots\,, y^n)$
  $$
    R^{f;\,q_j}[k]\;
	  =\;  \sum_{\scriptsizeboldd} c_{j, \scriptsizeboldd}\,\boldy^{\scriptsizeboldd}
  $$
  with the coefficients $c_{j, \scriptsizeboldd}\in C^{\infty}(U^{\prime}_f)$.
 Recall from Sec.\ 4.1
  the decomposition\\
  $\varphi_T^{\sharp}=(\varphi_T^{\sharp, (1)},\,\cdots\,,\, \varphi_T^{\sharp,(s)})$,
   which induces the decomposition
     $\tilde{\varphi}_T^{\sharp}
	   =(\tilde{\varphi}_T^{\sharp, (1)},\,\cdots\,,\, \tilde{\varphi}_T^{\sharp,(s)})$.
 Since $f\in C^{\infty}(Y)$,
  $$
    \partialdot f=0
  $$
 for $\partialdot$ any of $\partial/\partial t^1, \,\cdots\,,\,\partial/\partial t^l$.
This gives a collection of identities for each $j$:
   \begin{eqnarray*}
   0\;\;=   & \tilde{\varphi}_T^{\sharp, (j)}(\partialboldalpha f)
     & =\;\;
	     \sum_{\scriptsizeboldd}
	         \partialboldalpha c_{j, \scriptsizeboldd}
	            \cdot  \varphi_T^{\sharp, (j)}(\boldy^{\scriptsizeboldd})\;,   \\
   0\;\;=
     &  \partial^{\scriptsizeboldalpha_1}
         (\tilde{\varphi}_T^{\sharp, (j)}
		      (\partial^{\scriptsizeboldalpha - \scriptsizeboldalpha_1}f))
	 &=\;\;
	    \tilde{\varphi}_T^{\sharp, (j)}(\partialboldalpha f)\;
	     +\; \sum_{\scriptsizeboldd}
		         \partial^{\scriptsizeboldalpha - \scriptsizeboldalpha_1}c_{j, \scriptsizeboldd}
				 \cdot \partial^{\scriptsizeboldalpha_1}
				     \varphi_T^{\sharp, (j)}(\boldy^{\scriptsizeboldd})\;, \\
   0\;\;=
     &  \partial^{\scriptsizeboldalpha_2}
         (\tilde{\varphi}_T^{\sharp, (j)}
		      (\partial^{\scriptsizeboldalpha - \scriptsizeboldalpha_2}f))   \\
	&& \hspace{-11.9em}
	    =\;\;\;
	    \tilde{\varphi}_T^{\sharp, (j)}(\partialboldalpha f)\;
		 +\; \sum_{\mbox{\tiny $\begin{array}{l}
		                         \tinyboldalpha^{\prime} \prec  \tinyboldalpha_2, \\[-.2ex] 								 
								 |\tinyboldalpha^{\prime}|=1
								\end{array}$}}
            m_{\scriptsizeboldalpha^{\prime} \prec  \scriptsizeboldalpha_2}\cdot								
				\sum_{\scriptsizeboldd}
		          \partial^{\scriptsizeboldalpha - \scriptsizeboldalpha^{\prime}}c_{j, \scriptsizeboldd}
				  \cdot \partial^{\scriptsizeboldalpha^{\prime}}
				      \varphi_T^{\sharp, (j)}(\mbox{\boldmath $y$}^{\scriptsizeboldd})\;													  
	     +\; \sum_{\scriptsizeboldd}
		         \partial^{\mbox{\boldmath\scriptsize $\alpha$}
		                                - \scriptsizeboldalpha_2}c_{j, \scriptsizeboldd}
				 \cdot \partial^{\scriptsizeboldalpha_2}
				     \varphi_T^{\sharp, (j)}(\boldy^{\scriptsizeboldd})\;,			\\
  0\;\;=
     &  \partial^{\scriptsizeboldalpha_3}
         (\tilde{\varphi}_T^{\sharp, (j)}
		      (\partial^{\scriptsizeboldalpha - \scriptsizeboldalpha_3}f))   \\
	&& \hspace{-11.9em}
	    =\;\;\;
	    \tilde{\varphi}_T^{\sharp, (j)}(\partialboldalpha f)\;
		 +\; \sum_{\mbox{\tiny $\begin{array}{l}
		                        \tinyboldalpha^{\prime}\prec  \tinyboldalpha_3, \\[-.2ex] 								 
								 |\tinyboldalpha^{\prime}|=1
								\end{array}$}}
                m_{\scriptsizeboldalpha^{\prime} \prec  \scriptsizeboldalpha_3}\cdot								
				\sum_{\scriptsizeboldd}
		          \partial^{\scriptsizeboldalpha - \scriptsizeboldalpha^{\prime}}c_{j, \scriptsizeboldd}
				  \cdot \partial^{\scriptsizeboldalpha^{\prime}}
				      \varphi_T^{\sharp, (j)}(\boldy^{\scriptsizeboldd})\;		\\											  
	  && \hspace{-8em}
       +\; \sum_{\mbox{\tiny $\begin{array}{l}
		                       \tinyboldalpha^{\prime\prime} \prec  \tinyboldalpha_3, \\[-.2ex] 								 
								 |\tinyboldalpha^{\prime\prime}|=2
								\end{array}$}}
                 m_{\scriptsizeboldalpha^{\prime\prime}\prec  \scriptsizeboldalpha_3}\cdot								
				\sum_{\scriptsizeboldd}
		          \partial^{\scriptsizeboldalpha - \scriptsizeboldalpha^{\prime\prime}}c_{j, \scriptsizeboldd}
				  \cdot \partial^{\scriptsizeboldalpha^{\prime\prime}}
				      \varphi_T^{\sharp, (j)}(\boldy^{\scriptsizeboldd})\;	
	  +\; \sum_{\scriptsizeboldd}
		         \partial^{\scriptsizeboldalpha - \scriptsizeboldalpha_3}c_{j, \scriptsizeboldd}
				 \cdot \partial^{\scriptsizeboldalpha_3}
				     \varphi_T^{\sharp, (j)}(\boldy^{\scriptsizeboldd})\;,	 \\[2ex]
   & \cdots\cdots\cdots\cdots  &	\\[3ex]
  0\;\;=
     &  \partial^{\scriptsizeboldalpha_{k^{\prime}}}
         (\tilde{\varphi}_T^{\sharp, (j)}
		      (\partial^{\scriptsizeboldalpha - \scriptsizeboldalpha_{k^{\prime}}}f))   \\
	&& \hspace{-11.9em}
	    =\;\;
	     \sum_{k^{\prime\prime}=0}^{k^{\prime}-1}
		 \sum_{\mbox{\tiny $\begin{array}{l}
		                       \tinyboldalpha^{\prime\prime}\prec  \tinyboldalpha_{k^{\prime}}, \\
								 |\tinyboldalpha^{\prime\prime}|=k^{\prime\prime}
								\end{array}$}}
                m_{\scriptsizeboldalpha^{\prime\prime}\prec  \scriptsizeboldalpha_{k^{\prime}}}
				\cdot								
				\sum_{\scriptsizeboldd}
		          \partial^{\scriptsizeboldalpha - \scriptsizeboldalpha^{\prime\prime}}c_{j, \scriptsizeboldd}
				  \cdot \partial^{\scriptsizeboldalpha^{\prime\prime}}
				      \varphi_T^{\sharp, (j)}(\boldy^{\scriptsizeboldd})\;		
	  +\; \sum_{\scriptsizeboldd}
		         \partial^{\scriptsizeboldalpha - \scriptsizeboldalpha_{k^{\prime}}}c_{j, \scriptsizeboldd}
				 \cdot \partial^{\scriptsizeboldalpha_{k^{\prime}}}
				     \varphi_T^{\sharp, (j)}(\boldy^{\scriptsizeboldd})\;,	 \\[2ex]
   & \cdots\cdots\cdots\cdots  &	\\[-2ex]			
   \end{eqnarray*}
  where, with a slight abuse of the labelling index $k^{\prime}$,
   $\,\boldalpha_{k^{\prime}}$
     runs over all $\boldalpha_{k^{\prime}}  \prec \mbox{\boldmath $\alpha$}$
	 with  $|\mbox{\boldmath $\alpha$}_{k^{\prime}}|=k^{\prime}$,
   $k^{\prime}=1,\,\ldots\,, |\mbox{\boldmath $\alpha$}|-1$.
 Here,
  $$
      m_{\mbox{\boldmath\scriptsize $\alpha$}^{\prime\prime}
			         \prec \mbox{\boldmath\scriptsize $\alpha$}^{\prime}}\;
	  =\; \mbox{\footnotesize
	       $\left(\!\!\begin{array}{l}\alpha_1^{\prime} \\  \alpha^{\prime\prime}_1 \end{array}\!\!\right)
	          \,\cdots\,
			 \left(\!\!\begin{array}{l}\alpha_l^{\prime} \\  \alpha^{\prime\prime}_l \end{array}\!\!\right)$}
  $$
  counts the number of ways to choose
    $\partial^{\scriptsizeboldalpha^{\prime\prime}}$ from $\partial^{\scriptsizeboldalpha^{\prime}}$
   for
    $\mbox{\boldmath $\alpha$}^{\prime\prime}
	    :=(\alpha^{\prime\prime}_1,\,\cdots\,,\,\alpha^{\prime\prime}_l)
        \prec \mbox{\boldmath $\alpha^{\prime}$}
		              := (\alpha^{\prime}_1,\,\cdots\,,\,\alpha^{\prime}_l)$.
   
 Now observe that
  the above system of identities is equivalent to the following system of identities for each $j$:
 $$
   \sum_{\scriptsizeboldd}
		         \partial^{\mbox{\boldmath\scriptsize $\alpha$}
		                             - \mbox{\boldmath\scriptsize $\alpha$}^{\prime\prime}}
				      c_{j, \scriptsizeboldd}
				 \cdot \partial^{\mbox{\boldmath\scriptsize $\alpha$}^{\prime\prime}}
				     \varphi_T^{\sharp, (j)}(\mbox{\boldmath $y$}^{\scriptsizeboldd})\;
    =\; 0\hspace{2em}
	\mbox{for all
	  $\;\mbox{\boldmath $\alpha$}^{\prime\prime}\,
	        \prec\, \mbox{\boldmath $\alpha$}\,$,
	  $\;0\le |\mbox{\boldmath $\alpha$}^{\prime\prime}|< |\mbox{\boldmath $\alpha$}|-1$}\,.	
 $$
 It follows that
   \begin{eqnarray*}
     \partialboldalpha (\varphi_T^{\sharp}(f))
	 & =	
     &  \sum_{j=1}^s
	      \sum_{k^{\prime}=0}^{|\scriptsizeboldalpha|-1}
  	      \!\!\sum_{\mbox{\tiny
		                        $\begin{array}{l}
		                             \tinyboldalpha^{\prime} \prec  \tinyboldalpha, \\[-.2ex] 								 
								   | \tinyboldalpha^{\prime}|=k^{\prime}
								  \end{array}$}}          								
		  \!\!
		   m_{\scriptsizeboldalpha^{\prime}\prec  \scriptsizeboldalpha}\cdot
		  \sum_{\scriptsizeboldd}
		          \partial^{\scriptsizeboldalpha - \scriptsizeboldalpha^{\prime}}c_{j, \scriptsizeboldd}
				  \cdot \partial^{\scriptsizeboldalpha^{\prime}}
				      \varphi_T^{\sharp, (j)}(\boldy^{\scriptsizeboldd})\;		
	     +\; 	 \sum_{j=1}^s\sum_{\scriptsizeboldd}
		          c_{j, \scriptsizeboldd}
				   \cdot
				   \partialboldalpha \varphi_T^{\sharp, (j)}(\boldy^{\scriptsizeboldd}) \\
     & =
	   &  \sum_{j=1}^s
	        \sum_{\scriptsizeboldd}
		       c_{\scriptsizeboldd}
			   \cdot
			   \partialboldalpha \varphi_T^{\sharp, (j)}(\boldy^{\scriptsizeboldd})\;\;\;
            =\;\;\;
 			   \sum_{j=1}^s
   	           R^{f;q_j}[k]
			   \!\!\left.\rule{0ex}{0.8em}\right|_{\scriptsizeboldy^{\tinyboldd}
 			                    \rightsquigarrow\,
								\partial^{\tinyboldalpha}
								   (\varphi_T^{\sharp, (j)}(\scriptsizeboldy^{\tinyboldd}))}\;.
  \end{eqnarray*}
 This concludes the proof of the proposition.
   
\end{proof}

\bigskip

\noindent
{\it Remark 4.2.3.2. $[\,$case $k=1\,]$}$\;$  {\rm
 Since
  $Q^{f;\,q_j}_{(i_1,\,\cdots\,,\, i_k)}
     = R^{f;\,q_j}_{(i_1,\,\cdots\,,\,i_k)}
	    + \sum_{i_{k+1}=1}^n
		     Q^{f;\,q_j}_{(i_1,\,\cdots\,,\,i_k, i_{k+1})}\,\chi_{\varphi_T}^{(i_{k+1})}$
     and $\tilde{\varphi}^{\sharp, (j)}_T(\chi_{\varphi_T}^{(i_{k+1})})=0\,$
	  for $i_{k+1}=1,\,\ldots\,,n$,      	
  one has	
   $$
      \tilde{\varphi}_T^{\sharp, (j)}(Q^{f;\,q_j}_{(i_1,\,\cdots\,,\, i_k)})\;
      =\;  \tilde{\varphi}_T^{\sharp, (j)}(R^{f;\,q_j}_{(i_1,\,\cdots\,,\,i_k)})\,.
   $$
 In particular, for $k=1$,
   \begin{eqnarray*}
    \lefteqn{
    (R^{f;q_j}_{(i)}\,\chi_{\varphi_T}^{(i)})
         |_{\scriptsizeboldy^{\tinyboldd}\,
		          \rightsquigarrow\,
				  \partialdot \varphi_T^{\sharp, (j)}(\scriptsizeboldy^{\tinyboldd})}  }\\
	   && =\;\;
 	   (R^{f;q_j}_{(i)}
                     |_{\scriptsizeboldy^{\tinyboldd}\,					 					
		                      \rightsquigarrow\,
				               \partialdot \varphi_T^{\sharp, (j)}(\scriptsizeboldy^{\tinyboldd})})\,
                       \cdot \tilde{\varphi}_T^{\sharp, (j)}( \chi_{\varphi_T}^{(i)})\;
                +\; \tilde{\varphi}_T^{\sharp, (j)}(R^{f;q_j}_{(i)})
		               \cdot  (\chi_{\varphi_T}^{(i)}
                                       |_{\scriptsizeboldy^{\tinyboldd}\,
		                                        \rightsquigarrow\,
				                                \partialdot \varphi_T^{\sharp, (j)}(\scriptsizeboldy^{\tinyboldd})})\\
    && =\;\;
       	\tilde{\varphi}_T^{\sharp, (j)}(Q^{f;q_j}_{(i)})
		               \cdot  (\chi_{\varphi_T}^{(i)}
                                       |_{\scriptsizeboldy^{\tinyboldd}\,
		                                        \rightsquigarrow\,
				                                \partialdot \varphi_T^{\sharp, (j)}(\scriptsizeboldy^{\tinyboldd})})
  \end{eqnarray*}
 and
 Proposition~4.2.3.1 resumes to Proposition~4.2.2.1.
} 
  
\bigskip

\noindent
{\it Remark 4.2.3.3. $[\,$case $k=0\,]$}$\;$  {\rm
 Setting the convention that for $|\mbox{\boldmath $\alpha$}|=0$, i.e.\
   $\mbox{\boldmath $\alpha$}=(0,\,\cdots\,,\,0)$,
  $\partialboldalpha (\,\cdots\,):= (\,\cdots\,)$.
 Then, for $|\mbox{\boldmath $\alpha$}|=0$,
  Proposition~4.2.3.1 resumes to the case studied in [L-Y6] (D(11.3.1)), reviewd in Sec.\ 4.1;
  cf.\ the formula
     $$		
	   \varphi_T^{\sharp}(f)\;
	   =\;
	     \sum_{j=1}^s
		 \sum_{\mbox{\tiny $\begin{array}{c}0\le d_i\le r-1 \\[1.2ex]  1\le i\le n \end{array}$}}
			  a^{f;q_j}_{(d_1,\,\cdots\,,\,d_n)}\cdot
			    (\varphi_T^{\sharp, (j)}(y^1))^{d_1}\,\cdots\,
				(\varphi_T^{\sharp, (j)}(y^n))^{d_n}
	 $$
 at the end of Sec.\ 4.1, which is simply
 $\sum_{j=1}^sR^{f;\,q_j}[0]|
                              _{\scriptsizeboldy^{\tinyboldd}\,
                                      \rightsquigarrow\,
								      \varphi_T^{\sharp, (j)}(\scriptsizeboldy^{\tinyboldd})}$.
}
		
\bigskip

\noindent
{\it Remark 4.2.3.4. $[\,$A second look at Proposition~4.2.3.1 from a comparison with the commutative case$\,]$}$\;$
{\rm															
 In the commutative case,
  let
   $X={\Bbb R}^m$ with coordinates $\boldx=(x^1,\,\cdots\,, x^m)$,
   $Y={\Bbb R}^n$  with coordinates $\boldy=(y^1,\,\cdots\,, y^n)$,
   $f\in C^{\infty}(Y)$, and
   $h :=(h^1,\,\cdots\,, h^n):  X \rightarrow Y$ be a differentiable map.
 Let
   $T$ be a small neighborhood of the origin $0\in{\Bbb R}^l$ with coordinates $\boldt=(t^1,\,\cdots\,, t^l)$
   and
   $h_T :=(h_T^1,\,\cdots\,, h_T^n):X\rightarrow Y$ be a $T$-family of differentiable maps from $X$ to $Y$
     that extends $h=: h_0$.	
 For $\boldalpha=(\alpha_1,\,\cdots\,,\alpha_l)\in {\Bbb Z}_{\ge 0}^l$,
  let $\partial^{\scriptsizeboldalpha}_{\scriptsizeboldt,0}
          := \partial^{|\scriptsizeboldalpha|}/(\partial t^1)^{\alpha_1}\,\cdots\,(\partial t^l)^{\alpha_l}$
		 at $\boldt=0$.
 For $\boldd=(d_1,\,\cdots\,,d_n)\in {\Bbb Z}_{\ge 0}^n$,
  let $\partial^{\scriptsizeboldd}_{\scriptsizeboldy}
          := \partial^{|\scriptsizeboldd|}/(\partial y^1)^{d_1}\,\cdots\,(\partial y^n)^{d_n}$,
  where $|\boldd|:=d_1+\,\cdots\,+d_n$.		
 Then, it follows from the chain rule and the Leibniz rule that
  $$
    \partial^{\scriptsizeboldalpha}_{\scriptsizeboldt,0}(f(h_T(\boldx)))
  $$
   is a summation over ${\Bbb Z}_{\ge 0}$ of terms of the following form
   $$
      (\partial^{\scriptsizeboldd}_{\scriptsizeboldy}f)
	     (h(\boldx))
		 \cdot
		  \partial^{\scriptsizeboldalpha_{i_1}} _{\scriptsizeboldt,0}
		   h^{i_1}_T(\boldx)\,\cdots\,
		  \partial^{\scriptsizeboldalpha_{i_I}} _{\scriptsizeboldt,0}
		  h^{i_I}_T(\boldx)
   $$
   with
    $|\boldd|\le |\boldalpha|$,
    $\boldalpha_{i_1}+\,\cdots\,+\boldalpha_{i_I}=\boldalpha$ and
	$1\le i_1<\,\cdots\,<i_I\le n$, $1\le I\le n$.
 In particular, it is a linear combination of such
   $\partial^{\mbox{\scriptsize\boldmath $\alpha$}_{i_1}} _{\mbox{\scriptsize\boldmath $t$},0}
		   h^{i_1}_T(\boldx)\,\cdots\,
		  \partial^{\mbox{\scriptsize\boldmath $\alpha$}_{i_I}} _{\mbox{\scriptsize\boldmath $t$},0}
		  h^{i_I}_T(\boldx)$
  with coefficients all depending universally on $f$ and $h$ alone
  (i.e.\ with coefficients not depending on how $h$ is extended to $h_T$).
 This universal identity can be made precise as follows.
  \begin{itemize}
   \item[\LARGE $\cdot$]
   Let $|\boldalpha|=k$.
   The fiberwise Taylor Theorem applied to $f$ as a differentiable function on $(X\times Y)/X$
    in a neighborhood of the locus
    $\{y^1=h^1(\boldx),\,\cdots\,,\, y^n=h^n(\boldx)\}\subset X\times Y$
     gives
	 \begin{eqnarray*}
       f(\boldy)
	   & =  &
		 \sum_{d=0}^k\,
		  \frac{1}{d!}
		 \sum_{\;\scriptsizeboldd,\,|\scriptsizeboldd|=d}\,
           m_{\scriptsizeboldd}
		    \cdot		
		  (\partial^{\scriptsizeboldd}_{\scriptsizeboldy}f)(h(\boldx))
		        (\boldy-h(\boldx))^{\scriptsizeboldd}    \\
        && \hspace{7em}				
             +\, \frac{1}{(k+1)!}
			      \sum_{\;\scriptsizeboldd,\, |\scriptsizeboldd|=k+1}
			        m_{\scriptsizeboldd}
					 \cdot
					 Q_{\scriptsizeboldd}(h(\boldx))
					  (\boldy-h(\boldx))^{\scriptsizeboldd}
     \end{eqnarray*}
   for some $Q_{\scriptsizeboldd}\in C^{\infty}(Y)$.
  Here,
    $m_{\scriptsizeboldd}$ is the multiplicity factor associated to $\boldd$  and\\
    $(\boldy-h(\boldx))^{\scriptsizeboldd}
	    := (y^1-h^1(\boldx))^{d_1}\,\cdots\,(y^n-h^n(\boldx))^{d_n}$.	
   
  \item[\LARGE $\cdot$]
  When $h=h_0$ is extended to $h_T$, then for $\boldt \in T$ close enough to $0$,
	 \begin{eqnarray*}
       f(h_{\scriptsizeboldt}(\boldx))
	   & =  &   f(\boldy)|
	                _{\scriptsizeboldy\rightsquigarrow h_{\tinyboldt}(\scriptsizeboldx)}   \\
	   & =  &
		 \sum_{d=0}^k\,
		  \frac{1}{d!}
		 \sum_{\;\scriptsizeboldd,\,|\scriptsizeboldd|=d}\,
           m_{\scriptsizeboldd}
		    \cdot		
		  (\partial^{\scriptsizeboldd}_{\scriptsizeboldy}f)(h(\boldx))
		        (h_{\scriptsizeboldt}(\boldx)-h(\boldx))^{\scriptsizeboldd}    \\
        && \hspace{7em}				
             +\, \frac{1}{(k+1)!}
			      \sum_{\;\scriptsizeboldd,\, |\scriptsizeboldd|=k+1}
			        m_{\scriptsizeboldd}
					 \cdot
					 Q_{\scriptsizeboldd}(h(\boldx))
					  (h_{\scriptsizeboldt}(\boldx)-h(\boldx))^{\scriptsizeboldd}\,,
     \end{eqnarray*}
	 where
	  $(h_{\scriptsizeboldt}(\boldx)-h(\boldx))^{\scriptsizeboldd}
	     :=(h^1_{\scriptsizeboldt}(\boldx)-h^1(\boldx))^{d_1}\,
		          \cdots\,(h^n_{\scriptsizeboldt}(\boldx)-h^n(\boldx))^{d_n}$,
	and, hence,
     $$
      \partial^{\mbox{\scriptsize\boldmath $\alpha$}}_{\mbox{\scriptsize\boldmath $t$},0}
	   (f(h_T(\boldx)))\;
	   =\;  	
	      \sum_{d=0}^k\,
		  \frac{1}{d!}
		 \sum_{\;\scriptsizeboldd,\,|\scriptsizeboldd|=d}\,
           m_{\scriptsizeboldd}
		    \cdot		
		  (\partial^{\scriptsizeboldd}_{\scriptsizeboldy}f)(h(\boldx))
		     \cdot
			   \partial^{\scriptsizeboldalpha}_{\scriptsizeboldt,0}
			     \left(
 			   (h_{\scriptsizeboldt}(\boldx)-h(\boldx))^{\scriptsizeboldd} \right)\,.
     $$	
 \end{itemize}
 From this aspect,
  Proposition~4.2.3.1 is nothing but the equal of the above identity in our particular noncommutative situation,
  with
    the map $h:X\rightarrow Y$ replaced by the map $\varphi:(X^{\!A\!z},E)\rightarrow Y$   and
    the extension $h_T$ of $h$ replaced by the extension of $\varphi_T$ of $\varphi$. 				 
 
 {\it However}, caution that
 in the commutative situation,
  $\partial^{\mbox{\scriptsize\boldmath $\alpha$}}_{\mbox{\scriptsize\boldmath $t$},0}
   (f(h_T(\boldx)))$
  involves only $\partial^{\scriptsizeboldd}_{\scriptsizeboldy}f$ along the graph of $h$ up to
  (and including) order $|\alpha|$
  (i.e.\ restriction of $f$ to the $|\boldalpha|$-th infinitesimal neighborhood of the graph of $h$)
 while in our noncommutative situation,
   $\partial^{\scriptsizeboldalpha}_{\scriptsizeboldt,0}(\varphi_T^{\sharp}(f))$
   may involve $\partial^{\scriptsizeboldd}_{\scriptsizeboldy}f$
   along the support  $\Supp(\tilde{\cal E}_{\varphi})\subset X\times Y$
   of the graph $\tilde{\cal E}_{\varphi}$ of $\varphi$ up to (and including) order
   $r|\alpha|$, where $r$ is the rank of $E$ as a complex vectior bundle on $X$.
 The detail depends on the nilpotency of the structure sheaf
  ${\cal O}_{\scriptsizeSupp(\tilde{\cal E}_{\varphi})}$
  of $\Supp(\tilde{\cal E}_{\varphi})\subset X\times Y$.
}

\bigskip

\noindent
{\bf Remark/Notation 4.2.3.5. [equivalent form: general order]}$\;$
 (Cf.\ Remark/Notation~4.1.1.)
 Recall from the proof of Proposition~4.2.3.1 the expression
  $$
    R^{f;q_j}[k]\;
	=\; \sum_{\scriptsizeboldd}c_{j,\scriptsizeboldd}\,\boldy^{\scriptsizeboldd}\;\;
    \in\; C^{\infty}(U^{\prime}_f)[y^1,\,\cdots\,, y^n]\,.
  $$
 As in Remark/Notation~4.1.1,
 define
  $$
    R^{f}[k]\;
	 =\;   \sum_{\scriptsizeboldd}\,
	          \left(\rule{0ex}{1em}\right.
			    \sum_{j=1}^s
				   c_{j,\scriptsizeboldd}
				      \cdot
					   \Id_{(E_T|_{U^{\prime}}^{(j)})|_{U^{\prime}_f}}
			  \left.\rule{0ex}{1em}\right)
			   \cdot  \boldy^{\scriptsizeboldd}\,.
  $$
 Then
  Proposition~4.2.3.1, with Remark~4.2.3.3, can be stated equivalently as
  $$
    \partial^{\scriptsizeboldalpha} (\varphi_T^{\sharp}(f))
    =\;
   	           R^f[k]\!\!\left.\rule{0ex}{0.8em}\right|
			    _{\scriptsizeboldy^{\tinyboldd}
 			                    \rightsquigarrow\,
								\partial^{\mbox{\boldmath\tiny $\alpha$}}
								   (\varphi_T^{\sharp}(\scriptsizeboldy^{\tinyboldd}))}
    \hspace{2em}								
    \mbox{for $\boldalpha$ with $|\boldalpha|=k\in {\Bbb Z}_{\ge 0}$}\,.
 $$				
 This generalizes Remark/Notation~4.1.1.
 
 Furthermore, for $k=1$,
  recall Notation~4.2.1.7    and
  let $\partial_t$ be any of $\partial_{t^1},\,\cdots\,,\,\partial_{t^l}$  and
       $\boldalpha\in{\Bbb Z}_{\ge 0}^l$ be associated to $\partial_t$.
 Then, one has the following expansion of $\partial_t(\varphi_T^{\sharp}(f))$,
   linearly in
   $(\partial_t\varphi_T^{\sharp}(y^1),\,\cdots\,,\, \partial_t\varphi_T^{\sharp}(y^n))\,$:
  \begin{eqnarray*}
   \lefteqn{\partial_t(\varphi_T^{\sharp}(f))
     =\;
     \sum_{i=1}^n\,
	   \sum_{d=0}^{\mbox{\tiny $\bullet$}}
	   \sum_{\;\,\scriptsizeboldd,\, |\scriptsizeboldd|=d}
	   \sum_{\;\vec{\pi}
	                         \in \vec{\mbox{\scriptsize\it Ptn}}
				            (\mbox{\scriptsize\boldmath $\alpha$}, d),\,
							i_{(\vec{\pi},\,\tinyboldd)}=i}	      }  \\
    && \hspace{7em}	
	  ([\partial^{\vec{\pi}}_{y^i}]
	             R^f[1]_{(\scriptsizeboldd)})^L(\varphi_T^{\sharp}(\boldy))
	   \cdot
	   \partial_t \varphi_T^{\sharp}(y^i)
	   \cdot
	 ([\partial^{\vec{\pi}}_{y^i}]
	             R^f[1]_{(\scriptsizeboldd)})^R(\varphi_T^{\sharp}(\boldy))\,,
  \end{eqnarray*}
 where
  $R^f[1]_{(\scriptsizeboldd)}$ is the multi-degree-$\boldd$ component of $R^f[1]$
     as a polynomial in $(y^1,\,\cdots\,,\, y^n)$
 and 	
  $(\,\cdots\,)^{L,R}(\varphi_T^{\sharp}(\boldy))
    := (\,\cdots\,)^{L,R}(\varphi_T^{\sharp}(y^1),\,\cdots\,,\,\varphi_T^{\sharp}(y^n))
     =(\,\cdots\,)^{L,R}|
			    _{\scriptsizeboldy^{\tinyboldd^{\prime}}   \rightsquigarrow\,
				          \varphi_T^{\sharp}(\scriptsizeboldy^{\tinyboldd^{\prime}})}$.

\bigskip

The following is an immediate consequence of Remark~4.2.3.2:
																	
\bigskip																			 

\noindent 
{\bf Corollary 4.2.3.6. [chain rule under trace]}$\;$ {\it 	
 Under the trace map 
  $\Tr:C^{\infty}(\End_{\Bbb C}(E))\rightarrow C^{\infty}(X)^{\Bbb C}$, 
 the chain rule for a first-order derivation holds:
  $$
    \Tr\left(  \partialdot(\varphi_T^{\sharp}(f))    \right)\;
     =\; \Tr\left(\sum_{i=1}^n 
	             \frac{\partial f}{\partial y^i}
	            (\varphi_T^{\sharp}(y^1),\,\cdots\,,\, \varphi_T^{\sharp}(y^n))
				 \cdot \partialdot\varphi_T^{\sharp}(y^i)     
				 \right),	
  $$
 where $\partialdot$ is any of $\partial/\partial t^1,\,\cdots\,,\, \partial/\partial t^l$. 
}

\bigskip

\begin{proof}
 Let $\partial_{y_{\LARGEdot}}$ be any of $\partial/\partial y^1,\,\cdots\,, \partial/\partial y^n$.
 Then, with the notation in Remark/Notation~4.2.3.5,  
 observe that, for $f\in C^{\infty}(Y)$, 
  $$
   \varphi_T^{\sharp}(\partial_{y_{\LARGEdot}}\!f)\;
    =\; \varphi_T^{\sharp}(\partial_{y_{\LARGEdot}}\!(R^f[1]))\;
	=\; (\partial_{y_{\LARGEdot}}\!(R^f[1]))|
	          _{\scriptsizeboldy^{\tinyboldd}\,
			           \rightsquigarrow\, \varphi_T^{\sharp}(\scriptsizeboldy^{\tinyboldd})}\;
    =\; 	R^f[1]|_{\scriptsizeboldy^{\tinyboldd}\,\rightsquigarrow\, 
	                             \varphi_T^{\sharp}(
 								    \partial_{y_{\Largedot}}\!( \scriptsizeboldy^{\tinyboldd}  )
									                                     )  }\,. 	
  $$ 
 Since 
  $$
    \Tr\left(  \partialdot(\varphi_T^{\sharp}(f))    \right)\;
	  =\; \Tr  \left(\rule{0ex}{1em}\right. 
	      \sum_{i=1}^n\, 
		      	R^f[1]|_{\scriptsizeboldy^{\tinyboldd}\,\rightsquigarrow\, 
	                             \varphi_T^{\sharp}(
 								    \partial_{y_i}\!( \scriptsizeboldy^{\tinyboldd}  )
									                                     )  }\,
                \partial_{\LARGEdot}\varphi_T^{\sharp}(y^i) 
	                  \left.\rule{0ex}{1em}\right)\,,
  $$ 
 the corollary follows.
 
\end{proof}

\bigskip

\section{The first variation of the Dirac-Born-Infeld action and the equations of motion for D-branes}

We discuss in this section the first variation of the Dirac-Born-Infeld action (Sec.\ 5.2)
  and its consequence, the equations of motion of D-branes in our setting (Sec.\ 5.3).
We begin with a few remarks
  on variations and infinitesimal deformations in $C^{\infty}$-algebraic geometry (Sec.\ 5.1).

\bigskip

\subsection{Remark on deformation problems in $C^{\infty}$-algebraic geometry}

From the viewpoint of $C^{\infty}$-algebraic geometry,
 it is very natural to address a deformation problem
 as an extension problem over a non-reduced $C^{\infty}$-scheme,
 as did in the setting Grothendieck's Modern Algebraic Geometry, e.g.\ [Il], [Ser], [Schl].
On the other hand,
  for a variation problem in differential or symplectic geometry, it is customary
  to consider a $1$- or $2$-parameter family of objects in question and then take derivatives.
The following very elementary example indicates that the former is more general than the latter:
  
\bigskip

\begin{example} {\bf [infinitesimal extension vs.\
        extension over $(-\delta,\delta)\subset {\Bbb R}^1$]} {\rm
 Let
  $$
    \xymatrix{
     (p, \End_{\Bbb C}({\Bbb C})\simeq{\Bbb C},{\Bbb C})\; \ar[rr]^-{\varphi}
	   &&   \; Y={\Bbb R^1}
	}
  $$
 be a map from an Azumaya/matrix point of rank $1$ (ie.\ a ${\Bbb C}$-point) to $Y={\Bbb R^1}$,
 defined by a ring-homomorphism over ${\Bbb R}\subset {\Bbb C}$:
  $$
    \xymatrix @ R=-.2ex {	
	 \;\End_{\Bbb C}({\Bbb C})\simeq{\Bbb C}\;
	   && \; C^{\infty}({\Bbb R}^1) \ar[ll]_-{\varphi^{\sharp}}\; \\
     \hspace{3em}\lambda\hspace{3em}
	    && \hspace{2em}y\hspace{2em} \ar@{|->}[ll] 	  \\
     \hspace{2.3em}f(\lambda)\hspace{2.3em}
	    && \hspace{1.3em}f(y)\hspace{1.3em} \ar@{|->}[ll] 	   		
	}
  $$
  for a $\lambda\in{\Bbb R}\subset {\Bbb C}$ fixed.
 Let
   $T_1:=\Spec^{\Bbb R}({\Bbb R}[t]/(t^2))
             =: \Spec^{\Bbb R}({\Bbb R}[\epsilon])$, $\epsilon^2=0$,
    be a dual-point
	and
   $T_2 :=(-\delta,\delta)\subset {\Bbb R}$ be a $1$-manifold with parameter $t$,
   where $\delta>0$ small.
 $T_1\subset T_2$ as $C^{\infty}$-subscheme.
 Treat $\varphi^{\sharp}$ as a ring-homomorphism over
 $\Spec^{\Bbb R}{\Bbb R}\subset T_1$.
 Then,
   the following is an infinitesimal extension of $\varphi^{\sharp}$
     to a ring-homomorphism $\varphi_{T_1}^{\sharp}$ over the base $T_1$:
   $$
    \xymatrix @ R=-.2ex {	
	 \;\End_{{\Bbb C}[\epsilon]}({\Bbb C}[\epsilon])\simeq{\Bbb C}[\epsilon]\;
	   && \; C^{\infty}({\Bbb R}^1) \ar[ll]_-{\varphi_{T_1}^{\sharp}}\; \\
     \hspace{2.6em}\lambda+ \sqrt{-1}\epsilon\hspace{2.6em}
	    && \hspace{2em}y\hspace{2em} \ar@{|->}[ll] 	  \\
     \hspace{.9em}f(\lambda)+ f^{\prime}(\lambda)\sqrt{-1}\epsilon\hspace{.9em}
	    && \hspace{1.2em}f(y)\hspace{1.2em}. \ar@{|->}[ll]   	   		
	}
  $$
 On the other hand,
  let $E_{T_2}$ be the trivialized complex line bundle over $T_2$.
 Then,
 since $T_2$ is a manifold,
  any extension of $\varphi^{\sharp}=:\varphi_0^{\sharp}$
    to a ring-homomorphism $\varphi_{T_2}^{\sharp}$ over the base $T_2$
	must be of the following form
   $$
    \xymatrix @ R=-.2ex {	
	 \;C^{\infty}(End_{\Bbb C}(E_{T_2}))\;
	   && \; C^{\infty}({\Bbb R}^1) \ar[ll]_-{\varphi_{T_2}^{\sharp}}\; \\
     \hspace{3em}h(t)\hspace{3em}
	    && \hspace{2em}y\hspace{2em} \ar@{|->}[ll] 	  \\
     \hspace{2.4em}f(h(t))\hspace{2.4em}
	    && \hspace{1.2em}f(y)\hspace{1.2em}, \ar@{|->}[ll]
	}
  $$
 for $t\in T_2$, where $h\in C^{\infty}(T_2)$  with $h(0)=\lambda$.
 Whose associated infinitesimal deformation of $\varphi^{\sharp}$
  is given by $\varphi_{T_2}^{\sharp}|_{T_1}$:
  $$
    \xymatrix @ R=-.2ex {	
	 \;\End_{{\Bbb C}[\epsilon]}({\Bbb C}[\epsilon])\simeq {\Bbb C}[\epsilon]\;
	   && \; C^{\infty}({\Bbb R}^1) \ar[ll]_-{\varphi_{T_2}^{\sharp}|_{T_1}}\; \\
     \hspace{2.5em}\lambda+h^{\prime}(0)\epsilon\hspace{2.5em}
	    && \hspace{2em}y\hspace{2em} \ar@{|->}[ll] 	  \\
     \hspace{.8em}f(\lambda)+ f^{\prime}(\lambda)h^{\prime}(0)\epsilon\hspace{.8em}
	    && \hspace{1.2em}f(y)\hspace{1.2em}. \ar@{|->}[ll]
	}
 $$
 Since $f^{\prime}(\lambda)h^{\prime}(0)\in {\Bbb R}$,
  this can never be the given $\varphi_{T_1}^{\sharp}$.
 In other words,
  $\varphi_{T_1}^{\sharp}$  cannot be extended further
   to a ring-homomorphism over $T_2\supset T_1$.
}\end{example}

\bigskip

The above example demonstrates the fact that
 there can be infinitesimal deformations in a moduli problem that do not arise from a smooth family.
Such a phenomenon may be unfamiliar to differential geometers or string-theorists
 but is completely normal to algebraic geometers.
It only means that the associated moduli stack is singular at the point representing that object in question
 and hence some infinitesimal deformations of that object can be obstructed from further extensions.
From this point of view, our treatment of the variation problem below through a smooth family
 is not yet the most general one.
But we will focus only on such unobstructed deformations for the current notes.
The more general, possibly obstructed, deformations in our problem
  and their consequences should be understood better in the future.

\bigskip
 
\subsection{The first variation of the Dirac-Born-Infeld action}
Given an admissible Lorentzian map,
 $$
   \varphi\; :\; (X^{\!A\!z},E;\nabla)\; \longrightarrow\; (Y,g,B,\Phi)\,,
 $$
let
  $T:=(-\varepsilon, \varepsilon) \subset {\Bbb R}^1$
and
 $\varphi_t:(X^{\!A\!z},E;\nabla^t)\rightarrow (Y,g,B,\Phi)$, $t\in T$,
 be a differentiable $T$-family of admissible Lorentzian maps that deforms $\varphi=: \varphi_0$.
In this subsection we   derive in steps the first variation
  $$
    \left.\mbox{\Large $\frac{d}{dt}$}\right|_{t=0}\,
	 S_{\DBI}^{(\Phi,g,B)}(\varphi_t,\nabla^t)
  $$
  of the Dirac-Born-Infeld action.
The derivation for the other two situations:
  $(Y,g)$ Lorentzian and $\varphi_t$ spacelike, and $(Y,g)$ Riemannian and $\varphi_t$ Riemannian, 
 are completely the same.
  
As the major part of the discussion is local and around $0\in T$,
we will assume that $\varepsilon$ is small enough
 and set the computation over a small enough coordinate chart $U\subset X$ (with coordinate functions
  $\mbox{\boldmath $x$}= (x^1,\,\cdots\,,\,x^m)$
  so that
 $E|_U$ is trivializable and trivialized,  and
 $\varphi_t(U)$ is contained in a coordinate chart $V\subset Y$ (with coordinate functions
  $\mbox{\bf $y$}=(y^1,\,\cdots\,,\, y^n))$.
Recall from Sec.\ 3.2 that, over $U$,
 {\footnotesize
 \begin{eqnarray*}
  S_{\mbox{\it\tiny DBI}\,}^{(\Phi,g,B)\,}|_U(\varphi_t,\nabla^t)
    & =
	& -\,T_{m-1}\,
	     \int_U\,\Real\left( \Tr\left(
		     e^{-\varphi_t^{\diamond}\Phi}\,
               \sqrt{-\,\SymDet_U\left(
			                        \varphi_t^{\diamond}(g+B)+ 2\pi\alpha^{\prime}F_{\nabla^t}
									                \right)\,}			   \right)\right) \\
  && \hspace{-9em}=\;\;
    -\, T_{m-1}\,\int_U\,  \Real  \left( \Tr \left(
          e^{-\varphi_t^{\sharp}(\Phi)}\,
		    \sqrt{-\SymDet  \left(\rule{0ex}{1.2em} \right.
			     \sum_{i,j}
				         \varphi_t^{\sharp}(E_{ij})
						               D^t_{\mu}\varphi_t^{\sharp}(y^i)D^t_{\nu}\varphi_t^{\sharp}(y^j)\,
                              +\, 2\pi\alpha^{\prime}\,[\nabla^t_{\mu},\nabla^t_{\nu}]	
			                                   \left.\rule{0ex}{1.2em}\right)_{\mu\nu} \,}\right )\right)
	   d^m\mbox{\boldmath $x$}\,.
 \end{eqnarray*}}
Here,  we set the notation for the tensors and connections involved as follows:
  \begin{itemize}
   \item[\LARGE $\cdot$]
     $g+B \;=\; \sum_{i,j}(g_{ij}+ B_{ij})\,dy^i\otimes dy^j  \;
	   =:\; \sum_{i,j}E_{ij}dy^i \otimes dy^j$,
	  with $g_{ij}=g_{ji}$, $B_{ij}=-B_{ji}$,
	
   \item[\LARGE $\cdot$]	
     $\nabla^t \; =\;  d+ A^t \; =\;  \sum_{\mu} (\partial_{\mu}+A^t_{\mu} )\,dx^{\mu}\;$
	 is the connection on $E|_U$,
	
   \item[\LARGE $\cdot$]
     $D^t \;=\;  d+[A^t,\,\cdot\,] \;=\; \sum_{\mu}(\partial_{\mu}+[A^t_{\mu},\,\cdot\,])\,dx^{\mu}\;$
       is the $\nabla^t$-induced connection on $\End_{\Bbb C}(E|_U)$,
	
   \item[\LARGE $\cdot$]	
     $d^m\mbox{\boldmath $x$}\; :=\; dx^1\wedge\,\cdots\,\wedge dx^m\;$
	  is compatible with the orientation on $U$;
  \end{itemize}
  and, for later use,
   $$
     \dot{\varphi}^{\sharp}(y^i)\;
	    :=\;  \left.\mbox{\Large $\frac{d}{dt}$}\right|_{t=0}
		              \left(\varphi_T^{\sharp}(y^i)\right)\,, \hspace{2em}	
    ^{\mbox{\Large $\cdot$}}\!
	  (\varphi^{\sharp}(\boldy^{\scriptsizeboldd}))
	  :=   \left.\mbox{\Large $\frac{d}{dt}$}\right|_{t=0}
	                 (\varphi_T^{\sharp}(\boldy^{\scriptsizeboldd}))\,,
	   \hspace{2em}					
     \dot{A}_{\mu}\;
	    :=\;  \left.\mbox{\Large $\frac{d}{dt}$}\right|_{t=0}A^T_{\mu}\,.
   $$
We assume further that the local chart $U$ and $\varphi>0$ are small enough so that
 the construction over $U_T:=U\times (-\varepsilon,\varepsilon)$ in Sec.\ 4.1,
    with $p\in  U\times \{0\}\subset  U_T$,
  applies simultaneously to $e^{-\Phi}$ and $E_{ij}$, $i,j=1,\,\cdots\,,\,n$,
  to give the local expression of
 $\varphi_T^{\sharp}(\Phi)$ and $\varphi_T^{\sharp}(E_{ij})$, $i,j=1\,\ldots\,,\,n$,
   in terms of  elements in the polynomial ring over $C^{\infty}(U_T)$
    $$
	  \varphi_T^{\sharp}(\Phi)\,,\;
      \varphi_T^{\sharp}(E_{ij})\;   \in\; 	
	    \left(
		  \oplus_{j=1}^s C^{\infty}(U_T)\cdot \Id_{E_T^{(j)}}
		 \right)
	     [\,\varphi_T^{\sharp}(y^1), \,\cdots\,,\, \varphi_T^{\sharp}(y^n)   \,]
	$$
   of multi-degree $\le (r-1,\,\cdots\,,\,r-1)$.
Associated to these settings and with the notation from Remark/Notataion 4.2.3.5,
 recall that																	
 \begin{eqnarray*}
  e^{-\varphi_T^{\sharp}(\Phi)}
   & =
   & \varphi_T^{\sharp}(e^{-\Phi})\hspace{3.5em}
	  =\;\;  R^{\,e^{-\Phi}}[0]|
	         _{\scriptsizeboldy^{\tinyboldd}
			        \rightsquigarrow \varphi_T^{\sharp}(\scriptsizeboldy)^{\tinyboldd}}\,,   \\
   \left.\mbox{\Large $\frac{d}{dt}$}\right|_{t=0}e^{-\varphi_T^{\sharp}(\Phi)}
   & =
   &  \left.\mbox{\Large $\frac{d}{dt}$}\right|_{t=0}\varphi_T^{\sharp}(e^{-\Phi})\;\;
       =\;\;
       R^{\,e^{-\Phi}}[1]|
	         _{\scriptsizeboldy^{\tinyboldd}
			        \rightsquigarrow\,
				    ^{\mbox{\Large $\cdot$}}
					(\varphi_T^{\sharp}(\scriptsizeboldy)^{\tinyboldd})}\,;
 \end{eqnarray*}										
 and
 \begin{eqnarray*}
  \varphi_T^{\sharp}(E_{ij})
    & =
	& R^{\,E_{ij}}[0]|
	         _{\scriptsizeboldy^{\tinyboldd}
			        \rightsquigarrow \varphi_T^{\sharp}(\scriptsizeboldy)^{\tinyboldd}}\,,   \\	
  \left.\mbox{\Large $\frac{d}{dt}$}\right|_{t=0}\varphi_T^{\sharp}(E_{ij})
    & =
	&  R^{\,E_{ij}}[1]|
	         _{\scriptsizeboldy^{\tinyboldd}
			        \rightsquigarrow\,
				    ^{\mbox{\Large $\cdot$}}
					(\varphi_T^{\sharp}(\scriptsizeboldy)^{\tinyboldd})}\,,
 \end{eqnarray*}
 for $i,j=1,\,\ldots\,,\, n$.
For simplicity of notation,
 it is understood that $R^{\,E_{ij}}[1]$ is evaluated at $t=0$ in the expression
  $R^{\,E_{ij}}[1]|
	         _{\scriptsizeboldy^{\tinyboldd}
			        \rightsquigarrow
				    ^{\mbox{\Large $\cdot$}}
					(\varphi_T^{\sharp}(\scriptsizeboldy)^{\tinyboldd})}$;
 and similarly for induced expressions that follow this.

\vspace{6em}

\begin{flushleft}
{\bf Basic identities}
\end{flushleft}
Basic identities that will be used in the calculation are collected here for reference.

\bigskip

\noindent
$(a)$ {\it Differentiation of a square root}\hspace{1em}
Let $M(t)\in C^{\infty}(\End_{\Bbb C}(E))$,
    $t\in T:=(-\varepsilon, \varepsilon)\subset {\Bbb R}$,
  be a $T$-family of invertible endomorphisms of $E$ such that $\sqrt{M(t)}$ is well-defined,
  cf.\ Sec.\ 3.1.4.
Denote ${\left.\frac{d}{dt}\right|}_{t=0}\,M(t)$ by $\dot{M}(0)$.
Then, $\sqrt{M(t)}$, $t\in(-\varepsilon, \varepsilon)$, is also invertible  and
 $$
    \sqrt{M(0)}^{\,-1}
	   \left( \left.\mbox{\Large $\frac{d}{dt}$}\right|_{t=0}\,\sqrt{M(t)}\right)\,
	    \sqrt{M(0)}\,
	    +\,     \left( \left. \mbox{\Large $\frac{d}{dt}$}\right|_{t=0}\,\sqrt{M(t)} \right)\;
	 =\;   \sqrt{M(0)}^{\,-1}\,\dot{M}(0)\,.	
 $$
It follows that
 $$
   \Tr\left( \left.\mbox{\Large $\frac{d}{dt}$}\right|_{t=0}\,\sqrt{M(t)}\right)\;
	 =\;  \mbox{\Large $\frac{1}{2}$}\,\Tr\left(\sqrt{M(0)}^{\,-1}\,\dot{M}(0)\right)\,.	
 $$
Slightly more generally,
 if $C\in C^{\infty}(\End_{\Bbb C}(E))$ commutes with $\sqrt{M(0)}$,
 then
 $$
    \sqrt{M(0)}^{\,-1}
	   \left(C\cdot \left.\mbox{\Large $\frac{d}{dt}$}\right|_{t=0}\,\sqrt{M(t)}\right)\,
	    \sqrt{M(0)}\,
	    +\,     \left( C\cdot\left. \mbox{\Large $\frac{d}{dt}$}\right|_{t=0}\,\sqrt{M(t)} \right)\;
	 =\;  C\cdot \sqrt{M(0)}^{\,-1}\,\dot{M}(0)\,.	
 $$
It follows that
 $$
   \Tr\left( C\cdot \left.\mbox{\Large $\frac{d}{dt}$}\right|_{t=0}\,\sqrt{M(t)}\right)\;
	 =\;  \mbox{\Large $\frac{1}{2}$}\,
	         \Tr\left(C\cdot\sqrt{M(0)}^{\,-1}\,\dot{M}(0)\right)\,.	
 $$

\bigskip

\noindent
$(b)$ {\it Identities on  symmetrized determinant and its differentiation}\hspace{1em}
The Leibniz rule holds for a symmetric product:
 $$
  \partialdot(a_1\odot\,\cdots\,\odot  a_m)\;
   =\;  \sum_{\mu=1}^m
            a_1\odot\,\cdots\,\odot a_{\mu-1}
			\odot (\partialdot a_{\mu})
			\odot a_{\mu+1}\odot\,\cdots\,\odot a_m\,.
 $$
It follows that
 if let $M=[M^{(1)},\,\cdots\,,\,M^{(m)}]$
  be the presentation of an $m\times m$ matrix $M$ in terms of its column vectors,
 then
 $$
  \partialdot \SymDet(M)\;
   =\; \sum_{\nu=1}^m
          \SymDet(
          [M^{(1)},\, \cdots\,,\, M^{(\nu-1)},\,
            \partialdot M^{(\nu)},\,
 		   M^{(\nu+1)},\,\cdots\,,\, M^{(m)}])\,.
 $$
 Similarly, for $M$ presented in terms of its row vectors.
  
\bigskip

\noindent
$(c)$ {\it Trace and Lie bracket}\hspace{1em}
For $r\times r$ matrices or matrix-valued functions $A$, $B$, and $C$,
 $$
   \Tr(A\,[B,C])\;=\; \Tr([A,B]\,C)\,.
 $$

\bigskip

\noindent
$(d)$ {\it $\partialdot \Tr=\Tr D_{\mbox{\LARGE $\cdot$}}$}\hspace{1em}
Recall the induced connection $D$ on $\End_{\Bbb C}(E)$ from $\nabla$ on $E$.
\begin{itemize}
 \item[\LARGE $\cdot$] {\it
 Let $s\in C^{\infty}(\End_{\Bbb C}(E))$.
 Then $\partialdot\Tr(s) = \Tr(D_{\mbox{\LARGE $\cdot$}}s)$.}
\end{itemize}
 
\begin{proof}
 In any local presentation of $E$, let $\nabla=d+A$,
  where $A$ is the $\End_{\Bbb C}(E)$-valued connection $1$-form on $X$
  with respect to the local trivialization.
 Then $D=d+[A,\,\cdot\,]$ with respect to the induced local trivialization of $\End_{\Bbb C}(E)$.
 It follows that
  $$
    \Tr(D_{\mbox{\LARGE $\cdot$}}(s))\;
	 =\; \Tr(\partialdot s\,+\, [A_{\mbox{\LARGE $\cdot$}}, s])\;
	 =\; \Tr(\partialdot s)\;
     =\; \partialdot\Tr(s)\,.	
  $$
\end{proof}

\bigskip

\begin{flushleft}
{\bf The first variation of each ingredient in the Dirac-Born-Infeld action}
\end{flushleft}
$(a)$ {\it The first variation of
    $e^{-\varphi^{\sharp}(\Phi)}$
	and $\varphi^{\sharp}(E_{ij})$}\hspace{1em}
Recall Remark/Notation 4.2.3.5.
Then, it follows from Proposition~4.2.3.1 that
 %
 \begin{eqnarray*}
  \lefteqn{
   \left.\mbox{\Large $\frac{d}{dt}$}\right|_{t=0}
      \left(e^{-\varphi_T^{\sharp}(\Phi)}\right)\;\;
    =\;\;
       R^{\,e^{-\Phi}}[1]|
	      _{\scriptsizeboldy^{\tinyboldd}
 			                    \rightsquigarrow\,
								   ^{\mbox{\large $\cdot$}}\!
								     (\varphi^{\sharp}(\scriptsizeboldy^{\tinyboldd}))}     }\\
   && =\;\;
       \sum_{i^{\prime}=1}^n\,
	   \sum_{d=0}^{\mbox{\tiny $\bullet$}}
	   \sum_{\;\,\scriptsizeboldd,\, |\scriptsizeboldd|=d}
	   \sum_{\;\vec{\pi}
	                         \in \vec{\mbox{\scriptsize\it Ptn}}(1, d),\,
					   i_{(\vec{\pi},\,\tinyboldd)}=i^{\prime}}	                           \\
    && \hspace{7em}	
	  ([\partial^{\vec{\pi}}_{y^{i^{\prime}}}]
	             R^{\,e^{-\Phi}}[1]_{(\scriptsizeboldd)})^L(\varphi^{\sharp}(\boldy))
	   \cdot
	   \dot{\varphi}^{\sharp}(y^{i^{\prime}})
	   \cdot
	 ([\partial^{\vec{\pi}}_{y^{i^{\prime}}}]
	             R^{\,e^{-\Phi}}[1]_{(\scriptsizeboldd)})^R(\varphi^{\sharp}(\boldy)) \\									 
  && =:\:\;
    \sum_{i^{\prime}=1}^n\,
	\sum_{\; d,\,\scriptsizeboldd,\,\vec{\pi};\,
	                      |\scriptsizeboldd|=d,\, i_{(\vec{\pi},\tinyboldd)}=i^{\prime}}\,
	  R^{\,e^{-\Phi}}[1]_{(\scriptsizeboldd,\,\vec{\pi})}^{\;L}(\varphi^{\sharp}(\boldy))
	  \cdot \dot{\varphi}^{\sharp}(y^{i^{\prime}})
	  \cdot
	  R^{\,e^{-\Phi}}[1]_{(\scriptsizeboldd,\,\vec{\pi})}^{\;R}(\varphi^{\sharp}(\boldy))
 \end{eqnarray*}
 and
 \begin{eqnarray*}
  \lefteqn{
   \left.\mbox{\Large $\frac{d}{dt}$}\right|_{t=0}
          \left(\varphi_T^{\sharp}(E_{ij})\right)\;\;
   =\;\;  R^{\,E_{ij}}[1]|
             _{\scriptsizeboldy^{\tinyboldd}
 			                    \rightsquigarrow\,
								   ^{\mbox{\large $\cdot$}}\!
								     (\varphi^{\sharp}(\scriptsizeboldy^{\tinyboldd}))}}\\
  && =\;\;
       \sum_{i^{\prime}=1}^n\,
	   \sum_{d=0}^{\mbox{\tiny $\bullet$}}
	   \sum_{\;\,\scriptsizeboldd,\, |\scriptsizeboldd|=d}
	   \sum_{\;\vec{\pi}
	                         \in \vec{\mbox{\scriptsize\it Ptn}}(1, d),\,
					   i_{(\vec{\pi},\,\tinyboldd)}=i^{\prime}}	                           \\
    && \hspace{7em}	
	  ([\partial^{\vec{\pi}}_{y^{i^{\prime}}}]
	             R^{\,E_{ij}}[1]_{(\scriptsizeboldd)})^L(\varphi^{\sharp}(\boldy))
	   \cdot
	   \dot{\varphi}^{\sharp}(y^{i^{\prime}})
	   \cdot
	 ([\partial^{\vec{\pi}}_{y^{i^{\prime}}}]
	             R^{\,E_{ij}}[1]_{(\scriptsizeboldd)})^R(\varphi^{\sharp}(\boldy)) \\									 
  && =:\:\;
    \sum_{i^{\prime}=1}^n\,
	\sum_{\;d,\,\scriptsizeboldd,\,\vec{\pi};\,
	                    |\scriptsizeboldd|=d,\, i_{(\vec{\pi}, \tinyboldd)}=i^{\prime} }\,
	  R^{\,E_{ij}}[1]_{(\scriptsizeboldd,\,\vec{\pi})}^{\;L}(\varphi^{\sharp}(\boldy))
	  \cdot \dot{\varphi}^{\sharp}(y^{i^{\prime}})
	  \cdot
	  R^{\,E_{ij}}[1]_{(\scriptsizeboldd,\,\vec{\pi})}^{\;R}(\varphi^{\sharp}(\boldy))\,.	
 \end{eqnarray*}

\bigskip

\noindent
$(b)$
{\it The first variation of $D_{\mu}\varphi^{\sharp}(y^i)$ and $F_{\mu\nu}$}\hspace{1em}
By straightforward computation,
\begin{eqnarray*}
 \left.\mbox{\Large $\frac{d}{dt}$}\right|_{t=0}
      \left(D^T_{\mu}\varphi_T^{\sharp}(y^i)\right)
   & = & \left.\mbox{\Large $\frac{d}{dt}$}\right|_{t=0}
                 \left(\partial_{\mu}\varphi_T^{\sharp}(y^i)\,
				            +\, [A^T_{\mu}, \varphi_T^{\sharp}(y^i)]\right)  \\
 &   = & D_{\mu}\dot{\varphi}^{\sharp}(y^i)\,
               -\, [\varphi^{\sharp}(y^i), \dot{A}_{\mu}]					
\end{eqnarray*}
and
\begin{eqnarray*}
 \left.\mbox{\Large $\frac{d}{dt}$}\right|_{t=0}F^T_{\mu\nu}
  & = &  \left.\mbox{\Large $\frac{d}{dt}$}\right|_{t=0}
                 [\nabla^T_{\mu},\nabla^T_{\nu}]\\
 &  = & D_{\mu}\dot{A}_{\nu}\,-\, D_{\nu}\dot{A}_{\mu}\,.
\end{eqnarray*}

\bigskip

\begin{flushleft}
{\bf The first variation of the Dirac-Born-Infeld action}
\end{flushleft}
With all the ingredients prepared, the computation of the first variation of $S_{\DBI}(\varphi,\nabla)$
 is now straightforward, though some of the expressions may look complicated due to noncommutativity.
We proceed in five steps.
 
\bigskip

\noindent
{\it Step $(1):$ Input from all the pieces}
  
\medskip

\noindent
Let
 $$
    \mbox{\boldmath $M$}_{\mu\nu}(t)\;
	  :=\;   \sum_{i,j}
				 \varphi_t^{\sharp}(E_{ij})
	                  D^t_{\mu}\varphi_t^{\sharp}(y^i)D^t_{\nu}\varphi_t^{\sharp}(y^j)\,
                   +\, 2\pi\alpha^{\prime}\,[\nabla^t_{\mu},\nabla^t_{\nu}]\;
	 \in\; C^{\infty}(\End_{\Bbb C}(E|_U))
 $$
 and
  $\boldM(t) := [\boldM_{\mu\nu}(t)]_{\mu\nu}$
     the $m\times m$ matrix with $(\mu,\nu)$-entry $\boldM_{\mu\nu}(t)$.
Then,
{\footnotesize
 \begin{eqnarray*}
  \left.\mbox{\Large $\frac{d}{dt}$}\right|_{t=0}\,
	 S_{\DBI}(\varphi_t,\nabla^t)
  & = & 	
	-T_{m-1}\,\left.\mbox{\Large $\frac{d}{dt}$}\right|_{t=0}\,
	\int_U\,  \Real  \left( \Tr \left(
          e^{-\varphi_t^{\sharp}(\Phi)}\,
		    \sqrt{-\SymDet(\boldM(t))  \,}\right )\right)
	   d^m\mbox{\boldmath $x$}   \\
 &= &
     -T_{m-1}\,\int_U\,  \Real  \left( \Tr 	
	          \left.\mbox{\Large $\frac{d}{dt}$}\right|_{t=0}\,	
	      \left(
                e^{-\varphi_t^{\sharp}(\Phi)}\,
		          \sqrt{-\SymDet(\boldM(t))  \,}\right )\right)
	         d^m\mbox{\boldmath $x$}\,;
 \end{eqnarray*}
}
{\footnotesize
\begin{eqnarray*}
 && \hspace{-2em}
       \Tr 	
	    \left.\mbox{\Large $\frac{d}{dt}$}\right|_{t=0}\,	
	      \left(
                e^{-\varphi_t^{\sharp}(\Phi)}\,
		          \sqrt{-\SymDet(\boldM(t))  \,}\right )\\
 &&	=\;
      \Tr \left(
	     (R^{\,e^{-\Phi}}[1]_{t=0})		
			   \!\!\left.\rule{0ex}{0.8em}\right|_{\mbox{\boldmath\scriptsize $y$}^I
 			                    \rightsquigarrow\,
								   ^{\mbox{\large $\cdot$}}\!
								     (\varphi^{\sharp}(\mbox{\boldmath\scriptsize $y$}^I))}\,
       \cdot\, \sqrt{-\SymDet(\boldM(0))  \,}
           \right)\;
	+\; 	
	  \Tr 	
	     \left(
                e^{-\varphi^{\sharp}(\Phi)}\,
		     \left.\mbox{\Large $\frac{d}{dt}$}\right|_{t=0}\!\!
		          \sqrt{-\SymDet(\boldM(t))  \,}\right ).
\end{eqnarray*}
}
Since $(\varphi,\nabla)$ is admissible,
  $e^{-\varphi^{\sharp}(\Phi)}$ and $\sqrt{-\SymDet(\boldM(t))}$ commute.
Thus,
{\footnotesize
 $$
    \Tr 	\left(
                e^{-\varphi^{\sharp}(\Phi)}\,
		     \left.\mbox{\Large $\frac{d}{dt}$}\right|_{t=0}\!\!
		          \sqrt{-\SymDet(\boldM(t))  \,}\right )\;
  =\; \frac{-1}{2}\, \Tr\left(
                e^{-\varphi^{\sharp}(\Phi)}\,
		          \sqrt{-\SymDet(\boldM(0))  \,}^{\,-1}\cdot				
		     \left.\mbox{\Large $\frac{d}{dt}$}\right|_{t=0}\!\!
		        \SymDet(\boldM(t))  	
			                                            \right).
 $$
}
Denote by $[\,\cdot\,]^{\transpose}$  the transpose of the matrix $[\,\cdot\,]$
and let
 $$
    \boldM(t)\;
	 =\;
	   \left[
	      \begin{array}{c} \boldM_{(1)}(t)\\    \cdots \\    \boldM_{(m)}(t)  \end{array}		
	   \right] \;
	 =\;
    \left[ \boldM_{(1)}^{\transpose},\,\cdots\,,\,\boldM_{(m)}^{\transpose} \right]^{\transpose}	 
 $$
   be the presentation of $\boldM(t)$ in terms of its row vectors   and
 denote $\left.\frac{d}{dt}\right|_{t=0}\boldM_{(\mu)}(t)$
   by $\dot{\boldM}_{(\mu)}(0)$,
   for $\mu=1,\,\ldots\,,\,m$.
Then
{\footnotesize
 $$
   \left.\mbox{\Large $\frac{d}{dt}$}\right|_{t=0}\!\!
		        \SymDet(\boldM(t))\;  	
    =\; \sum_{\mu=1}^m
               \SymDet(
			     [\boldM_{(1)}(0)^{\transpose},\,\cdots\,,\, \boldM_{(\mu-1)}(0)^{\transpose},\,
			      \dot{\boldM}_{(\mu)}(0)^{\transpose},\,
				  \boldM_{(\mu+1)}(0)^{\transpose},\,
				  \cdots\,,\, \boldM_{(m)}(0)^{\transpose}]^{\transpose})\,.
 $$
}
Denote
  $\left.\frac{d}{dt}\right|_{t=0}\boldM_{\mu\nu}(t)$ by $\dot{\boldM}_{\mu\nu}(0)$,
   for $\mu,\,\nu=1,\,\ldots\,,\,m$.
Then, the $\nu$-th entry in $\dot{\boldM}_{(\mu)}(0)$ is given by
{\footnotesize
 \begin{eqnarray*}
   \dot{\boldM}_{\mu\nu}(0)
    & = &   \sum_{i,j}	
	             R^{\,E_{ij}}[1]|
			              _{\scriptsizeboldy^{\tinyboldd}
 			                    \rightsquigarrow\,
								   ^{\mbox{\large $\cdot$}}\!
								     (\varphi^{\sharp}(\scriptsizeboldy^{\tinyboldd}))}
                    \cdot									
	                  D_{\mu}\varphi^{\sharp}(y^i)D_{\nu}\varphi^{\sharp}(y^j) \\
    &&      +\; 	
	             \sum_{i,j}	
				 \varphi^{\sharp}(E_{ij})
				   \cdot(
				     D_{\mu}\dot{\varphi}^{\sharp}(y^i)\,
				            -\, [\varphi^{\sharp}(y^i), \dot{A}_{\mu}])		
				   \cdot					
					  D_{\nu}\varphi^{\sharp}(y^j)      \\
    &&      +\;
	             \sum_{i,j}	
				 \varphi^{\sharp}(E_{ij})
	               D_{\mu}\varphi^{\sharp}(y^i)					
					\cdot (D_{\nu}\dot{\varphi}^{\sharp}(y^j)\,
                                           -\, [\varphi^{\sharp}(y^j), \dot{A}_{\nu}])\;
               +\, 2\pi\alpha^{\prime}\,
			        (D_{\mu}\dot{A}_{\nu}-D_{\nu}\dot{A}_{\mu})\,.
 \end{eqnarray*}
}
With
 $$
    \mbox{\boldmath $M$}_{\mu\nu}(0)\;
	  =\;   \sum_{i,j}
				 \varphi^{\sharp}(E_{ij})
	                  D_{\mu}\varphi^{\sharp}(y^i)D_{\nu}\varphi^{\sharp}(y^j)\,
                   +\, 2\pi\alpha^{\prime}\,[\nabla_{\mu},\nabla_{\nu}]\,,
 $$
 one has altogether:
{\footnotesize
 \begin{eqnarray*}
   \left.\mbox{\Large $\frac{d}{dt}$}\right|_{t=0}\,
   S_{\DBI}(\varphi_t,\nabla^t)
    & = & 	
	 -T_{m-1}\,\left.\mbox{\Large $\frac{d}{dt}$}\right|_{t=0}\,
	  \int_U\,  \Real  \left( \Tr \left(
          e^{-\varphi_t^{\sharp}(\Phi)}\,
		    \sqrt{-\SymDet(\boldM(t))  \,}\right )\right)
	   d^m\mbox{\boldmath $x$}   \\[1.2ex]
   && \hspace{-10em}
    =\;  	
	-T_{m-1}\,
	     \int_U\,  \Real\!\left(\rule{0ex}{1em}\right.\!\!\Tr\!\left(\rule{0ex}{1em}\right.
          R^{\,e^{-\Phi}}[1]|
	                    _{\scriptsizeboldy^{\tinyboldd}
 			                    \rightsquigarrow\,
								   ^{\mbox{\large $\cdot$}}\!
								     (\varphi^{\sharp}(\scriptsizeboldy^{\tinyboldd}))}\,
       \cdot\, \sqrt{-\SymDet(\boldM(0))  \,}   \\
    && \hspace{-7em}
	  -\, \mbox{\large $\frac{1}{2}$}\,
             e^{-\varphi^{\sharp}(\Phi)}\,
		     \sqrt{-\SymDet(\boldM(0))  \,}^{\,-1}\\
    &&	\hspace{-5em}		
		     \cdot		
              \sum_{\mu=1}^m
               \SymDet(
			     [\boldM_{(1)}(0)^{\transpose},\,\cdots\,,\, \boldM_{(\mu-1)}(0)^{\transpose},\,
			      \dot{\boldM}_{(\mu)}(0)^{\transpose},\,
				  \boldM_{(\mu+1)}(0)^{\transpose},\,\cdots\,,\,
				  \boldM_{(m)}(0)^{\transpose}]^{\transpose})
          \left.\rule{0ex}{1em}\right )\!\!\!\left.\rule{0ex}{1em}\right)
	                 d^m\mbox{\boldmath $x$}                       \\[1.2ex] 
    && \hspace{-10em}
    =\;  	
	-T_{m-1}\,
	     \int_U\,  \Real\!\left(\rule{0ex}{1em}\right.\!\!\Tr\!\left(\rule{0ex}{1em}\right.
      (R^{\,e^{-\Phi}}[1])|
	         _{\scriptsizeboldy^{\tinyboldd}
 			                    \rightsquigarrow\,
								   ^{\mbox{\large $\cdot$}}\!
								     (\varphi^{\sharp}(\scriptsizeboldy^{\tinyboldd}))}\,
       \cdot\, \sqrt{-\SymDet(\boldM(0))  \,}   \\
    && \hspace{-7em}
	  -\, \mbox{\large $\frac{1}{2}$}\,
             e^{-\varphi^{\sharp}(\Phi)}\,
		     \sqrt{-\SymDet(\boldM(0))  \,}^{\,-1}\\
    &&	\hspace{-5em}		
		     \cdot		
              \sum_{\mu=1}^m
			   \sum_{\sigma\in \tinySym_m}
			    (-1)^{\sigma}
			     \boldM_{1\,\sigma(1)}(0)
				      \odot\,\cdots\,\odot \boldM_{(\mu-1)\,\sigma(\mu-1)}(0)\\[-3ex]
    &&  \hspace{7em}
    	    \odot \dot{\boldM}_{\mu\,\sigma(\mu)}(0)
				  \odot \boldM_{(\mu+1)\,\sigma(\mu+1)}(0)
				      \odot\,\cdots\,\odot \boldM_{m\,\sigma(m)}(0)				
          \left.\rule{0ex}{1em}\right )\!\!\!\left.\rule{0ex}{1em}\right)
	                 d^m\mbox{\boldmath $x$}\,.     					
 \end{eqnarray*}
}

\bigskip

\noindent
{\it Step $(2):$} \parbox[t]{37em}{\it
      Arrangement to boundary terms and the linear functional
	  $\delta S_{\DBI}(\varphi,\nabla)/\delta(\varphi,\nabla)$ on\\
      $(\dot{\varphi}^{\sharp}(y^1),\,\cdots\,,\, \dot{\varphi}^{\sharp}(y^n);\,
              \dot{A}_1,\,\cdots\,,\, \dot{A}_m)$}
  
\medskip

\noindent
Summands from the first cluster
 $$
   R^{\,e^{-\Phi}}[1]|
                _{\scriptsizeboldy^{\tinyboldd}
 			                    \rightsquigarrow\,
								   ^{\mbox{\large $\cdot$}}\!
								     (\varphi^{\sharp}(\scriptsizeboldy^{\tinyboldd}))}\,
       \cdot\, \sqrt{-\SymDet(\boldM(0))  \,}
 $$
 contain only $\dot{\varphi}^{\sharp}(y^i)$, $i=1,\,\ldots\,,\,n$,
  from $(R^{\,e^{-\Phi}}[1])|
                            _{\boldy^{\tinyboldd}
 			                    \rightsquigarrow\,
								   ^{\mbox{\large $\cdot$}}\!
								     (\varphi^{\sharp}(\scriptsizeboldy^{\tinyboldd}))}$.
Hence,
 it contributes solely to the linear functional
   $\delta S_{\DBI}^{(\Phi,g,B)}(\varphi,\nabla)/\delta(\varphi,\nabla)$
   on
   $(\dot{\varphi}^{\sharp}(y^1),\,\cdots\,,\, \dot{\varphi}^{\sharp}(y^n);\,
              \dot{A}_1,\,\cdots\,,\, \dot{A}_m)$ 		
   and, hence,
 to the equations of motion for $(\varphi,\nabla)$.
 
On the other hand, summands from the expansion of the second cluster
  \begin{eqnarray*}
    && \hspace{-1em}
    -\, \mbox{\large $\frac{1}{2}$}\,
             e^{-\varphi^{\sharp}(\Phi)}\,
		     \sqrt{-\SymDet(\boldM(0))  \,}^{\,-1}\\
    &&	\hspace{1em}
		     \cdot		
              \sum_{\mu=1}^m
			   \sum_{\sigma\in \tinySym_m}
			     (-1)^{\sigma}
			          \boldM_{1\,\sigma(1)}(0)
				      \odot\,\cdots\,\odot \boldM_{(\mu-1)\,\sigma(\mu-1)}(0)\\[-3ex]
    &&  \hspace{13em}
    	    \odot \dot{\boldM}_{\mu\,\sigma(\mu)}(0)
				  \odot \boldM_{(\mu+1)\,\sigma(\mu+1)}(0)
				      \odot\,\cdots\,\odot \boldM_{m\,\sigma(m)}(0)	
  \end{eqnarray*}					
 are of two types:
   \begin{itemize}
    \item[\LARGE $\cdot$]
     One contains a factor in the list
       $\dot{\varphi}^{\sharp}(y^i)$, $i=1,\,\ldots\,,\,n$,
    	$\dot{A}_{\mu}$, $\mu=1,\,\ldots\,,\, m$
		from some $\dot{\boldM}_{\mu^{\prime}\nu^{\prime}}(0)$,
		 $\mu^{\prime},\, \nu^{\prime}=1,\,\ldots\,,\,m$.
     They contribute to the linear functional
	   $\delta S_{\DBI}(\varphi,\nabla)/\delta(\varphi,\nabla)$ on
       $(\dot{\varphi}^{\sharp}(y^1),\,\cdots\,,\, \dot{\varphi}^{\sharp}(y^n);\,
              \dot{A}_1,\,\cdots\,,\, \dot{A}_m)$ and, hence,
      to the equations of motion for $(\varphi,\nabla)$.
	
	\item[\LARGE $\cdot$]
	 The other contains a factor in the list
       $D_{\mu}\dot{\varphi}(y^i)$, $i=1,\,\ldots\,,\, n$, $\mu=1,\,\ldots\,,\,m$,
       $D_{\mu}\dot{A}_{\nu}$, $\mu,\,\nu=1,\,\ldots\,,\,m$,
       from some $\dot{\boldM}_{\mu^{\prime}\nu^{\prime}}(0)$,
		 $\mu^{\prime},\, \nu^{\prime}=1,\,\ldots\,,\,m$.
     After integration by parts,	
	   each contributes a boundary term in an integral $\int_{\partial U}(\,\cdots\,)$
	    and a term in the linear functional
	    $\delta S_{\DBI}(\varphi,\nabla)/\delta(\varphi,\nabla)$ on
        $(\dot{\varphi}^{\sharp}(y^1),\,\cdots\,,\, \dot{\varphi}^{\sharp}(y^n);\,
              \dot{A}_1,\,\cdots\,,\, \dot{A}_m)$.
	  The latter contributes then to the equations of motion for $(\varphi,\nabla)$.	     	   
   \end{itemize}	

We now proceed to study their details.

\bigskip

\noindent
{\it Step $(3):$ Details for the first cluster}

\medskip

\noindent
For the first cluster,
 
{\footnotesize
 \begin{eqnarray*}
  \lefteqn{
    \Tr\!\left(\rule{0ex}{1em}\right.\!\!
       R^{\,e^{-\Phi}}[1]|
	         _{\scriptsizeboldy^{\tinyboldd}
 			                    \rightsquigarrow\,
								   ^{\mbox{\large $\cdot$}}\!
								     (\varphi^{\sharp}(\scriptsizeboldy^{\tinyboldd}))}\,
       \cdot\, \sqrt{-\SymDet(\boldM(0))  \,}
	      \!\left.\rule{0ex}{1em}\right)     }\\
   && =\:\:
     \Tr\!\left(\rule{0ex}{1em}\right.\!\!
       \left(\rule{0ex}{1em}\right.
	    \sum_{i^{\prime}=1}^n\,
	     \sum_{\;d,\,\scriptsizeboldd,\,\vec{\pi};\,
		                     |\scriptsizeboldd|=d,\, i_{(\vec{\pi}, \tinyboldd)}=i^{\prime} }\,
	      R^{\,e^{-\Phi}}[1]_{(\scriptsizeboldd,\,\vec{\pi})}^{\;L}(\varphi^{\sharp}(\boldy))
	    \cdot \dot{\varphi}^{\sharp}(y^{i^{\prime}})
	    \cdot
	    R^{\,e^{-\Phi}}[1]_{(\scriptsizeboldd,\,\vec{\pi})}^{\;R}(\varphi^{\sharp}(\boldy))
	    \left.\rule{0ex}{1em}\right)
       \cdot\, \sqrt{-\SymDet(\boldM(0))  \,}
	      \!\left.\rule{0ex}{1em}\right)                   \\
   && =\:\:
     \Tr\!\left(\rule{0ex}{1em}\right.
	    \sum_{i^{\prime}=1}^n
		  \left(\rule{0ex}{1em}\right.
	     \sum_{\;d,\,\scriptsizeboldd,\,\vec{\pi};\,
		                     |\scriptsizeboldd|=d,\,  i_{(\vec{\pi},\tinyboldd)}=i^{\prime}}\,
	      R^{\,e^{-\Phi}}[1]_{(\scriptsizeboldd,\,\vec{\pi})}^{\;R}(\varphi^{\sharp}(\boldy))
          \cdot
		    \sqrt{-\SymDet(\boldM(0))  \,}
		  \cdot	
	    R^{\,e^{-\Phi}}[1]_{(\scriptsizeboldd,\,\vec{\pi})}^{\;L}(\varphi^{\sharp}(\boldy))
		  \left.\rule{0ex}{1em}\right)
	      \cdot
		  \dot{\varphi}^{\sharp}(y^{i^{\prime}})
	      \!\left.\rule{0ex}{1em}\right)                     \\
   && =:\;
     \Tr \!\left(\rule{0ex}{1em}\right.
      \sum_{i^{\prime}=1}^n
	    \NL^{1,\,(\Phi,g,B)}_{i^{\prime}}(\varphi,\nabla)
		 \cdot  \dot{\varphi}^{\sharp}(y^{i^{\prime}})
	       \left.\rule{0ex}{1em}\right).
 \end{eqnarray*}
}

\bigskip

\noindent
{\it Step $(4):$ Details for the second cluster}

\medskip

\noindent
For the second cluster, recall Lemma~3.1.3.8.
Then
																		
{\footnotesize																			
 \begin{eqnarray*}
   && \SymDet(
			     [\boldM_{(1)}(0)^{\transpose},\,\cdots\,,\, \boldM_{(\mu-1)}(0)^{\transpose},\,
			      \dot{\boldM}_{(\mu)}(0)^{\transpose},\,
				  \boldM_{(\mu+1)}(0)^{\transpose},\,
				  \cdots\,,\, \boldM_{(m)}(0)^{\transpose}]^{\transpose})\;\;=\\
   &&  \hspace{-.6em}
          \mbox{\large $\frac{1}{m!}$}
              \sum_{\mu^{\prime}=1}^m
              \sum_{\mbox{\tiny
			                   $\begin{array}{c} \sigma\in\tinySym_m\\
							                                              \sigma(\mu^{\prime})=\mu \end{array}$}}\!\!
                (-1)^{\sigma}
				\Det(
			     [\boldM_{(\sigma{(1)})}(0)^{\transpose},\,
				   \cdots\,,\, \boldM_{(\sigma(\mu^{\prime}-1))}(0)^{\transpose},\,
			      \dot{\boldM}_{(\mu)}(0)^{\transpose},\,
				  \boldM_{(\sigma(\mu^{\prime}+1))}(0)^{\transpose},\,
				     \cdots\,,\, \boldM_{(\sigma(m))}(0)^{\transpose}]^{\transpose})\,.
 \end{eqnarray*}} 
Thus, denoting the factor {\footnotesize
 $\;-\,\frac{1}{2}\,
       e^{-\varphi^{\sharp}(\Phi)}\,\sqrt{-\SymDet(\boldM(0))  \,}^{\,-1}\;$}
 by $F_2(\varphi,\nabla;\Phi,g,B)$,
 %
{\footnotesize
 \begin{eqnarray*}
    && \hspace{-2em}
	   \Tr\!\left(\rule{0ex}{1em}\right.
	    -\, \mbox{\large $\frac{1}{2}$}\,
             e^{-\varphi^{\sharp}(\Phi)}\,
		     \sqrt{-\SymDet(\boldM(0))  \,}^{\,-1}\\
    &&	\hspace{3em}	
		     \cdot\;		
              \sum_{\mu=1}^m
               \SymDet(
			     [\boldM_{(1)}(0)^{\transpose},\,\cdots\,,\, \boldM_{(\mu-1)}(0)^{\transpose},\,
			      \dot{\boldM}_{(\mu)}(0)^{\transpose},\,
				  \boldM_{(\mu+1)}(0)^{\transpose},\,
				  \cdots\,,\, \boldM_{(m)}(0)^{\transpose}]^{\transpose})
          \left.\rule{0ex}{1em}\right )\\
   && \hspace{-1em}
     =\;
	   \Tr\!\left(\rule{0ex}{1em}\right.
	     F_2(\varphi,\nabla;\Phi,g,B)\,       	
		     \cdot\,
              \sum_{\mu=1}^m
               \SymDet(
			     [\boldM_{(1)}(0)^{\transpose},\,\cdots\,,\, \boldM_{(\mu-1)}(0)^{\transpose},\,
			      \dot{\boldM}_{(\mu)}(0)^{\transpose},\,
				  \boldM_{(\mu+1)}(0)^{\transpose},\,
				  \cdots\,,\, \boldM_{(m)}(0)^{\transpose}]^{\transpose})
          \left.\rule{0ex}{1em}\right )\\		
   && \hspace{-1em}
     =\;
	   \Tr\!\left(\rule{0ex}{1em}\right.
	    \mbox{\large $\frac{1}{m!}$}\,
	     F_2(\varphi,\nabla;\Phi,g,B)\,       	
		     \cdot\,			
              \sum_{\mu=1}^m\,					
              \sum_{\mu^{\prime}=1}^m
              \sum_{\mbox{\tiny
			                   $\begin{array}{c} \sigma\in\tinySym_m\\
							                                              \sigma(\mu^{\prime})=\mu \end{array}$}} \\
   && \hspace{3em}																		
                (-1)^{\sigma}\,
				\Det(
			     [\boldM_{(\sigma{(1)})}(0)^{\transpose},\,
				   \cdots\,,\, \boldM_{(\sigma(\mu^{\prime}-1))}(0)^{\transpose},\,
			      \dot{\boldM}_{(\mu)}(0)^{\transpose},\,
				  \boldM_{(\sigma(\mu^{\prime}+1))}(0)^{\transpose},\,
				     \cdots\,,\, \boldM_{(\sigma(m))}(0)^{\transpose}]^{\transpose})		
          \left.\rule{0ex}{1em}\right )\\		  		  	
   && \hspace{-1em}
     =\;
	   \Tr\!\left(\rule{0ex}{1em}\right.
	    \mbox{\large $\frac{1}{m!}$}\,
              \sum_{\mu=1}^m\,					
              \sum_{\mu^{\prime}=1}^m
              \sum_{\mbox{\tiny
			                   $\begin{array}{c} \sigma\in\tinySym_m\\
							                                              \sigma(\mu^{\prime})=\mu \end{array}$}}
             (-1)^{\sigma}																		  \\
   && \hspace{2em}\cdot\,																		
				\Det(
			     [ F_2(\varphi,\nabla;\Phi,g,B)\,\boldM_{(\sigma{(1)})}(0)^{\transpose},\,
				   \cdots\,,\, \boldM_{(\sigma(\mu^{\prime}-1))}(0)^{\transpose},\,
			      \dot{\boldM}_{(\mu)}(0)^{\transpose},\,
				  \boldM_{(\sigma(\mu^{\prime}+1))}(0)^{\transpose},\,
				     \cdots\,,\, \boldM_{(\sigma(m))}(0)^{\transpose}]^{\transpose})		
          \left.\rule{0ex}{1em}\right )\\		
   && \hspace{-1em}
     =\;
	   \Tr\!\left(\rule{0ex}{1em}\right.
	    \mbox{\large $\frac{1}{m!}$}\,
              \sum_{\mu=1}^m\,					
              \sum_{\mu^{\prime}=1}^m
              \sum_{\mbox{\tiny
			                   $\begin{array}{c} \sigma\in\tinySym_m\\
							                                              \sigma(\mu^{\prime})=\mu \end{array}$}}
             (-1)^{\sigma}\,(-1)^{\mu^{\prime}(m-\mu^{\prime})}																		  \\
   && \hspace{2em}\cdot\,																		
				\Det(
			     [ 	\boldM_{(\sigma(\mu^{\prime}+1))}(0)^{\transpose},\,
				     \cdots\,,\, \boldM_{(\sigma(m))}(0)^{\transpose},\,				
				  F_2(\varphi,\nabla;\Phi,g,B)\,\boldM_{(\sigma{(1)})}(0)^{\transpose},\,
				    \cdots\,,\, \boldM_{(\sigma(\mu^{\prime}-1))}(0)^{\transpose},\,
			      \dot{\boldM}_{(\mu)}(0)^{\transpose}]^{\transpose})		
          \left.\rule{0ex}{1em}\right )\\		  		
   && \hspace{1em}\mbox{\normalsize (by the invariance of trace under cyclic permutations)}\,.		  
 \end{eqnarray*}
}

\noindent
Note that $\dot{\boldM}_{\mu\nu}(0)$, $\mu,\nu=1,\,\ldots\,,\, m$,
 now appear uniformly as the last factor in the summands from the expansion of
 $\Det([\,\cdots\,]^{\transpose})$ above.
Let $\Minor(\varphi,\nabla;\Phi,g,B\,|\,\mu^{\prime},\sigma)_{\mu\nu}$
 be the $(m,\nu)$-minor of
 {\small
  $[ 	\boldM_{(\sigma(\mu^{\prime}+1))}(0),\,
				     \cdots\,,\, \boldM_{(\sigma(m))}(0),\,				
				  F_2(\varphi,\nabla;\Phi,g,B)\,\boldM_{(\sigma{(1)})}(0),\,
				    \cdots\,,\, \boldM_{(\sigma(\mu^{\prime}-1))}(0),\,
			      \dot{\boldM}_{(\mu)}(0)]^{\transpose}$}. Then:
 
{\footnotesize
 \begin{eqnarray*}
   && \hspace{-.6em}
     =\;
	   \Tr\!\left(\rule{0ex}{1em}\right.
	    \mbox{\large $\frac{1}{m!}$}\,
              \sum_{\mu=1}^m\,					
              \sum_{\mu^{\prime}=1}^m
              \sum_{\mbox{\tiny
			                   $\begin{array}{c} \sigma\in\tinySym_m\\
							                                              \sigma(\mu^{\prime})=\mu \end{array}$}}
             (-1)^{\sigma}\,(-1)^{\mu^{\prime}(m-\mu^{\prime})}\,
          \cdot\,																		
           \sum_{\nu=1}^m			
              (-1)^{m+\nu}			
			  \Minor(\varphi,\nabla;\Phi,g,B\,|\,\mu^{\prime},\sigma)_{\mu\nu}
			   \dot{\boldM}_{\mu\nu}(0)
          \left.\rule{0ex}{1em}\right )\\	
   && \hspace{-.6em}
      =\;
        \Tr\!\left(\rule{0ex}{1em}\right.
              \sum_{\mu=1}^m\,	\sum_{\nu=1}^m
		       \ComboMinor(\varphi,\nabla;\Phi,g,B)_{\mu\nu}\,	\dot{\boldM}_{\mu\nu}(0)
          \left.\rule{0ex}{1em}\right ),  \\	
        &&  \hspace{1em}
   		        \mbox{\normalsize where}\;\;
		        \ComboMinor(\varphi,\nabla;\Phi,g,B)_{\mu\nu}  \\
		&& \hspace{7em}		
		           :=\; \mbox{\large $\frac{1}{m!}$}\,
                   \sum_{\mu^{\prime}=1}^m
                   \sum_{\mbox{\tiny
			                        $\begin{array}{c} \sigma\in\tinySym_m\\
						      	                                              \sigma(\mu^{\prime})=\mu \end{array}$}}
                  (-1)^{\sigma}\,(-1)^{\mu^{\prime}(m-\mu^{\prime})+m+\nu}\,
                    \Minor(\varphi,\nabla;\Phi,g,B\,|\,\mu^{\prime},\sigma)_{\mu\nu}\,,  \\
   && \hspace{-.6em}			
      =\;
        \Tr\!\left(\rule{0ex}{1em}\right.
              \sum_{\mu=1}^m\,	\sum_{\nu=1}^m
		       \ComboMinor(\varphi,\nabla;\Phi,g,B)_{\mu\nu}			   			  \\
   &&  \hspace{6em}
           \cdot
		    \left(\rule{0ex}{1em}\right.
 			  \sum_{i,j}	
	           R^{\,E_{ij}}[1]|
			              _{\scriptsizeboldy^{\tinyboldd}
 			                    \rightsquigarrow\,
								   ^{\mbox{\large $\cdot$}}\!
								     (\varphi^{\sharp}(\scriptsizeboldy^{\tinyboldd}))}
                    \cdot									
	                  D_{\mu}\varphi^{\sharp}(y^i)D_{\nu}\varphi^{\sharp}(y^j) \\
    && \hspace{8em}
	        +\; 	
	             \sum_{i,j}	
				 \varphi^{\sharp}(E_{ij})
				   \cdot(
				     D_{\mu}\dot{\varphi}^{\sharp}(y^i)\,
				            -\, [\varphi^{\sharp}(y^i), \dot{A}_{\mu}])		
				   \cdot					
					  D_{\nu}\varphi^{\sharp}(y^j)      \\
    && \hspace{8em}
          	+\;
	             \sum_{i,j}	
				 \varphi^{\sharp}(E_{ij})
	               D_{\mu}\varphi^{\sharp}(y^i)					
					\cdot (D_{\nu}\dot{\varphi}^{\sharp}(y^j)\,
                                           -\, [\varphi^{\sharp}(y^j), \dot{A}_{\nu}])\;
               +\, 2\pi\alpha^{\prime}\,
			        (D_{\mu}\dot{A}_{\nu}-D_{\nu}\dot{A}_{\mu})
			  \left.\rule{0ex}{1em}\right)\!\!
          \left.\rule{0ex}{1em}\right )\\	
    && \hspace{-.6em}		
	      =\;\;  \mbox{(I)}\; +\; \mbox{(II)}\; +\; \mbox{(III)}\; +\; \mbox{(IV)}\hspace{2em}
		            \mbox{(defined in  {\it Step} (4.1) -- {\it Step} (4.4) below)}\,.
 \end{eqnarray*}
  }
 
Let us now study each of the four subclusters of the second cluster separately.

\vspace{6em}

\noindent
{\it Step $(4.1):$ The subcluster $\,${\rm (I)}}

{\footnotesize
\begin{eqnarray*}
  \mbox{(I)}
   &:=\; & 			
        \Tr\!\left(\rule{0ex}{1em}\right.
              \sum_{\mu=1}^m\,	\sum_{\nu=1}^m
		       \ComboMinor(\varphi,\nabla;\Phi,g,B)_{\mu\nu}	
              \cdot			
 			  \sum_{i,j}	
	           R^{\,E_{ij}}[1]|
			          _{\scriptsizeboldy^{\tinyboldd}
 			                    \rightsquigarrow\,
								   ^{\mbox{\large $\cdot$}}\!
								     (\varphi^{\sharp}(\scriptsizeboldy^{\tinyboldd}))}
                    \cdot									
	                  D_{\mu}\varphi^{\sharp}(y^i)D_{\nu}\varphi^{\sharp}(y^j)
          \left.\rule{0ex}{1em}\right )\\	
   &= &		
	   \Tr\!\left(\rule{0ex}{1em}\right.
              \sum_{\mu=1}^m\,	\sum_{\nu=1}^m
		       \ComboMinor(\varphi,\nabla;\Phi,g,B)_{\mu\nu} \\
     &&\hspace{3em}			
			   \cdot			   			
 			  \sum_{i,j}
			    \left(\rule{0ex}{1em}\right.	
			      \sum_{i^{\prime}=1}^n\,
	              \sum_{\;d,\,\scriptsizeboldd,\,\vec{\pi};\,
				                       |\scriptsizeboldd|=d,\, i_{(\vec{\pi},\tinyboldd)}=i^{\prime}}\,
	               R^{\,E_{ij}}[1]_{(\scriptsizeboldd,\,\vec{\pi})}^{\;L}(\varphi^{\sharp}(\boldy))
	                \cdot     \dot{\varphi}^{\sharp}(y^{i^{\prime}})
	                \cdot
	                R^{\,E_{ij}}[1]_{(\scriptsizeboldd,\,\vec{\pi})}^{\;R}(\varphi^{\sharp}(\boldy))
				   \left.\rule{0ex}{1em}\right)
                    \cdot									
	                  D_{\mu}\varphi^{\sharp}(y^i)D_{\nu}\varphi^{\sharp}(y^j)
          \left.\rule{0ex}{1em}\right )\\	
   &= &		
	   \Tr\!\left(\rule{0ex}{1em}\right.
	     \sum_{i^{\prime}=1}^n\,
            \left(\rule{0ex}{1em}\right.		
              \sum_{\mu=1}^m\,	\sum_{\nu=1}^m\,
			   \sum_{i,j}\,
			     \sum_{\;d,\,\scriptsizeboldd,\,\vec{\pi};\,
				                     |\scriptsizeboldd|=d,\, i_{(\vec{\pi},\tinyboldd)}=i^{\prime}}\,
                  R^{\,E_{ij}}[1]_{(\scriptsizeboldd,\,\vec{\pi})}^{\;R}(\varphi^{\sharp}(\boldy))
                    \cdot									
	                  D_{\mu}\varphi^{\sharp}(y^i)D_{\nu}\varphi^{\sharp}(y^j)    \\
     &&\hspace{14.7em}			
	           \cdot\,					
		       \ComboMinor(\varphi,\nabla;\Phi,g,B)_{\mu\nu}
			   \cdot			   			
	               R^{\,E_{ij}}[1]_{(\scriptsizeboldd,\,\vec{\pi})}^{\;L}(\varphi^{\sharp}(\boldy))
				 \left.\rule{0ex}{1em}\right)
	                \cdot     \dot{\varphi}^{\sharp}(y^{i^{\prime}})
          \left.\rule{0ex}{1em}\right )\\	
   &\;=: &
         \Tr \left(\rule{0ex}{1em}\right.
            \sum_{i^{\prime}=1}^n
			  \NL^{2.I,\,(\Phi,g,B)}_{i^{\prime}}(\varphi, \nabla)
			    \cdot    \dot{\varphi}^{\sharp}(y^{i^{\prime}})
			   \!\left.\rule{0ex}{1em}\right).
\end{eqnarray*}
}

\bigskip
\bigskip

\noindent
{\it Step $(4.2):$ The subcluster $\,${\rm (II)}}

\medskip

\noindent
This subcluster contributes also to boundary terms.

{\footnotesize
\begin{eqnarray*}
 \mbox{(II)}
   &:=\; &
        \Tr\!\left(\rule{0ex}{1em}\right.
              \sum_{\mu=1}^m\,	\sum_{\nu=1}^m
		       \ComboMinor(\varphi,\nabla;\Phi,g,B)_{\mu\nu}			   			  \\[-2.4ex]
   &&  \hspace{8em}
           \cdot
		    \left(\rule{0ex}{1em}\right.			 								  					  
	             \sum_{i,j}	
				 \varphi^{\sharp}(E_{ij})
				   \cdot
				     D_{\mu}\dot{\varphi}^{\sharp}(y^i)\,				          		
				   \cdot					
					  D_{\nu}\varphi^{\sharp}(y^j)\;
          	+\;
	             \sum_{i,j}	
				 \varphi^{\sharp}(E_{ij})
	               D_{\mu}\varphi^{\sharp}(y^i)					
					\cdot D_{\nu}\dot{\varphi}^{\sharp}(y^j)					
			  \left.\rule{0ex}{1em}\right)\!\!
          \left.\rule{0ex}{1em}\right )\\	
   &= &
        \Tr\!\left(\rule{0ex}{1em}\right.
              \sum_{\mu=1}^m\,	\sum_{\nu=1}^m
	             \sum_{i,j}	
				  \left(\rule{0ex}{1em}\right.
				 D_{\mu}\varphi^{\sharp}(y^i)
				 \cdot
                   \ComboMinor(\varphi,\nabla;\Phi,g,B)_{\nu\mu}\,				
				 \varphi^{\sharp}(E_{ji})   \\[-3ex]
   && \hspace{14em}				
          	+\;
				   \ComboMinor(\varphi,\nabla;\Phi,g,B)_{\mu\nu}
				 \varphi^{\sharp}(E_{ij})
	               D_{\mu}\varphi^{\sharp}(y^i)					
				  \left.\rule{0ex}{1em}\right)
					\cdot D_{\nu}\dot{\varphi}^{\sharp}(y^j)					
          \left.\rule{0ex}{1em}\right )\\	
  &= &
        \Tr\!\left(\rule{0ex}{1em}\right.
              \sum_{\nu=1}^m\,	\sum_{\mu=1}^m\, \sum_{i,j}
				   D_{\nu}\!
				    \left[\rule{0ex}{1em}\right.
					 \left(\rule{0ex}{1em}\right.
				 D_{\mu}\varphi^{\sharp}(y^i)\,
                   \ComboMinor(\varphi,\nabla;\Phi,g,B)_{\nu\mu}\,				
				 \varphi^{\sharp}(E_{ji})   \\[-3ex]
  && \hspace{14em}
            +\;
				   \ComboMinor(\varphi,\nabla;\Phi,g,B)_{\mu\nu}\,
				 \varphi^{\sharp}(E_{ij})\,
	               D_{\mu}\varphi^{\sharp}(y^i)					
				  \left.\rule{0ex}{1em}\right)
				  \dot{\varphi}^{\sharp}(y^j)
			  \left.\rule{0ex}{1em}\right]					
          \left.\rule{0ex}{1em}\right )\\	
  &&
       -\;
        \Tr\!\left(\rule{0ex}{1em}\right.
              \sum_{j=1}^n\,
			   \left(\rule{0ex}{1em}\right.
			   \sum_{\mu=1}^m\,	\sum_{\nu=1}^m\, \sum_{i=1}^n	
				 D_{\nu}\!
				  \left[\rule{0ex}{1em}\right.\!\!
				 D_{\mu}\varphi^{\sharp}(y^i)
				 \cdot
                   \ComboMinor(\varphi,\nabla;\Phi,g,B)_{\nu\mu}\,				
				 \varphi^{\sharp}(E_{ji})   \\[-3ex]
    && \hspace{15em}
          	+\;
				   \ComboMinor(\varphi,\nabla;\Phi,g,B)_{\mu\nu}
				 \varphi^{\sharp}(E_{ij})
	               D_{\mu}\varphi^{\sharp}(y^i)\!\!					
				  \left.\rule{0ex}{1em}\right]\!\!
				    \left.\rule{0ex}{1em}\right)
					\cdot \dot{\varphi}^{\sharp}(y^j)					
          \left.\rule{0ex}{1em}\right )                \\
  &= &
       \sum_{\nu=1}^m \partial_{\nu}
        \Tr\!\left(\rule{0ex}{1em}\right.
              \sum_{\mu=1}^m\, \sum_{i,j}
				    \left[\rule{0ex}{1em}\right.
					 \left(\rule{0ex}{1em}\right.
				 D_{\mu}\varphi^{\sharp}(y^i)\,
                   \ComboMinor(\varphi,\nabla;\Phi,g,B)_{\nu\mu}\,				
				 \varphi^{\sharp}(E_{ji})   \\[-3ex]
  && \hspace{14em}
            +\;
				   \ComboMinor(\varphi,\nabla;\Phi,g,B)_{\mu\nu}\,
				 \varphi^{\sharp}(E_{ij})\,
	               D_{\mu}\varphi^{\sharp}(y^i)					
				  \left.\rule{0ex}{1em}\right)
				  \dot{\varphi}^{\sharp}(y^j)
			  \left.\rule{0ex}{1em}\right]					
          \left.\rule{0ex}{1em}\right )\\	
  &&
       -\;
        \Tr\!\left(\rule{0ex}{1em}\right.
              \sum_{j=1}^n\,
			   \left(\rule{0ex}{1em}\right.
			   \sum_{\mu=1}^m\,	\sum_{\nu=1}^m\, \sum_{i=1}^n	
				 D_{\nu}\!
				  \left[\rule{0ex}{1em}\right.\!\!
				 D_{\mu}\varphi^{\sharp}(y^i)
				 \cdot
                   \ComboMinor(\varphi,\nabla;\Phi,g,B)_{\nu\mu}\,				
				 \varphi^{\sharp}(E_{ji})   \\[-3ex]
    && \hspace{15em}
          	+\;
				   \ComboMinor(\varphi,\nabla;\Phi,g,B)_{\mu\nu}
				 \varphi^{\sharp}(E_{ij})
	               D_{\mu}\varphi^{\sharp}(y^i)\!\!					
				  \left.\rule{0ex}{1em}\right]\!\!
				    \left.\rule{0ex}{1em}\right)
					\cdot \dot{\varphi}^{\sharp}(y^j)					
          \left.\rule{0ex}{1em}\right )                \\
    &\;=:&
		  \sum_{\nu=1}^m (-1)^{\nu-1}\, \partial_{\nu}
			   (\BT^{2.I\!I,\,(\varphi,\nabla;\Phi,g,B)}_{\nu}(\dot{\varphi}^{\sharp}(\boldy)))\;  		  
	      +\, \Tr \left(\rule{0ex}{1em}\right.
		          \sum_{j=1}^n
				   \NL^{2.I\!I,\,(\Phi,g,B) }_j (\varphi,\nabla)
				       \cdot \dot{\varphi}^{\sharp}(y^j)
				      \!\left.\rule{0ex}{1em}\right).
\end{eqnarray*}
}

\vspace{18em}

\noindent
{\it Step $(4.3):$ The subcluster $\,${\rm (III)}}

{\footnotesize
\begin{eqnarray*}
 \mbox{(III)}
   &:=\; &
        \Tr\!\left(\rule{0ex}{1em}\right.
          -\,    \sum_{\mu=1}^m\,	\sum_{\nu=1}^m
		                 \ComboMinor(\varphi,\nabla;\Phi,g,B)_{\mu\nu}			   			  \\
   &&  \hspace{6em}
           \cdot
		    \left(\rule{0ex}{1em}\right.								
	             \sum_{i,j}	
				 \varphi^{\sharp}(E_{ij})
				   \cdot [\varphi^{\sharp}(y^i), \dot{A}_{\mu}]
				   \cdot D_{\nu}\varphi^{\sharp}(y^j)\;
          	+\;
	             \sum_{i,j}	
				 \varphi^{\sharp}(E_{ij})\,
	               D_{\mu}\varphi^{\sharp}(y^i)					
					\cdot [\varphi^{\sharp}(y^j), \dot{A}_{\nu}]
			  \left.\rule{0ex}{1em}\right)\!\!
          \left.\rule{0ex}{1em}\right )\\
   &= &
        \Tr\!\left(\rule{0ex}{1em}\right.
            \sum_{\nu=1}^m\,	\sum_{\mu=1}^m\, \sum_{i,j}	   		
		    \left(\rule{0ex}{1em}\right.								
				\ComboMinor(\varphi,\nabla;\Phi,g,B)_{\mu\nu}\,		
				 \varphi^{\sharp}(E_{ij})\,
	               D_{\mu}\varphi^{\sharp}(y^i)\,					
					 \varphi^{\sharp}(y^j)       \\[-2.4ex]
   && \hspace{9em}
        -\;
			   \varphi^{\sharp}(y^j)\,
				\ComboMinor(\varphi,\nabla;\Phi,g,B)_{\mu\nu}\,		
				 \varphi^{\sharp}(E_{ij})\,
	               D_{\mu}\varphi^{\sharp}(y^i)\,	         \\
   && \hspace{9em}						
	          -\,
                 D_{\mu}\varphi^{\sharp}(y^j)\,			
			     \ComboMinor(\varphi,\nabla;\Phi,g,B)_{\nu\mu}\,	
				 \varphi^{\sharp}(E_{ij})\,
				   \varphi^{\sharp}(y^i)                 \\
    && \hspace{9em}			
          	+\;
			    \varphi^{\sharp}(y^i)\, D_{\mu}\varphi^{\sharp}(y^j)\,
			    \ComboMinor(\varphi,\nabla;\Phi,g,B)_{\nu\mu}\,		
				 \varphi^{\sharp}(E_{ij})
				   \left.\rule{0ex}{1em}\right)				
				   \cdot   \dot{A}_{\nu}					
          \left.\rule{0ex}{1em}\right )          \\
   &\;=: &
        \Tr \left(\rule{0ex}{1em}\right.
		  \sum_{\nu=1}^m
		    \NL^{2.I\!I\!I,\, (\Phi,g,B)}_{\nu}(\varphi,\nabla)
			  \cdot \dot{A}_{\nu}
		      \!\left.\rule{0ex}{1em}\right).		
\end{eqnarray*}
}

\bigskip
\bigskip

\noindent
{\it Step $(4.4):$ The subcluster $\,${\rm (IV)}}

\medskip

\noindent
This subcluster contributes also to boundary terms.

{\footnotesize
\begin{eqnarray*}
 \mbox{(IV)}
   &:=\; &
        \Tr\!\left(\rule{0ex}{1em}\right.\!
		   2\pi\alpha^{\prime}\,
              \sum_{\mu=1}^m\,	\sum_{\nu=1}^m
		       \ComboMinor(\varphi,\nabla;\Phi,g,B)_{\mu\nu}			   			
           \cdot (D_{\mu}\dot{A}_{\nu}-D_{\nu}\dot{A}_{\mu})\!
          \left.\rule{0ex}{1em}\right )\\	
   &= &		
        \Tr\!\left(\rule{0ex}{1em}\right.\!
		   2\pi\alpha^{\prime}\,
              \sum_{\mu=1}^m\,	\sum_{\nu=1}^m
		        \left(\rule{0ex}{1em}\right.\!\!
			       \ComboMinor(\varphi,\nabla;\Phi,g,B)_{\mu\nu}\,
             		-\, \ComboMinor(\varphi,\nabla;\Phi,g,B)_{\nu\mu}
					\!\!\left.\rule{0ex}{1em}\right)	
           \cdot D_{\mu}\dot{A}_{\nu}\!
          \left.\rule{0ex}{1em}\right )\\	
   &= & 	
        \Tr\!\left(\rule{0ex}{1em}\right.\!
		   2\pi\alpha^{\prime}\,
              \sum_{\mu=1}^m\,	\sum_{\nu=1}^m
		       D_{\mu}\left[
			    \left(\rule{0ex}{1em}\right.\!\!
			       \ComboMinor(\varphi,\nabla;\Phi,g,B)_{\mu\nu}\,
             		-\, \ComboMinor(\varphi,\nabla;\Phi,g,B)_{\nu\mu}
					 \!\!\left.\rule{0ex}{1em}\right)	
           \cdot \dot{A}_{\nu}
		                             \right]\!\!
          \left.\rule{0ex}{1em}\right )\\
   && 	
      -\;
        \Tr\!\left(\rule{0ex}{1em}\right.\!
		   2\pi\alpha^{\prime}\,
              \sum_{\nu=1}^m\,	
			   \left(\rule{0ex}{1em}\right.
			     \sum_{\mu=1}^m
		          D_{\mu}
			        \left[\rule{0ex}{1em}\right.\!\!
				        \ComboMinor(\varphi,\nabla;\Phi,g,B)_{\mu\nu}\,
             	  	       -\, \ComboMinor(\varphi,\nabla;\Phi,g,B)_{\nu\mu}
				    \!\!\left.\rule{0ex}{1em}\right]\!\!	
			  \left.\rule{0ex}{1em}\right)\!			
                  \cdot \dot{A}_{\nu}\!
          \left.\rule{0ex}{1em}\right )   \\
   &= &
        \sum_{\mu=1}^m
		 \partial_{\mu}
		  \Tr\!\left(\rule{0ex}{1em}\right.\!
		   2\pi\alpha^{\prime}\,
              \sum_{\nu=1}^m
			    \left(\rule{0ex}{1em}\right.\!\!
				   \ComboMinor(\varphi,\nabla;\Phi,g,B)_{\mu\nu}\,
             		-\, \ComboMinor(\varphi,\nabla;\Phi,g,B)_{\nu\mu}
			    \!\!\left.\rule{0ex}{1em}\right)	
           \cdot \dot{A}_{\nu}
          \left.\rule{0ex}{1em}\right )\\
   && 	
      -\;
        \Tr\!\left(\rule{0ex}{1em}\right.\!
		   2\pi\alpha^{\prime}\,
              \sum_{\nu=1}^m\,	
			   \left(\rule{0ex}{1em}\right.
			     \sum_{\mu=1}^m
		          D_{\mu}
			        \left[\rule{0ex}{1em}\right.\!\!
					  \ComboMinor(\varphi,\nabla;\Phi,g,B)_{\mu\nu}\,
             	  	    -\, \ComboMinor(\varphi,\nabla;\Phi,g,B)_{\nu\mu}
				    \!\!\left.\rule{0ex}{1em}\right]\!\!	
			  \left.\rule{0ex}{1em}\right)\!			
                  \cdot \dot{A}_{\nu}\!
          \left.\rule{0ex}{1em}\right )  \\
	&\;=: &
	    \sum_{\mu=1}^m (-1)^{\mu-1}\, \partial_{\mu}
		   (\BT^{2.I\!V,\,(\varphi,\nabla;\Phi,g,B)}_{\mu}(\dot{\boldA}))\;
	             +\; \Tr \left(\rule{0ex}{1em}\right.
				          \sum_{\nu=1}^m
						    \NL^{2.I\!V,\,(\Phi,g,B)}_{\nu}(\varphi,\nabla)
							   \cdot \dot{A}_{\nu}
						     \!\left.\rule{0ex}{1em}\right).		
\end{eqnarray*}
}

\vspace{8em}

\noindent
{\it Step $(5):$ The final formula}

\medskip

\noindent
In summary,
 with the notation introduced for the various nonlinear first-order and second-order differential expressions
 on $(\varphi,\nabla)$ that depend on $(\Phi,g,B)$ and appear in the calculation
 (subject to a relabelling of the dummy $i^{\prime}$ index),
one has

{\footnotesize
\begin{eqnarray*}
   \left.\mbox{\Large $\frac{d}{dt}$}\right|_{t=0}\,
   S_{\DBI}(\varphi_t,\nabla^t)
    & = & 	
	 -T_{m-1}\,\left.\mbox{\Large $\frac{d}{dt}$}\right|_{t=0}\,
	  \int_U\,  \Real  \left( \Tr \left(
          e^{-\varphi_t^{\sharp}(\Phi)}\,
		    \sqrt{-\SymDet(\boldM(t))  \,}\right )\right)
	   d^{\,m}\boldx   \\[1.2ex]
  && \hspace{-10em}
    =\;  	
	-T_{m-1}\,\int_U
	  \Real
       \!\left(\rule{0ex}{1em}\right.	
	    \sum_{\mu=1}^m (-1)^{\mu-1}\,
		  \partial_{\mu}
		   \left(\rule{0ex}{1em}\right.
		     \BT^{2.I\!I,\,(\varphi,\nabla;\Phi,g,B)}_{\mu}(\dot{\varphi}^{\sharp}(\boldy))\,
			 +\, \BT^{2.I\!V,\,(\varphi,\nabla;\Phi,g,B)}_{\mu}(\dot{\boldA})
		   \left.\rule{0ex}{1em}\right)		
	   \!\!\!\left.\rule{0ex}{1em}\right)	
	          d^{\,m}\boldx                 \\
  && \hspace{-8em}
    -T_{m-1}\,\int_U
       \Real	
	    \!\left(\rule{0ex}{1em}\right.
	    \!\!\Tr
	    \!\left(\rule{0ex}{1em}\right.
	     \!\sum_{j=1}^n
		  ( \NL^{1,\,(\Phi,g,B)}_j(\varphi,\nabla)\,
		     +\, \NL^{2.I,\,(\Phi,g,B)}_j(\varphi,\nabla)\,
		     +\, \NL^{2.I\!I,\,(\Phi,g,B)}_j(\varphi,\nabla)   )\, \cdot\,
		      \dot{\varphi}^{\sharp}(y^j)  \\[-2ex]
   && \hspace{7em}			
		 +\;
		 \sum_{\nu=1}^m
          ( \NL^{2.I\!I\!I,\,(\Phi,g,B)}_{\nu}(\varphi,\nabla)
		        +\, \NL^{2.I\!V,\,(\Phi,g,B)}_{\nu}(\varphi,\nabla) )\, \cdot\,
		       \dot{A}_{\nu}
		\!\left.\rule{0ex}{1em}\right)		
	    \!\!\!\left.\rule{0ex}{1em}\right)	
			d^{\,m}\boldx    \\
  && \hspace{-10em}
    =:\;  	
	-T_{m-1}\,\int_{\partial U}
	  \Real(\BT^{(\varphi,\nabla;\Phi,g,B)}(\dot{\varphi}^{\sharp}(\boldy),\dot{\boldA}))\\
   && \hspace{-4em}
    -T_{m-1}\,\int_U
       \Real	
	    \!\left(\rule{0ex}{1em}\right.
	    \!\!\Tr
	    \!\left(\rule{0ex}{1em}\right.
	     \!\sum_{j=1}^n
		  \NL^{(\Phi,g,B);\delta\varphi}_j(\varphi,\nabla)\,
		     \cdot\,  \dot{\varphi}^{\sharp}(y^j)
		 +\;
		 \sum_{\nu=1}^m
           \NL^{(\Phi,g,B);\delta\nabla}_{\nu}(\varphi,\nabla)\,
		     \cdot\, \dot{A}_{\nu}
		\!\left.\rule{0ex}{1em}\right)		
	    \!\!\!\left.\rule{0ex}{1em}\right)	
			d^{\,m}\boldx\,.
\end{eqnarray*}
}
  
\noindent
Here,
{\footnotesize
 \begin{eqnarray*}
   \BT^{(\varphi,\nabla;\Phi,g,B)}(\dot{\varphi}^{\sharp}(\boldy),\dot{\boldA})\\
   && \hspace{-10em}
   :=\;
	    \sum_{\mu=1}^m
		   \left(\rule{0ex}{1em}\right.\!\!
		     \BT^{2.I\!I,\,(\varphi,\nabla;\Phi,g,B)}_{\mu}(\dot{\varphi}^{\sharp}(\boldy))\,
			 +\, \BT^{2.I\!V,\,(\varphi,\nabla;\Phi,g,B)}_{\mu}(\dot{\boldA})
		   \!\!\left.\rule{0ex}{1em}\right)		
	          dx^1\wedge\,\cdots\,\wedge dx^{\mu-1}\wedge \widehat{dx^{\mu}}
			  \wedge dx^{\mu+1}\,\cdots\,\wedge dx^m\,,
 \end{eqnarray*}}
 with the $\widehat{dx^{\mu}}$ meaning the removal of $dx^{\mu}$,
 is a complex-valued $(m-1)$-form on $U$
 that depends linearly on $(\dot{\boldy},\dot{\boldA})$
 and whose real part gives the total boundary term (up to the factor $-T_{m-1}$)
  of the first variation of $S_{\DBI}^{(\Phi,g,B)}(\varphi,\nabla)$ with respect to $(\varphi,\nabla)$.

\bigskip

\subsection{The equations of motion for D-branes}

\begin{remark} $[\,$effect of $\Real(\,\cdot\,)$ in action to equations of motion$\,]\;$ {\rm
 Due to the operation `\hspace{.1ex}{\sl Taking the real part of}\hspace{.6ex}' $\Real(\,\cdot\,)$,
  to go from the the first variation formula
     to the expression for the equations of motion
  there is a detail that depends on how the space of pairs $(\varphi,\nabla)$
  and its tangents $(\delta\varphi, \delta\nabla)$ are parameterized;
 (cf.\      $\Real(e^{\sqrt{-1}\theta}z)
                     =\cos\theta\cdot \Real(z) - \sin\theta\cdot \Imaginary(z)$).
 \begin{itemize}
  \item[(1)]
   For the $\varphi$-part,
     first,  caution that it is {\it not} that
       just because $\varphi^{\sharp}(y^i)$, $i=1, \,\ldots\,, n$,
          take values in a ring over ${\Bbb C}$
       (i.e.\ $C^{\infty}(\End_{\Bbb C}(E))$)
       that the space $\Map((X^{\!A\!z},E),Y)$ of all such $\varphi$'s becomes a complex space.
   Indeed,
      due to the fact that
     all the eigenvalues of $\varphi^{\sharp}(f)$, $f\in C^{\infty}(Y)$ are real
     (cf.\ [L-Y4: Sec.\ 3], D(11.1)),
   $\Map((X^{\!A\!z},E),Y)$ is intrinsically a real space and
      there is no natural complex-space structure on it (even if exists)
      that can be made compatible with the underlying moduli problem
    since
      if $\delta\varphi$ is an unobstructed tangent to $\Map((X^{\!A\!z},E),Y)$, then
      $\sqrt{-1}\delta\varphi$ can never be an unobstructed tangent to $\Map((X^{\!A\!z}),Y)$.	
   So this part is good in the sense that if we fix a real presentation for $\varphi$'s in the study,
    then $\Real(\delta S_{\DBI}/\delta\varphi)$ gives the system of equations of motion for $\varphi$.
	
  \item[(2)]
   For the $\nabla$-part,
     if alone, the parameter space is complex in nature in our most general setting.
   When $E$ is Hermitian and $\nabla$ is required to be compatible with the Hermitian structure, 
     the resulting parameter space becomes intrinsically real.
   In the latter case, depending on the convention in presenting a unitary gauge theory
    (mathematicians vs.\ physicists),
    one may take either
     $\Real(\delta S_{\DBI}^{(\Phi,g,B)}/\delta\nabla)$ or
     $\Imaginary(\delta S_{\DBI}^{(\Phi,g,B)})/\delta\nabla$  as the system of equations for $\nabla$.
   {\it However}, this is not the full story as we imposed the admissible condition
     $\nabla_{\mbox{\tiny $\bullet$}}{\cal A}_{\varphi}\subset {\cal A}_{\varphi}$
	 on $\nabla$.
   Details on writing the equations of motion will have to depend on how we present this condition.
 \end{itemize}
 Not to let this additional detail to distract us in this first work in the D(13) subseries,
  we present for the current notes the system of equations of motion that remove the effect of
  $\Real(\,\cdot\,)$ in $S_{\DBI}^{(\Phi,g,B)}$.
 In other words, a true system of equations of motion will involve only a combination of what are given below.
}\end{remark}

\bigskip

It follows from the study in Sec.\ 5.2 that
the equations of motion for D-branes from the Dirac-Born-Infeld action,
 with the D-brane world-volume modelled in the current context as an admissible map
  $$
    \varphi:(X^{A\!z},E;\nabla)\; \longrightarrow\;
     (Y,\Phi,g,B)
  $$
  from an Azumaya/matrix manifold with a fundamental module with a connection
  $(X^{A\!z},E;\nabla)$
  to a space-time $Y$ with massless background fields $(\Phi,g,B)$ from closed string excitations,  
 are given by the following system of second-order nonlinear partial differential equations on $(\varphi,\nabla)$:
 %
 %
 $$
  \left\{
   \!\!\begin{array}{lll}
    \NL^{(\Phi,g,B);\delta\varphi}_j (\varphi,\nabla)            & \!\!\!=\; 0\,,
  	  & \mbox{for $j=1,\,\ldots\,,\, n$}\,;                  \\[1.2ex]
    \NL^{(\Phi,g,B);\delta\nabla}_{\nu} (\varphi,\nabla) & \!\!\!=\; 0\,,
	 & \mbox{for $\nu=1,\,\ldots\,,\, m$}\,.
   \end{array}
  \right.
 $$
Here, for the first subsystem,
 $$
   \NL^{(\Phi,g,B);\delta\varphi}_j (\varphi,\nabla)\;
    =\;  \NL^{1,\,(\Phi,g,B)}_j (\varphi,\nabla)\,
		     +\, \NL^{2.I,\,(\Phi,g,B)}_j (\varphi,\nabla)
		     +\, \NL^{2.I\!I,\,(\Phi,g,B)}_j (\varphi,\nabla)
 $$
 with
{\footnotesize
 \begin{eqnarray*}
  \NL^{1,\,(\Phi,g,B)}_j (\varphi,\nabla)
    &= &
	     \sum_{d,\,\scriptsizeboldd,\,\vec{\pi};\,
		                       |\scriptsizeboldd|=d,\, i_{(\vec{\pi}, \tinyboldd)}=j }\,
	      R^{\,e^{-\Phi}}[1]_{(\scriptsizeboldd,\,\vec{\pi})}^{\;R}(\varphi^{\sharp}(\boldy))
          \cdot
		    \sqrt{-\SymDet(\boldM(0))  \,}
		  \cdot	
	    R^{\,e^{-\Phi}}[1]_{(\scriptsizeboldd,\,\vec{\pi})}^{\;L}(\varphi^{\sharp}(\boldy))\,,\\
  \NL^{2.I,\,(\Phi,g,B)}_j (\varphi,\nabla)
    &= &
         \sum_{\mu=1}^m\,	\sum_{\nu=1}^m\,
			   \sum_{i^{\prime},j^{\prime}}\,
			     \sum_{d,\,\scriptsizeboldd,\,\vec{\pi};\,
				                         |\scriptsizeboldd|=d,\, i_{(\vec{\pi},\tinyboldd)}=j}\,
                  R^{\,E_{i^{\prime}j^{\prime}}}[1]_{(\scriptsizeboldd,\,\vec{\pi})}^{\;R}
				      (\varphi^{\sharp}(\boldy))
                    \cdot									
	                  D_{\mu}\varphi^{\sharp}(y^i)D_{\nu}\varphi^{\sharp}(y^j)    \\
     &&\hspace{12em}			
	           \cdot\,					
		       \ComboMinor(\varphi,\nabla;\Phi,g,B)_{\mu\nu}
			   \cdot			   			
	               R^{\,E_{i^{\prime}j^{\prime}}}[1]_{(\scriptsizeboldd,\,\vec{\pi})}^{\;L}
				   (\varphi^{\sharp}(\boldy))  \,,            \\	            		     		
  \NL^{2.I\!I,\,(\Phi,g,B)}_j (\varphi,\nabla)
    &= &
	    -\,\sum_{\mu=1}^m\,	\sum_{\nu=1}^m\, \sum_{i=1}^n	
				 D_{\nu}\!
				  \left(\rule{0ex}{1em}\right.\!\!
				 D_{\mu}\varphi^{\sharp}(y^i)
				 \cdot
                   \ComboMinor(\varphi,\nabla;\Phi,g,B)_{\nu\mu}\,				
				 \varphi^{\sharp}(E_{ji})   \\[-3ex]
    && \hspace{8em}
          	+\;
				   \ComboMinor(\varphi,\nabla;\Phi,g,B)_{\mu\nu}
				 \varphi^{\sharp}(E_{ij})
	               D_{\mu}\varphi^{\sharp}(y^i)\!\!					
				  \left.\rule{0ex}{1em}\right);
 \end{eqnarray*}
}
and, for the second subsystem,
$$
  \NL^{(\Phi,g,B);\delta\nabla}_{\nu} (\varphi,\nabla) \;
   =\;   \NL^{2.I\!I\!I,\,(\Phi,g,B)}_{\nu} (\varphi,\nabla)
		        +\, \NL^{2.I\!V,\,(\Phi,g,B)}_{\nu} (\varphi,\nabla)
$$
with
{\footnotesize
 \begin{eqnarray*}
    \NL^{2.I\!I\!I,\,(\Phi,g,B)}_{\nu} (\varphi,\nabla)
  &= &
        \sum_{\mu=1}^m\, \sum_{i,j}	   		
		    \left(\rule{0ex}{1em}\right.\!\!								
				\ComboMinor(\varphi,\nabla;\Phi,g,B)_{\mu\nu}\,		
				 \varphi^{\sharp}(E_{ij})\,
	               D_{\mu}\varphi^{\sharp}(y^i)\,					
					 \varphi^{\sharp}(y^j)       \\[-2.4ex]
   && \hspace{4.2em}
        -\;
			   \varphi^{\sharp}(y^j)\,
				\ComboMinor(\varphi,\nabla;\Phi,g,B)_{\mu\nu}\,		
				 \varphi^{\sharp}(E_{ij})\,
	               D_{\mu}\varphi^{\sharp}(y^i)\,	         \\
   && \hspace{4.2em}						
	          -\,
                 D_{\mu}\varphi^{\sharp}(y^j)\,			
			     \ComboMinor(\varphi,\nabla;\Phi,g,B)_{\nu\mu}\,	
				 \varphi^{\sharp}(E_{ij})\,
				   \varphi^{\sharp}(y^i)                 \\
    && \hspace{4.2em}			
          	+\;
			    \varphi^{\sharp}(y^i)\, D_{\mu}\varphi^{\sharp}(y^j)\,
			    \ComboMinor(\varphi,\nabla;\Phi,g,B)_{\nu\mu}\,		
				 \varphi^{\sharp}(E_{ij})
				   \!\!\left.\rule{0ex}{1em}\right),				  \\
  \NL^{2.I\!V,\,(\Phi,g,B)}_{\nu} (\varphi,\nabla)\;
   & =  &
    2\pi\alpha^{\prime}\,
	\sum_{\mu=1}^m
	  D_{\mu}
	     (\ComboMinor(\varphi,\nabla;\Phi,g,B)_{\nu\mu}\,
             	  	    -\, \ComboMinor(\varphi,\nabla;\Phi,g,B)_{\mu\nu})\,.
 \end{eqnarray*}}
 
In both subsystems,
{\footnotesize
 \begin{eqnarray*}
  \ComboMinor(\varphi,\nabla;\Phi,g,B)_{\mu\nu}
   &= &
	   \mbox{\large $\frac{1}{m!}$}\,
              \sum_{\mu^{\prime}=1}^m
              \sum_{\mbox{\tiny
			                   $\begin{array}{c} \sigma\in\tinySym_m\\
							                                              \sigma(\mu^{\prime})=\mu \end{array}$}}
             (-1)^{\sigma}\,(-1)^{\mu^{\prime}(m-\mu^{\prime})+m+\nu}\,
              \Minor(\varphi,\nabla;\Phi,g,B\,|\,\mu^{\prime},\sigma)_{\mu\nu}\,,
 \end{eqnarray*}
}
where
{\footnotesize
 \begin{eqnarray*}
   \lefteqn{
    \Minor(\varphi,\nabla;\Phi,g,B\,|\,\mu^{\prime},\sigma)_{\mu\nu}\;
       =\; \mbox{the $(m,\nu)$-minor of}     } \\
    &&
      [ \boldM_{(\sigma(\mu^{\prime}+1))}(0)^{\transpose},\,
				     \cdots\,,\, \boldM_{(\sigma(m))}(0)^{\transpose},\,				
				  F_2(\varphi,\nabla;\Phi,g,B)\,\boldM_{(\sigma{(1)})}(0)^{\transpose},\,
				    \cdots\,,\, \boldM_{(\sigma(\mu^{\prime}-1))}(0)^{\transpose},\,
			      \dot{\boldM}_{(\mu)}(0)^{\transpose}]^{\transpose}
 \end{eqnarray*}}
with
{\footnotesize
 \begin{eqnarray*}
   F_2(\varphi,\nabla;\Phi,g,B)
     & = &
	  -\,\mbox{\large $\frac{1}{2}$}\,
       e^{-\varphi^{\sharp}(\Phi)}\,\sqrt{-\SymDet(\boldM(0))  \,}^{\,-1}\,,  \\	
 {\boldM}_{(\,\mbox{\tiny $\bullet$}\,)}(0)
     &=  &
	 \mbox{the {\tiny $\bullet\,$}-th row vector of $\boldM(0)$}\,,     \\
   \boldM_{\mu\nu}(0)
	 &= & \mbox{the $(\mu,\nu)$-entry of $\boldM(0)$}\;\; =\;\;
   	 \sum_{i^{\prime},j^{\prime}}
				 \varphi^{\sharp}(E_{i^{\prime}j^{\prime}})\,
	                  D_{\mu}\varphi^{\sharp}(y^{i^{\prime}})\,
					  D_{\nu}\varphi^{\sharp}(y^{j^{\prime}})\,
                   +\, 2\pi\alpha^{\prime}\,[\nabla_{\mu},\nabla_{\nu}]\,,   \\
  \dot{\boldM}_{\mu\nu}(0)
    & = &   \sum_{i^{\prime},j^{\prime}}	
	             R^{\,E_{i^{\prime}j^{\prime}}}[1]|
			              _{\scriptsizeboldy^{\tinyboldd}
 			                    \rightsquigarrow\,
								   ^{\mbox{\large $\cdot$}}\!
								     (\varphi^{\sharp}(\scriptsizeboldy^{\tinyboldd}))}
                    \cdot									
	                  D_{\mu}\varphi^{\sharp}(y^{i^{\prime}})
					  D_{\nu}\varphi^{\sharp}(y^{j^{\prime}}) \\
    &&      +\; 	
	             \sum_{i^{\prime},j^{\prime}}	
				 \varphi^{\sharp}(E_{i^{\prime}j^{\prime}})
				   \cdot(
				     D_{\mu}\dot{\varphi}^{\sharp}(y^{i^{\prime}})\,
				            -\, [\varphi^{\sharp}(y^{i^{\prime}}), \dot{A}_{\mu}])		
				   \cdot					
					  D_{\nu}\varphi^{\sharp}(y^{j^{\prime}})      \\
    &&      +\;
	             \sum_{i^{\prime},j^{\prime}}	
				 \varphi^{\sharp}(E_{i^{\prime}j^{\prime}})
	               D_{\mu}\varphi^{\sharp}(y^{i^{\prime}})					
					\cdot (D_{\nu}\dot{\varphi}^{\sharp}(y^{j^{\prime}})\,
                                           -\, [\varphi^{\sharp}(y^{j^{\prime}}), \dot{A}_{\nu}])\;
               +\, 2\pi\alpha^{\prime}\,
			        (D_{\mu}\dot{A}_{\nu}-D_{\nu}\dot{A}_{\mu})\,.	
 \end{eqnarray*}
 }

\bigskip

\begin{remark}  $[\,$origin/correction from anomaly equations for open strings$\,]\;$  {\rm
 From the string-theory point of view,
  it is very important to understand further
  how such systems of differential equations on the pair $(\varphi,\nabla)$
  can arise from or be correced/improved by the anomaly-free conditions in open-string theory.
 Cf.\ Issue (7), Sec.\ 1.
 %
}\end{remark}

\medskip

\begin{remark} {\rm [\hspace{.1ex}}the case of Hermitian/unitary D-branes{\rm\hspace{.1ex}]}$\;$
{\rm
 When in addition $E$ is equipped with a Hermitian structure
    and $\varphi$ is Hermitian and $\nabla$ is unitary,
  the Dirac-Born-Infeld action functional
   $S^{(\Phi,g,B)}_{(\varphi,\nabla)}$ and, hence, the resulting equations of motion
   can be simplified.
 The detail should be studied further.
 Cf.\  Remark~2.3.8 and Remark~3.2.5.
}\end{remark}

\bigskip

\section{Remarks on the Chern-Simons/Wess-Zumino term}

In view of Polchinski's realization ([Po1]) that
  a D-brane world-volume can couple to a Ramond-Ramond field in superstring theory
  (cf.\ {\sc Figure}~6-0-1),
 the Chern-Simons/Wess-Zumino term $S_{\CSWZ}$ for D-branes is also an indispensable part
 to understand the dynamics of D-branes.
 %
 %
 \begin{figure} [htbp]
  \bigskip
  \centering

  \includegraphics[width=0.8\textwidth]{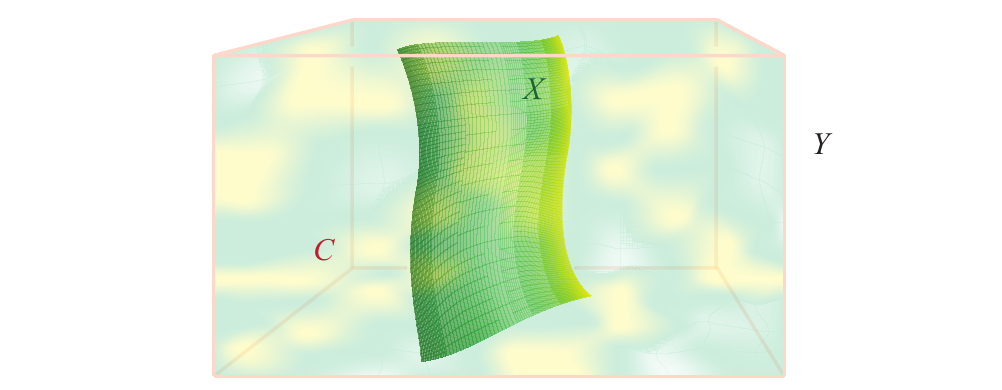}

  \bigskip
  \bigskip
  \centerline{\parbox{13cm}{\small\baselineskip 12pt
   {\sc Figure}~6-0-1.
    In superstring theory, a D-brane world-volume $X$ can couple to a Ramond-Ramond field $C$,
	  created by closed superstrings, on the space-time $Y$.
    Such coupling influences the dynamics of the D-brane as well. 	
    In the {\sc Figure},
 	 the Ramond-Ramond field is indicated by an etherlike foggy background with varying density.
       }}
  \bigskip
 \end{figure}	
 With the same essence as for the construction of $S^{(\Phi, g, B)}_{\DBI}(\varphi,\nabla)$,
 we
 construct in this section
 the Chern-Simons/Wess-Zumino action $S_{\CSWZ}(\varphi,\nabla)$
    for lower-dimensional D-branes, in which cases anomaly issues do not occur,
 derive their first variation formula and, hence,
 obtain their contribution to the equations of motions for D-branes.

To begin, with anomalies taken into account,
the coupling of a simple embedded D-brane
 $$
   f\; :\;  X\;  \hookrightarrow\;  Y
 $$
 with the Ramond-Ramond field $C$ on $Y$ (with a $B$-field background $B$),  is encoded in
  the Chern-Simons/Wess-Zumino action for D-branes,
  which takes the form
 $$
   S^{(C,B)}_{\CSWZ}(f,\nabla)\;  =\;
     T_{m-1}\,\int_X  \left(
	    f^{\ast}C
	    \wedge e^{2\pi\alpha^{\prime}F_{\nabla}+ f^{\ast}B}
	    \wedge \sqrt{\hat{A}(X)/\hat{A}(N_{X/Y})\,}\right)_{(m)}\,,
 $$
  where
   \begin{itemize}
    \item[\LARGE $\cdot$]
	 $m=\dimm X$, $\,T_{m-1}$ the D$(m-1)$-brane tension,
	 $\hat{A}(\,\cdot\,)$ the $\hat{A}$-class of the bundle in question,\\
     $N_{X/Y}$ the normal bundle of $X$ in $Y$ along $f$,
	
	\item[\LARGE $\cdot$]
     $(\,\cdots\,)_{(m)}$ is the degree-$m$ component of a differential form $(\,\cdots\,)$ on $X$.
   \end{itemize}	
The fact that the over coupling strength is identical with the D-brane tension $T_{m-1}$
 is a consequence of supersymmetry.
Readers are referred to, e.g.\ [Bac], [Joh], [Po3:vol.\ II], [Sz] for more details and references.

With the lesson already learned from studying the Dirac-Born-Infeld action,
formally the Chern-Simons/Wess-Zumino action generalizes to the case of coincident D-brane in our setting
 $$
   \varphi:(X^{\!A\!z},E;\nabla)\; \longrightarrow\; Y\,,
 $$	
 as
 {\small
 $$
  S_{\CSWZ}^{(C,B)}(f,\nabla)\;
   \stackrel{\mbox{\tiny{\rm formally}}}{=}\;
     T_{m-1}\,\int_X  \Real\left( \Tr
    	 \left(\varphi^{\diamond}C
		            \wedge e^{2\pi\alpha^{\prime}F_{\nabla}+\varphi^{\diamond}B}
	                \wedge  \sqrt{\hat{A}(X^{\!A\!z})/\hat{A}(N_{X^{\!A\!z}/Y})\,}
					\right) \right)_{(m)}\,.
 $$}
One now has to resolve in addition the following issues:
  \begin{itemize}
   \item[(8)]
    the anomaly factor
	 ``$\,\sqrt{\hat{A}(X^{\!A\!z})/\hat{A}(N_{X^{\!A\!z}/Y})}\,$",
     which presumably is an $\End_{\Bbb C}(E)$-valued differential form on $X$;	
	
   \item[(9)]	
    wedging of of $\End_{\Bbb C}(E)$-valued differential forms on $X\,$:\\
	$\varphi^{\diamond}C\wedge e^{2\pi\alpha^{\prime}F_{\nabla}+\varphi^{\diamond}B}
	  \wedge \sqrt{\hat{A}(X)/\hat{A}(N_{X^{\!A\!z}/Y})\,}$.	
  \end{itemize}

\bigskip

\subsection{Resolution of issues in the Chern-Simons/Wess-Zumino term}

We address in this subsection the resolution of Issue (9) in a way
 that is compatible with how we treat/interpret the Dirac-Born-Infled action in Sec.\ 3.
This gives us a version of the Chern-Simons/Wess-Zumino term $S_{\CSWZ}^{(C,B)}$
 for D-branes of dimension $-1,\, 0,\, 1$, and $2$
 that matches the Dirac-Born-Infeld action $S_{\DBI}^{(\Phi,g,B)}$ constructed in Sec.\ 3.

\vspace{6em}
	
\begin{flushleft}
{\bf From determinant function to wedge product of differential forms}
\end{flushleft}
For an ordinary differentiable manifold $M$,
  the wedge product of differential forms is determined by the wedge product
 of a collection of $1$-forms and the latter is set by the determinat function through the following rule
 $$
   (\omega^1\wedge\,\cdots\,\wedge\omega^s)(e_1\wedge\,\cdots\,\wedge e_s)\;=\;
     \Det(\omega^i(e_j))\,.
 $$
Here,
  $e_1,\,\cdots\,,\, e_s$ are vector fields on $M$,
  $\omega^1,\,\cdots\,,\, \omega^s$ are $1$-forms on $M$,
  $e_1\wedge\,\cdots\,\wedge e_s
      := \sum_{\sigma\in\scriptsizeSym_s}(-1)^{\sigma}
	        e_{\sigma(1)}\otimes\,\cdots\,\otimes e_{\sigma(s)}$, and
  $(\omega^i(e_j))$ is the $s\times s$ matrix with the $(i,j)$-entry $\omega^i(e_j)$.						
When $\omega^1,\,\cdots\,,\,\omega^s$ are enhanced to $1$-forms with value in a noncommutative ring $R$,
 the original determinant function $\Det(\,\cdot\,)$	needs to be enhanced/generalized as well
 to a determinant function for matrices with entries in $R$
  since now $\omega^j(e_i)\in R$, for $i,j=1,\,\ldots\,,\, s$.

Recall that in the study of non-Abelian Dirac-Born-Infeld action for the pair $(\varphi,\nabla)$,
 we ran into the need for such a generalization, too, and
 introduced the notion of symmetrized determinant $\SymDet$; cf.\ Definition~3.1.3.6.
There, we propose an Ansatz that this is the determinant function
  for the construction of the non-Abelian Dirac-Born-Infeld action,
  cf.\ Ansatz 3.1.3.11.
It is very natural to suggest that the same notion of determinant function is applied to
 both the Dirac-Born-Infeld term and the Chern/Simons/Wess-Zumino term  in the full action for D-branes:

\bigskip

\begin{ansatz}  {\bf [wedge product in the Chern-Simons/Wess-Zumino action]}$\;$ {\rm
 We interpret the wedge products that appear in the formal expresion for the Chern-Simons/Wess-Zumino term
  $S_{\CSWZ}^{(C,B)}$
  through the symmetrized determinant that applies to the above defining identities for wedge product;
  namely, we require that
  $$
   (\omega^1\wedge\,\cdots\,\wedge\omega^s)(e_1\wedge\,\cdots\,\wedge e_s)\;=\;
     \SymDet(\omega^i(e_j))
 $$
  for $\End_{\Bbb C}(E)$-valued $1$-forms $\omega^1,\,\cdots\,,\, \omega^s$ on $X$.
 Denote this generalized wedge product by $\odotwedge$.
}\end{ansatz}
				
\medskip

\begin{example} {\bf [$C_{(1)}\odotwedge F\odotwedge F$]}$\;$ {\rm
 Let
   $C_{(1)}=\sum_{\mu} C_{\mu}dx^{\mu}$ and
   $F=\sum_{\mu^{\prime},\nu^{\prime}}F_{\mu^{\prime}\nu^{\prime}}
      dx^{\mu^{\prime}}\wedge dx^{\nu^{\prime}}$
   be an $\End_{\Bbb C}(E)$-valued $1$-form and $2$-form respectively,
  then
  $$
    C_{(1)}\odotwedge F \odotwedge F\;=\;
     \sum_{\mu,\, \mu^{\prime},\, \nu^{\prime},\, \mu^{\prime\prime},\, \nu^{\prime\prime}}
      (C_{\mu}\odot F_{\mu^{\prime}\nu^{\prime}}
	       \odot F_{\mu^{\prime\prime}\nu^{\prime\prime}})\,
		dx^{\mu}\wedge dx^{\mu^{\prime}}\wedge dx^{\nu^{\prime}}
		\wedge dx^{\mu^{\prime\prime}}\wedge dx^{\nu^{\prime\prime}}\,,		
  $$
  where, recall that,
   $C_{\mu}\odot F_{\mu^{\prime}\nu^{\prime}}
      \odot F_{\mu^{\prime\prime}\nu^{\prime\prime}}$ is the symmetrized product of the triple
  $(C_{\mu}, F_{\mu^{\prime}\nu^{\prime}},
     F_{\mu^{\prime\prime}\nu^{\prime\prime}})$.	
}\end{example}	
 
\medskip

\begin{remark} {\it $[\,$on the ring
      $(C^{\infty}(\bigwedge^{\bullet}T^{\ast}X\otimes_{\Bbb R}\End_{\Bbb C}(E)),
           +, \odotwedge)\,]$}$\;$ {\rm (Cf.\ Remark 3.1.3.10.)
 Properties of $\odotwedge$ follow from
   properties of $\odot$ on $C^{\infty}(\End_{\Bbb C}(E))$ and
   properties of $\wedge$ on $C^{\infty}(\bigwedge^{\bullet}T^{\ast}X)$ .
 In particular, for example,
  $C_{(1)}\odotwedge F \odotwedge F$
    is directly defined for the triple $(C_{(1)}, F, F)$
	of $\End_{\Bbb C}(E)$-valued differential forms on $X$,
  rather than through a train of applications of a binary operation.
The three elements in $\bigwedge^5T^{\ast}X\otimes_{\Bbb R} \End_{\Bbb C}(E)$
  $$
      C_{(1)}\odotwedge F \odotwedge F\,, \hspace{2em}
	  (C_{(1)}\odotwedge F) \odotwedge F\,, \hspace{2em}
	  C_{(1)}\odotwedge (F \odotwedge F)
  $$
  in general  are all different.
The ring
  $(C^{\infty}(\bigwedge^{\bullet}T^{\ast}X\otimes_{\Bbb R}\End_{\Bbb C}(E)),
      +, \odotwedge)$
 is ${\Bbb Z}_2$-graded, ${\Bbb Z}_2$-commutative, but not associative.			
}\end{remark}

\medskip

\begin{lemma} {\bf [$\,\varphi^{\diamond}$, $\wedge$, and $\odotwedge\,$]}
 Let
   $\varphi:(X^{\!A\!z},E,\nabla)\rightarrow Y$ be an admissible map  and
   $\zeta_1, \,\cdots\,,\zeta_k$ differential forms on $Y$.
 Then
   $$
      \varphi^{\diamond}\zeta_1\odotwedge\,\cdots\,\odotwedge \varphi^{\diamond}\zeta_k\;
	  =\; \varphi^{\diamond}(\zeta_1\wedge\,\cdots\,\wedge\zeta_k)\,.
   $$
\end{lemma}

\medskip

\begin{proof}
 Recall the surrogate $X_{\varphi}$ of $X^{\!A\!z}$ specified by $\varphi$
 and the built-in maps
 $$
   \xymatrix{
    X_{\varphi}\ar[rr]^-{f_{\varphi}}   \ar[d]^-{\pi_{\varphi}}
	   && Y\\
    X   &&\hspace{1.6em}.
   }
 $$
 Since
   the function-ring $A_{\varphi}:= C^{\infty}(X)\langle\Image\varphi^{\sharp}\rangle$
     of $X_{\varphi}$ is commutative,
   for differential forms $\zeta_1^{\prime},\,\cdots\,,\, \zeta_k^{\prime}$ on $X_{\varphi}$,
   $$
      {\pi_{\varphi}}_{\ast}\zeta_1^{\prime}
         \odotwedge\, \cdots\, \odotwedge
		   {\pi_{\varphi}}_{\ast}\zeta_k^{\prime} \;
		 =\;     {\pi_{\varphi}}_{\ast}\zeta_1^{\prime}
                         \wedge\, \cdots\, \wedge  {\pi_{\varphi}}_{\ast}\zeta_k^{\prime} \;
		 =\; {\pi_{\varphi}}_{\ast} (\zeta_1^{\prime}\wedge\,\cdots\,\wedge \zeta_k^{\prime})\,.
   $$
 It follows that
  \begin{eqnarray*}
   \lefteqn{
       \varphi^{\diamond}\zeta_1\odotwedge\,\cdots\,\odotwedge \varphi^{\diamond}\zeta_k\;\;
	 =\;\;  {\pi_{\varphi}}_{\ast}(f_{\varphi}^{\ast}\zeta_1)
	            \odotwedge\,\cdots\, \odotwedge
				{\pi_{\varphi}}_{\ast}(f_{\varphi}^{\ast}\zeta_k)     }\\[.6ex]
	&& =\;\;
	  {\pi_{\varphi}}_{\ast}
	     (f_{\varphi}^{\ast}\zeta_1\wedge\,\cdots\,\wedge f_{\varphi}^{\ast}\zeta_k)\;\;
	      =\;\;  {\pi_{\varphi}}_{\ast}
		                (f_{\varphi}^{\ast}(\zeta_1\wedge\,\cdots\,\wedge \zeta_k))\;\;
		  =\;\; \varphi^{\diamond}(\zeta_1\wedge\,\cdots\,\wedge\zeta_k)\,.
  \end{eqnarray*}	
\end{proof}

\bigskip

\begin{flushleft}
{\bf The Chern-Simons/Wess-Zumino action for lower dimensional D-branes}
\end{flushleft}
For a simple D-brane  world-volume $f:X\hookrightarrow Y$, the anomaly factor
 $\sqrt{\hat{A}(X)/\hat{A}(N_{X/Y})\,}=1$, for $\dimm X = m\le 3$.
This may not hold for $\varphi$
 since $\varphi(X^{\!A\!z})$ can have fuzzy/nilpotent structure of nilpotency $\le r$
 (the rank of $E$ as a complex vector bundle on $X$), which can be large
 even when the dimension $m$ of $X$ is small.
However, if one formally assume that the same is true,
 then for lower dimensional D-branes
 (i.e.\ D$(-1)$-, D$0$-, D$1$-, D$2$-branes),
 one has:
 (Assuming that $B=\sum_{i,j}B_{ij}dy^i\otimes dy^j$, $B_{ji}=-B_{ij}$)
 
  \begin{itemize}
   \item[\LARGE $\cdot$]
    For {\it D$(-1)$-brane} world-point $(m=0)\,$:
     $$
	   S_{\CSWZ}^{(C_{(0)})}(\varphi)\;
	   =\; T_{-1}\,\cdot\, \Tr(\varphi^{\diamond}C_{(0)})\;
	   =\; T_{-1}\,\cdot\, \Tr(\varphi^{\sharp}(C_{(0)}))\,.
	 $$
   
   \item[\LARGE $\cdot$]
    For {\it D-particle} world-line $(m=1)\,$:
	 Assume that  $C_{(1)}=\sum_{i=1}^nC_i\,dy^i$ locally; then
     $$
	   S_{\CSWZ}^{(C_{(1)})}(\varphi,\nabla)\;
	   =\; T_0 \int_X  \Tr(  \varphi^{\diamond}C_{(1)}) \;\;\;	
	   \stackrel{\mbox{\tiny locally}}{=}\;
	   T_0\int_U\Tr
	      \!\left(\rule{0ex}{1em}\right.
		     \sum_{i=1}^n \varphi^{\sharp}(C_i) \cdot D_x\varphi^{\sharp}(y^i)
			  \!\left.\rule{0ex}{1em}\right) dx\,.	
	 $$
   Here, $D_x:=D_{\partial/\partial x}$.
   
   \item[\LARGE $\cdot$]
    For {\it D-string} world-sheet $(m=2)\,$:
	Assume that $C_{(2)}= \sum_{i,j=1}^nC_{ij}\,dy^i\otimes dy^j$ locally,\\
	 with $C_{ij}=-C_{ji}$;
	then
	{\footnotesize	
    \begin{eqnarray*}
	 \lefteqn{
	 S_{\tinyCSWZ}^{(C_{(0)}, C_{(2)}, B)}(\varphi,\nabla)\;\;
	   = \;\; T_1 \int_X    \Real( \Tr(
	                \varphi^{\diamond}C_{(2)}\, +\,  \varphi^{\diamond}(C_{(0)}B)
			         +\, 2\pi\alpha^{\prime} \varphi^{\sharp}(C_{(0)})\odot F_{\nabla} ))   }\\
     && \hspace{.7em}=\hspace{1.3em}
	           T_1 \int_X    \Real( \Tr(
	                \varphi^{\diamond}(C_{(2)}+C_{(0)}B   )\,
			         +\, \pi\alpha^{\prime} \varphi^{\sharp}(C_{(0)})F_{\nabla}\,
					 +\, \pi\alpha^{\prime}F_{\nabla}\varphi^{\sharp}(C_{(0)})
				      ))    \\
     && \stackrel{\mbox{\tiny locally}}{=}\;\;
	      T_1\int_U
	       \Real    \!\left(\rule{0ex}{1em}\right.\!\!
		      \Tr  \!\left(\rule{0ex}{1em}\right.
		      \sum_{i,j=1}^n\,
			      \varphi^{\sharp}(C_{ij}+C_{(0)}B_{ij})\,
				    D_{\!x^1}\varphi^{\sharp}(y^i)\,D_{\!x^2}\varphi^{\sharp}(y^j)\,  \\
     && \hspace{4em}					
                +\, \pi\alpha^{\prime}
				         \varphi^{\sharp}(C_{(0)})\,[\nabla_{\!x^1},\nabla_{\!x^2}]\,
			    +\, \pi\alpha^{\prime}
					     [\nabla_{\!x^1},\nabla_{\!x^2}]\, \varphi^{\sharp}(C_{(0)})		
			   \left.\rule{0ex}{1em}\right)\!\!
			        \left.\rule{0ex}{1em}\right)
		      d^2\!\boldx\,.		
	 \end{eqnarray*}}
	 Here, $D_{\!x^1}:= D_{\partial/\partial x^1}$, $D_{\!x^2}:=D_{\partial/\partial x^2}$
	   and  $\nabla_{\!x^1}:= \nabla_{\partial/\partial x^1}$,
	           $\nabla_{\!x^2}:= \nabla_{\partial/\partial x^2}$.
	
   \item[\LARGE $\cdot$]
    For {\it D-membrane} world-volume $(m=3)\,$:
	Assume that
	   $C_{(1)}=\sum_{i=1}^nC_i\,dy^i$ and
	   $C_{(3)}=\sum_{i,j,k=1}^nC_{ijk}\,dy^i\otimes dy^j\otimes dy^k$ locally,
	   with $C_{ijk}$ alternating with respect to $ijk$;
	then
  {\footnotesize
	\begin{eqnarray*}
	 \lefteqn{
	  S_{\tinyCSWZ}^{(C_{(1)}, C_{(3)}, B)}(\varphi,\nabla)\;\;
	   =\;\;  T_2 \int_X    \Real( \Tr(
	                \varphi^{\diamond}C_{(3)}\,
					+\,  \varphi^{\diamond}(C_{(1)}\wedge B)\,
	                +\, 2\pi\alpha^{\prime}\,
					            \varphi^{\diamond}C_{(1)}\odotwedge F_{\nabla}\,
			                ))     }   \\
     && \stackrel{\mbox{\tiny locally}}{=}\;\;
	   T_2\int_U
	       \Real    \!\left(\rule{0ex}{1em}\right.\!\!
		      \Tr  \!\left(\rule{0ex}{1em}\right.
			     \sum_{i,j,k=1}^n
				    \varphi^{\sharp}(C_{ijk}+C_iB_{jk}+C_jB_{ki}+C_kB_{ij})\,
					   D_{\!x^1}\varphi^{\sharp}(y^i)\,
					   D_{\!x^2}\varphi^{\sharp}(y^j)\,
					   D_{\!x^3}\varphi^{\sharp}(y^k) \\
     && \hspace{4em}					
		      +\, \pi\alpha^{\prime}
			        \sum_{(\lambda\mu\nu)\in\scriptsizeSym_3}
			          \sum_{i=1}^n\,
			           (-1)^{(\lambda\mu\nu)}
					     \left(\rule{0ex}{1em}\right.
						     \varphi^{\sharp}(C_i)\,D_{\!x^{\lambda}}(\varphi^{\sharp}(y^i))\,
						           [\nabla_{\!x^{\mu}}, \nabla_{\!x^{\nu}}]\,
                            +\, [\nabla_{\!x^{\mu}}, \nabla_{\!x^{\nu}}]\,
							          \varphi^{\sharp}(C_i)\,D_{\!x^{\lambda}}\varphi^{\sharp}(y^i)
                           \!\left.\rule{0ex}{1em}\right)
			   \!\!\left.\rule{0ex}{1em}\right)\!\!
			        \left.\rule{0ex}{1em}\right)
		      d^3\!\boldx\,.		    							
	\end{eqnarray*}}
  \end{itemize}
The technical issue of anomaly is the focus of another work.
For the moment, we will take the above as our {\it working Anzatz}
 for the Chern-Simons/Wess-Zumino action for lower-dimensional D-branes.

\bigskip

\begin{remark} {\it $[$What is Ramond-Ramond field?$\,]$}\; {\rm
 From the way a D-brane couple to a Ramond-Ramond field,
  one learns that a Ramond-Ramond field is to a D-brane as a $B$-field is to a fundamental string.
 In the latter case, while a $B$-field is taken to be a $2$-form on the space-time $Y$ to begin with,
   after years of development one learns
      that the meaning/precise definition of $B$-field goes much beyond just a $2$-form on $Y$.
 It's not yet settled what it really is, but it is known that structures on loop spaces and gerbes are involved
  (e.g.\ [Bry]).
 One expects thus that, in parallel,
   a Ramond-Ramond field go beyond just a differential form on the space-time $Y$.
 Under our setting, the loop space in the case of $B$-field is expected to be replaced by
 a map-space $\Map((Z^{\!A\!z},E;\nabla),Y)$,
 where $Z^{\!A\!z}$ is an Azumaya/matrix manifold represemnting a D-brane (not D-brane world-volume).
 For example, the Ramond-Ramond $2$-field $C_{(2)}$ in the Type IIB superstring theory, when fully developed,
  is expected to be related to a matrix-loop space $\Map((S^{1, A\!z},E,\nabla), Y)$
  and structures thereupon.
 Furthermore,
  when $E$ has rank $>1$,
  one expects also that $C_{(2)}$,  being a field sourced by D-strings, is enhanced to non-Abelian-valued.
 All these issues, and beyond, remain to be understood.
}\end{remark}

\bigskip
 
\subsection{The first variation and the contribution to the equations of motion}

Under the same setup as in Sec.~5.2,
  we derive in this subsection
  the first variation of the Chern-Simons/Wess-Zumino action $S_{\CSWZ}^{(C,B)}$
   for lower-dimensional D-brane world-volumes.
The additional contribution to the equations of motion for such lower-dimensional D-branes
  due to the additional term $S_{\CSWZ}^{(C,B)}$
    in the total action for D-brane world-volume would then follow.

\bigskip

\subsubsection{D$(-1)$-brane world-point $(m=0)$}

For a D$(-1)$-brane world-point,
 $\dimm X = 0$, $\nabla=0$, and
 $S_{\CSWZ}^{(C_{(0)})}(\varphi)
    = T_{-1}\,\cdot\, \Tr(\varphi^{\sharp}(C_{(0)}))$.
It follows that
{\footnotesize
 \begin{eqnarray*}
  \lefteqn{
   \left.\mbox{\Large $\frac{d}{dt}$}\right|_{t=0}\,
	 S_{\tinyCSWZ}^{(C_{(0)})}(\varphi_T)\;
     =\; 	
	 T_{-1}\,
	  \left.\mbox{\Large $\frac{d}{dt}$}\right|_{t=0}\,
	    \Tr(\varphi_T^{\sharp}(C_{(0)}))    \;\;\;
	 =\;\;\;
	  T_{-1}\,	
	    \Tr
		 \left(\rule{0ex}{1em}\right.
		 \left.\mbox{\Large $\frac{d}{dt}$}\right|_{t=0}\,
	      \varphi_T^{\sharp}(C_{(0)})
		 \left.\rule{0ex}{1em}\right)                 }\\
  &&=\;\:
     T_{-1}\,	
	    \Tr
		 \left(\rule{0ex}{1em}\right.
     \sum_{j=1}^n
	  \left(\rule{0ex}{1em}\right.	
	   \sum_{d,\,\scriptsizeboldd,\,\vec{\pi}}\,
	     R^{\,C_{(0)}}[1]_{(\scriptsizeboldd,\,\vec{\pi})}^{\;R}(\varphi^{\sharp}(\boldy))
	     \cdot
 	     R^{\,C_{(0)}}[1]_{(\scriptsizeboldd,\,\vec{\pi})}^{\;L}(\varphi^{\sharp}(\boldy))
	  \left.\rule{0ex}{1em}\right)
	  \cdot  \dot{\varphi}^{\sharp}(y^j)
		 \left.\rule{0ex}{1em}\right)\\
  &&\:=:
     T_{-1}\Tr
	    \left(\rule{0ex}{1em}\right.
      	  \sum_{j=1}^n\,
              \NL^{(C_{(0)});\delta\varphi}_j (\varphi) \cdot   \dot{\varphi}^{\sharp}(y^j)
		 \left.\rule{0ex}{1em}\right)\,.
 \end{eqnarray*}
}
Here, the following identities are employed:
\begin{eqnarray*}
  \lefteqn{
   \left.\mbox{\Large $\frac{d}{dt}$}\right|_{t=0}
          \left(\varphi_T^{\sharp}(C_{(0)})\right)\;\;
   =\;\;  R^{\,C_{(0)}}[1]|
             _{\scriptsizeboldy^{\tinyboldd}
 			                    \rightsquigarrow\,
								   ^{\mbox{\large $\cdot$}}\!
								     (\varphi^{\sharp}(\scriptsizeboldy^{\tinyboldd}))}}\\
  && =\;\;
       \sum_{j=1}^n\,
	   \sum_{d=0}^{\mbox{\tiny $\bullet$}}
	   \sum_{\;\,\scriptsizeboldd,\, |\scriptsizeboldd|=d}
	   \sum_{\;\vec{\pi}
	                         \in \vec{\mbox{\scriptsize\it Ptn}}(1, d),\,
					   i_{(\vec{\pi},\,\tinyboldd)}= j}	                           \\
    && \hspace{7em}	
	  ([\partial^{\vec{\pi}}_{y^j}]
	             R^{\,C_{(0)}}[1]_{(\scriptsizeboldd)})^L(\varphi^{\sharp}(\boldy))
	   \cdot
	   \dot{\varphi}^{\sharp}(y^j)
	   \cdot
	 ([\partial^{\vec{\pi}}_{y^j}]
	             R^{\,C_{(0)}}[1]_{(\scriptsizeboldd)})^R(\varphi^{\sharp}(\boldy)) \\									 
  && =:\:\;
    \sum_{j=1}^n\,
	\sum_{d,\,\scriptsizeboldd,\,\vec{\pi}}\,
	  R^{\,C_{(0)}}[1]_{(\scriptsizeboldd,\,\vec{\pi})}^{\;L}(\varphi^{\sharp}(\boldy))
	  \cdot \dot{\varphi}^{\sharp}(y^j)
	  \cdot
	  R^{\,C_{(0)}}[1]_{(\scriptsizeboldd,\,\vec{\pi})}^{\;R}(\varphi^{\sharp}(\boldy))\,.	
 \end{eqnarray*}

In this case, $S_{\DBI}^{(\Phi,g,B)}(\varphi)=0$ always    and
 the full action
   $S_{\DBI}^{(\Phi,g,B)}+S_{\CSWZ}^{(C_{(0)})}$
    is simply $S_{\CSWZ}^{(C_{(0)})}$.
The full system of equations of motion is thus
 $$
  \NL_j^{(\Phi, g, B, C_{(0)}); \delta\varphi}(\varphi)\;
   :=\;   \NL^{(C_{(0)});\delta\varphi}_j (\varphi)\;\;=\;\;0\,,
 $$			
 $j=1,\,\ldots\,,\,n\,$,
 for D$(-1)$-brane.
Such world-points give rise to instantons in space-time.

\bigskip

\subsubsection{D-particle world-line $(m=1)$}

For a D-particle world-line,  $\dimm X=1$ and
 {\footnotesize
  $$
    S_{\tinyCSWZ}^{(C_{(1)})}(\varphi,\nabla)\;
	 =\;
  	 T_0\int_U\Tr
	      \!\left(\rule{0ex}{1em}\right.
		     \sum_{i=1}^n \varphi^{\sharp}(C_i) \cdot D_x\varphi^{\sharp}(y^i)
			  \!\left.\rule{0ex}{1em}\right) dx
  $$}
  locally over $X$.
It follows that
{\footnotesize
 \begin{eqnarray*}
  \lefteqn{
    \left.\mbox{\Large $\frac{d}{dt}$}\right|_{t=0}\,
	 S_{\tinyCSWZ}^{(C_{(1)})}(\varphi_T,\nabla^T)\;
	  = \;
	  T_0\int_U\Tr
	      \!\left(\rule{0ex}{1em}\right.
		     \sum_{i=1}^n
			    \dot{\varphi}^{\sharp}(C_i) \cdot D_x\varphi^{\sharp}(y^i)\,
				 +\, \varphi^{\sharp}(C_i)
				         \cdot (D_x\dot{\varphi}^{\sharp}(y^i)\,
						     -\, [\varphi^i(\sharp)(y^i),  \dot{A}_x])				
			  \!\left.\rule{0ex}{1em}\right) dx                    }\\
    &&
	  = \;
	   T_0\, \Tr
	       \!\left(\rule{0ex}{1em}\right.\!
		      \sum_{i=1}^n
			   \varphi^{\sharp}(C_i)\,\dot{\varphi}^{\sharp}(y^i)			
		   \!\!\left.\rule{0ex}{1em}\right)\!\!\left.\rule{0ex}{1em}\right|_{\partial U}  \\
    && \hspace{6em}		
	   +\,
	  T_0\int_U\Tr
	      \!\left(\rule{0ex}{1em}\right.
		     \sum_{i=1}^n
			    \dot{\varphi}^{\sharp}(C_i) \cdot D_x\varphi^{\sharp}(y^i)\,
				 -\, D_x\varphi^{\sharp}(C_i)\cdot \dot{\varphi}^{\sharp}(y^i)\,
                 -\, \varphi^{\sharp}(C_i)\cdot [\varphi^i(\sharp)(y^i), \dot{A}_x]				
			  \!\left.\rule{0ex}{1em}\right) dx                    \\			
	&& =\; T_0\,\BT^{(\varphi;C_{(1)})}(\dot{\varphi}^{\sharp}(\boldy))|_{\partial U}\,
	              +\,   T_0 \int_U\,
				    \Tr
				     \left(\rule{0ex}{1em}\right.
					  \sum_{j=1}^n\,
					    \NL^{(C_{(1)});\delta\varphi}_j(\varphi,\nabla)
					    \cdot   \dot{\varphi}^{\sharp}(y^j)\;
						+\; \NL^{(C_{(1)});\delta\nabla}_x(\varphi)\cdot\dot{A}_x
					 \!\left.\rule{0ex}{1em}\right)dx\,,
 \end{eqnarray*}}
 where
{\footnotesize
 \begin{eqnarray*}
  \BT^{(\varphi;C_{(1)})}(\dot{\varphi}^{\sharp}(\boldy))
   & =
   & \Tr
	       \!\left(\rule{0ex}{1em}\right.\!
		      \sum_{i=1}^n
			   \varphi^{\sharp}(C_i)\,\dot{\varphi}^{\sharp}(y^i)			
		   \!\!\left.\rule{0ex}{1em}\right)\,, \\
  \NL^{(C_{(1)});\delta\varphi}_j(\varphi,\nabla) \\[-1ex]
    && \hspace{-6em}=\;\;
	  -\, D_x\varphi^{\sharp}(C_j)\,
	    +\,\sum_{i=1}^n
	          \sum_{d, \scriptsizeboldd,\vec{\pi};\; |\scriptsizeboldd|=d, i_{(\vec{\pi},\tinyboldd)}=j}
			  R^{C_i}[1]_{(\scriptsizeboldd,\vec{\pi})}^R(\varphi^{\sharp}(\boldy))
			    \cdot  D_x\varphi^{\sharp}(y^i)
				\cdot R^{C_i}[1]_{(\tinyboldd,\vec{\pi})}^L(\varphi^{\sharp}(\boldy))\,,    \\
  \NL^{(C_{(1)});\delta\nabla}_x(\varphi)
    & = & \sum_{i=1}^n\,[\varphi^{\sharp}(y^i)\,,\, \varphi^{\sharp}(C_i)]
	           \;\; =\;\; 0 \,.
 \end{eqnarray*}}

The full action
 $S_{\DBI}^{(\Phi,g,B)}(\varphi,\nabla)+S_{\CSWZ}^{(C_{(1)})}(\varphi,\nabla)$
 gives the system of equations of motion for a D-particle moving in $Y$:
 \begin{eqnarray*}
  \NL^{(\Phi, g, B, C_{(1)}); \delta\varphi}_j(\varphi,\nabla)
   & :=  &  \NL^{(\Phi,g,B); \delta\varphi}_j(\varphi,\nabla)\,
                   +\, \NL^{(C_{(1)});\delta\varphi}_j (\varphi,\nabla)\;\;=\;\;0\,,  \\
  \NL^{(\Phi,g,B,C_{(1)});\delta\nabla}_x(\varphi,\nabla)
   & := &  \NL^{(\Phi,g,B);\delta\nabla}_x(\varphi,\nabla)  \,     \hspace{9em}=\;\;0\,,
 \end{eqnarray*}
 $j=1,\,\ldots\,,\,n\,$.
 
For the current case, the curvature $F_{\nabla}$ of $\nabla$ is zero and
 the above system  may still involves $A_x$ but not its differentials with respect to $x$.
I.e.\  it is a system of differential equations on $\varphi$ but non-differential equations on $\nabla$.
 $\nabla$ is thus non-dynamical, as is anticipated.
Thus, after a re-trivialization fo the fundamental module $E$ on $X$, one may assume that $A_x\equiv 0$
 and the above system is reduced to a system
  $$
    \NL^{(\Phi,g,B,C_{(1)});\delta\varphi}(\varphi)\;=\;0\,,
	  \hspace{3em}j=1,\,\ldots\,,\, n\,,
  $$
 of second-order nonlinear differential equations that involve $\varphi$ alone.

\bigskip

\subsubsection{D-string world-sheet $(m=2)$}

Denote
 $$
  \breve{C}_{(2)}\;
   :=\;  C_{(2)}+C_{(0)}B\;
   =\; \sum_{ij}(C_{ij}+C_{(0)}B_{ij})\, dy^i\otimes dy^j\;
   =\; \sum_{i,j}\breve{C}_{ij}dy^i\otimes dy^j
 $$
  in local coordinates of $Y$.
Then,
for a D-string world-sheet, $\dimm X=2$ and
{\footnotesize
  \begin{eqnarray*}
	 S_{\tinyCSWZ}^{(C_{(0)},C_{(2)},B)}(\varphi,\nabla)
	   & =
	   & T_1\int_U
	       \Real    \!\left(\rule{0ex}{1em}\right.\!\!
		      \Tr  \!\left(\rule{0ex}{1em}\right.
		      \sum_{i,j=1}^n\,
			      \varphi^{\sharp}(\breve{C}_{ij})\,
				    D_1\varphi^{\sharp}(y^i)\,D_2\varphi^{\sharp}(y^j)\,  \\
     && \hspace{8em}					
                +\, \pi\alpha^{\prime}
				         \varphi^{\sharp}(C_{(0)})\,F_{12}\,
			    +\, \pi\alpha^{\prime}
					     F_{12}\, \varphi^{\sharp}(C_{(0)})		
			   \left.\rule{0ex}{1em}\right)\!\!
			        \left.\rule{0ex}{1em}\right)
		      d^2\!\boldx		     	
  \end{eqnarray*}}
 
 \noindent
 locally over $X$.
(Here,
     $D_1:=D_{\partial/\partial x^1}$, $D_2:=D_{\partial/\partial x^2}$, and
     $F_{12}:=[\nabla_{\!x^1},\nabla_{\!x^2}]$ is the curvature of $\nabla$.)
It follows then from a straightforward computation that

 {\footnotesize
 \begin{eqnarray*}
  \lefteqn{
    \left.\mbox{\Large $\frac{d}{dt}$}\right|_{t=0}\,
	 S_{\tinyCSWZ}^{(C_{(0)},C_{(2)},B)}(\varphi_T,\nabla^T) }\\
  &&  =
	   T_1\int_U
	       \Real    \!\left(\rule{0ex}{1em}\right.\!\!
		      \Tr  \!\left(\rule{0ex}{1em}\right.
			   \sum_{i,j=1}^n
			        \!\left(\rule{0ex}{1em}\right.\!\!
					\dot{\varphi}^{\sharp}(\breve{C}_{ij})\,
				    D_1\varphi^{\sharp}(y^i)\,D_2\varphi^{\sharp}(y^j)\,
                  +\, \varphi^{\sharp}(\breve{C}_{ij})\,
				    (D_1\dot{\varphi}^{\sharp}(y^i)- [\varphi^{\sharp}(y^i), \dot{A}_1])\,
					  D_2\varphi^{\sharp}(y^j)\,           \\[-2.4ex]
    && \hspace{18em}					
                  +\, \varphi^{\sharp}(\breve{C}_{ij})\, D_1\varphi^{\sharp}(y^i)\,
					(D_2\dot{\varphi}^{\sharp}(y^j) - [\varphi^{\sharp}(y^j), \dot{A}_2] )	
                     \!\!\left.\rule{0ex}{1em}\right)						\\
	&& \hspace{8em}	      	
	    +\, \pi\alpha^{\prime}
				\dot{\varphi}^{\sharp}(C_{(0)})\,F_{12}\,
        +\, \pi\alpha^{\prime}
				 \varphi^{\sharp}(C_{(0)})\, (D_1\dot{A}_2 - D_2\dot{A}_1)\,	\\
    && \hspace{8em}				
	    +\, \pi\alpha^{\prime}
			 (D_1\dot{A}_2 - D_2\dot{A}_1)\, \varphi^{\sharp}(C_{(0)})\,
        +\, \pi\alpha^{\prime}
			   F_{12}\, \dot{\varphi}^{\sharp}(C_{(0)})		\,	
			   \left.\rule{0ex}{1em}\right)
			        \!\!\!\left.\rule{0ex}{1em}\right)
		      d^2\!\boldx	         \\
   &&
    =\; T_1\,\int_{\partial U}
	  \Real(\BT^{(\varphi,\nabla; C_{(0)}, C_{(2)},B)}
	                   (\dot{\varphi}^{\sharp}(\boldy),\dot{\boldA}))\\
   && \hspace{3.7em}
    +\, T_1\,\int_U
       \Real	
	    \!\left(\rule{0ex}{1em}\right.
	    \!\!\Tr
	    \!\left(\rule{0ex}{1em}\right.
	      \sum_{j=1}^n
		  \NL^{(C_{(0)}, C_{(2)}, B);\delta\varphi}_j(\varphi,\nabla)\,
		     \cdot\,  \dot{\varphi}^{\sharp}(y^j)
		 +\;
		 \sum_{\nu=1}^2
           \NL^{(C_{(0)});\delta\nabla}_{\nu}(\varphi,\nabla)\,
		     \cdot\, \dot{A}_{\nu}
		\!\left.\rule{0ex}{1em}\right)		
	    \!\!\!\left.\rule{0ex}{1em}\right)	
	    d^2\!\boldx\,,	 		  			 			
 \end{eqnarray*}}
where

\bigskip

${\LARGE}\cdot$ the {\it boundary term} is given by

{\footnotesize
\begin{eqnarray*}
 \lefteqn{
  \BT^{(\varphi,\nabla; C_{(0)}, C_{(2)}, B)}
	                   (\dot{\varphi}^{\sharp}(\boldy),\dot{\boldA})\,,
      \hspace{1em}\mbox{a $1$-form on $U$}\,,					   }\\
   &&  =\;
     \Tr \left(\rule{0ex}{1em}\right.
	    \sum_{j=1}^n
		   \left(\rule{0ex}{1em}\right.   		
		     \sum_{i=1}^n
			  D_2\varphi^{\sharp}(y^i)\,\varphi^{\sharp}(\breve{C}_{ji})
           \left.\rule{0ex}{1em}  \right) 	
				\cdot  \dot{\varphi}^{\sharp}(y^j)\,
			 +\, 2\pi\alpha^{\prime} \varphi^{\sharp}(C_{(0)})\cdot \dot{A}_2		
	       \left.\rule{0ex}{1em}  \right)
		   dx^2  \\
   &&  \hspace{6em}
     -\, \Tr \left(\rule{0ex}{1em}\right.
	    \sum_{j=1}^n
		   \left(\rule{0ex}{1em}\right.   		
		     \sum_{i=1}^n
			  \varphi^{\sharp}(\breve{C}_{ij})\, D_1\varphi^{\sharp}(y^i)
           \left.\rule{0ex}{1em}  \right) 	
				\cdot  \dot{\varphi}^{\sharp}(y^j)\,
			 -\, 2\pi\alpha^{\prime} \varphi^{\sharp}(C_{(0)})\cdot \dot{A}_1		
	       \left.\rule{0ex}{1em}  \right)
		   dx^1\,,
\end{eqnarray*}}

\bigskip

{\LARGE $\cdot$} the {\it subsystem associated to variations of $\varphi\,$}:

{\footnotesize
\begin{eqnarray*}
 \lefteqn{
  \NL^{(C_{(0)}, C_{(2)},B);\delta\varphi}_j(\varphi,\nabla)\;    }\\
  &&
   =\;
	 \sum_{i^{\prime}, j^{\prime}=1}^n
	  \sum_{\;d,\scriptsizeboldd, \vec{\pi};\,
	                      |\scriptsizeboldd|=d, i_{(\vec{\pi},\tinyboldd)}=j }
            R^{\breve{C}_{i^{\prime}j^{\prime}}}[1]_{(\scriptsizeboldd,\vec{\pi})}^R
			    (\varphi^{\sharp}(\boldy))\,
            D_1\varphi^{\sharp}(y^{i^{\prime}})\,
            D_2\varphi^{\sharp}(y^{j^{\prime}})\,				
			R^{\breve{C}_{i^{\prime}j^{\prime}}}[1]_{(\scriptsizeboldd,\vec{\pi})}^L
                 (\varphi^{\sharp}(\boldy))		    \\
    &&  \hspace{2em}
	  -\, \sum_{i=1}^n
	         \left(\rule{0ex}{1em}\right.
             D_2\varphi^{\sharp}(y^i)\,D_1\varphi^{\sharp}(\breve{C}_{ji})\,
             +\, D_2\varphi^{\sharp}(\breve{C}_{ij})\, D_1\varphi^{\sharp}(y^i)\,
             +\, [F_{12}, \varphi^{\sharp}(y^i)]	\cdot \varphi^{\sharp}(\breve{C}_{ji})
             \left.\rule{0ex}{1em}\right)			 \\			
    && \hspace{2em}
	  +\, 2\pi\alpha^{\prime}
	        \sum_{\;d,\scriptsizeboldd, \vec{\pi};\,
	                      |\scriptsizeboldd|=d, i_{(\vec{\pi},\tinyboldd)}=j }
            R^{C_{(0)}}[1]_{(\scriptsizeboldd,\vec{\pi})}^R
			    (\varphi^{\sharp}(\boldy))\,
            F_{12}\,				
			R^{C_{(0)}}[1]_{(\scriptsizeboldd,\vec{\pi})}^L
                 (\varphi^{\sharp}(\boldy))\,,		    			
\end{eqnarray*}}
 
\bigskip
 
{\LARGE $\cdot$} the {\it subsystem associated to variations of $\nabla\,$}:
  
{\footnotesize
$$
 \NL^{(C_{(0)});\delta\nabla}_1(\varphi,\nabla) \;
    =\;    2\pi\alpha^{\prime}D_2\varphi^{\sharp}(C_{(0)})\,,
  \hspace{2em}
  \NL^{(C_{(0)});\delta\nabla}_2(\varphi,\nabla) \;
    =\;    -\,2\pi\alpha^{\prime}D_1\varphi^{\sharp}(C_{(0)})\,.
$$}

\begin{itemize}
 \item[]
 Note that, as a consequence of Leibniz rule or integration by parts, there are at first summands 
  \begin{eqnarray*}
   -D_2\varphi^{\sharp}(y^j)\,\varphi^{\sharp}(\breve{C}_{ij})\,\varphi^{\sharp}(y^i)\,
    +\, \varphi^{\sharp}(y^i)\, D_2\varphi^{\sharp}(y^j)\, \varphi^{\sharp}(\breve{C}_{ij})
	  &&  \mbox{in $\;\NL^{(C_{(0)});\delta\nabla}_1(\varphi,\nabla)$}\,,  \\
   -\varphi^{\sharp}(\breve{C}_{ij})\, D_1\varphi^{\sharp}(y^i)\, \varphi^{\sharp}(y^j)\,
    +\, \varphi^{\sharp}(y^j)\, \varphi^{\sharp}(\breve{C}_{ij})\, D_1\varphi^{\sharp}(y^i)
	  &&  \mbox{in $\;\NL^{(C_{(0)});\delta\nabla}_2(\varphi,\nabla)$}\,,
  \end{eqnarray*}
  respectively.
 However, they vanish for $(\varphi, \nabla)$ admissible.
 Thus, the $2$-forms $C_{(2)}$ and $B$ has no consequence to
   the variation of $S^{(C_{(0)}, C_{(2)}, B)}_{\CSWZ}$ with respect to $\nabla$.
 This is anticipated since there is no coupling term between $C_{(2)}$, $B$ and $\nabla$
  in $S_{\CSWZ}^{(C_{(0)}, C_{(2)}, B)}$.
\end{itemize}

The contribution of the Chern-Simon/Wess-Zumino term
 $S_{\CSWZ}^{(C_{(0)}, C_{(2)}, B)}$ to the equations of motion for
 a $D$-string follows immediately.

\bigskip

\subsubsection{D-membrane world-volume $(m=3)$}

Denote
 \begin{eqnarray*}
  \breve{C}_{(3)}
   & :=\:   &  C_{(3)}+C_{(1)}\wedge B  \\
   & =
      & \sum_{i,j,k}
	       (C_{ijk}+C_iB_{jk}+C_jB_{ki}+C_kB_{ij})\, dy^i\otimes dy^j\otimes dy^k\;\;
       =\;\; \sum_{i,j,k}\breve{C}_{ijk}dy^i\otimes dy^j\otimes dy^k
 \end{eqnarray*}
 in local coordinates of $Y$.
Then,
for D-membrane world-volume, $\dimm X=3$ and
 {\footnotesize
 \begin{eqnarray*}
  \lefteqn{
  S_{\tinyCSWZ}^{(C_{(1)}, C_{(3)},B)}(\varphi,\nabla)} \\
   &&  =
	   T_2\int_U
	       \Real    \!\left(\rule{0ex}{1em}\right.\!\!
		      \Tr  \!\left(\rule{0ex}{1em}\right.
			     \sum_{i,j,k=1}^n
				    \varphi^{\sharp}(\breve{C}_{ijk})\,
					   D_1\varphi^{\sharp}(y^i)\,
					   D_2\varphi^{\sharp}(y^j)\,
					   D_3\varphi^{\sharp}(y^k) \\
     && \hspace{7em}					
		      +\, 2\pi\alpha^{\prime}
			        \sum_{(\lambda\mu\nu)\in\scriptsizeSym_3}
			          \sum_{i=1}^n\,
			           (-1)^{(\lambda\mu\nu)}
					      \!\left(\rule{0ex}{1em}\right.\!
						     \varphi^{\sharp}(C_i)\,D_{\lambda}\varphi^{\sharp}(y^i)\, F_{\mu\nu}\,							   							
                             \!\!\left.\rule{0ex}{1em}\right)					
			   \!\!\!\left.\rule{0ex}{1em}\right)
			        \!\!\!\left.\rule{0ex}{1em}\right)
		      d^3\!\boldx	
 \end{eqnarray*}}
 
 \noindent
 locally over $X$.
(Here,  $F_{\mu\nu}:=[\nabla_{\!x^{\mu}},\nabla_{\!x^{\nu}}]$ is the curvature of $\nabla$.)
It follows then from a straightforward computation that

{\footnotesize
 \begin{eqnarray*}
  \lefteqn{
    \left.\mbox{\Large $\frac{d}{dt}$}\right|_{t=0}\,
	 S_{\tinyCSWZ}^{(C_{(1)}, C_{(3)}, B)}(\varphi_T,\nabla^T) }\\
  &&  =
	   T_2\int_U
	       \Real    \!\left(\rule{0ex}{1em}\right.\!\!
		      \Tr  \!\left(\rule{0ex}{1em}\right.
			     \sum_{i,j,k=1}^n
                   \left(\rule{0ex}{1em}\right.				
				    \dot{\varphi}^{\sharp}(\breve{C}_{ijk})\,
					   D_1\varphi^{\sharp}(y^i)\,
					   D_2\varphi^{\sharp}(y^j)\,
					   D_3\varphi^{\sharp}(y^k)\,     \\[-2ex]
    && \hspace{11em}					
				 +\, \varphi^{\sharp}(\breve{C}_{ijk}
				          \cdot
						  (D_1\dot{\varphi}^{\sharp}(y^i)
							 - [\varphi^{\sharp}(y^i), \dot{A}_1] )
						  \cdot
                          D_2\varphi^{\sharp}(y^j)\,D_3\varphi^{\sharp}(y^k)\,	\\
    && \hspace{13em}						
				 +\, \varphi^{\sharp}(\breve{C}_{ijk})\, D_1\varphi^{\sharp}(y^i)\,
				        \cdot
						 (D_2\dot{\varphi}^{\sharp}(y^j)
						   - [\varphi^{\sharp}(y^j), \dot{A}_2] )
						\cdot
                        D_3\varphi^{\sharp}(y^k)\,  \\
    && \hspace{15em}
				 +\, \varphi^{\sharp}(\breve{C}_{ijk})\,
				            D_1\varphi^{\sharp}(y^i)\, D_2\varphi^{\sharp}(y^j)
                        \cdot				
						(D_3\dot{\varphi}^{\sharp}(y^k) - [\varphi^{\sharp}(y^k), \dot{A}_3]  )
					   \left.\rule{0ex}{1em}\right)\\
     && \hspace{4em}					
		      +\, 2\pi\alpha^{\prime}
			        \sum_{(\lambda\mu\nu)\in\scriptsizeSym_3}
			          \sum_{i=1}^n\,
			           (-1)^{(\lambda\mu\nu)}
					      \!\left(\rule{0ex}{1em}\right.\!				
						     \dot{\varphi}^{\sharp}(C_i)\,
							     D_{\lambda}\varphi^{\sharp}(y^i)\, F_{\mu\nu}\,	 \\
	 && \hspace{8em}
   	            +\,  \varphi^{\sharp}(C_i)
							       \cdot
							      ( D_{\lambda}\dot{\varphi}^{\sharp}(y^i)
								       - [\varphi^{\sharp}(y^i), \dot{A}_{\lambda}] )
							       \cdot   F_{\mu\nu}\,	
						  +\,  \varphi^{\sharp}(C_i)\,D_{\lambda}\varphi^{\sharp}(y^i)
						          \cdot   (D_{\mu}\dot{A}_{\nu} -  D_{\nu}\dot{A}_{\mu}) \,	
                             \!\!\left.\rule{0ex}{1em}\right)					
			   \left.\rule{0ex}{1em}\right)
			        \!\!\!\left.\rule{0ex}{1em}\right)
		      d^3\!\boldx	         \\
   &&
    =\; T_2\,\int_{\partial U}
	  \Real(\BT^{(\varphi,\nabla; C_{(1)}, C_{(3)},B)}
	                   (\dot{\varphi}^{\sharp}(\boldy),\dot{\boldA}))\\
   && \hspace{3.7em}
    +\, T_2\,\int_U
       \Real	
	    \!\left(\rule{0ex}{1em}\right.
	    \!\!\Tr
	    \!\left(\rule{0ex}{1em}\right.
	      \sum_{j=1}^n
		  \NL^{(C_{(1)}, C_{(3)},B);\delta\varphi}_j(\varphi,\nabla)\,
		     \cdot\,  \dot{\varphi}^{\sharp}(y^j)
		 +\;
		 \sum_{\nu=1}^3
           \NL^{(C_{(1)});\delta\nabla}_{\nu}(\varphi,\nabla)\,
		     \cdot\, \dot{A}_{\nu}
		\!\left.\rule{0ex}{1em}\right)		
	    \!\!\!\left.\rule{0ex}{1em}\right)	
	    d^3\!\boldx\,,	 		  			 			
 \end{eqnarray*}}
where

\bigskip

${\LARGE}\cdot$ the {\it boundary term} is given by

{\footnotesize
\begin{eqnarray*}
 \lefteqn{
  \BT^{(\varphi,\nabla; C_{(1)}, C_{(3)},B)}
	                   (\dot{\varphi}^{\sharp}(\boldy),\dot{\boldA})\,,
      \hspace{1em}\mbox{a $2$-form on $U$}\,,					   }\\
   && =\;
     \Tr \left(\rule{0ex}{1em}\right.
	    \sum_{j=1}^n
		   \left(\rule{0ex}{1em}\right.
		    \sum_{i,k=1}^n
			   D_2\varphi^{\sharp}(y^i)\,D_3\varphi^{\sharp}(y^k)\,
			              \varphi^{\sharp}(\breve{C}_{jik})\,
			   +\, 4\pi\alpha^{\prime}\,F_{23}\,\varphi^{\sharp}(C_j)
		      \left.\rule{0ex}{1em}\right) \cdot \dot{\varphi}^{\sharp}(y^j) \,  \\
    && \hspace{3.6em}			
			 +\, 4\pi\alpha^{\prime}
			       \left(\rule{0ex}{1em}\right.
			     \sum_{i=1}^n
			       \varphi^{\sharp}(C_i)\,D_3\varphi^{\sharp}(y^i)
				   \left.\rule{0ex}{1em}\right)
				   \cdot  \dot{A}_2\,
			 -\, 4\pi\alpha^{\prime}
			    \left(\rule{0ex}{1em}\right.
				  \sum_{i=1}^n
				    \varphi^{\sharp}(C_i)\,D_2\varphi^{\sharp}(y^i)
			     \left.\rule{0ex}{1em}\right)  \cdot \dot{A}_3			
	       \left.\rule{0ex}{1em}  \right) d^2\wedge dx^3\,                   \\
   && \hspace{1.7em}		
       - \, \Tr  \left(\rule{0ex}{1em}\right.
	        \sum_{j=1}^n
			    \left(\rule{0ex}{1em}\right.
                   \sum_{i,k=1}^n
                     D_3\varphi^{\sharp}(y^k)\,
					   \varphi^{\sharp}(\breve{C}_{ijk})\, D_1\varphi^{\sharp}(y^i)\,
					    -\, 4\pi\alpha^{\prime}\, F_{13}\, \varphi^{\sharp}(C_j)
                 \left.\rule{0ex}{1em}  \right) \cdot  \dot{\varphi}^{\sharp}(y^j)\,   \\
    && \hspace{5.2em}				
             -\, 4\pi\alpha^{\prime}
                 \left(\rule{0ex}{1em}\right.
                   \sum_{i=1}^n
                     \varphi^{\sharp}(C_i)\, D_3\varphi^{\sharp}(y^i)				
                 \left.\rule{0ex}{1em}  \right) 	  \cdot \dot{A}_1\,
			 +\, 4\pi\alpha^{\prime}
                     \left(\rule{0ex}{1em}\right.
                      \sum_{i=1}^n
                         \varphi^{\sharp}(C_i)\, D_1\varphi^{\sharp}(y^i)					
                 \left.\rule{0ex}{1em}  \right) 	\cdot \dot{A}_3		
	       \left.\rule{0ex}{1em}  \right)  d^1\wedge dx^3\,                 \\
   && \hspace{1.7em}		
	   +\, \Tr  \left(\rule{0ex}{1em}\right.
	        \sum_{j=1}^n
			    \left(\rule{0ex}{1em}\right.
                  \sum_{i,k=1}^n
                    \varphi^{\sharp}(\breve{C}_{ikj})\,
					    D_1\varphi^{\sharp}(y^i)\,  D_2\varphi^{\sharp}(y^k)\,
						 +\, 4\pi\alpha^{\prime}\, F_{12}\, \varphi^{\sharp}(C_j)
                 \left.\rule{0ex}{1em}  \right) \cdot  \dot{\varphi}^{\sharp}(y^j)\,  \\
    &&  \hspace{5.2em}				
             +\, 4\pi\alpha^{\prime}
                 \left(\rule{0ex}{1em}\right.
                   \sum_{i=1}^n
                     \varphi^{\sharp}(C_i)\, D_2\varphi^{\sharp}(y^i)				
                 \left.\rule{0ex}{1em}  \right) 	  \cdot \dot{A}_1\,
			 -\, 4\pi\alpha^{\prime}
                     \left(\rule{0ex}{1em}\right.
                     \sum_{i=1}^n
					   \varphi^{C_i}\, D_1\varphi^{\sharp}(y^i)
                 \left.\rule{0ex}{1em}  \right) 	\cdot \dot{A}_2		
	       \left.\rule{0ex}{1em}  \right) dx^1\wedge dx^2\,,
\end{eqnarray*}}

\vspace{12em}

{\LARGE $\cdot$} the {\it subsystem associated to variations of $\varphi\,$}:

{\footnotesize
\begin{eqnarray*}
 \lefteqn{
  \NL^{(C_{(1)}, C_{(3)},B);\delta\varphi}_j(\varphi,\nabla)\;    }\\
  &&
   =\;
	 \sum_{i^{\prime}, j^{\prime}, k^{\prime}=1}^n
	  \sum_{\;d,\scriptsizeboldd, \vec{\pi};\,
	                      |\scriptsizeboldd|=d, i_{(\vec{\pi},\tinyboldd)}=j }
            R^{\breve{C}_{i^{\prime}j^{\prime}k^{\prime}}}[1]
			      _{(\scriptsizeboldd,\vec{\pi})}^R
			    (\varphi^{\sharp}(\boldy))\,
            D_1\varphi^{\sharp}(y^{i^{\prime}})\,
            D_2\varphi^{\sharp}(y^{j^{\prime}})\,
            D_3\varphi^{\sharp}(y^{k^{\prime}})						
			R^{\breve{C}_{i^{\prime}j^{\prime}k^{\prime}}}[1]
			       _{(\scriptsizeboldd,\vec{\pi})}^L
                 (\varphi^{\sharp}(\boldy))		    \\
  && \hspace{2em}
      -\, \sum_{i,k=1}^n
              \left(\rule{0ex}{1em}\right.
			   D_2\varphi^{\sharp}(y^i)\, D_3\varphi^{\sharp}(y^k)\,
  			      D_1\varphi^{\sharp}(\breve{C}_{jik})\,      						
			 +\, D_3\varphi^{\sharp}(y^k)\,
			       D_2\varphi^{\sharp}(\breve{C}_{ijk})\, D_1\varphi^{\sharp}(y^i)\,     \\[-2ex]
    && \hspace{7em}				
	         +\, D_3\varphi^{\sharp}(\breve{C}_{ikj})\, D_1\varphi^{\sharp}(y^i)\,
			       D_2\varphi^{\sharp}(y^k)\,
			 +\, [F_{12}, \varphi^{\sharp}(y^i)]
			        \cdot D_3\varphi^{\sharp}(y^k)\, \varphi^{\sharp}(\breve{C}_{jik})\,  \\
    && \hspace{10em}					
             +\, [F_{23}, \varphi^{\sharp}(y^k)]
                    \cdot \varphi^{\sharp}(\breve{C}_{ijk})\, D_1\varphi^{\sharp}(y^i)\,
             +\, [F_{31}, \varphi^{\sharp}(y^i)]
                    \cdot D_2\varphi^{\sharp}(y^k)\, \varphi^{\sharp}(\breve{C}_{ikj})			 			      
			  \left.\rule{0ex}{1em}\right)	    \\
 && \hspace{2em}
    +\, 2\pi\alpha^{\prime}
           \sum_{i=1}^n
		    \sum_{\;(\lambda\mu\nu)\in \scriptsizeSym_3}
			 \sum_{\; d,\scriptsizeboldd, \vec{\pi};
			                  |\scriptsizeboldd|=d, i_{(\vec{\pi},\tinyboldd)}=j}
			 (-1)^{(\lambda\mu\nu)}
			  R^{C_i}[1]_{(\scriptsizeboldd, \vec{\pi})}^R(\varphi^{\sharp}(\boldy))\,
			  D_{\lambda}\varphi^{\sharp}(y^i)\, F_{\mu\nu}\,
			  R^{C_i}[1]_{(\scriptsizeboldd, \vec{\pi})}^L(\varphi^{\sharp}(\boldy))\\
 && \hspace{2em}
     -\, 2\pi\alpha^{\prime}
           \sum_{(\lambda\mu\nu)\in\scriptsizeSym_3}
            (-1)^{(\lambda\mu\nu)}
                \left(\rule{0ex}{1em}\right.			
				  F_{\mu\nu}\,  D_{\lambda}\varphi^{\sharp}(C_j)\,
				  +\, D_{\lambda}F_{\mu\nu} \varphi^{\sharp}(C_j)
				\left.\rule{0ex}{1em}\right)\,,
\end{eqnarray*}}
 
\bigskip
 
{\LARGE $\cdot$} the {\it subsystem associated to variations of $\nabla\,$}:
  
{\footnotesize
\begin{eqnarray*}
 \lefteqn{
   \NL^{(C_{(1)});\delta\nabla}_{\lambda}(\varphi,\nabla) }\\
 && =\;
   2\pi\alpha^{\prime}
     \sum_{i=1}^n\,
      [\varphi^{\sharp}(y^i)\,,\, F_{\mu\nu}\varphi^{\sharp}(C_i)]	 \,  \\
 && \hspace{2em}	
   +\, 4\pi\alpha^{\prime}
        \sum_{i=1}^n\,
		 \left(\rule{0ex}{1em}\right.
		  D_{\mu}\varphi^{\sharp}(C_i)\, D_{\nu}\varphi^{\sharp}(y^i)\,
		  -\, D_{\nu}\varphi^{\sharp}(C_i)\, D_{\mu}\varphi^{\sharp}(y^i)\,
	      +\,\varphi^{\sharp}(C_i)
		       \cdot [F_{\mu\nu}, \varphi^{\sharp}(y^i)]
		 \left.\rule{0ex}{1em}\right),  \\[1ex]
  && \hspace{-3.6em}
    \mbox{where $\;\;(\lambda\mu\nu)\;=\; (123),\, (231),\, (312)$}\,.		
\end{eqnarray*}}

\begin{itemize}
 \item[]
 Note that, as a consequence of Leibniz rule or integration by parts, there are at first summands 
  \begin{eqnarray*}
    \sum_{i,j,k=1}^n
        [\varphi^{\sharp}(y^i),
	     D_2\varphi^{\sharp}(y^j)D_3\varphi^{\sharp}(y^k)\varphi^{\sharp}(\breve{C}_{ijk})]
	  &&  \mbox{in $\;\NL^{(C_{(1)});\delta\nabla}_1(\varphi,\nabla)$}\,,  \\
	 \sum_{i,j,k=1}^n
        [\varphi^{\sharp}(y^j),
	     D_3\varphi^{\sharp}(y^k)\varphi^{\sharp}(\breve{C}_{ijk})D_1\varphi^{\sharp}(y^i)]
	  &&  \mbox{in $\;\NL^{(C_{(1)});\delta\nabla}_2(\varphi,\nabla)$}\,,  \\
    \sum_{i,j,k=1}^n
        [\varphi^{\sharp}(y^k),
	     \varphi^{\sharp}(\breve{C}_{ijk})D_1\varphi^{\sharp}(y^i)D_2\varphi^{\sharp}(y^j)]
	  &&  \mbox{in $\;\NL^{(C_{(1)});\delta\nabla}_3(\varphi,\nabla)$}
  \end{eqnarray*}
  respectively.
 However, they vanish for $(\varphi, \nabla)$ admissible.
  Thus, the $3$-forms $C_{(3)}$ and $C_{(1)}\wedge B$
  have no consequence to the variation of $S_{\CSWZ}^{(C_{(1)},C_{(3)},B)}$
  with respect to $\nabla$.
 This is anticipated since there is no coupling term
   between $C_{(3)}$, $C_{(1)}\wedge B$ and $\nabla$
   in $S_{\CSWZ}^{(C_{(1)},C_{(3)},B)}$.
\end{itemize}

The contribution of the Chern-Simon/Wess-Zumino term
  $S_{\CSWZ}^{(C_{(1)},C_{(3)},B)}$
  to the equations of motion for a $D$-membrane follows immediately.
 
\bigskip

\begin{ssremark} $[$contribution only to first-order terms in EOM$\,]\;$ {\rm
As observed from these examples, for lower dimensional D-branes,
 the  Chern-Simons/Wess-Zumino term $S_{\CSWZ}^{(C,B)}$ in the action contributes
 an additional set of {\it first-order} nonlinear differential-expression terms to the system of equations
 of motion fo D-branes.
In particular,
 they preserve the signature of the original system
 from the Dirac-Born-Infeld term $S_{\DBI}^{(\Phi,g,B)}$ in the action.
}\end{ssremark}

\bigskip
\bigskip

The current notes D(13.1)
 lay down some foundation toward the dynamics of D-branes along the line of our D-project.
Solutions to the system of  equations of motion from the total action
  $S_{\DBI}^{(\Phi, g,B)}(\varphi,\nabla)+S_{\CSWZ}^{(C,B)}(\varphi,\nabla)$
  for a D-brane world-volume
  should be thought of as
   an Azumaya/matrix version of minimal submanifolds or harmonic maps,
     twisted/bent,
	   on  one hand,  by the (dynamical) gauge field $\nabla$ on the domain manifold $X$
	     with a (noncommutative) endomorphism/matrix function-ring
	   and,
      on the other hand, by the background field $(\Phi,g,B,C)$,
	  created by closed (super)strings, on the target space(-time) $Y$.
Further details, issues, and examples are the focus of the sequels.

 \vspace{10em}
\baselineskip 13pt
{\footnotesize

\vspace{1em}

\noindent
chienhao.liu@gmail.com, chienliu@math.harvard.edu; \\
yau@math.harvard.edu

}

\end{document}